\documentclass[phd,11pt,hidelinks]{psuthesis}

\pdfoutput=1

\usepackage{amsmath,epsfig}
\usepackage{amssymb}
\usepackage{dsfont}
\usepackage{stmaryrd}
\usepackage[mathscr]{eucal}
\usepackage{subfigure}
\usepackage{bm}
\usepackage{cite}
\usepackage{array}
\usepackage{multirow}
\usepackage{tikz}
\usetikzlibrary{arrows}
\usepackage{verbatim}
\usepackage{algorithm,algorithmic}
\usepackage{html,makeidx}
\usepackage{amsthm}

\newtheorem{lem}{Lemma}

\usepackage{graphicx}

\def\mb{\mathbf}

\def\bs{\boldsymbol}
\def\ds{\displaystyle}

\def\RCMLEL{$\text{RCML}_\text{EL}$}

\def\RCMLC{$\text{RCML}_\text{Chen}$ }
\def\CNCEL{$\text{CNCML}_\text{EL}$ }
\def\CNCML{$\text{CNCML}$ }
\def\CNCtrue{$\text{CNCML}_\text{true}$ }

\def\ben{\begin{equation*}}
\def\een{\end{equation*}}
\def\be{\begin{equation}}
\def\ee{\end{equation}}
\def\beaa{\begin{eqnarray*}}
\def\eeaa{\end{eqnarray*}}
\def\bea{\begin{eqnarray}}
\def\eea{\end{eqnarray}}
\def\biea{\begin{IEEEeqnarray}{rCl}}
\def\eiea{\end{IEEEeqnarray}}

\DeclareMathOperator{\rank}{rank}
\DeclareMathOperator{\sinc}{sinc}
\DeclareMathOperator{\lr}{LR}
\DeclareMathOperator{\tr}{tr}
\DeclareMathOperator{\diag}{diag}
\DeclareMathOperator{\Kmax}{K_{\max}}

\usepackage[Lenny]{fncychap}
\ChTitleVar{\Huge\sffamily\bfseries}

\title{Robust Covariance Matrix Estimation for Radar Space-Time Adaptive Processing (STAP)}

\author{Bosung Kang}
\dept{Electrical Engineering}
\degreedate{August 2015}
\copyrightyear{2015}

\documenttype{Dissertation}

\submittedto{The Graduate School}

\numberofreaders{4}

\honorsadviser{Honors P. Adviser}

\secondthesissupervisor{Second T. Supervisor}

\honorsdepthead{Department Q. Head}

\advisor[Dissertation Advisor, Chair of Committee]
        {Vishal Monga}
        {Associate Professor of Electrical Engineering}

\readerone[]
          {Constantino M. Lagoa}
          {Professor of Electrical Engineering}

\readertwo[]
          {David J. Miller}
          {Professor of Electrical Engineering}

\readerthree[]
            {Jesse L. Barlow}
            {Professor of Computer Science and Engineering}

\readerfour[Head of the Department of Electrical Engineering]
        {Kultegin Aydin}
        {Professor of Electrical Engineering}
\begin{document}
%%%%%%%%%%%%%%%%%%%%%%%%
% Preliminary Material %
%%%%%%%%%%%%%%%%%%%%%%%%
% This command is needed to properly set up the frontmatter.
\frontmatter

%%%%%%%%%%%%%%%%%%%%%%%%%%%%%%%%%%%%%%%%%%%%%%%%%%%%%%%%%%%%%%
% IMPORTANT
%
% The following commands allow you to include all the
% frontmatter in your thesis. If you don't need one or more of
% these items, you can comment it out. Most of these items are
% actually required by the Grad School -- see the Thesis Guide
% for details regarding what is and what is not required for
% your particular degree.
%%%%%%%%%%%%%%%%%%%%%%%%%%%%%%%%%%%%%%%%%%%%%%%%%%%%%%%%%%%%%%
% !!! DO NOT CHANGE THE SEQUENCE OF THESE ITEMS !!!
%%%%%%%%%%%%%%%%%%%%%%%%%%%%%%%%%%%%%%%%%%%%%%%%%%%%%%%%%%%%%%

% Generates the signature page. This is not bound with your
% thesis.
%\psusigpage

% Generates the title page based on info you have provided
% above.
\psutitlepage

% Generates the committee page -- this is bound with your
% thesis. If this is an baccalaureate honors thesis, then
% comment out this line.
\psucommitteepage

% Generates the abstract. The argument should point to the
% file containing your abstract.
%
\pagestyle{plain}
\chapter*{Abstract}
    \begin{singlespace}
        Estimating the disturbance or clutter covariance is a centrally important problem in radar space time adaptive processing (STAP) since estimation of the disturbance or interference covariance matrix plays a central role on radar target detection in the presence of clutter, noise and a jammer. The disturbance covariance matrix should be inferred from training sample observations in practice. Traditional maximum likelihood (ML) estimators are effective when homogeneous (target free) training data is abundant but lead to poor estimates, degraded false alarm rates, and detection loss in the regime of limited training. However, large number of homogeneous training samples are generally not available because of difficulty of guaranteeing target free disturbance observation, practical limitations imposed by the spatio-temporal nonstationarity, and system considerations. The problem has been exacerbated by recent advances that have led to more antenna elements ($J$) and higher temporal resolution ($P$) time epochs resulting in a large dimension ($N = JP)$.

In this dissertation, we look to address the aforementioned challenges by exploiting physically inspired constraints into ML estimation. While adding constraints is beneficial to achieve satisfactory performance in the practical regime of limited training, it leads to a challenging problem. Unlike unconstrained estimators, a vast majority of constrained radar STAP estimators are iterative and expensive numerically, which prohibits practical deployment. We focus on breaking this classical trade-off between computational tractability and desirable performance measures, particularly in training starved regimes. In particular, we exploit both the structure of the disturbance covariance and importantly the knowledge of the clutter rank to yield a new rank constrained maximum likelihood (RCML) estimator of clutter/disturbance covariance. We demonstrate that the rank-constrained estimation problem can in fact be cast in the framework of a tractable convex optimization problem, and derive closed form expressions for the estimated covariance matrix. In addition, we derive a new covariance estimator for STAP that jointly considers a Toeplitz structure and a rank constraint on the clutter component. Past work has shown that in the regime of low training, even handling each constraint individually is hard and techniques often resort to slow numerically based solutions. Our proposed solution leverages the rank constrained ML estimator (RCML) of structured covariances to build a computationally friendly approximation that involves a cascade of two closed form solutions. Performance analysis using the KASSPER data set (where ground truth covariance is made available) shows that the proposed RCML estimator vastly outperforms state-of-the art alternatives even for low training including the notoriously difficult regime of $K \leq N$ training regimes and for the experiments considering real-world scenarios such as target detection performance and the case that some of training samples are corrupted by target information.

Finally, we address the problem of working with inexact physical radar parameters under a practical radar environment. As shown in this dissertation, employing practical constraints such as a rank of the clutter subspace and a condition number of disturbance covariance leads to a practically powerful estimator as well as a closed form solution. While the rank and the condition number are very effective constraints, often practical non-ideality makes it difficult to be known precisely using physical models. We propose a robust covariance estimation method via an expected likelihood (EL) approach. We analyze covariance estimation algorithms under three different cases of imperfect constraints: 1) only rank constraint, 2) both rank and noise power constraint, and 3) condition number constraint. For each case, we formulate estimation of the constraint as an optimization problem with the expected likelihood criterion and formally derive and prove a significant analytical result such as uniqueness of the solution. Through experimental results from a simulation model and the KASSPER data set, we show the estimator with optimal constraints obtained by the EL approach outperforms alternatives in the sense of a normalized signal-to-interference and noise ratio (SINR). 
    \end{singlespace}
\newpage

% Generates the Table of Contents
\thesistableofcontents

% Generates the List of Figures
\thesislistoffigures

% Generates the List of Tables
\thesislistoftables

% Generates the List of Symbols. The argument should point to
% the file containing your List of Symbols.
%\thesislistofsymbols{SupplementaryMaterial/ListOfSymbols}

% Generates the Acknowledgments. The argument should point to
% the file containing your Acknowledgments.
%\thesisacknowledgments{SupplementaryMaterial/Acknowledgments}

% Generates the Epigraph/Dedication. The first argument should
% point to the file containing your Epigraph/Dedication and
% the second argument should be the title of this page.
%\thesisdedication{SupplementaryMaterial/Dedication}{Dedication}

%%%%%%%%%%%%%%%%%%%%%%%%%%%%%%%%%%%%%%%%%%%%%%%%%%%%%%
% This command is needed to get the main part of the %
% document going.                                    %
%%%%%%%%%%%%%%%%%%%%%%%%%%%%%%%%%%%%%%%%%%%%%%%%%%%%%%
\thesismainmatter

%%%%%%%%%%%%%%%%%%%%%%%%%%%%%%%%%%%%%%%%%%%%%%%%%%
% This is an AMS-LaTeX command to allow breaking %
% of displayed equations across pages. Note the  %
% closing the "}" just before the bibliography.  %
%%%%%%%%%%%%%%%%%%%%%%%%%%%%%%%%%%%%%%%%%%%%%%%%%%
\allowdisplaybreaks{
%
%%%%%%%%%%%%%%%%%%%%%%
% THE ACTUAL CONTENT %
%%%%%%%%%%%%%%%%%%%%%%
% Chapters
%SourceDoc ../YourName-Dissertation.tex
\chapter{Introduction}
\label{Ch:Introduction}

\section{Motivation}
\label{Sec:Motivation}

Space-time adaptive processing (STAP), joint adaptive processing in the spatial and temporal domains \cite{Monzingo80,Guerci03,Klemm02}, is the cornerstone of radar signal processing and creates the ability to suppress interfering signals while simultaneously preserving gain on the desire signal. However, for STAP to be successful, interference statistics, specifically interference covariance matrix, must be estimated from observed signals. In the absence of prior knowledge about the interference environment, a large number of target free disturbance training samples are required to obtain accurate estimates. A compelling challenge emerges because such a large number of homogeneous training samples are generally not available in practice \cite{Himed97}. Therefore, recent research in radar STAP has focused on overcoming this practical issue. Knowledge-based processing which uses a priori information about the radar environment is one of the most popular approaches to this problem. Figure \ref{Fig:PreviousWorks} shows previous works related to the knowledge-based processing.

\begin{figure}
\centering
\includegraphics[scale=0.6]{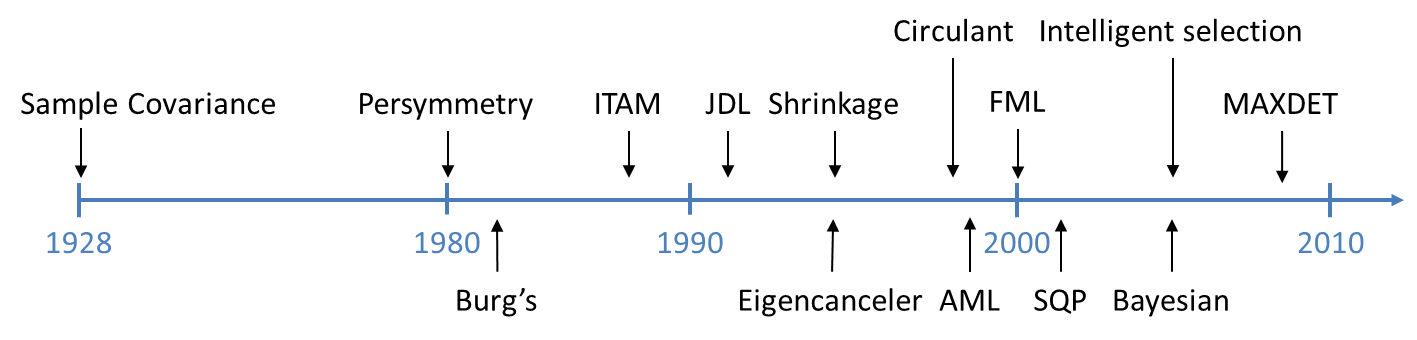}
\caption{Knowledge-base processing algorithms}
\label{Fig:PreviousWorks}
\end{figure}

A subset of these techniques includes intelligent training selection \cite{Guerci06} and spatio-temporal degree reduction \cite{Wicks06,Wang91,Gini08}. Enforcing structure on covariance matrices, which is the main focus of this dissertation, has also been pursued. Examples of structure include persymmetry \cite{Nitzberg80}, Toeplitz structure \cite{Burg82,Li99,AlHomidan02,Abramovich98,Wilkes88}, circulant structure \cite{Conte98}, and eigenstructure \cite{Hoffbeck96,Abramovich81,Carlson88,Haimovich96,Steiner00,Chen10,Aubry12}. The fast maximum likelihood (FML) method \cite{Steiner00} which enforces a special eigenstructure is shown to be the most competitive technique experimentally \cite{Rangaswamy04Sep,Gini08}. In particular, the disturbance covariance matrix represents the clutter matrix which is a low rank and positive semidefinite plus an identity matrix scaled by noise power which represents noise subspace. The FML method ensures that the estimated covariance matrix has eigenvalues all greater than the noise level. Under ideal conditions which means there are no mutual coupling between array elements and no internal clutter motion, the rank of the clutter matrix is easily calculated by the Brennan rule \cite{Ward94}. Even if there is mutual coupling in practice, it is well known that the rank of the clutter matrix is much less than the spatio-temporal degree (the dimension of the covariance matrix).

Since the covariance matrix from a stationary stochastic signal is Hermitian and Toeplitz, estimation of Toeplitz covariance benefits many applications such as array processing and time series analysis. Though various estimation and approximation techniques of Toeplitz covariance matrices have been proposed \cite{Miller87,Little02,Forster89,Jansson00}, it is well-known that there is no closed-form solution of the ML estimation of a Toeplitz covariance matrix. Therefore, many Toeplitz covariance estimation methods require large number of training samples for computational tractability \cite{Li99}, otherwise it involves iterations or numerical approaches \cite{AlHomidan02}. In particular, the iterated Toeplitz approximation method (ITAM) \cite{Wilkes88} and the iterative approach by Forster \emph{et al.} \cite{Forster89} deal with both the clutter rank and Toeplitz structure jointly. However, both approaches are based on a computationally expensive iterative procedure. The ITAM estimator has been shown to be effective under very low training, but does not show performance improvements for realistic training regime \cite{Kang15}.

Most covariance estimation techniques which exploit the knowledge of radar environment assume that the constraint used in the estimation problem is known or can be estimated from radar physics under ideal condition. For example, the Brennan rule says the rank is easily calculated. However, in practice the clutter rank departs from the Brennan rule prediction due to antenna errors and internal clutter motion. In this case, the rank is not known precisely and needs to be determined. Though determination of the number of signals is a classical eigenvalue problem, it is important to note that the problem does not have a simple and unique solution. A noise level and a condition number of a disturbance covariance matrix should be estimated as well if they are unknown or not precisely known in practice. The expected likelihood (EL) approach \cite{Abramovich07} proposed by Abramovich \emph{et al.} provides a novel criterion to determine a regularization parameter based on the statistical invariance property of the likelihood ratio (LR) values. The regularization parameters are selected so that the LR value of the estimate agree as closely as possible with the statistics of the LR value of the true covariance matrix which does not depend on the true covariance itself.

\section{Overview of Dissertation Contributions}
\label{Sec:Contributions}

In view of the aforementioned observations, we develop structured covariance estimation methods which exploit practical and powerful constraints, that is, the rank of the clutter matrix and Toeplitz structure. Furthermore, we also develop covariance estimation methods which automatically and adaptively determine the values of practical constraints via an expected likelihood approach. The main contributions of this dissertation are described in detail in the following sections.

\subsection{Estimation of Structured Covariance Matrices}
\label{Sec:Contribution1}

We first set up the optimization problem to estimate the disturbance covariance matrix with a structural constraint and the rank constraint. The estimation problem is unfortunately not convex, because neither the cost function nor the rank constraint are convex. However, we show that using a transformation of variables, reduction to a convex form is possible. Furthermore, by invoking Karush-Kuhn-Tucker (KKT) \cite{Boyd04} for the resulting convex problems, it is possible to derive a closed-form solution. Akin to the FML, we initially assume that the noise power is known. Then we extend our results to the case of unknown noise variance. In that case, we assume that only a lower bound on the noise power is available.

Our second contribution aims to break the classical trade-off between performance under low training regimes and computational complexity. We develop a computationally efficient approximation of structured Toeplitz covariance under a rank (EASTR) constraint. EASTR satisfies both the Toeplitz structure property (at least approximately) and the rank information of the clutter subspace at the same time. Decades of research has shown that enforcing even each constraint individually can be quite onerous. We propose to decouple the rank and Toeplitz constraints, which lends analytical tractability. Crucially, the EASTR solution does not need iterative steps, that is, involves a cascade of two steps in which a closed form solution is available in each step. First, a closed form solution using ML employing the rank constraint, is obtained from the RCML estimator. Next we propose a new method to perturb the eigenvalues of the RCML estimator in a rank preserving manner so as to impose the Toeplitz structure. The merits of the RCML estimator and EASTR are also verified experimentally over both simulated data and the knowledge-aided sensor signal processing expert reasoning (KASSPER) data set.

\subsection{Estimation of Constraints under Imperfect Knowledge}
\label{Sec:Contribution2}

We develop covariance estimation methods which automatically and adaptively determine the values of practical constraints via an expected likelihood approach for practical radar STAP. The proposed methods guide the selection of the constraints via the expected likelihood criteria in the case that the knowledge of the constraints is imperfectly known in practice. We consider three different cases of the constraints: 1) only the clutter rank constraint, 2) both the clutter rank and the noise power constraints, and 3) the condition number constraint. For each case, we develop significant analytical results. First we formally prove that the rank selection problem has a unique solution. Second, we derive a closed form solution of the optimal noise power for a given rank, which means an iterative or numerical method is not required to find the optimal noise power and enables fast implementation. Finally we also prove there exists the unique condition number for the condition number estimation method. Experimental investigation on a simulation model and on the KASSPER data set shows that the proposed methods for three different cases outperform alternatives such as the FML, leading rank selection methods in radar literature and statistics, and the ML estimation of the condition number constraint.

\section{Organization}
\label{Sec:Organization}

A snapshot of the main contributions of this dissertation is presented next. Publications related to the contribution in each chapter are also listed where applicable.

\textbf{Chapter \ref{Ch:Background}} discusses background of radar space time adaptive processing (STAP) and introduce previous works for knowledge-based processing in radar applications. First, significance and benefit of radar STAP which is main application of the proposed covariance estimation methods are presented. We also discuss the practical limitations of radar STAP and approaches to overcome the limitation. One of the approaches is to exploit knowledge of radar environment. In particular, many researchers have enforced structure on covariance and shrinkage estimation techniques have also been considered. A more thorough review of these techniques is provided. We introduce an expected likelihood approach which is useful to determine parameters in optimization problem. Our third contribution, the method of automatic selection of constraints, is based on this approach.

In \textbf{Chapter \ref{Ch:RCML}}, we propose a novel covariance estimation method which exploits the rank of the clutter subspace, rank constrained maximum likelihood (RCML) estimation. We formulate the RCML estimation problem, and subsequent derivation is provided. The solutions of the estimation and optimization problems are derived for the two cases of both known and unknown (known lower bound) noise levels. Experimental validation reports on the performance of the proposed estimator and compares it against well-known existing methods in terms of signal to interference and noise ratio. In addition, we discuss practical merits of the proposed method, such as probability of detection and whiteness test. The material related to this method was presented at the 2012 IEEE International Radar Conference \cite{Monga12} and the 2013 IEEE International Radar Conference \cite{Kang13} and was published in the IEEE Transactions on Aerospace and Electronic Systems \cite{Kang14}.

Our second contribution presented in \textbf{Chapter \ref{Ch:EASTR}} is towards breaking a classical trade-off between computational complexity and performance in low or realistic training regimes for estimation problem of Toeplitz disturbance covariance matrix. We develop a computationally efficient approximation of structured Toeplitz covariance under a rank constraint (EASTR). We consider two cases: 1) when the Toeplitz constraint is satisfied exactly, we obtain the exact Toeplitz estimate, satisfying the rank constraint and Toeplitz property and 2) when the Toeplitz constraint is not exactly satisfied, we make slight modification on the Toeplitz and obtain approximately the Toeplitz estimate. Experimental validation of the proposed method is also provided to compare the proposed method against widely used existing radar STAP covariance estimators. This material was presented at the 2013 IEEE Asilomar Conference on Signals, Systems and Computers \cite{Kang13Asilomar}, the 2013 IEEE 5th International Workshop on Computational Advances in Multi-Sensor Adaptive Processing (CAMSAP) \cite{Kang13CAMSAP}, and the 2014 IEEE International Radar Conference \cite{Kang14Radarcon} and was published in the IEEE Transactions on Aerospace and Electronic Systems \cite{Kang15} and in the IEEE Aerospace and Electronic Systems Magazine \cite{Kang15Magazine}.

\textbf{Chapter \ref{Ch:EL}} presents our third contribution considering a robust covariance estimation method under imperfectly known constraints. We propose a robust covariance estimation method via an expected likelihood approach and solutions of the constraint estimation problems for three different cases of the constraints. Experimental validation of our proposed methods is provided to report the performance of the proposed methods and compare them against existing parameter selection such as rank selection and maximum likelihood estimators on both a simulation model and the KASSPER data set. The material was presented at the 2015 IEEE International Radar Conference \cite{Kang15Radarcon} and has been submitted to the IEEE Transactions on Aerospace and Electronic Systems in May 2015 \cite{Kang16}.

% In \textbf{Chapter \ref{Ch:Conclusion}}, the main contributions of this dissertation are summarized. 
%SourceDoc ../YourName-Dissertation.tex
\chapter{Background}
\label{Ch:Background}

\section{Space-Time Adaptive Processing (STAP)}
\label{Sec:STAP}
Signal detection using adaptive processing in both spatial and temporal domains offers significant benefits in a variety of applications including radar, sonar, satellite communications, and seismic systems \cite{Monzingo80}. Specifically, the space-time adaptive processing (STAP) enables to suppress interference while preserving gain on desired signal. Consider the airborne array radar with $J$ elements. The radar transmits a pulse in a certain direction and the transmitted signal is reflected by various objects such as buildings, land, water, vegetation, and one or more targets of interest. On receive, the radar samples this reflected wave at a high rate, with each of the $K$ samples corresponding directly to reflections from a specific range. The received signal may include not only desired signals but also undesired interfering effects from extraneous objects (clutter) and electronic counter-measures (ECM), such as jamming. In addition, the background white noise caused by the radar receiver circuitry as well as man-made sources and machinery is also included in the received returns. This process is repeated for $P$ pulses transmitted. Therefore the entire received data can be $J \times P \times K$ data cube.

More precisely, the radar receiver front end consists of an array of $J$ antenna elements, which receives signals from targets, clutter, and jammers. These reflections induce a voltage at each element of the antenna array, which constitutes the measured array data at a given time instant. Snapshots of the measured data collected at $P$ successive time epochs give rise to the spatio-temporal nature of the received radar data. The spatio-temporal product $JP=N$ is defined to be the system dimensionality. Figure \ref{Fig:interference} uses the angle-Doppler space to illustrate the need for space-time adaptive processing (STAP). A target at a specific angle and traveling at a specific velocity (corresponding to a Doppler frequency) occupies a single point in this space. A jammer originates from a particular angle but is temporally (in Doppler domain) white and the clutter occupies a ridge in this 2-D space. Consequently the target signal is masked by white jammer in Doppler domain and the clutter in spatial domain and therefore, carrying out merely temporal (Doppler domain) or spatial (angle domain) processing fails to separate the target from the interference \cite{Ward94}. On the other hand, joint domain processing in angle and Doppler enables target detection as shown in Figure \ref{Fig:interference}.

\begin{figure}[!t]
\centering
\includegraphics[scale=0.6]{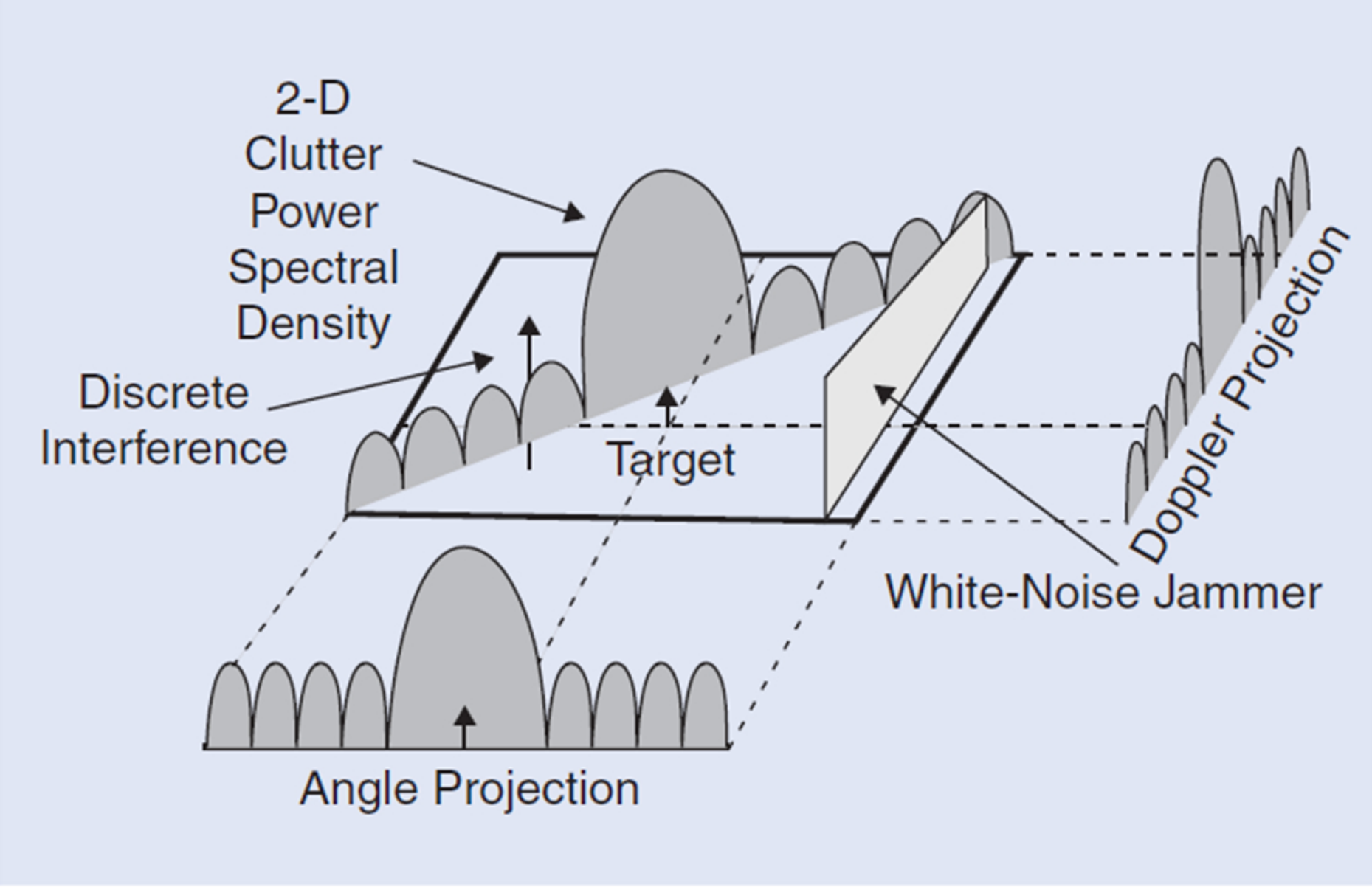}
\caption{The target and interference scenario in an airborne radar.}
\label{Fig:interference}
\end{figure}

The target detection problem can be cast in the framework of a statistical hypothesis test of the form
\bea
H_0 : \mb{x} & = & \mb{d} \; = \; \mb{c}+\mb{j}+\mb{n}\label{Eq:h0}\\
H_1 : \mb{x} & = & \alpha \mb{s}(\theta_t,f_t) + \mb{d} \; = \; \alpha \mb{s}(\theta_t,f_t) + \mb{c}+\mb{j}+\mb{n}\label{Eq:h1}
\eea
where $\mb{x} \in \mathds{C}^{JP \times 1}$ is a vector form of the received data under either hypothesis, $\mb{d}$ represents the overall disturbance which is the sum of $\mb{c}$, clutter, $\mb{j}$, jammers, and $\mb{n}$, the background white noise. The vector $\mb{s}$ is a known spatio-temporal steering vector that represents the signal returned from the target for a specific angle and Doppler and $\alpha$ is the unknown target complex amplitude. The steering vector $\mb{s}$ is defined in the case of an equally spaced linear array by
\bea
\mb{s} & = & \mb{s}_t \otimes \mb{s}_s \label{Eq:steeringvector1}\\
\mb{s}_t & = & \big[ 1 \; z_t \; z_t^2 \; \ldots \; z_t^{(N-1)} \big]^T, \label{Eq:steeringvector2}\\
\mb{s}_s & = & \big[ 1 \; z_s \; z_s^2 \; \ldots \; z_s^{(J-1)} \big]^T, \label{Eq:steeringvector3}\\
\mb{z}_s & = & e^{j2\pi f_s} = e^{(j2\pi \frac{d}{\lambda} \sin\phi_t)},\qquad z_t = e^{j2\pi f_t / f_R}, \label{Eq:steeringvector4}
\eea
where $\phi_t$ and $f_t$ are the angle and Doppler frequency respectively, and $\otimes$ denotes the Kronecker vector product, $f_R$ the pulse repetition frequency (PRF), and $\lambda$ the wavelength of operation. Finally, $\mb{s}_t$ and $\mb{s}_s$ represent the temporal and spatial steering vectors, respectively.

Now, a test statistic $\Lambda$ using a weight vector $\mb{w}$ is
\be
\Lambda = | \mb{w}^H\mb{x}|^2 \overset{H_1}{\underset{H_0}{\gtrless}} \lambda,
\ee
where $H$ denotes the Hermitian transpose and $\lambda$ a \emph{threshold} which is determined by the probability of false alarm. The whiten-and-match filter (MF) for detecting a rank-1 signal is the optimum processing method for Gaussian interference statistics. It is given by \cite{Robey92}
\be
\label{Eq:MF}
\mb{w} = \dfrac{\mb{R}_d^{-1}\mb{s}}{\sqrt{\mb{s}^H\mb{R}_d^{-1}\mb{s}}} \; \Rightarrow \Lambda_{MF} = \dfrac{|\mb{s}^H\mb{R}_d^{-1}\mb{x}|^2}{\mb{s}^H\mb{R}_d^{-1}\mb{s}} \overset{H_1}{\underset{H_0}{\gtrless}} \lambda_{MF}
\ee
where $\mb{R}_d$ is a known interference covariance matrix. Eq. (\ref{Eq:MF}) represents the matched filtering of the whitened data $\breve{\mb{x}} = \mb{R}_d^{-1/2}\mb{x}$ and whitened steering vector $\breve{\mb{s}} = \mb{R}_d^{-1/2}\mb{s}$. From Eq. (\ref{Eq:MF}), it turns out that the interference covariance matrix $\mb{R}_d$ is crucial in the detection statistic.

\section{Estimation of Covariance Matrices}
\label{Sec:CovarianceEstimation}

%\subsection{Estimation of Covariance Matrices}
%\subsubsection{Structured Covariance Estimation}
\subsection{Structured Covariance Estimation}
To overcome the practical issue of lack of generous training data, researchers have developed approaches using signal processing and statistical learning techniques, covariance matrix estimation techniques that enforce and exploit particular structure have been pursued. In this section, we discuss representative algorithms, the fast maximum likelihood (FML), shrinkage estimation methods, and eigencanceler.

Steiner and Gerlach proposed a maximum likelihood (ML) solution for a structured covariance matrix that has the form of the identity matrix plus an unknown positive semi-definite Hermitian (PSDH) matrix under the assumption that the input interference is Gaussian \cite{Steiner00}. In particular, the disturbance covariance matrix $\mb R$ represents the exhibits the following structure
\be
\mb R = \sigma^2 \mb I + \mb R_c
\ee
where $\mb R_c$ denotes the clutter matrix which has a low rank and is positive semi-definite and $\mb I$ is an identity matrix. This covariance matrix form is often valid in realistic interference scenarios for radar and communication systems.

The ML estimate for the covariance matrix is given by
\be
\hat{\mb R} = \arg \ds\min_{\mb R_c} [tr \{\mb{R}^{-1}\mb{S}\} + \log(|\mb{R}|)]
\ee
Consider an eigenvalue decomposition of the sample covariance matrix (SCM), $\mb S = \mb\Phi \bs\Lambda \mb\Phi^H$, where $\bs\Lambda$ is an $N \times N$ diagonal matrix with diagonal entries $\lambda_1 \geq \lambda_2 \geq \cdots \geq \lambda_N$ that are eigenvalues of $\mb S$ and $\mb\Phi$ is an $N \times N$ unitary eigenmatrix. Then the ML solution for $\mb R$ is given by
\be
\hat{\mb R} = \mb\Phi \bs\Lambda_0 \mb\Phi^H
\ee
where $\mb\Phi$ is the unitary eigenmatrix associated with the eigenvalue decomposition of the SCM,
\be
\bs\Lambda_0 \equiv \text{Diag} (\lambda_1,\lambda_2,\ldots,\lambda_M,\sigma^2,\sigma^2,\ldots,\sigma^2)
\ee
is an $N \times N$ diagonal matrix, $\lambda_1 \geq \lambda_2 \geq \ldots \geq \lambda_N$ are the eigenvalues of the SCM, and $M$ is the number of eigenvalues greater than $\sigma^2$.
Therefore, the FML technique ensures that the estimated covariance matrix has eigenvalues all greater than $\sigma^2$ by assuming that its value is known, which is sometimes unrealistic.

Shrinkage estimators are commonly used approach to the problem of covariance estimation for high dimensional data \cite{Cao08}. In general, shrinkage estimators shrink the sample covariance matrix to a target structure
\be
\hat{\mb R} = \alpha \mb D + (1-\alpha)\mb S
\ee
where $\mb D$ and $\mb S$ are a positive definite matrix and the sample covariance matrix, respectively.

There are several methods of shrinkage estimators according to how to define the positive definite matrix $\mb D$. For instance, $\mb D$ is considered as the identity matrix in some papers, that is,
\be
\hat{\mb R} = \alpha \mb I + (1-\alpha)\mb S
\ee
This approach is usually used in Ridge regression and Tikhonov regularization. Second, others refer $\mb D$ as an scaled identity matrix by the average of the eigenvalues,
\be
\hat{\mb R} = \alpha \Big(\frac{tr (\mb S)}{N}\Big)\mb I + (1-\alpha)\mb S
\ee
Finally, $\mb D$ can be set to the diagonal of $\mb S$,
\be
\hat{\mb R} = \alpha \text{diag}(\mb S) + (1-\alpha) \mb S
\ee
Here, the shrinkage intensity parameter $\alpha$ is chosen by the leave-one-out likelihood (LOOC), specifically, it is chosen so that the average likelihood of omitted samples is maximized as suggested in \cite{Hoffbeck96}.

The eigencanceler \cite{Haimovich96} provides simultaneous rejection of both clutter and directional interferences by adaptive processing in the spatial and Doppler domains. The eigencanceler uses eigendata to suppress clutter and directional interferences while minimizing noise contributions and maintaining specified beam pattern constraints. They first consider the space-time correlation matrix as the form of
\bea
\mb R & = & E\{\mb x \mb x^H\}\\
& = & \mb R_J + \mb R_c + \mb M_\nu
\eea
where $\mb R_J$, $\mb R_c$, $\mb R_\nu$ are the correlation matrices of the jammers, clutter, noise, respectively and $\mb x$ are the received signals. Each of the correlation matrices can be then written by
\be
\mb R_J = \sum_{\theta_i} \int_{B_\nu} S_{J,i} (\nu) \mb v(\theta_i,\nu) \mb v^H(\theta_i,\nu)d\nu
\ee
\be
\mb R_C = \int_\Theta \int_{B_\nu} S_C (\theta,\nu) \mb v(\theta,\nu) \mb v^H(\theta,\nu)d\nu d\theta
\ee
\be
\mb R_\nu = \sigma_\nu^2 \mb I
\ee
where $\mb v$ is the position vector on the angle $\theta_i$ and Doppler frequency $\nu$ and $S$ is the power spectral density.

The eigencanceler is based on the eigenanalysis which suggests a small number of eigenvalues contain all the information about interferences (jammers and clutter), and therefore, the span of the eigenvectors associated with these significant eigenvalues includes all the position vectors that comprise the interference signals. Under assumption that $r$ dominant eigenvalues and eigenvectors represent interferences, the covariance matrix can be expressed by
\be
\mb R = \sum_{i=1}^r p_i \mb v_i \mb v_i^H + \sigma^2 \mb I
\ee
where $p_i$ and $\mb v_i$ are the clutter power and the eigenvector corresponding to $r$ dominant eigenvalues, respectively.

Aubry \emph{et al.} proposed the method of a structured covariance matrix under a condition number upper-bound constraint \cite{Aubry12}. The initial non-convex optimization problem is
\be
\label{Eq:InitialProblemRCML}
\left\{ \begin{array}{cc}
\ds \max_{\mb R} & f(\mb Z) = \frac{1}{\pi^{NK}|\mb R|^K}\exp(-\tr\{\mb Z^H \mb R^{-1} \mb Z\})\\
s.t. & \mb{R} = \sigma^2 \mb{I} + \mb{R}_c\\
 & \frac{\lambda_{\max}(\mb R)}{\lambda_{\min} (\mb R)} \leq K_{\max}\\
 & \mb{R}_c \succeq \mb 0\\
 & \sigma^2 \geq c \end{array} \right.
\ee

The authors showed that the optimization problem falls within the class of MAXDET problems \cite{Vandenberghe98,DeMaio09} and developed an efficient procedure for its solution in closed form which is given by

\be
\label{Eq:CN1}
\mb{R}^\star = \mb{V}{\mb\Lambda^\star}^{-1}\mb{V}^H
\ee
where
\be
\mb\Lambda^\star = \diag\big(\bs\lambda^\star (\bar u)\big),
\ee
$\bs\lambda^\star (\bar u) = [\lambda_1^\star (\bar u),\ldots,\lambda_N^\star (\bar u)]$ with
\be
\lambda_i^\star (\bar u) = \min \bigg(\min (K_{\max} \bar u, 1), \max \Big(\bar u,\frac{1}{\bar d_i}\Big) \bigg),
\ee
$K_{\max}$ is a condition number constraint, and $\bar u$ is an optimal solution of the following optimization problem,
\be
\label{Eq:OptimizationU}
\left\{ \begin{array}{cc}
\ds\min_u & \sum_{i=1}^N G_i(u)\\
s.t. & 0 < u \leq 1 \end{array} \right.
\ee
where
\be
\label{Eq:Giu1}
G_i(u) = \left\{ \begin{array}{ll}
\log \Kmax - \log u + \Kmax \bar d_i u & \text{if} \quad 0 < u \leq \frac{1}{\Kmax}\\
\bar d_i & \text{if} \quad \frac{1}{\Kmax} \leq u \leq 1 \end{array} \right.
\ee
for $\bar d_i \leq 1$, and
\beaa
\label{Eq:Giu2}
\lefteqn{G_i(u)}\nonumber\\ & = & \left\{ \begin{array}{ll}
\log \Kmax - \log u + \Kmax \bar d_i u & \text{if} \quad 0 < u \leq \frac{1}{\Kmax \bar d_i}\\
\log \bar d_i + 1 & \text{if} \quad \frac{1}{\Kmax \bar d_i} < u \leq \frac{1}{\bar d_i}\\
\frac{1}{\bar d_i} + \bar d_i u & \text{if} \quad \frac{1}{\bar d_i} \leq u \leq 1 \end{array} \right.
\eeaa
for $\bar d_i > 1$. Similarly to the RCML estimator, the ML solution of the eigenvalue is a fucntion of $\bar d_i$'s and the condition number $K_{\max}$.

%\subsubsection{Toeplitz Covariance Estimation}
\subsection{Toeplitz Covariance Estimation}
Many Toeplitz covariance estimation methods have been published in the literature on structured covariance matrix estimation. An important technique is the maximum likelihood (ML) approach. However, it is well known that the ML estimation of a Hermitian Toeplitz covariance matrix has no closed-form solution \cite{Fuhrmann91}. Some previous works consider only the Toeplitz matrix constraint \cite{Li99,Burg82,Miller87}. Li \textit{et al.} proposed asymptotic maximum likelihood (AML) estimation and its closed form solution \cite{Li99}. However, they did not consider the rank constraint and the closed form solution is applied only when the number of sample data is large enough. On the other hand, Al-Homidan proposed an algorithm to find the nearest Toeplitz matrix with rank $r$ given a certain matrix by using LDL$^{\text{T}}$ decomposition \cite{AlHomidan02}. ITAM (Iterated Toeplitz Approximation Method) \cite{Wilkes88} was proposed to find the approximate rank $r$ Toeplitz matrix under structural constraint but it assumes the signal is harmonic with complex exponentials which we do not assume necessarily. We briefly discuss Toeplitz covariance estimation algorithms in this section.

Burg \emph{et al.} proposed a method for estimating a covariance matrix of specified structure from vector samples of the random process \cite{Burg82}. It appears to be the first to study the ML method in its full generality. They assumed that the random process is zero-mean multivariate Gaussian and found the maximum-likelihood covariance matrix that has the specified structure. They used the fact that the gradient must be orthogonal to variations in a feasible set. To derive the necessary conditions, they define the variation of $\mb{R}$ to be
\be
\partial\mb{R} =
\left(
  \begin{array}{cccc}
    \partial \mb{R}(1,1) & \partial \mb{R}(1,2) & \cdots & \partial \mb{R}(1,N) \\
    \partial \mb{R}(2,1) & \partial \mb{R}(2,2) & \cdots & \partial \mb{R}(2,N) \\
    \vdots & \vdots & \ddots & \vdots \\
    \partial \mb{R}(N,1) & \partial \mb{R}(2,N) & \cdots & \partial \mb{R}(N,N) \\
  \end{array}
\right)
\ee
Then, the variation of the determinant of $\mb{R}$ in terms of the variation of $\mb{R}$ will be
\be
\label{Eq:varofdetofR}
\partial|\mb{R}| = |\mb{R}|\text{tr}(\mb{R}^{-1}\partial\mb{R})
\ee
From Eq. (\ref{Eq:varofdetofR}),
\be
\label{Eq:varoflogofdetofR}
\partial \log |\mb{R}| = \text{tr}(\mb{R}^{-1}\partial\mb{R}).
\ee
We have another useful matrix theorem. From $\mb{R}\mb{R}^{-1} = \mb{I}$,
\be
\label{Eq:varofinvofR}
\partial \mb{R}^{-1} = -\mb{R}^{-1}\partial\mb{R}\mb{R}^{-1}
\ee
Now the variation of the cost function is derived by
\beaa
\partial \log (|\mb R|) - \partial \tr (\mb R^{-1} \mb S) & = & -\tr(\mb R^{-1} \partial \mb R) - \tr [\partial (\mb R^{-1}) \mb S]\\
& = & -\tr(\mb R^{-1}\partial \mb R - \mb R^{-1}\partial \mb R \mb R^{-1} \mb S)\\
& = & -\tr (\mb R^{-1}\mb S\mb R^{-1} \partial \mb R - \mb R^{-1}\partial\mb R)
\eeaa
The expression is thus neatly written as
\be
\label{Eq:varofcostfunction}
\text{tr}[(\mb{R}^{-1} - \mb{R}^{-1}\mb{S}\mb{R}^{-1})\partial\mb{R}]
\ee
The condition for minimization is that the variation of the cost function is zero for any feasible variation of $\mb{R}$. Thus the equation to be solved is
\be
\text{tr}[(\mb{R}^{-1} - \mb{R}^{-1}\mb{S}\mb{R}^{-1})\partial\mb{R}] = 0
\ee
Then they used the inverse iteration algorithm to find the estimate of $\mb{R}$. However, we have to find the basis in the feasible set which is the set of all of Hermitian Toeplitz matrices with rank $r$ in order to perform the iteration.

In the case of a signal comprised only of complex exponentials, the true $N \times N$ covariance matrix $\mb{R}_S$ will have the additional property that its rank is equal to $r$, the number of complex exponentials in the signal. If white noise having a variance of $\sigma^2$ is added to this signal, the covariance matrix becomes
\be
\mb{R} = \mb{R}_S + \mb{R}_N
\ee
where $\mb{R}_N = \sigma^2 \mb{I}$. ITAM is a heuristically appealing technique, that uses repeated application of the singular value decomposition (SVD) to produce a rank $r$ Hermitian Toeplitz estimate of $\mb{R}_S$.\\
For the harmonic retrieval problem, the observed signal is assumed to have the form of a sum of $r$ complex exponentials in white noise
\be
x(n) = \ds \sum_{i=1}^r A_i e^{jw_in} + w(n)
\ee
The noise $w(n)$ is assumed to be zero mean and have a variance of $\sigma^2$. The goal of Toeplitz approximation is to find the best model $\hat{\mb{R}}$ for $\mb{R}$ where the model is constrained to have the form
\be
\hat{\mb{R}} = \mb{R}_S + \sigma^2\mb{I}
\ee
where
\bea
\mb{R}_S & = & \ds\sum_{i=1}^r P_i \mb{q}_i \mb{q}_i^H\\
\mb{q} & = & [1 \; e^{j\hat{w}_i} \; e^{j2\hat{w}_i} \; \cdots \; e^{j(N-1)\hat{w}_i}]^T
\eea
First estimate $\hat{\mb{R}}_S$ from the observed covariance matrix $\mb{R}$.
\bea
\mb{R} & = & \mb{U}\mb{\Sigma}\mb{U}^H\\
\mb{U} & = & [\mb{U}_1 \; \mb{U}_2]\\
\mb{\Sigma} & = &
\left[
\begin{array}{cc}
\mb{\Sigma}_1 & \mb{0}\\
\mb{0} & \mb{\Sigma}_2
\end{array} \right]
\eea
\be
\hat{\mb{R}}_S = \mb{U}_1 (\mb{\Sigma}_1 - \lambda_{av}\mb{I})\mb{U}_1^H
\ee
where
\be
\lambda_{av} = \ds\dfrac{1}{M-r} \ds\sum_{i=r+1}^N \lambda_i
\ee
Now, since $\hat{\mb{R}}_S$ is not in general a Toeplitz matrix, the next step is to find the Toeplitz matrix, $\mb T$, that is closest to $\hat{\mb{R}}_S$. This is done by averaging along the diagonals of $\hat{\mb{R}}_S$ and using these average values to define the corresponding diagonals of the new Toeplitz matrix., $\hat{\mb T}$. The problem now is that $\hat{\mb T}$ does not have a rank of $p$, so the next iteration begins by finding the rank $p$ approximation to $\hat{\mb T}$. This establishes the ITAM algorithm \cite{Wilkes88} which is illustrated in Figure \ref{Fig:ITAM}.

\begin{figure}[ht]
\centering
\includegraphics[width=2.5in]{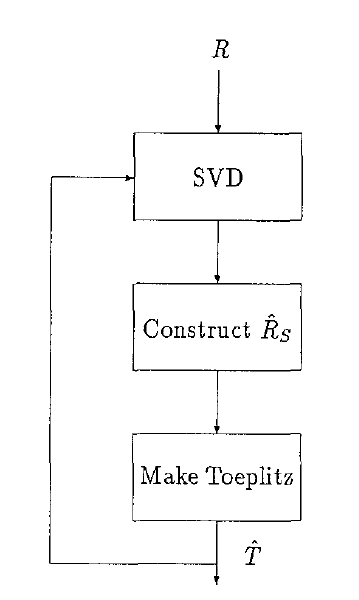}
\caption{The procedure of the ITAM algorithm.}
\label{Fig:ITAM}
\end{figure}

Li \emph{et al.} proposed a computationally efficient method that provides asymptotic (for large samples) maximum likelihood (AML) estimation for structured covariance matrices \cite{Li99}. They provided a closed-form formula for estimating Hermitian Toeplitz covariance matrices by invoking the extended invariance principle (EXIP) \cite{Stoica89} for parameter estimation. They denote Hermitian Toeplitz covariance matrix $\mb{R}(\bs{\phi})$ that is a known function of an unknown parameter vector $\bs{\phi} \in \mathbb{R}^{(2N-1) \times 1}$. Usually, a Hermitian Toeplitz matrix is parameterized such that $\bs{\phi}$ consists of the real and imaginary parts of the first column (or row) of $\mb{R}$.\\
The exact ML estimate $\hat{\bs{\phi}}$ of $\bs{\phi}$ is obtained by maximizing the likelihood function, which is equivalent to
\be
\hat{\bs{\phi}} = \ds\arg\min_{\bs{\phi} \in \mathbb{R}^{(2N-1)\times 1}} L_\phi (\bs{\phi}) = \ln |\mb{R}(\bs{\phi})| + tr \big[ \mb{R}^{-1}(\bs{\phi})\mb{S} \big]
\ee
Let $\mb{r} =  \text{vec} (\mb{R}) \in \mathbb{C}^{N^2 \times 1}$ and $\bs{\gamma} \in D_\gamma \subset \mathbb{R}^{N^2 \times 1}$ denote the vector made from the real and imaginary parts of the elements of $\mb{R}$ both above and on the main diagonal. Evidently, there is an $N^2 \times N^2$ matrix $\mb{F}$ such that
\be
\bs{\gamma} = \mb{F}\mb{r}
\ee
Note that $\bs\gamma$ represents a reparameterization of the covariance matrix $\mb{R}$ originally parameterized by $\bs\phi$. Hence there exists a mapping from $\bs\phi$ to $\bs\gamma$, i.e.,
\be
\bs\gamma = f(\bs\phi) \in D_\gamma, \qquad \forall \bs\phi \in D_\phi
\ee
such that
\be
L_\gamma ( f(\bs\phi)) = L_\phi(\bs\phi), \qquad \forall \bs\phi \in D_\phi
\ee
Let $\tilde{\bs\gamma} = \mb{F}\text{vec}(\mb{S})$. Then
\be
\label{Eq:approxprob1}
\tilde{\bs\phi} = \ds\arg\min_{\bs\phi \in D_\phi} \big[ \tilde{\bs\gamma} - f(\bs\phi)\big]^T\mb{\Gamma}^{-1}\big[ \tilde{\bs\gamma} - f(\bs\phi)\big]
\ee
is an asymptotically (in $K$) valid approximation of the ML estimate $\hat{\bs\phi}$, where $\mb\Gamma$ is a consistent (in $K$) estimate of the covariance matrix of $\tilde{\bs\gamma}$.\\
Let $\tilde{\mb{r}} = \text{vec}(\tilde{\mb{R}})$ where $\tilde{\mb{R}}$ is a sample covariance matrix. Then $\tilde{\bs{\gamma}} = \mb{F} \tilde{\mb{r}}$. Since $\mb\Gamma = \mb{F}\mb{C}\mb{F}^H$ where $\mb{C} = \text{cov}(\tilde{\mb{r}})$, Eq. (\ref{Eq:approxprob1}) is equivalent to
\be
\tilde{\bs\phi} = \ds\arg\min_{\bs\phi \in D_\phi} \big[ \tilde{\mb{r}} - \mb{r}(\bs\phi)]^H\mb{C}^{-1}\big[ \tilde{\mb{r}} - \mb{r}(\bs\phi)\big]
\ee
Finally, they obtain the final formula of AML for estimating a Hermitian Toeplitz covariance matrix.
\be
\hat{\bs\phi} = (\mb\Psi^H \tilde{\mb C}^{-1} \mb\Psi)^{-1}(\mb\Psi^H \tilde{\mb C}^{-1} \tilde{\mb r})
\ee

Al-Homidan proposed an algorithm that finds the nearest symmetric positive semi-definite Toeplitz matrix $\mb{T}$ to given a matrix $\mb{F}$ \cite{AlHomidan02}. The problem is formulated as a nonlinear minimization problem with positive semi-definite Toeplitz matrix as constraints. An algorithm with rapid convergence is obtained by $l_1$ Sequential Quadratic Programming (SQP) method.\\
Consider the following problem: Given a data matrix $\mb{F} \in \mathbb{R}^{N \times N}$, find the nearest symmetric positive semi-definite Toeplitz matrix $\mb{T}$ to $\mb{F}$ and $\text{rank}(\mb{T})=r$, that is,
\bea
\min & \phi = ||\mb{F} - \mb{T}||\\
\text{subject to} & \mb{T} \in T, \quad \rank(\mb{T}) = r\nonumber
\eea
Using partial $LDL^T$ factorization of $\mb{T}$, the partial factors $\mb{T} = \mb{L}\mb{D}\mb{L}^T$ can be calculated such that
\be
\mb{L} = \left[ \begin{array}{cc}
                  \mb{L}_{11} & \mb{0}\\
                  \mb{L}_{21} & \mb{I}_{N-r} \end{array} \right], \;
\mb{D} = \left[ \begin{array}{cc}
                  \mb{D}_{1} & \mb{0}\\
                  \mb{0} & \mb{D}_{2} \end{array} \right], \;
\mb{T} = \left[ \begin{array}{cc}
                  \mb{T}_{11} & \mb{T}_{21}^T\\
                  \mb{T}_{21} & \mb{T}_{22} \end{array} \right]
\ee
where $\mb{L}_{11}$, $\mb{D}_1$ and $\mb{T}_{11}$ are $r \times r$ matrices; $\mb{D}_2$ and $\mb{T}_{22}$ are $(n-r) \times (n-r)$ matrices; $\mb{L}_{21}$ and $\mb{T}_{21}$ are $(n-r) \times r$ matrices; $\mb{D}_1$ is diagonal and positive definite and $\mb{D}_2 = \mb{0}$. Since
\be
\mb{D}_2(\mb{T}) = \mb{T}_{22} - \mb{T}_{21}\mb{T}_{11}^{-1}\mb{T}_{21}^T
\ee
and if the structure of the matrix $\mb{T}$ is in a Toeplitz form, i.e.,
\be
\mb{T} = \left[ \begin{array}{cccc}
x_1 & x_2 & \cdots & x_N\\
x_2 & x_1 & \cdots & x_{N-1}\\
\vdots & \vdots & \ddots & \vdots\\
x_N & x_{N-1} & \cdots & x_1
\end{array} \right]
\ee
the constraint $\mb{T} \in T$ to be written in the form
\be
\mb{D}_2(\mb{T}(\mb{x})) = \mb{0}
\ee
Hence, the optimization problem can be expressed as
\bea
\min & \phi\\
\text{subject to} & \mb{D}_2(\mb{T}) = \mb{0} = \mb{Z}^T\mb{T}\mb{Z} \nonumber
\eea
where
\ben
\mb{Z} = \left[ \begin{array}{c}
                  -\mb{T}_{11}^{-1}\mb{T}_{21}^T\\
                  \mb{I} \end{array} \right]
\een
Now the cost function is
\beaa
\phi & = & \ds\sum_{i,j=1}^{N} (f_{ij} - t_{ij})^2\\
     & = & \ds\sum_{i,j=1}^{N} (f_{ij} - x_{|i-j+1|})^2
\eeaa
and the constraints $\mb{D}_2(\mb{T}) = \mb{0}$ can be written in the following form:
\be
d_{ij}(\mb{x}) = x_{|i-j+1|} - \ds\sum_{k,l=1}^{r}x_{i-k+1}[\mb{T}_{11}^{-1}]_{kl}x_{j-l+1} = 0
\ee
where $i,j=r+1,\cdots,N$ and $[\mb{T}_{11}^{-1}]_{kl}$ denotes the element of $\mb{T}_{11}^{-1}$ in $kl$-position.
Thus
\bea
\min & \phi = \ds\sum_{i,j=1}^{N} (f_ij - x_{|i-j+1|})^2\\
\text{subject to} & d_{ij}(\mb{x}) = 0\nonumber
\eea
the SQP method applied to above problem requires the solution of the QP subproblem
\bea
\ds\min_{\bs{\delta}} & \phi + \nabla\phi^T\bs{\delta} + \dfrac{1}{2}\bs{\delta}^TW\bs{\delta} \qquad \bs{\delta} \in \mathbb{R}^r\\
\text{subject to} & d_{ij} + \nabla d_{ij}^T \bs{\delta} = 0 \qquad i,j = r+1,\ldots,N \nonumber
\eea
giving a correction vector $\bs{\delta}^{(k)}$, so that $\mb{x}^{(k+1)} = \mb{x}^{(k)} + \bs{\delta}^{(k)}$.

\section{Expected Likelihood Approach}
\label{Sec:EL_Background}

Abramovich \emph{et al.} \cite{Abramovich07} proposed the expected likelihood estimation criterion which inherently justifies the appropriate selection of parameters such as loading factor based on direct likelihood matching. Expected likelihood approach is motivated by invariance properties of the likelihood ratio (LR) value which is given by
\bea
\lr(\mb R, \mb Z) & \equiv & \Big[\dfrac{f(\mb Z | \mb R)}{f(\mb Z | \mb S)}\Big]^{1/K}\\
& = & \frac{|\mb R^{-1} \mb S| \exp N}{\exp[\tr(\mb R^{-1} \mb S)]}
\eea
under a Gaussian assumption on the observations, $\mb z_i$'s. Furthermore, the unconstrained ML solution $\mb S$ has the LR value of 1. That is,
\be
\ds\max_{\mb R} \lr (\mb R, \mb Z) = \lr (\mb S, \mb Z) = 1
\ee

However, as shown in \cite{Abramovich07} the LR values of the true covariance matrix $\mb R_0$ are much lower than that of the ML solution $\mb S$. Therefore, it seems natural to replace the ML estimate by one that generates LR values consistent with what is expected for the true covariance matrix. More importantly, Abramovich \emph{et al.} showed \cite{Abramovich07} that the pdf of the LR for the true covariance matrix, which is given by
\bea
\lr (\mb R_0, \mb Z) & = & \frac{|\mb R_0^{-1} \mb S| \exp N}{\exp[\tr(\mb R_0^{-1} \mb S)]}\\
& = & \frac{|\mb R_0^{-1/2} \mb S \mb R_0^{-1/2}| \exp N}{\exp[\tr(\mb R_0^{-1/2} \mb S \mb R_0^{-1/2})]}
\eea
does not depend on the true covariance itself since
\be
\hat{\mb C} \equiv N \mb R_0^{-1/2} \mb S \mb R_0^{-1/2} \sim \mathcal{CW}(K,N,\mb I)
\ee
where $\mathcal{CW}$ represents complex Wishart distribution which is determined entirely by $K$ and $N$ and does not need $\mb R_0$. Therefore, the pdf of LR values for the true covariance matrix can be precalculated for given $K$ and $N$ and indeed the moments of distribution of the LR values were derived by Abramovich \emph{et al.} in their paper \cite{Abramovich07}.

Based on the invariance of the pdf of LR values, the EL approach can be used to determine values of parameters in estimation problems. For instance, the EL estimator for a diagonally loaded SMI technique under homogeneous interference training conditions and fluctuating target with known power is given by \cite{Abramovich07}
\be
\hat{\mb R}_{\text{LSMI}} = \hat\beta \mb I + \mb S
\ee
where
\be
\label{Eq:OptimalBeta}
\hat\beta \equiv \ds\arg_\beta \Bigg\{ \dfrac{|(\beta \mb I + \mb S)^{-1} \mb S| \exp N}{\exp \big( \tr[(\beta \mb I + \mb S)^{-1} \mb S] \big)} \equiv \lr_0 \Bigg\}
\ee
and $\lr_0$ is the reference median statistic, which can be precalculated from the pdf of the LR values
\be
\int_0^{\lr_0} f\big[\lr(\mb R_0, \mb Z)\big] d\lr = 0.5
\ee
where $f\big[\lr(\mb R_0, \mb Z)\big]$ is the invariant pdf of the LR values. 
\chapter{Rank Constrained Maximum Likelihood Estimation}
\label{Ch:RCML}

\section{Introduction}
\label{Sec:RCML_Intro}

%\IEEEPARstart{T}{his} demo file is intended to serve as a ``starter file''
%for IEEE journal papers produced under \LaTeX\ using
%IEEEtran.cls version 1.7 and later.
% Rewording.

Space time adaptive processing (STAP), i.e. joint adaptive processing in the spatial and temporal domains \cite{Monzingo04,Guerci03,Klemm02} is the cornerstone of radar signal processing and creates the ability to suppress interfering signals while simultaneously preserving gain on the desired signal. For STAP to be successful though, interference statistics must be estimated from training samples. Training, therefore plays a pivotal role in adaptive radar systems.

%Signal detection using adaptive processing in spatial and temporal domains offers significant benefits in a variety of application domains including radar, sonar, satellite communications and seismic systems \cite{Monzingo04}. This paper concerns itself with signal processing for radar systems using multiple antenna elements that coherently process multiple pulses. An adaptive array of spatially distributed sensors, which process multiple temporal snapshots, overcomes the directivity and resolution limits of a single sensor.

The radar STAP filter rejects clutter and noise by using an adaptive multidimensional finite impulse response (FIR) filter structure which consists of $P$ time taps and $J$ spatial taps.  The optimal STAP complex weight vector which maximizes the return signal from a desired target requires formation and inversion of the disturbance covariance matrix. Further, this covariance matrix is central to test statistics involved in target detection, such as the normalized matched filter (NMF) \cite{Rangaswamy04Sep} and the generalized likelihood ratio test (GLRT) \cite{Robey92}.

%The time and spatial taps support Doppler filtering and angle processing, respectively. Totally, $JP$ taps (or spatio-temporal degrees of freedom - DOF) make up the multidimensional adaptive FIR filter. The spatial DOF corresponds to the number of individual antenna array elements ($J$) and the temporal DOF corresponds to the number of sampled pulses ($P$).

%% needs editing to include SP magazine and related references
The disturbance covariance matrix must in practice be estimated. In the absence of any prior knowledge about the interference environment, a large number of homogeneous (target free) disturbance training samples are required to obtain accurate estimates. The quality and quantity of training data are governed by the scale with which the disturbance statistics change with respect to space and time as well as by systems considerations such as bandwidth. A compelling challenge for radar STAP emerges since the availability of generous homogeneous training data is often unrealistic \cite{Himed97}. This problem is exacerbated because the estimation process must be repeated for each Doppler and range bin of interest. Therefore, much recent research in radar STAP has focused on circumventing the lack of generous homogeneous training. One approach to this problem includes the use of {\emph {a priori}} information about the radar environment and is widely referred to in the literature as knowledge-based processing \cite{Capraro06,Haykin06,Guerci06,Wicks06,Miranda06,DeMaio09,DeMaio10}.
A subset of these techniques concern themselves with intelligent training selection \cite{Guerci06}. Other methods try to reduce the spatio-temporal DOF to reduce both the number of required training samples and computational cost \cite{Wang91,Wicks06,Gini08}. Enforcing structure on covariance matrices (such as Toeplitz) and shrinkage estimation techniques have also been considered \cite{Li99,Chen10}. A more thorough review of these techniques is provided in Section \ref{motivation}.

%% okay, rewording and slightly more detail - make messages compelling
Our work introduces a new advance in robust covariance estimation for STAP, known as the rank-constrained maximum likelihood (RCML) estimator of structured covariance matrices. The disturbance covariance matrix exhibits a structure that comprises the sum of noise (white) and clutter covariances with the clutter component being positive semi-definite and rank-deficient. What is of particular interest is that in airborne radar scenarios involving land clutter, this rank can in fact be determined using the Brennan rule \cite{Ward94} under nominal conditions. Incorporation of this rank information in our research is a logical evolution of the pioneering work by Steiner and Gerlach \cite{Steiner00} wherein they demonstrated that the fast maximum likelihood estimator (FML) which is cognizant of the eigenstructure of the disturbance covariance can in fact outperform competing approaches in the literature.

\noindent \textbf{Contributions:} From an analytic viewpoint, the contributions of this work are in formulating and solving the rank-constrained ML estimation problem for structured covariance matrices. We note that the value of the rank has been identified in statistics \cite{Anderson63} and in radar signal processing via eigencancelation approaches \cite{Haimovich96}, \cite{Guerci03}. Our work overlaps with these techniques in that we use eigenvectors identical to those of the sample covariance matrix for the estimated covariance matrix whose optimal choice can formally be shown to agree with those obtained from an SVD of the data matrix of training data. We demonstrate additionally that despite the presence of the challenging (non-convex) rank-constraint, our estimation problem can in fact be reduced to a convex optimization problem over the eigenvalues and further, closed form expressions are derived for the estimator. Our central analytical result is the RCML estimator, which like FML and eigencancelation approaches above exploits the knowledge of the radar noise floor. The noise power is typically determined by placing the radar in ``receive only" mode prior to going active \cite{Skolnik90}. For mathematical completeness, we also derive another estimator called RCML$_{\text{LB}}$ for the case when the noise floor is assumed unknown and only a lower bound (LB) is available. We note our RCML$_{\text{LB}}$ estimator is in fact a generalization of the result reported by Wax and Kailath \cite{Wax85}, who quote Anderson \cite{Anderson63} for the result.

Our experimental contributions include an extensive performance analysis of the RCML method and performance comparison with a host of competing methods with demonstrated success. Our experimental data was obtained from the DARPA KASSPER data set where ground truth covariance is available which helps in evaluation via well known figures of merit such as the normalized signal to interference and noise ratio (SINR). The merits of RCML are most pronounced in the low training regime, particularly the case of $K < N$ training samples which is a stiff practical challenge for radar STAP. Additionally, the proposed RCML estimator is robust to perturbations against the true knowledge of the rank making it even more appealing from a practical standpoint.

The rest of the chapter is organized as follows. Section \ref{motivation} reviews existing covariance estimation literature and further motivates our contribution. Section \ref{contribution} formulates the RCML estimation problem and subsequent derivation is provided in Section \ref{sec:MLestimation}. The solutions of the estimation/optimization problem are derived for the two cases of both known and unknown (known lower bound) noise levels. Experimental validation of our methods is provided in Section \ref{results} wherein we report the performance of the proposed estimator and compare it against well-known existing methods in terms of SINR. In addition, rank sensitivity of the proposed RCML estimator is also evaluated. Finally, concluding remarks with directions for future work are presented in Section \ref{conclusion}.

\section{Motivation and Review}
\label{motivation}

Because the covariance matrix plays a crucial role in the detection statistic (see Eq. (\ref{Eq:MF})), it is very important to estimate it reliably. Widrow et al. and Applebaum proposed least-squares method \cite{Widrow67} and maximum signal-to-noise-ratio criterion \cite{Applebaum66}, respectively, using feedback loops. However, these methods were slow to converge to the steady-state solution. Reed, Mallet, and Brennan \cite{Reed74} verified that the sample matrix inverse (SMI) method demonstrated considerably better convergence. In the sample matrix inverse method, the disturbance covariance matrix can be estimated using $K$ data ranges for training
\be
\label{Eq:SCM}
\hat{\mb{R}}_d = \dfrac{1}{K}\ds\sum_{k=1}^K \mb{x}_k \mb{x}_k^H = \dfrac{1}{K}\mb{X}\mb{X}^H
\ee
where $K$ is the number of training data we received, $\mb{x}_k \in \mathds{C}^{N}, N = JP$ is the $k$th vector of training data, and $\mb{X} = [ \mb{x}_1 \; \mb{x}_2 \; \ldots \; \mb{x}_K ] \in \mathds{C}^{N \times K}$. It is well known that the sample covariance is the {\emph{unconstrained}} maximum likelihood estimator when $K \geq N$. Despite this virtue, there remain fundamental problems with the SMI approach. First, typically $K >N$ training samples are needed to guarantee the non-singularity of the estimated covariance matrix. In fact, when $K < N$ the estimate is singular and therefore precludes implementation of the STAP processor. As much past research has shown \cite{Gini08}, the estimate also does quite poorly in the vicinity of $K =N$ training samples.

As observed in Section \ref{Sec:RCML_Intro}, large number of homogeneous training samples are generally not available \cite{Himed97}. To overcome the practical issue of limited training, covariance matrix estimation techniques that enforce and exploit particular structure have been pursued. Examples of structure include persymmetry \cite{Nitzberg80}, the Toeplitz property \cite{Fuhrmann91,Li99,Abramovich98}, circulant structure \cite{Conte98}, multichannel autoregressive models \cite{Roman00,Wang09} and physical constraints \cite{Kraay07}. The FML method \cite{Steiner00} which enforces special eigenstructure also falls in this category and in fact is shown to be the most competitive technique experimentally \cite{Rangaswamy04Sep,Gini08}. In particular, the disturbance covariance matrix $\mb{R}$ represents the exhibits the following structure
\be
\label{Eq:structuralconstraint}
\mb{R} = \sigma^2 \mb{I} + \mb{R}_c
\ee
where $\mb{R}_c$ denotes the clutter matrix which has a low rank and is positive semi-definite and $\mb{I}$ is an identity matrix. Steiner and Gerlach's FML technique ensures that the estimated covariance matrix has eigenvalues all greater than $\sigma^2$ . Recently, the work by Aubry \emph{et al.} \cite{Aubry12} has also improved upon FML by the introduction of a condition number constraint. Other approaches include Bayesian covariance matrix estimators \cite{DeMaio06,DeMaio07,Besson07,Guerci06,Wang10} and the use of knowledge-based covariance models \cite{Gurram06,Melvin06,DeMaio09,DeMaio10}. Finally, shrinkage estimation methods have been also considered \cite{Chen10,Ledoit03,Stoica08,Won09}.

\section{Rank Constrained ML Estimation of Structured Covariance Matrices}

\label{main}

\subsection{Overview of contribution}
\label{contribution}

The principal contribution of our work us to incorporate the rank of the clutter covariance matrix, $\mb{R}_c$, explicitly into ML estimation of the disturbance covariance matrix\footnote{Preliminary version of the work has appeared at the 2012 IEEE Radar conference \cite{Monga12}.}. Under ideal conditions (no mutual coupling between array elements and no internal clutter motion), Brennan rule \cite{Ward94} states that the rank of $\mb{R}_c$ in the airborne linear phased array radar problem is given by
\be
\label{Eq:Brennanrule}
\text{rank}(\mb{R}_c) = J + \gamma(P-1)
\ee
where $\gamma = 2 v_p T/d$ is the slope of the clutter ridge, with $v_p$ denoting the platform velocity, $T$ denoting the pulse repetition interval, and $d$ denoting the inter-element spacing. Even if there is mutual coupling in practice, $\mb{R}_c$ has rank $r$ which is much less than the spatio-temporal product $N = JP$ in many practical airborne radar applications. In addition, powerful techniques have been developed \cite{Rangaswamy04Sep} to determine the  rank fairly accurately.

We first set up the optimization problem to estimate the disturbance covariance matrix with a structural constraint on $\mb{R}$ and the rank constraint on $\mb{R}_c$. The estimation problem when seen as an optimization over $\mb{R}$ is unfortunately not a convex problem, since neither the cost function nor the constraints (rank) are convex (elaborated upon in Section \ref{sec:MLestimation}). We will however show that using a transformation of variables, reduction to a convex form is possible and further by invoking KKT conditions \cite{Boyd04} for the resulting convex problems, it is in fact possible to derive a closed form solution. Akin to FML, we initially assume (Section \ref{Sec:Known}) that the noise power $\sigma^2$ is known while setting up and solving the problem. Then we extend our results (Section \ref{Sec:Unknown}) to the case of unknown noise variance. In that case, we assume that only a lower bound on the noise power is available. That is, we know $\hat{c}$ where the covariance matrix is expressed by
\be
\label{Eq:lowerboundconstraint}
\mb{R} = c\mb{I} + \mb{R}_c
\ee
where $c$ is the noise variance (to be estimated/optimized) such that $c > \hat{c}$.

\subsection{ML Estimation}
\label{sec:MLestimation}

Let $\mb{z}_i \in \mathds{C}^N$ be the $i$th realization of the target-free (stochastic) disturbance vector and $K$ be the number of training samples. That is, $i=1,2,\ldots,K$ and $N=JP$. Therefore, under each training sample, $\mb{z}_i$, under assumption of zero mean, obeys
\be
f(\mb{z}_i) = \dfrac{1}{\pi^N |\mb{R}|} \exp(-\mb{z}_i^H \mb{R}^{-1}\mb{z}_i)
\ee
which comes from a zero-mean complex circular Gaussian distribution and $\mb{R}$ is the $N \times N$ disturbance covariance matrix. Further, $|\mb{R}|$ denotes the determinant of $\mb{R}$, $\mb{z}_i^H$ is the Hermitian (conjugate transpose) of $\mb{z}_i$. Let $\mb{Z}$ be the $N \times K$ complex matrix whose $i$-th column is the observed vector $\mb{z}_i$. Since each observations $\mb{z}_i$ are i.i.d, the likelihood of observing $\mb{Z}$ given $\mb{R}$ is given by
\bea
f(\mb R; \mb{Z}) & = & \frac{1}{\pi^{NK}}|\mb{R}|^{-K}\exp\big(-tr \{\mb{Z}^H\mb{R}^{-1}\mb{Z}\}\big)\\
& = & \frac{1}{\pi^{NK}}|\mb{R}|^{-K}\exp\big(-tr \{\mb{R}^{-1}\mb{Z}\mb{Z}^H\}\big)\\
& = & \frac{1}{\pi^{NK}}|\mb{R}|^{-K}\exp\big(-K \cdot tr \{\mb{R}^{-1}\mb{S}\}\big)
\eea
where $\mb{S} = \frac{1}{K}\mb{Z}\mb{Z}^H$ is the well-known sample covariance matrix. Our goal is to find the positive definite matrix $\mb{R}$ that maximizes the likelihood function $f_{(\mb{R})}(\mb{Z})$. The logarithm of the likelihood term is
\be
\log f(\mb R; \mb{Z}) = -K \cdot tr \{\mb{R}^{-1}\mb{S}\} - K \log(|\mb{R}|) - NK \log(\pi).
\ee
Maximizing the log-likelihood as a function of $\mb{R}$ is equivalent to minimizing the function given by
\be
\label{Eq:costfunctionR}
\tr \{\mb{R}^{-1}\mb{S}\} + \log(|\mb{R}|).
\ee
Therefore, the optimization problem for obtaining the maximum likelihood solution under the rank constraint is given by
\be
\left\{ \begin{array}{cc}
\ds \min_{\mb{R}} & \tr \{\mb{R}^{-1}\mb{S}\} + \log(|\mb{R}|)\\
s.t. & \mb{R} = \sigma^2 \mb{I} + \mb{R}_c\\
 & \rank(\mb{R}_c) = r\\
 & \mb{R}_c \succeq \mb{0} \end{array} \right.
\ee
 Since the cost function is not a convex function in $\mb{R}$, we apply a transformation variables i.e., let $\mb{X} = \sigma^2 \mb{R}^{-1}$ and $\bar{\mb{S}} = \dfrac{1}{\sigma^2}\mb{S}$. Then, the revised cost function in the optimization variable $\mb{X}$ becomes
\bea
\lefteqn{tr \{\mb{R}^{-1}\mb{S}\} + \log(|\mb{R}|)}\nonumber\\
 & = & tr \{\bar{\mb{S}}\mb{X}\} - \log(|\dfrac{1}{\sigma^2}\mb{X}|)\\
 & = & tr \{\bar{\mb{S}}\mb{X}\} - \log(|\mb{X}|) + \log \sigma^{2N}. \label{Eq:costfunctionX}
\eea
Since $\log \sigma^{2N}$ in Eq. (\ref{Eq:costfunctionX}) is a constant, the final cost function to be minimized is
\be
\label{Eq:finalcostfunction}
tr \{\bar{\mb{S}}\mb{X}\} - \log(|\mb{X}|).
\ee
Note that $tr \{\bar{\mb{S}}\mb{X}\} = \ds \sum_{i=1}^N \ds \sum_{j=1}^N \bar s_{ji}x_{ij}$ is affine and $\log (| \mb{X} |)$ is concave, which implies $- \log (| \mb{X} |)$ is convex. Therefore, the final cost function, Eq. (\ref{Eq:finalcostfunction}) is convex in the variable $\mb{X}$.

We now express $\mb{X}$ and $\bar{\mb{S}}$ in terms of their eigenvalue decomposition, i.e., $\mb{X} = \mb{\Phi}\mb{\Lambda}\mb{\Phi}^H$ and let $\bar{\mb{S}}$ be decomposed as $\bar{\mb{S}} = \mb{V}\mb{D}\mb{V}^H$ where $\mb{\Phi}$ and $\mb{V}$ are orthonormal eigenvector matrices of $\mb{X}$ and $\bar{\mb{S}}$, $\mb\Lambda$ and $\mb{D}$ are diagonal matrices with diagonal entries which are eigenvalues of $\mb{X}$ arranged in ascending order and $\bar{\mb{S}}$ in descending order, respectively. Using the eigendecompositions, the cost function can be simplified as
\bea
\lefteqn{ \tr \{\bar{\mb{S}}\mb{X}\} - \log (| \mb{X} |)}\nonumber\\
& = & \tr \{\mb{V}\mb{D}\mb{V}^H\mb{\Phi}\mb{\Lambda}\mb{\Phi}^H\} - \log (| \mb{\Phi}\mb{\Lambda}\mb{\Phi}^H |)\\
& = & \tr \{\mb{D}\mb{V}^H\mb{\Phi}\mb{\Lambda}\mb{\Phi}^H\mb{V}\} - \log (| \mb{\Lambda} |)% \\
% & \geq & tr \{\mb{D}\mb{\Lambda}\} - log (| \mb{\Lambda} |) \label{Eq:inequality}\\
% & = & \mb{d}^T \bs{\lambda} - \mb{1}^T\log \bs{\lambda} \label{Eq:vectorcostfunction}
\eea
Therefore, the optimization problem using the eigenvalues and the eigenvectors is given by
\be
\left\{ \begin{array}{cc}
\ds \min_{\mb{\Lambda},\mb\Phi} & \tr \{\mb{D}\mb{V}^H\mb{\Phi}\mb{\Lambda}\mb{\Phi}^H\mb{V}\} - \log (| \mb{\Lambda} |)\\
s.t. & \mb\Lambda \mb 1 \preceq \sigma^2 \mb 1\\
 & \mb U \mb\Lambda \mb 1 \preceq \mb 0\\
 & \mb\Lambda \mb e = \mb e\\
 & \mb\Phi^H \mb\Phi = \mb I \end{array} \right.
\ee
where $\mb 1$ and $\mb 0$ represent $N \times 1$ vectors with constant elements of 1 and 0, respectively, $\mb e = [0,\ldots,0_r,1,\ldots,1]^T$, and a matrix $\mb U$ is given in Eq. \eqref{Eq:U}. Since the first three constraints are only for the eigenvalues of $\mb X$, that is, $\mb\Lambda$, we can obtain the optimal eigenvector matrix $\mb\Phi$ by solving the following optimization problem
\be
\label{Eq:OptimizationPhi}
\left\{ \begin{array}{cc}
\ds \min_{\mb\Phi} & \tr \{\mb{D}\mb{V}^H\mb{\Phi}\mb{\Lambda}\mb{\Phi}^H\mb{V}\} - \log (| \mb{\Lambda} |)\\
s.t. & \mb\Phi^H \mb\Phi = \mb I \end{array} \right.
\ee
and the solution of the optimization problem \eqref{Eq:OptimizationPhi} is in fact fairly well known from standard unconstrained ML estimation. That is, over the space of unitary matrices the optimal $\mb{\Phi}$ that maximizes the likelihood is the one that matches with the eigenvector matrix of the sample covariance $\mb V$ \cite{Cao08}. The cost function can be more simplified by substituting $\mb\Phi$ with $\mb V$,
\bea
\tr \{\mb{D}\mb{V}^H\mb{\Phi}\mb{\Lambda}\mb{\Phi}^H\mb{V}\} - \log (| \mb{\Lambda} |) & = & \tr \{\mb{D}\mb{\Lambda}\} - log (| \mb{\Lambda} |) \label{Eq:inequality}\\
& = & \mb{d}^T \bs{\lambda} - \mb{1}^T\log \bs{\lambda} \label{Eq:vectorcostfunction}
\eea
where $\mb{d}$ and $\bs{\lambda}$ are vectors with entries of eigenvalues of $\bar{\mb{S}}$ and $\mb{X}$ respectively and $\log \bs{\lambda} = [\log \lambda_1, \log \lambda_2, \cdots, \log \lambda_N]^T$. Hence, the optimization may be focused on the vector of eigenvalues $\bs{\lambda}$. Many previous algorithms such as FML \cite{Steiner00}, eigencanceler \cite{Haimovich96}, and other eigen-based techniques \cite{Guerci03} have utilized the same eigenvectors without eigenvalue optimization.

\subsubsection{Known Noise Level Case}
\label{Sec:Known}

We first assume the noise power is known. Then the constraints of the optimization problem are
\be
\left\{ \begin{array}{c}
\mb{R} = \sigma^2 \mb{I} + \mb{R}_c\\
\rank(\mb{R}_c) = r\\
\mb{R}_c \succeq \mb{0}\\
\mb{R} \succeq \sigma^2 \mb{I} \end{array} \right. .
\ee
Since $rank(\mb{R}_c) = r$, $\mb{R}_c$ has $r$ non-negative eigenvalues and the rest eigenvalues are all zero. Hence, from Eq. (\ref{Eq:structuralconstraint}), $\mb{R}$ has $r$ eigenvalues which are greater than or equal to $\sigma^2$ and the rest eigenvalues equal to $\sigma^2$. That is, the eigenvalue matrix of $\mb{R}$ has structure of the form
\be
\label{Eq:formR}
\left(
  \begin{array}{ccccccc}
    \bar{\lambda}_1 & 0               & \cdots & 0               & 0        & \cdots & 0\\
    0               & \bar{\lambda}_2 & \cdots & 0               & 0        & \cdots & 0\\
    0               & 0               & \ddots & 0               & 0        & \cdots & 0\\
    0               & 0               & 0      & \bar{\lambda}_r & 0        & \cdots & 0\\
    0               & 0               & \cdots & 0               & \sigma^2 & \cdots & 0\\
    0               & 0               & \cdots & 0               & 0        & \ddots & 0\\
    0               & 0               & \cdots & 0               & 0        & \cdots & \sigma^2\\
  \end{array}
\right)
\ee
where $\bar{\lambda}_1 \geq \bar{\lambda}_2 \geq \cdots \geq \bar{\lambda}_r \geq \sigma^2$. Hence, $\mb\Lambda$, the eigenvalue matrix of $\mb{X} = \sigma^2\mb{R}^{-1}$, has the following structure
\be
\mb\Lambda = \left(
  \begin{array}{ccccccc}
    \lambda_1 & 0         & \cdots & 0         & 0 & \cdots & 0\\
    0         & \lambda_2 & \cdots & 0         & 0 & \cdots & 0\\
    0         & 0         & \ddots & 0         & 0 & \cdots & 0\\
    0         & 0         & 0      & \lambda_r & 0 & \cdots & 0\\
    0         & 0         & \cdots & 0         & 1 & \cdots & 0\\
    0         & 0         & \cdots & 0         & 0 & \ddots & 0\\
    0         & 0         & \cdots & 0         & 0 & \cdots & 1\\
  \end{array}
\right)
\ee
where $\lambda_1 \leq \lambda_2 \leq \cdots \leq \lambda_r \leq 1$. Now the constraints can be expressed in vector and matrix forms. The first constraint is $\lambda_1 \leq \lambda_2 \leq \cdots \leq \lambda_N$, that is, $\mb{U}\bs{\lambda} \preceq \mb{0}$ where\\
\be
\label{Eq:U}
\mb{U} =
\left(
  \begin{array}{ccccc}
    1 & -1 & 0 & \cdots & 0 \\
    0 & 1 & -1 & \cdots & \vdots \\
    0 & 0 & \ddots & \ddots & 0\\
    0 & 0 & \cdots & 1 & -1 \\
    0 & 0 & \cdots & 0 & -1 \\
  \end{array}
\right) \in \mathbb{R}^{N\times N}
%s.t. & \lambda_i = 1 \qquad \text{for}\; (N-r) \; \lambda_i\\
% & 0 < \lambda_i < 1 \qquad \text{for the rest}\; (r) \; \lambda_i \end{array} \right.
\ee
The second constraint is $0 < \lambda_i \leq 1$ which can be expressed by
\be
\bs{\varepsilon} \preceq \bs{\lambda} \preceq \mb{1}
\ee
where $\bs{\varepsilon}$ is a vector with all entries equal to the same constant $\epsilon$ such that $\epsilon$ is picked close to zero. Note, that
this is done to avoid solutions where any of the $\lambda_{i}$  exactly equals zero since that would lead to a singular $\mb{X}$.
The final constraint is $\lambda_{r+1} = \lambda_{r+2} = \cdots = \lambda_N = 1$, and is expressed as
\be
\label{Eq:Equalityconstraint}
\mb{E} \bs{\lambda} = \mb{h}
\ee
where
\be
\mb{E} =
\left(
  \begin{array}{cc}
    \mb{0}_{r\times r} & \mb{0}_{r \times (N-r)}\\
    \mb{0}_{(N-r)\times r} & \mb{I}_{N-r}\\
  \end{array}
\right) \in \mathds{R}^{N\times N}
\ee
and $\mb{h} = [0, 0, \cdots, 0_r, 1, 1, \cdots, 1]^T$.\\
We therefore have the following optimization problem.
\be
\left\{ \begin{array}{cc}
\ds \min_{\bs{\lambda}} & \mb{d}^T \bs{\lambda} - \mb{1}^T\log \bs{\lambda}\\
s.t. & \mb{U}\bs{\lambda} \preceq \mb{0}\\
 & -\bs{\lambda} \preceq -\bs{\varepsilon}\\
 & \bs{\lambda} \preceq \mb{1}\\
 & \mb{E} \bs{\lambda} = \mb{h} \end{array} \right.
\ee

\emph{Remark}: Note that the constraint $\bs{\lambda} \preceq \mb{1}$ in fact forces the maximum rank of $\mb{R_c}$ to be $r$. That is, we use a relaxation of the strict inequality which forces an exact rank constraint. As will be shown shortly, this relaxation allows us to obtain a closed form solution while the optimal solutions to the two problems (original vs. relaxed) are hardly distinguishable \cite{Kang12}.

Combining the inequality constraints into one, we have
\be
\label{Eq:optimizationproblem1}
\left\{ \begin{array}{cc}
\ds \min_{\bs{\lambda}} & \mb{d}^T \bs{\lambda} - \mb{1}^T\log \bs{\lambda}\\
s.t. & \mb{F}\bs{\lambda} \preceq \mb{g}\\
 & \mb{E} \bs{\lambda} = \mb{h} \end{array} \right.
\ee
where $\mb{F} = \left[
  \begin{array}{c}
    \mb{U}\\
    -\mb{I}\\
    \mb{I}\\
  \end{array}\right]$,
  $\mb{g} = \left[
  \begin{array}{c}
    \mb{0}\\
    -\bs{\varepsilon}\\
    \mb{1}\\
  \end{array}\right]$,
  $\mb{E} =
\left[
  \begin{array}{cc}
    \mb{0}_{r\times r} & \mb{0}_{r \times (N-r)}\\
    \mb{0}_{(N-r)\times r} & \mb{I}_{N-r}\\
  \end{array}
\right]$, and $\mb{h} = [0, 0, \cdots, 0_r, 1, 1, \cdots, 1]^T$. Here, $\bs{\lambda}$, $\mb{d}$, $\mb{h} \in \mathds{R}^N$, $\mb{g} \in \mathds{R}^{3N}$, $\mb{U}$, $\mb{E} \in \mathds{R}^{N \times N}$, and $\mb{F} \in \mathds{R}^{3N \times N}$. The optimization problem (\ref{Eq:optimizationproblem1}) is obviously a convex optimization problem because the cost function is a convex function and feasible constraint sets are convex as well.
%The \emph{Lagrangian} $L$ : $\mathds{R}^N \times \mathds{R}^{3N} \times \mathds{R}^N \mapsto \mathds{R}$ associated with the problem is:
%\be
%L(\bs{\lambda}, \bs{\mu}, \bs{\nu}) = \mb{d}^T \bs{\lambda} - \mb{1}^T \log \bs{\lambda} + \bs{\mu}^T(\mb{F}\bs{\lambda}-\mb{g}) + \bs{\nu}^T(\mb{E}\bs{\lambda}-\mb{h})
%\ee
%where $\bs{\mu} \in \mathds{R}^{3N}$ and $\bs{\nu} \in \mathds{R}^{N}$. The KKT conditions for any $\bs\lambda$ to be a minimizer are:

%\textbf{Primal inequality constraints:}
%\be
%\mb{F}\bs{\lambda} \preceq \mb{g}
%\ee

%\textbf{Primal equality constraints:}
%\be
%\mb{E} \bs{\lambda} = \mb{h}
%\label{eqconst}
%\ee

%\textbf{Dual inequality constraints:}
%\be
%\bs{\mu} \succeq \mb{0}
%\ee

%\textbf{Complementary slackness:}
%\beaa
%\mu_1(\lambda_1 - \lambda_2) & = & 0\\
%& \vdots & \\
%\mu_{N-1}(\lambda_{N-1} - \lambda_N) & = & 0\\
%\mu_{N}\lambda_{N} & = & 0\\
%\mu_{N+1}(\lambda_1 - \varepsilon) & = & 0\\
%& \vdots & \\
%\mu_{2N}(\lambda_N - \varepsilon) & = & 0\\
%\mu_{2N+1}(\lambda_1 - 1) & = & 0\\
%& \vdots &\\
%\mu_{3N}(\lambda_N - 1) & = & 0
%\eeaa

%\textbf{Stationary}:
%\be
%\nabla_{\bs\lambda}L(\bs{\lambda}, \bs{\mu}, \bs{\nu}) = \mb{a} - \frac{1}{\bs\lambda} + \mb{F}^T\bs{\mu} + \mb{E}^T\bs{\nu} = \mb{0}
%\label{stationary}
%\ee
%where $\dfrac{1}{\bs\lambda} = [\dfrac{1}{\lambda_1}, \dfrac{1}{\lambda_2}, \cdots, \dfrac{1}{\lambda_N}]^T$.

\begin{lem}
\label{Lemma3.1}
A closed form solution of (\ref{Eq:optimizationproblem1}) can be obtained by KKT conditions and is given by
\be
\lambda_i^\star = \left\{ \begin{array}{cc}
\min (1,\dfrac{1}{d_i}) & \text{for} \; i=1,2,\ldots,r\\
1 & \text{for} \; i=r+1,r+2,\ldots,N \end{array} \right.
\label{eqn:RCML}
\ee
\end{lem}
\begin{proof}
Before solving the problem using the KKT conditions, we can simplify the problem further. Note that the constraint (\ref{Eq:Equalityconstraint}) tells us we already have $\lambda_{r+1} = \lambda_{r+2} = \cdots = \lambda_{N} = 1$. This means we do not need to optimize $\lambda_{r+1}, \lambda_{r+2}, \ldots = \lambda_{N}$. Therefore, we can arrive at an $r$ variable-optimization problem given by
\be
\label{Eq:simplifiedoptimizationproblem}
\left\{ \begin{array}{cc}
\ds \min_{\bs{\lambda}} & \mb{d}^T \bs{\lambda} - \mb{1}^T\log \bs{\lambda}\\
s.t. & \mb{F}\bs{\lambda} \preceq \mb{g} \end{array} \right.
\ee
where $\mb{F} = \left[\begin{array}{ccc} \mb{U}^T & -\mb{I} & \mb{I}\end{array} \right]^T$, $\mb{g} = \left[\begin{array}{ccc} \mb{0}^T & -\bs{\varepsilon}^T & \mb{1}^T \end{array}\right]^T$. Here, $\bs{\lambda}$, $\mb{d} \in \mathds{R}^r$, $\mb{g} \in \mathds{R}^{3r}$, $\mb{U} \in \mathds{R}^{r \times r}$, and $\mb{F} \in \mathds{R}^{3r \times r}$. We can rewrite the Lagrangian and KKT conditions to solve the problem.\\

The \emph{Lagrangian} $L$ : $\mathds{R}^r \times \mathds{R}^{3r} \times \mathds{R}^r \mapsto \mathds{R}$ associated with the problem is:
\be
L(\bs{\lambda}, \bs{\mu}) = \mb{d}^T \bs{\lambda} - \mb{1}^T \log \bs{\lambda} + \bs{\mu}^T(\mb{F}\bs{\lambda}-\mb{g})
\ee
where $\bs{\mu} \in \mathds{R}^{3r}$. The KKT conditions for any $\bs\lambda$ to be a minimizer are:

\textbf{Primal inequality constraints:}
\be
\label{Eq:simplifiedprimal}
\mb{F}\bs{\lambda} \preceq \mb{g}
\ee

\textbf{Dual inequality constraints:}
\be
\label{Eq:simplifieddual}
\bs{\mu} \succeq \mb{0}, \; \bs{\nu} \succeq \mb{0}, \; \bs{\upsilon} \succeq \mb{0}
\ee

\textbf{Complementary slackness:}
\label{Eq:simplifiedslackness}
\beaa
\mu_1(\lambda_1 - \lambda_2) & = & 0\\
& \vdots & \\
\mu_{r-1}(\lambda_{r-1} - \lambda_r) & = & 0\\
\mu_{r}\lambda_{r} & = & 0\\
\nu_{1}(\lambda_1 - \varepsilon) & = & 0\\
& \vdots & \\
\nu_{r}(\lambda_r - \varepsilon) & = & 0\\
\upsilon_{1}(\lambda_1 - 1) & = & 0\\
& \vdots &\\
\upsilon_{r}(\lambda_r - 1) & = & 0
\eeaa

\textbf{Stationarity}:
\be
\label{Eq:simplifiedstationary}
\nabla_{\bs\lambda}L(\bs{\lambda}, \bs{\mu}) = \mb{d} - \frac{1}{\bs\lambda} + \mb{F}^T\bs{\mu} = \mb{0}
\ee
where $\dfrac{1}{\bs\lambda} = [\dfrac{1}{\lambda_1}, \dfrac{1}{\lambda_2}, \cdots, \dfrac{1}{\lambda_r}]^T$. To make the notation clear, we let $\bs{\mu} = [\mu_1, \mu_2, \ldots, \mu_r]^T$, $\bs{\nu} = [\mu_{r+1}, \mu_{r+2}, \ldots, \mu_{2r}]^T$, and $\bs{\upsilon} = [\mu_{2r+1}, \mu_{2r+2}, \ldots, \mu_{3r}]^T$. Then the primal inequality constraints mean
\be
\left\{\begin{array}{cc}
\lambda_1 \leq \lambda_2 \leq \cdots \leq \lambda_r & \\
\varepsilon \leq \lambda_i \leq 1 & \qquad \text{for} \; i=1,2,\ldots,r \end{array} \right. .
\ee
First, from complementary slackness, we get $\bs\nu = \mb{0}$ since we let $\varepsilon$ be a very small number which is enough close to zero such that any $\lambda_i$ cannot equal to $\varepsilon$. In addition, the slackness means
\be
\left\{\begin{array}{cc}
\mu_i = 0 & \qquad \text{if} \; \lambda_i \neq \lambda_{i+1}\\
\upsilon_i = 0 & \qquad \text{if} \; \lambda_i \neq 1 \end{array} \right. .
\ee
Next we analyze the stationarity condition
\beaa
d_1 - \dfrac{1}{\lambda_1} + \mu_1 - \upsilon_1 = 0\\
d_2 - \dfrac{1}{\lambda_2} - \mu_1 + \mu_2 - \upsilon_2 = 0\\
\vdots\\
d_r - \dfrac{1}{\lambda_r} - \mu_{r-1} - \mu_r - \upsilon_r = 0.\\
\eeaa
We first assume that $\lambda_1 \neq 1$ which leads to $\upsilon_1 = 0$. We now have two cases which are $\lambda_1 \neq \lambda_2$ and $\lambda_1 = \lambda_2$. In the former case, we also know $\mu_1 = 0$. Therefore,
\be
d_1 - \dfrac{1}{\lambda_1} = 0.
\ee
This leads to $\lambda_1 = \dfrac{1}{d_1}$ and by the primal inequality constraint, $\lambda_1 \leq 1$, finally $\lambda_1 = \min (\dfrac{1}{d_1}, 1)$. In the latter case, we can also consider two cases here which are $\lambda_1 = \lambda_2 \neq \lambda_3$ and $\lambda_1 = \lambda_2 = \lambda_3$. In the first case, $\mu_2 = \upsilon_1 = \upsilon_2 = 0$. Therefore, the first two equations in the stationarity condition become
\be
\left\{\begin{array}{c}
d_1 - \dfrac{1}{\lambda_1} + \mu_1 = 0\\
d_2 - \dfrac{1}{\lambda_1} - \mu_1 = 0 \end{array} \right. .
\ee
From these two equations, we get $\mu_1 = (d_2 - d_1)/2$. Because $d_1 \geq d_2$, this implies $\mu_1 \leq 0$. However, $\mu_1 \geq 0$ from the dual inequality constraints. Therefore, $\mu_1 = 0$ and $\lambda_1 = \min (\dfrac{1}{d_1}, 1)$. Next, the case of $\lambda_1 = \lambda_2 = \lambda_3$ can be split into two which are $\lambda_1 = \lambda_2 = \lambda_3 \neq \lambda_4$ and $\lambda_1 = \lambda_2 = \lambda_3 = \lambda_4$. By following an approach similar to solving simultaneous equations, we conclude that $\bs{\mu} = \mb{0}$ and therefore $\lambda_i = \min (\dfrac{1}{d_i}, 1)$.
\end{proof}

From Lemma \ref{Lemma3.1}, the optimal solution $\mb{X}^\star$ is
\be
\mb{X}^\star = \mb{V}\mb\Lambda^\star\mb{V}^H
\ee
and the optimal covariance matrix estimator $\mb{R}^\star$ is
\be
\label{Eq:RCMLsolution2}
\mb{R}^\star = \sigma^2 {\mb{X}^\star}^{-1} = \sigma^2 \mb{V}{\mb\Lambda^\star}^{-1}\mb{V}^H
\ee
where $\mb{V}$ is the eigenvector matrix of the sample covariance matrix $\mb{S}$ and $\mb\Lambda^\star$ is a diagonal matrix with diagonal entries $\lambda_i^\star$. It should be noted that this is a generalization of the FML solution in \cite{Steiner00} with the rank-information additionally incorporated.
\subsubsection{Unknown Noise Level Case}
\label{Sec:Unknown}

For mathematical completeness, we also derive the case when $\sigma^{2}$ is also estimated as a part of the estimation process. We however assume that a lower bound (LB) is known and call this estimator $\text{RCML}_{\text{LB}}$.

We proceed by defining $\mb{X} = \mb{R}^{-1}$ (instead of $\mb{X} = \sigma^2 \mb{R}^{-1}$ as in Section \ref{Sec:Known}) because we assume that the exact noise power is unknown. Using simplifications similar to the ones in Section \ref{Sec:Known}, we have the cost function to be minimized as
\be
\label{Eq:costfunctionX2}
tr \{\mb{R}^{-1}\mb{S}\} + \log(|\mb{R}|) = tr \{\mb{S}\mb{X}\} - \log(|\mb{X}|).
\ee
The righthand side of Eq. (\ref{Eq:costfunctionX2}) is also a convex function in $\mathbf X$ via the same reasoning as in Section \ref{Sec:Known}. Note that we do not use $\mb{S}'$ any more in this case. Now the constraints are written as
\be
\left\{ \begin{array}{c}
\mb{R} = c \mb{I} + \mb{R}_c\\
rank(\mb{R}_c) \leq r\\
\mb{R}_c \succeq \mb{0}\\
\mb{R} \succeq c \mb{I}\\
c \geq \hat{c} \end{array} \right. .
\label{eqn:OptProblemLB}
\ee
The eigenvalue matrix of $\mb{R}$ has structure similar to that in Eq. (\ref{Eq:formR})
\be
\label{Eq:formR2}
\left(
  \begin{array}{ccccccc}
    \bar{\lambda_1} & 0               & \cdots & 0               & 0 & \cdots & 0\\
    0               & \bar{\lambda_2} & \cdots & 0               & 0 & \cdots & 0\\
    0               & 0               & \ddots & 0               & 0 & \cdots & 0\\
    0               & 0               & 0      & \bar{\lambda}_r & 0 & \cdots & 0\\
    0               & 0               & \cdots & 0               & c & \cdots & 0\\
    0               & 0               & \cdots & 0               & 0 & \ddots & 0\\
    0               & 0               & \cdots & 0               & 0 & \cdots & c\\
  \end{array}
\right)
\ee
where $\bar{\lambda}_1 \geq \bar{\lambda}_2 \geq \cdots \geq \bar{\lambda}_r \geq c$, $\mb\Lambda$, and hence the eigenvalue matrix of $\mb{X} = \mb{R}^{-1}$, has the following structure
\be
\mb\Lambda = \left(
  \begin{array}{ccccccc}
    \lambda_1 & 0         & \cdots & 0         & 0            & \cdots & 0\\
    0         & \lambda_2 & \cdots & 0         & 0            & \cdots & 0\\
    0         & 0         & \ddots & 0         & 0            & \cdots & 0\\
    0         & 0         & 0      & \lambda_r & 0            & \cdots & 0\\
    0         & 0         & \cdots & 0         & \dfrac{1}{c} & \cdots & 0\\
    0         & 0         & \cdots & 0         & 0            & \ddots & 0\\
    0         & 0         & \cdots & 0         & 0            & \cdots & \dfrac{1}{c}\\
  \end{array}
\right)
\ee
%where $\lambda_1 \leq \lambda_2 \leq \cdots \leq \lambda_r \leq \dfrac{1}{c}$. The first constraint is $\lambda_1 \leq \lambda_2 \leq \cdots \leq \lambda_N$, that is, $\mb{U}\bs{\lambda} \preceq \mb{0}$ where\\
%\be
%\mb{U} =
%\left(
%  \begin{array}{ccccc}
%    1 & -1 & 0 & \cdots & 0 \\
%    0 & 1 & -1 & \cdots & \vdots \\
%    0 & 0 & \ddots & \ddots & 0\\
%    0 & 0 & \cdots & 1 & -1 \\
%    0 & 0 & \cdots & 0 & -1 \\
%  \end{array}
%\right) \in \mathbb{R}^{N\times N}
%%s.t. & \lambda_i = 1 \qquad \text{for}\; (N-r) \; \lambda_i\\
%% & 0 < \lambda_i < 1 \qquad \text{for the rest}\; (r) \; \lambda_i \end{array} \right.
%\ee
%The second constraint is $0 < \lambda_i < \dfrac{1}{c}$ which can be expressed by
%\be
%\bs{\varepsilon} \preceq \bs{\lambda} \preceq \dfrac{1}{c}\mb{1}
%\ee
%where $\bs{\varepsilon}$ is enough near constant vector to zero.
%The last constraint is $\lambda_{r+1} = \lambda_{r+2} = \cdots = \lambda_N = \dfrac{1}{c}$, that is,\\
%\be
%\mb{E} \bs{\lambda} = \mb{h}
%\ee
%where
%\be
%\mb{E} =
%\left(
%  \begin{array}{cc}
%    \mb{0}_{r\times r} & \mb{0}_{r \times (N-r)}\\
%    \mb{0}_{(N-r)\times r} & \mb{I}_{N-r}\\
%  \end{array}
%\right) \in \mathbb{R}^{N\times N}
%\ee
%and $\mb{h} = [0, 0, \cdots, 0_r, \dfrac{1}{c}, \dfrac{1}{c}, \cdots, \dfrac{1}{c}]^T$.\\
Using reductions similar to the ones in Section \ref{Sec:Known}, we can formulate the following optimization problem
\be
\left\{ \begin{array}{cc}
\ds \min_{\bs{\lambda},c} & \mb{d}^T \bs{\lambda} - \mb{1}^T\log \bs{\lambda}\\
s.t. & \mb{U}\bs{\lambda} \preceq \mb{0}\\
 & -\bs{\lambda} \preceq -\bs{\varepsilon}\\
 & \bs{\lambda} \preceq \dfrac{1}{c}\mb{1}\\
 & \mb{E} \bs{\lambda} = \mb{h}\\
 & c \geq \hat{c} \end{array} \right.
\ee
Note, $\mathbf U$ and $\mathbf E$ are same as in Section \ref{Sec:Known}, but $\mb{h} = [0, 0, \cdots, 0_r, \dfrac{1}{c}, \dfrac{1}{c}, \cdots, \dfrac{1}{c}]^T$.

Again, inequality constraints may be combined into one
\be
\label{Eq:optimizationproblem2}
\left\{ \begin{array}{cc}
\ds \min_{\bs{\lambda},c} & \mb{d}^T \bs{\lambda} - \mb{1}^T\log \bs{\lambda}\\
s.t. & \mb{F}\bs{\lambda} \preceq \mb{g}\\
 & \mb{E} \bs{\lambda} = \mb{h}\\
 & c \geq \hat{c} \end{array} \right.
\ee
where $\mb{F} = \left[
  \begin{array}{c}
    \mb{U}\\
    -\mb{I}\\
    \mb{I}\\
  \end{array}\right]$,
  $\mb{g} = \left[
  \begin{array}{c}
    \mb{0}\\
    -\bs{\varepsilon}\\
    \dfrac{1}{c}\mb{1}\\
  \end{array}\right]$,
  $\mb{E} =
\left[
  \begin{array}{cc}
    \mb{0}_{r\times r} & \mb{0}_{r \times (N-r)}\\
    \mb{0}_{(N-r)\times r} & \mb{I}_{N-r}\\
  \end{array}
\right]$, and $\mb{h} = [0, 0, \cdots, 0_r, \dfrac{1}{c}, \dfrac{1}{c}, \cdots, \dfrac{1}{c}]^T$. Here, $\bs{\lambda}$, $\mb{a}$, $\mb{h} \in \mathbb{R}^N$, $\mb{g} \in \mathbb{R}^{3N}$, $\mb{U}$, $\mb{E} \in \mathbb{R}^{N \times N}$, and $\mb{F} \in \mathbb{R}^{3N \times N}$.\\
\indent

To solve (\ref{Eq:optimizationproblem2}), we first fix the noise level $c$ to $\bar{c} \geq \hat{c}$ and solve the optimization problem with respect to $\bs\lambda$ to obtain the optimal solution $\bs\lambda^\star(\bar{c})$ which is a function of $\bar{c}$. By substituting the optimal $\bs\lambda^\star(\bar{c})$ back into the cost function, we can get an auxiliary optimization over the scalar variable $c$.

Once we fix $c = \bar{c}$, the optimization problem (\ref{Eq:optimizationproblem2}) can be reduced to a problem in just $\bs\lambda$ which is nearly the same problem as (\ref{Eq:optimizationproblem1})
\be
\label{Eq:optimizationproblem3}
\left\{ \begin{array}{cc}
\ds \min_{\bs{\lambda}} & \mb{d}^T \bs{\lambda} - \mb{1}^T\log \bs{\lambda}\\
s.t. & \mb{F}\bs{\lambda} \preceq \mb{g}\\
 & \mb{E} \bs{\lambda} = \mb{h} \end{array} \right.
\ee
where $\mb{F} = \left[
  \begin{array}{c}
    \mb{U}\\
    -\mb{I}\\
    \mb{I}\\
  \end{array}\right]$,
  $\mb{g} = \left[
  \begin{array}{c}
    \mb{0}\\
    -\bs{\varepsilon}\\
    \dfrac{1}{\bar{c}}\mb{1}\\
  \end{array}\right]$,
  $\mb{E} =
\left[
  \begin{array}{cc}
    \mb{0}_{r\times r} & \mb{0}_{r \times (N-r)}\\
    \mb{0}_{(N-r)\times r} & \mb{I}_{N-r}\\
  \end{array}
\right]$, and $\mb{h} = [0, 0, \cdots, 0_r, \dfrac{1}{\bar{c}}, \dfrac{1}{\bar{c}}, \cdots, \dfrac{1}{\bar{c}}]^T$. Here, $\bs{\lambda}$, $\mb{a}$, $\mb{h} \in \mathbb{R}^N$, $\mb{g} \in \mathbb{R}^{3N}$, $\mb{U}$, $\mb{E} \in \mathbb{R}^{N \times N}$, and $\mb{F} \in \mathbb{R}^{3N \times N}$. Note that (\ref{Eq:optimizationproblem1}) and (\ref{Eq:optimizationproblem3}) are the same except that (\ref{Eq:optimizationproblem3}) uses $\dfrac{1}{\bar{c}}\mb{1}$ and $\dfrac{1}{\bar{c}}$ instead of $\mb{1}$ and $1$ in $\mb{g}$ and $\mb{h}$. Using a derivation similar to the one in Lemma \ref{Lemma3.1}, the optimal solution $\bs\lambda^\star$ of (\ref{Eq:optimizationproblem3}) is
\be
\label{eqn:optSolutionLB}
\lambda_i^\star = \left\{ \begin{array}{cc}
\min (\dfrac{1}{\bar{c}},\dfrac{1}{d_i}) & \text{for} \; i=1,2,\ldots,r\\
\dfrac{1}{\bar{c}} & \text{for} \; i=r+1,r+2,\ldots,N \end{array} \right.
\ee

If we substitute the optimal $\bs{\lambda^{*}}$ as in Eq. (\ref{eqn:optSolutionLB}) back into (\ref{Eq:optimizationproblem2}), we get the following problem to solve in the unknown noise level variable $c$
\be
\label{Eq:OptimizationProblemC}
\left\{ \begin{array}{cc}
\ds\min_c & \ds\sum_{i=1}^N G_i(c)\\
s.t. & c \geq \hat{c} \end{array} \right.
\ee
where
\be
G_i(c) = \left\{ \begin{array}{cc}
1 + \log d_i & \text{if} \; c \leq d_i\\
\dfrac{d_i}{c} + \log c & \text{if} \; c>d_i \end{array} \right.
\ee
for $i=1,2,\ldots,r$ and
\be
G_i(c) = \dfrac{d_i}{c} + \log c
\ee
for $i=r+1,r+2,\ldots,N$.
This problem too admits a closed form which is derived in Lemma \ref{Lemma3.2}.
\begin{lem}
\label{Lemma3.2}
The optimal solution of $c$ in \eqref{Eq:OptimizationProblemC} is given by
\be
c^\star = \left\{ \begin{array}{cc}
\hat{c} & \text{if} \; \hat{c} > \dfrac{\ds \sum_{i=r+1}^N d_i}{N-r}\\
\dfrac{\ds \sum_{i=r+1}^N d_i}{N-r} & \text{if} \; \hat{c} \leq \dfrac{\ds \sum_{i=r+1}^N d_i}{N-r} \end{array} \right.
\ee
\end{lem}
\begin{proof}
Let $c^\star$ be an optimal solution to the following optimization problem
\be
\left\{ \begin{array}{cc}
\ds\min_c & \ds\sum_{i=1}^N G_i(c)\\
s.t. & c \geq \sigma^2 \end{array} \right.
\ee
where, $G_i(c) = d_i\lambda_i^\star - \log \lambda_i$, namely
\be
G_i(c) = \left\{ \begin{array}{cc}
1 + \log d_i & \text{if} \; c \leq d_i\\
\dfrac{d_i}{c} + \log c & \text{if} \; c>d_i \end{array} \right.
\ee
for $i=1,2,\ldots,r$ and
\be
G_i(c) = \dfrac{d_i}{c} + \log c
\ee
for $i=r+1,r+2,\ldots,N$.\\
\indent
First, note that $G_i(c)$ is a constant for $c<d_i$ and monotonically increasing in $c$ for $c \geq d_i$ when $i = 1,2,\ldots,r$ because the first derivative of the function $\dfrac{d_i}{c} + \log c$, $-\dfrac{d_i}{c^2} + \dfrac{1}{c} > 0$ for $c>d_i$. Then, for $i=1,\ldots,r$,
\be
G_i(c) = \left\{ \begin{array}{ccc}
1 + \log d_i & \text{if} \; c \leq d_i & : \; \text{constant}\\
\dfrac{d_i}{c} + \log c & \text{if} \; c>d_i & : \; \text{increasing} \end{array} \right.
\ee
Because we assumed $d_1 \geq d_2 \geq \cdots \geq d_r$, we have
\bea
\label{Eq:sumGi}
\lefteqn{\ds \sum_{i=1}^r G_i(c)}\nonumber\\
& = & \left\{ \begin{array}{ccc}
\ds \sum_{i=1}^r ( 1 + \log d_i ) & \text{if} \; c \leq d_r\\
\ds \sum_{i=1}^{r-1} ( 1 + \log d_i ) + \dfrac{d_r}{c} + \log c & \text{if} \; d_r \leq c \leq d_{r-1}\\
\vdots & & \vdots\\
\ds \sum_{i=1}^r ( \dfrac{d_i}{c} + \log c ) & \text{if} \; c>d_1 \end{array} \right. .
\eea
In Eq. (\ref{Eq:sumGi}), $\sum_{i=1}^r G_i(c)$ is a constant when $c \leq d_r$ and an increasing function otherwise.
Finally, since $\sum_{i=1}^r G_i(c)$ is continuous at all $d_i$, we can conclude that
\be
\ds \sum_{i=1}^r G_i(c) = \left\{ \begin{array}{cc}
\text{constant} & \text{if} \; c \leq d_r\\
\text{increasing} & \text{if} \; c > d_r \end{array} \right.
\label{condb4r}
\ee
Now, $G_i(c) = \dfrac{d_i}{c} + \log c$ for $i=r+1,r+2,\ldots,N$. Hence,
\bea
\ds \sum_{i=r+1}^N G_i(c) & = & \ds \sum_{i=r+1}^N (\dfrac{d_i}{c} + \log c )\\
& = & \dfrac{1}{c} \ds \sum_{i=r+1}^N d_i + (N-r) \log c
\eea
We can easily see that
\be
\ds \sum_{i=r+1}^N G_i(c) = \left\{ \begin{array}{cc}
\text{decreasing} & \text{if} \; c \leq \hat d\\
\text{increasing} & \text{if} \; c > \hat d \end{array} \right.
\label{condafterr}
\ee
where $\hat d = \dfrac{\sum_{i=r+1}^N d_i}{N-r}$. In addition, it is obvious $\hat d < d_r$ because $d_r > d_i$ for all $i=r+1,r+2,\ldots,N$ and $\hat d$ which is the mean value of $d_{r+1},d_{r+2},\ldots,d_N$. From Eq. \eqref{condb4r} and Eq. \eqref{condafterr},
\bea
G(c) & = & \ds \sum_{i=1}^N G_i(c) = \ds \sum_{i=1}^r G_i(c) + \ds \sum_{i=r+1}^N G_i(c)\nonumber\\
& = & \left\{ \begin{array}{cc}
\text{decreasing} & \text{if} \; c \leq \hat d\\
\text{increasing} & \text{if} \; c > \hat d \end{array} \right.
\eea
Consequently, the optimal solution $c^\star$ which minimizes the cost function $G(c)$ is determined as
\be
c^\star = \left\{ \begin{array}{cc}
\hat{c} & \text{if} \; \hat{c} > \hat d\\
\hat d & \text{if} \; \hat{c} \leq \hat d \end{array} \right.
\ee

\end{proof}

Therefore, the optimal solution $\bs\lambda^\star$ is
\be
\label{eqn:optimalLabmdainoptimalc}
\lambda_i^\star = \left\{ \begin{array}{cc}
\min (\dfrac{1}{c^\star},\dfrac{1}{d_i}) & \text{for} \; i=1,2,\ldots,r\\
\dfrac{1}{c^\star} & \text{for} \; i=r+1,r+2,\ldots,N \end{array} \right. .
\ee

Combining the optimal estimated eigenvalue set in (\ref{eqn:optimalLabmdainoptimalc}) (expressed in terms of the optimal noise level variable $c^\star$) with the eigenvectors of the sample covariance leads to the optimal estimate of the (inverse of) covariance estimate of $\mb R$ under the rank constraint. This result is generalization of the result in Wax and Kailath \cite{Wax85} by including and exploiting a lower bound on the noise floor when available. We discuss these two algorithms in Section \ref{Sec:Wax} in more detail.

%Therefore,
%\be
%\mb{X}^\star = \mb{V}\mb\Lambda^\star\mb{V}^H
%\ee
%and the optimal covariance matrix $\mb{R}^\star$ is
%\be
%\mb{R}^\star = {\mb{X}^\star}^{-1} = \mb{V}{\mb\Lambda^\star}^{-1}\mb{V}^H
%\ee
%where $\mb{V}$ is the eigenvector matrix of the sample covariance matrix $\mb{S}$ and $\mb\Lambda^\star$ is a diagonal matrix with diagonal entries $\lambda_i^\star$.

\section{Experimental Results}
\label{results}

\subsection{Experimental Setup and Methods Compared}

\noindent Data from the L-band data set of the Knowledge Aided Sensor Signal Processing and Expert Reasoning (KASSPER) program \cite{Bergin02} is used for the performance analysis discussed in this section. The KASSPER data is the result of a significant effort by DARPA to provide a publicly available resource for the evaluation and benchmarking of radar STAP algorithms. As elaborated in \cite{Guerci06}, the KASSPER data set was carefully captured to represent real-world ground clutter and captures variations in underlying terrain, foliage and urban/manmade structures. Further, the KASSPER data set exhibits two very desirable characteristics from the viewpoint of evaluating covariance estimation techniques: 1.) the low-rank structure of clutter in KASSPER has been verified by researchers before \cite{Rangaswamy04Sep,Guerci06}, and 2.) the true covariance matrices for each range bin have been made available - this facilitates comparisons via powerful figures of merit where the theoretical upper/lower bounds are known.

The L-band data set consists of a data cube of 1000 range bins corresponding to the returns from a single coherent processing interval from 11($=J$) channels and 32($=P$) pulses. Therefore, the dimension of observations (or the spatio-temporal product) $N$ is $11 \times 32 = 352$. Other key parameters are detailed in Table \ref{Tb:parameters}. Finally, a clutter rank\footnote{We set clutter ridge parameters so that $\gamma = 1$.   } of $r = J + P - 1 = 42$ was used by our RCML estimator in all the results to follow, unless explicitly stated otherwise.

%Now we examine the performance of the proposed algorithm and previous works using various number of training samples, $K$. That is, we evaluate the performance under 300 ($K < N$), 352 ($K=N$), 750 ($K\approx 2N$), and 3000 ($K\gg N$) training samples. In terms of measurement, since data from the KASSPER program makes the ''true" covariance matrix available, we can use the normalized SINR measure and estimation error variance to compare the performance of the algorithms. Now we introduce the other algorithms used in comparison with the proposed estimator.
\begin{table}[!t]
\begin{center}
\caption{KASSPER Dataset-1 parameters}
\label{Tb:parameters}
\begin{tabular}{l l}
  \hline
  Parameter & Value\\
  \hline
  Carrier Frequency & 1240 MHz \\
  Bandwidth (BW) & 10 MHz \\
  Number of Antenna Elements & 11 \\
  Number of Pulses & 32 \\
  Pulse Repetition Frequency & 1984 Hz \\
  1000 Range Bins & 35 km to 50 km \\
  91 Azimuth Angles & $87^{\circ}$, $89^{\circ}$, $\ldots$ $267^{\circ}$ \\
  128 Doppler Frequencies & -992 Hz, -976.38 Hz, $\ldots$, 992 Hz \\
  Clutter Power & 40 dB \\
  Number of Targets & 226 (~200 detectable targets) \\
  Range of Target Dop. Freq. & -99.2 Hz to 372 Hz\\
  \hline
\end{tabular}
\end{center}
\end{table}
%\begin{figure}[!t]
%\centering
%\includegraphics[width=2.5in]{fig/eigenspectrum.eps}
%\caption{KASSPER Dataset-1 Eigenspectrum.}
%\label{Fig:eigenspectrum}
%\end{figure}

We evaluate and compare four different covariance estimation techniques:

\begin{itemize}
\item \textbf{Sample Covariance Matrix:} The sample covariance matrix is given in Eq. (\ref{Eq:SCM}). It is well known that the sample covariance is the unconstrained maximum likelihood estimator under Gaussian disturbance statistics. Consistent with radar literature \cite{Reed74}, we'll refer to the use of this technique as SMI.

\item \textbf{Fast Maximum Likelihood:} The fast maximum likelihood (FML) \cite{Steiner00} uses the structural constraint of the covariance matrix which is given in Eq. (\ref{Eq:structuralconstraint}). The FML method just involves calculating the eigenvalue decomposition of the sample covariance and perturbing eigenvalues to conform to the structure in Eq. (\ref{Eq:structuralconstraint}). The noise variance $\sigma^2$ is assumed known or pre-estimated. FML's success in radar STAP is widely known \cite{Rangaswamy04Sep,Wicks06,Gini08}.

\item \textbf{Leave-one-out shrinkage estimator:} Shrinkage estimators are powerful estimators of covariance for high dimensional data that are known to also perturb the eigenstructure of the sample covariance matrix\footnote{Via this definition, the FML and RCML can also been seen as a special class of shrinkage estimators.} \cite{Chen10} - often to ensure non-singularity of the estimated covariance. While a variety of shrinkage techniques are known \cite{Chen10,Ledoit03,Stoica08,Won09}, we choose the leave-one-out covariance matrix estimate (LOOC) shrinkage estimator \cite{Hoffbeck96},
    \be
    \mb{R} = \beta diag(\mb{S}) + (1-\beta) \mb{S}
    \ee
The value of $\beta$ is determined via a cross-validation technique so that the average likelihood of omitted samples is maximized. We pick this estimator because it is most suited to the problem at hand and has demonstrated success in the $K \leq N$ training regime \cite{Hoffbeck96}.

\item \textbf{Eigencanceler:} The eigencanceler (EigC) is based on the eigenanalysis which suggests a small number of eigenvalues contain all the information about interferences (jammers and clutter), and therefore, the span of the eigenvectors associated with these significant eigenvalues includes all the position vectors that comprise the interference signals \cite{Haimovich96}. Since we assume that the rank is known a priori, the eigencanceler can be compared with our estimator as we use $r$ dominant eigenvectors as interference eigenvectors. The covariance matrix can be expressed by
    \be
    \mb{R} = \ds\sum_{i=1}^r p_i\mb{v}_i\mb{v}_i^H + \sigma^2 \mb{I}
    \ee
    where $p_i$ and $\mb{v}_i$ are the clutter power and the eigenvector corresponding to $r$ dominant eigenvalues, respectively. For $p_i \gg \sigma^2$, it follows from \cite{Kirsteins94,Rangaswamy04Sep} that the estimated inverse covariance matrix can be approximated as $\hat{\mb{R}}^{-1} \approx \dfrac{1}{\sigma^2}(\mb{I} - \mb{P})$ where $\mb{P} = \ds\sum_{i=1}^{r} \mb{v}_i \mb{v}_i^H$. We apply this inverse covariance matrix in computing the SINR.

\item \textbf{Rank Constrained Maximum Likelihood:} Our proposed estimator (abbreviated to RCML) incorporates the structural constraint and for the first time the information of the rank of the clutter component.

\end{itemize}

\subsection{Experimental Evaluation}

The normalized signal to interference and noise ratio (SINR) is used for evaluation the aforementioned covariance estimation techniques. The SINR is desired to be as high as possible. This figure of merit is plotted against azimuthal angle as well as Doppler frequency for distinct training regimes, i.e. low, representative and generous training. We also show the plot of SINR performance versus the number of training samples. Finally, we also evaluate the robustness of our RCML estimator against perturbations in the knowledge of the true rank.

\begin{figure}
\begin{center}
\subfigure[$K=352$]{\includegraphics[width=2.5in]{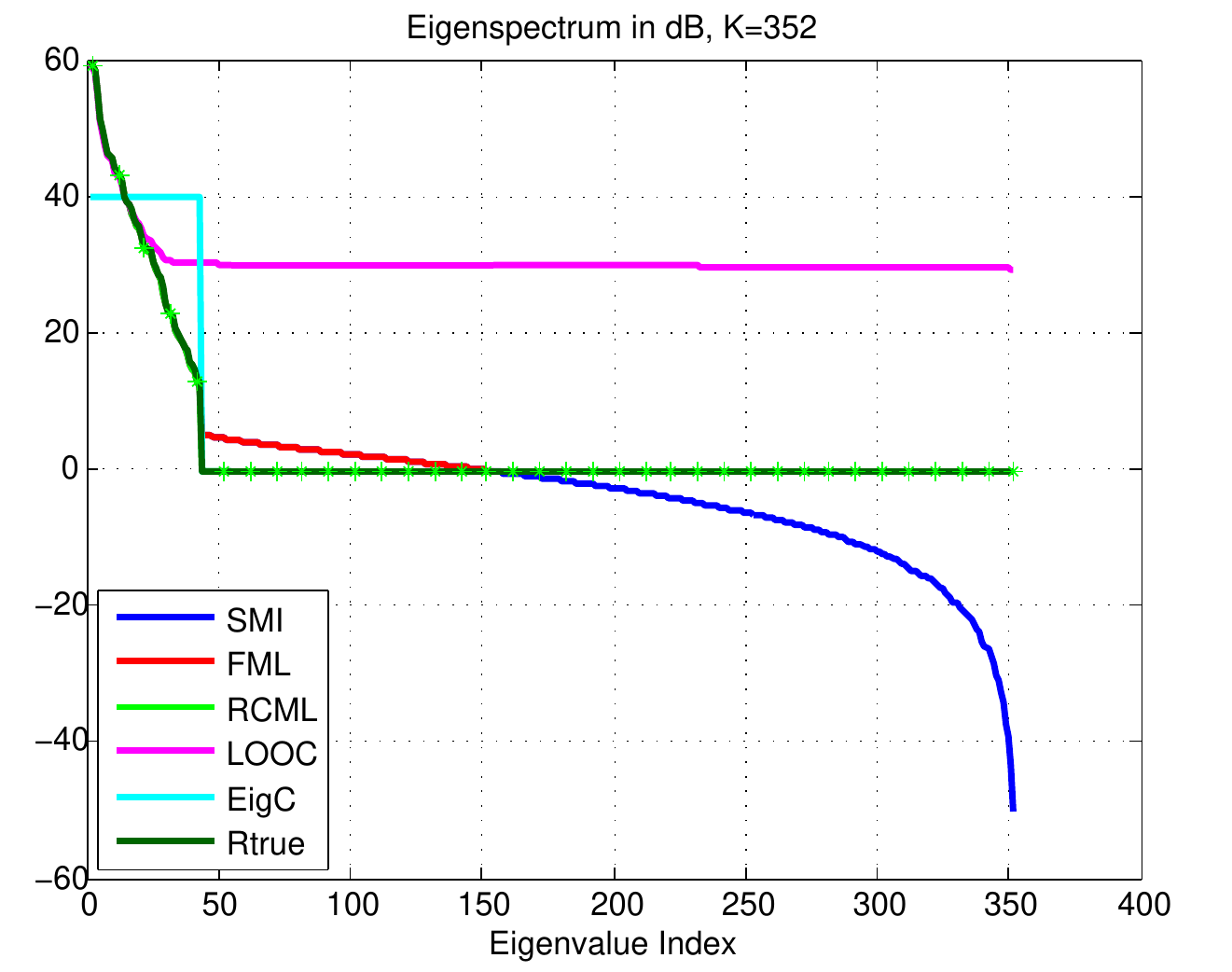}\label{Fig:Eigenspectrum352}}
\hfil
\subfigure[$K=750$]{\includegraphics[width=2.5in]{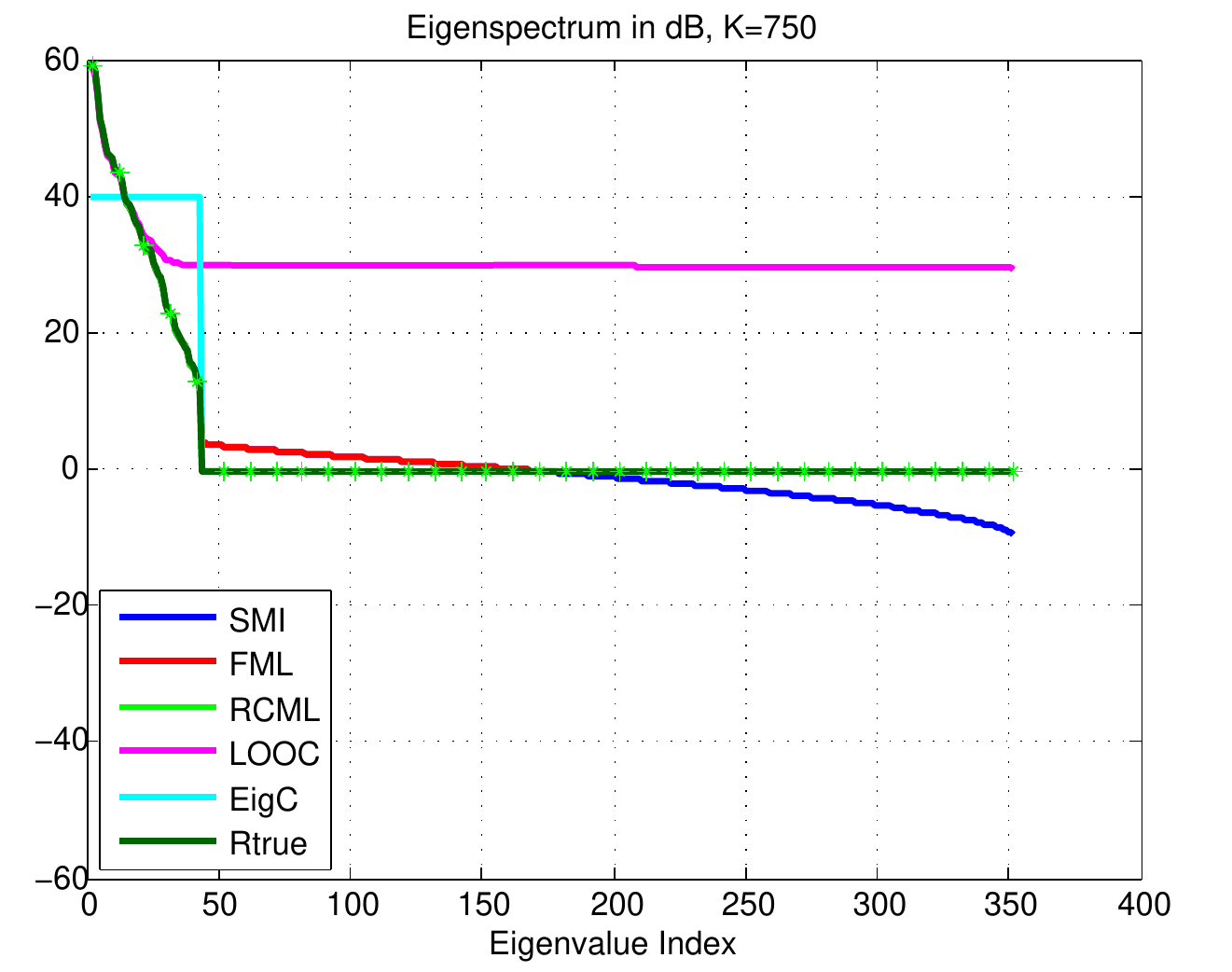}\label{Fig:Eigenspectrum750}}\\
\end{center}
\caption{Eigenspectra of the estimated covariance matrices.}
\label{Fig:Eigenspectrum}
\end{figure}

\subsubsection{Eigenspectrum}
\label{sec:Eigenspectrum}

First, we provide the eigenspectra plots, i.e. a plot of the eigenvalues of the estimated covariance matrices vs. those of the ground truth covariance. Figure \ref{Fig:Eigenspectrum} shows the eigenspectra of various estimators for $K=N=352$ and $K=750$ in dB. We can see the decay of eigenvalues of the true covariance matrix is readily apparent. The first few dominant eigenvalues are well predicted by every method but we notice that after the index 20 or so RCML shows the tightest overlap with the eigenvalues of the true covariance matrix with FML and EigC approaches not far behind. In particular, while FML shows slightly better agreement in capturing the decay, EigC like RCML is more accurate in capturing the rank. The SMI, FML and LOOC approaches are in fact very significantly off in their estimate of the clutter rank which is determined as an outcome for these methods. In the next Section, we will demonstrate the benefits of incorporating rank and structural information about the disturbance covariance for widely used figures of merit in the radar literature.
%
%All of compared methods (except EigC) show they have very close eigenvalues to those of the true covariance matrix for some dominant eigenvalues. FML overlaps SMI where the eigenvalues of the sample covariance are greater than $\sigma^2$ (index 200 for $K=352$ and 250 for $K=750$) and keeps the eigenvalues $\sigma^2$ after that. In the case of the eigencanceler, it has constant eigenvalues for the index less than the rank $r$ and overlaps the RCML for the rest. Though the RCML and FML are the most comparable to the true covariance matrix, FML shows bigger difference than RCML in the range of eigenvalue index 50 to 250 since it does not employ the rank information.

\subsubsection{Normalized SINR vs. angle and Doppler}
\label{sec:SINR}

The normalized SINR measure \cite{Monzingo04} is commonly used in the radar literature and is given by
\be
\label{Eq:SINR}
\eta = \dfrac{|\mb{s}^H\hat{\mb{R}}^{-1}\mb{s}|^2}{|\mb{s}^H\hat{\mb{R}}^{-1}\mb{R}\hat{\mb{R}}^{-1}\mb{s}||\mb{s}^H\mb{R}^{-1}\mb{s}|}
\ee
where $\mb{s}$ is the spatio-temporal steering vector, $\hat{\mb{R}}$ is an estimated covariance matrix, and $\mb{R}$ is the corresponding true covariance matrix. It is easily seen that $0 < \eta \leq 1$ and $\eta = 1$ if and only if $\hat{\mb{R}} = \mb{R}$. Since the steering vector is a function of both azimuthal angle and Doppler frequency, we evaluate the normalized SINR in both angle and Doppler domain. This would lead to a SINR surface as a function of azimuthal angle and Doppler and comparing surface plots across different covariance estimation techniques is cumbersome. We therefore obtain plots as a function of one variable (i.e. just angle/Doppler) by marginalizing (averaging) over the other variable. The SINR is plotted in dB. in all figures in this dissertation, that is, $\text{SINR}_{\text{dB}} = 10 \log_{10} \eta$. Therefore, $\text{SINR}_{\text{dB}} \leq 0$.

Figures \ref{Fig:SINR_RCML1} and \ref{Fig:SINR_RCML2} plot the variation of normalized SINR as a function of the azimuthal angle and the Doppler frequency for varying number of training samples, $K$. Specifically, Figures \ref{Fig:SINR300angle} and \ref{Fig:SINR300dop} are corresponding to $K =300 < N=352$, Figures \ref{Fig:SINR352angle} and \ref{Fig:SINR352dop} plot results for $K = 352 = N$, likewise Figures \ref{Fig:SINR750angle}, \ref{Fig:SINR750dop} and \ref{Fig:SINR3000angle}, \ref{Fig:SINR3000dop} are corresponding to  $K=750\approx 2N$ and $K=3000 \gg N$ respectively.

Figures \ref{Fig:SINR_RCML1} reports results for the challenging regime of $K \leq N$. When $K < N$ the sample covariance matrix is not invertible, hence for the results in Figures \ref{Fig:SINR_RCML1} and \ref{Fig:SINR_RCML2} we used its pseudo-inverse as a substitute. Unsurprisingly, the sample covariance technique suffers tremendously when $K \leq N$ as is evident from Figures \ref{Fig:SINR_RCML1}. LOOC shrinkage does considerably better than SMI because it forces a reasonably good eigenstructure. The informed estimators, i.e.\ FML, EigC, and RCML perform appreciably well with RCML affording the best overall performance. It is useful to note that RCML in fact offers about 1 dB improvement over FML.

Even for representative training in Figures \ref{Fig:SINR750angle} and \ref{Fig:SINR750dop}, the vastly superior performance of the FML, EigC, and RCML techniques is apparent. Again, by virtue of incorporating the rank information, the proposed RCML estimator outperforms the competing methods.  Finally, Figures \ref{Fig:SINR3000angle} and \ref{Fig:SINR3000dop} confirm the intuition that as training becomes close to asymptotic, the gap between the various methods begins to decrease - of course, such generous training is typically impossible to obtain in practice. This is due to the fact that all the covariance matrix estimates considered converge to the true covariance matrix in the limit of large training data.

\begin{figure}
\begin{center}
\subfigure[$K=300$]{\includegraphics[width=2.5in]{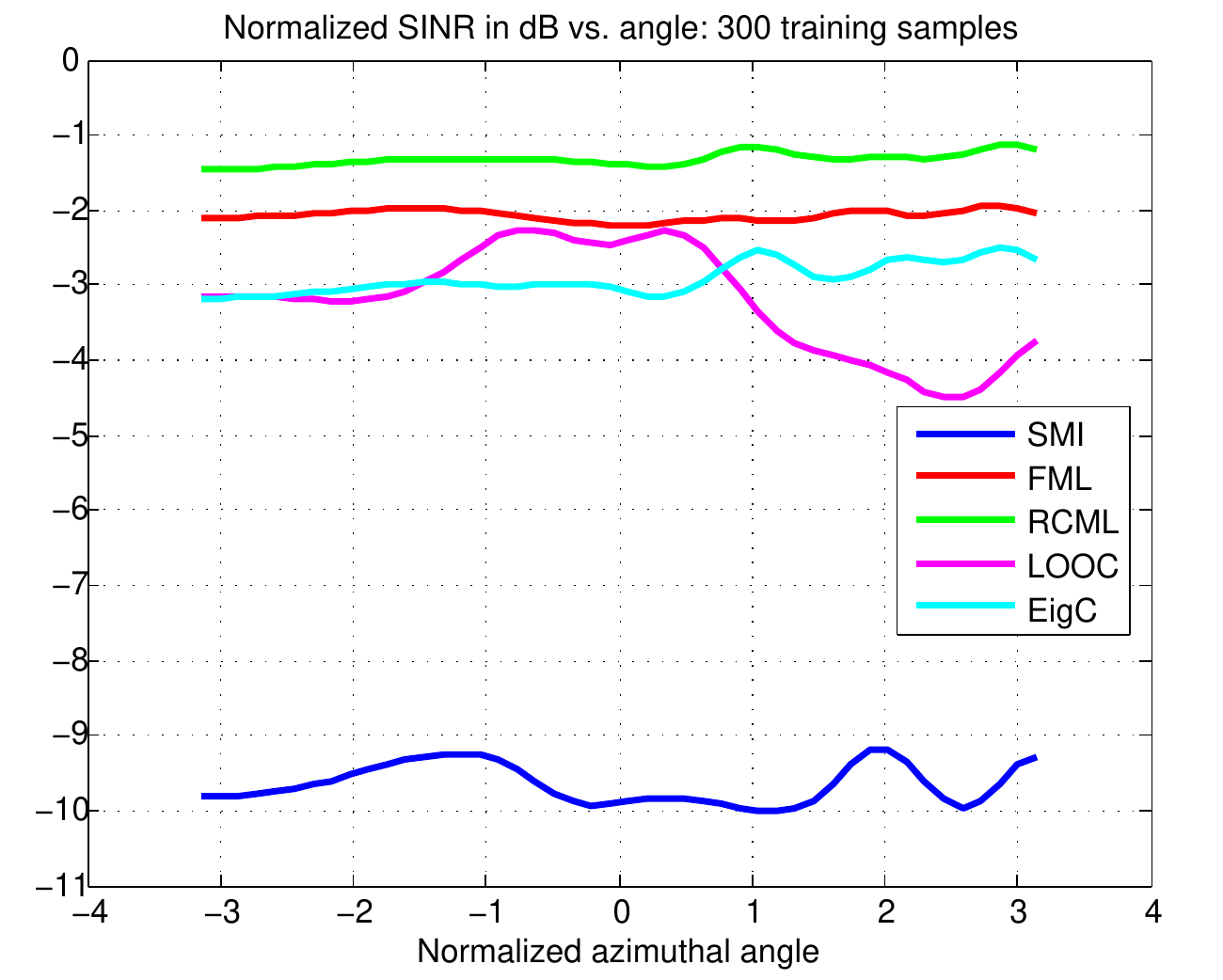}\label{Fig:SINR300angle}}
\hfil
\subfigure[$K=300$]{\includegraphics[width=2.5in]{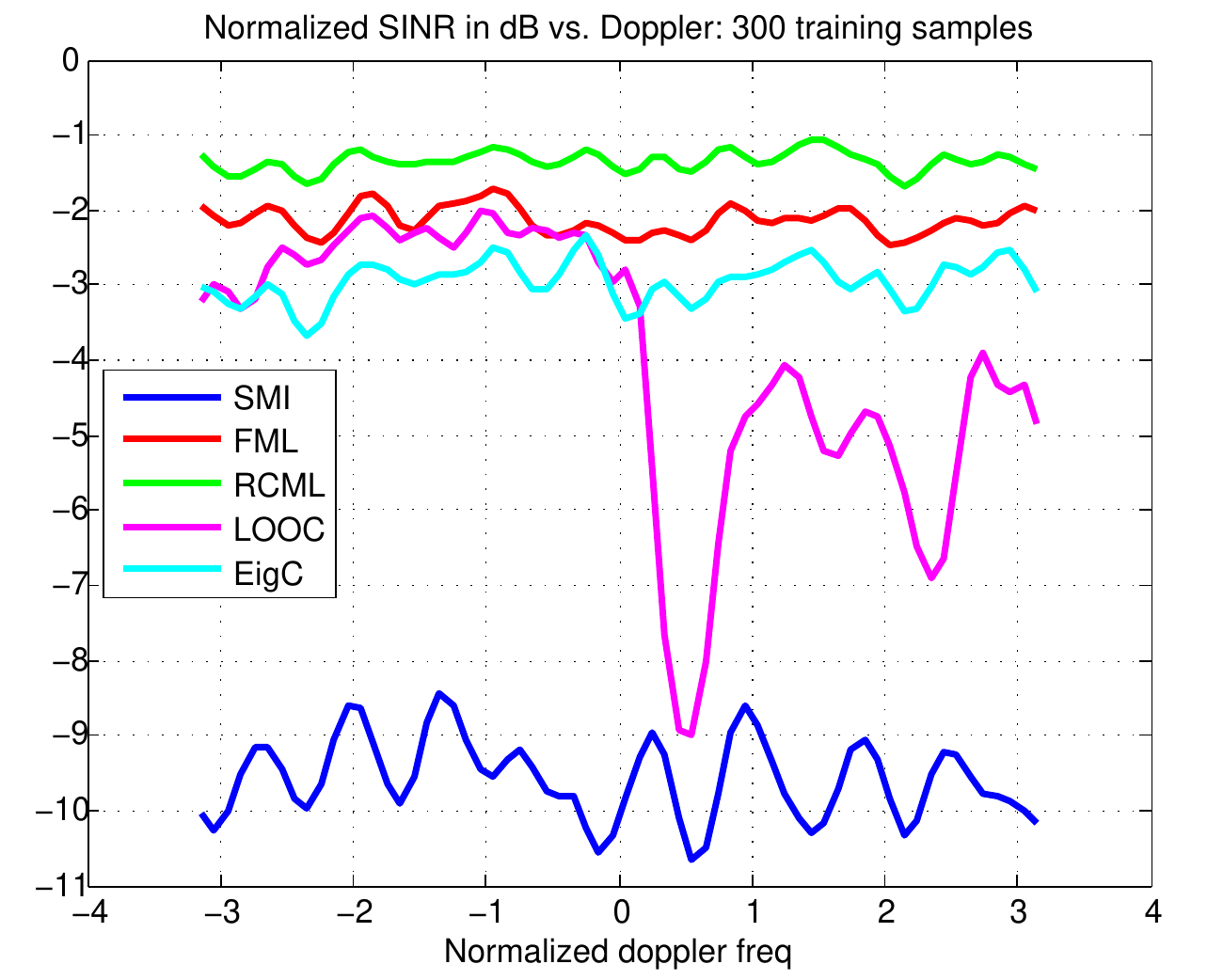}\label{Fig:SINR300dop}}\\
\subfigure[$K=352$]{\includegraphics[width=2.5in]{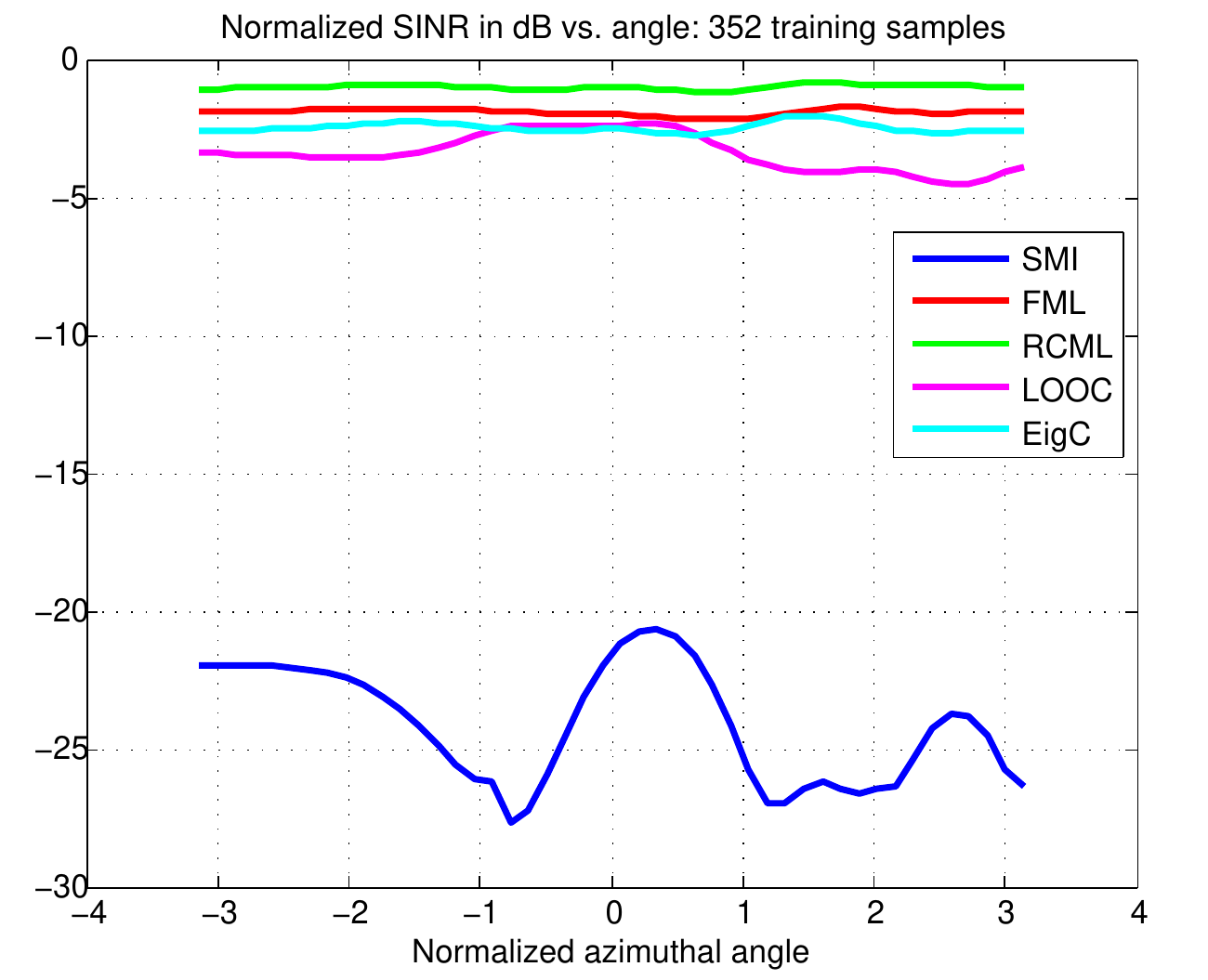}\label{Fig:SINR352angle}}
\hfil
\subfigure[$K=352$]{\includegraphics[width=2.5in]{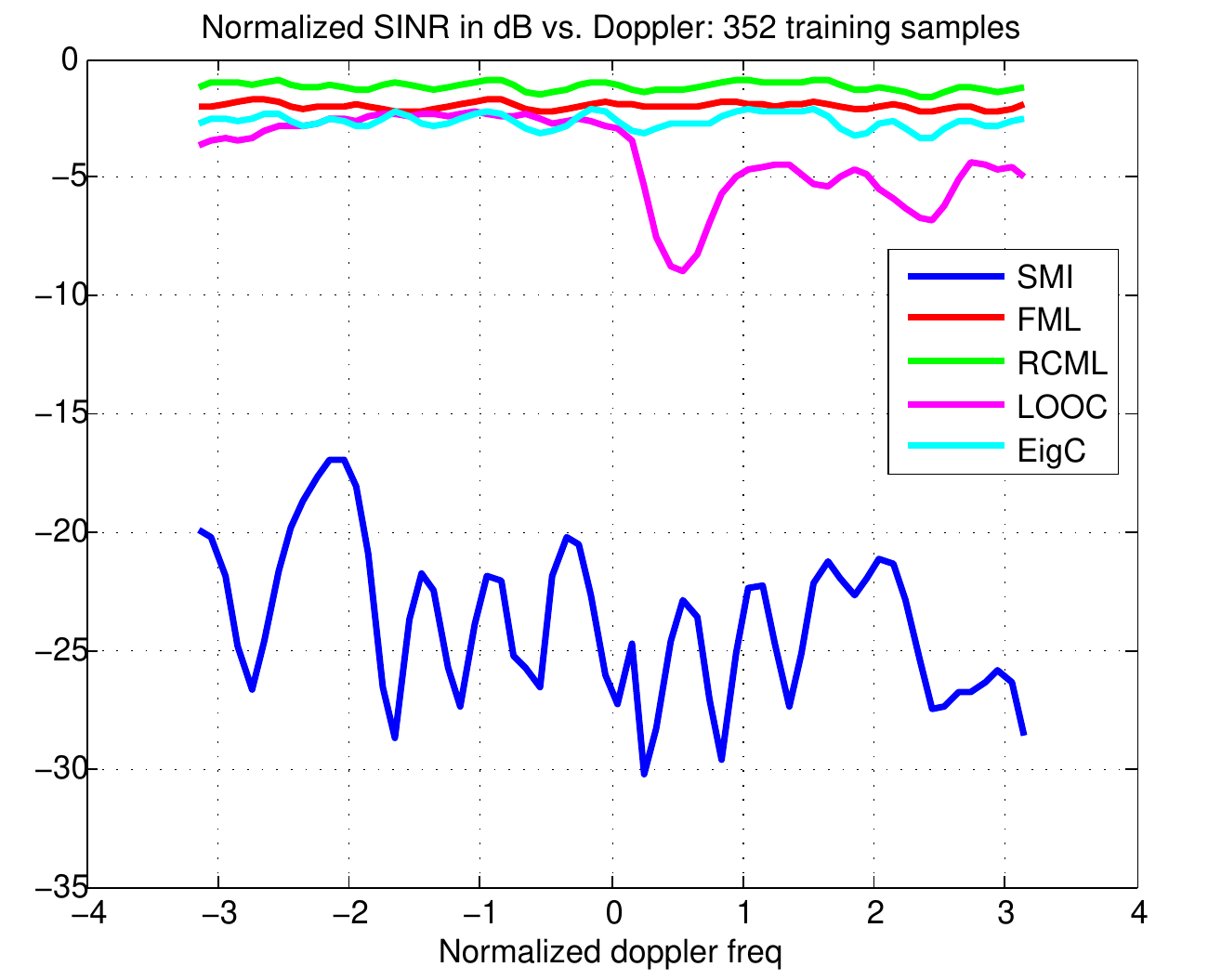}\label{Fig:SINR352dop}}\\
\end{center}
\caption{Normalized SINR vs. normalized azimuthal angle and doppler frequency.}
\label{Fig:SINR_RCML1}
\end{figure}

\begin{figure}
\begin{center}
\subfigure[$K=750$]{\includegraphics[width=2.5in]{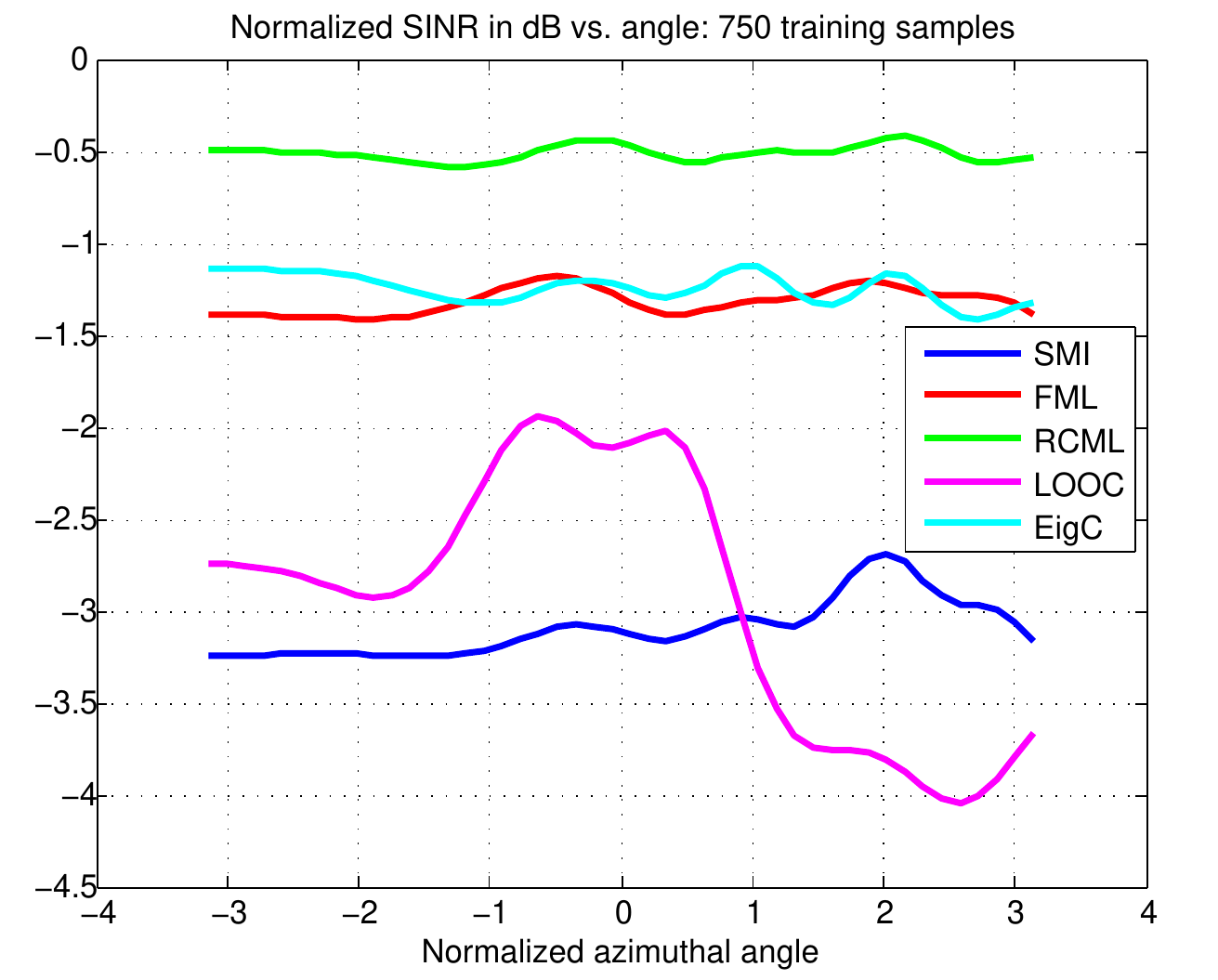}\label{Fig:SINR750angle}}
\hfil
\subfigure[$K=750$]{\includegraphics[width=2.5in]{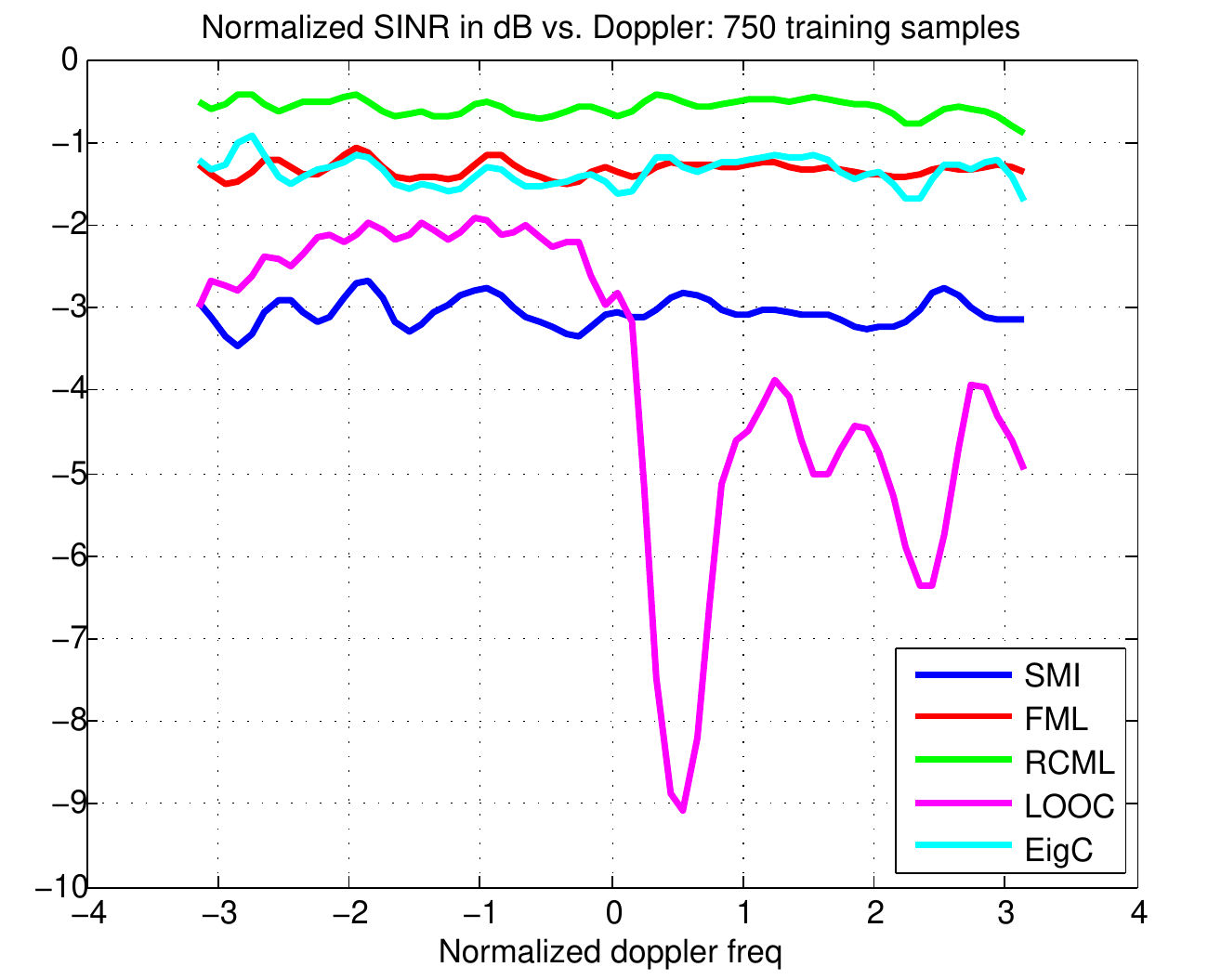}\label{Fig:SINR750dop}}\\
\subfigure[$K=3000$]{\includegraphics[width=2.5in]{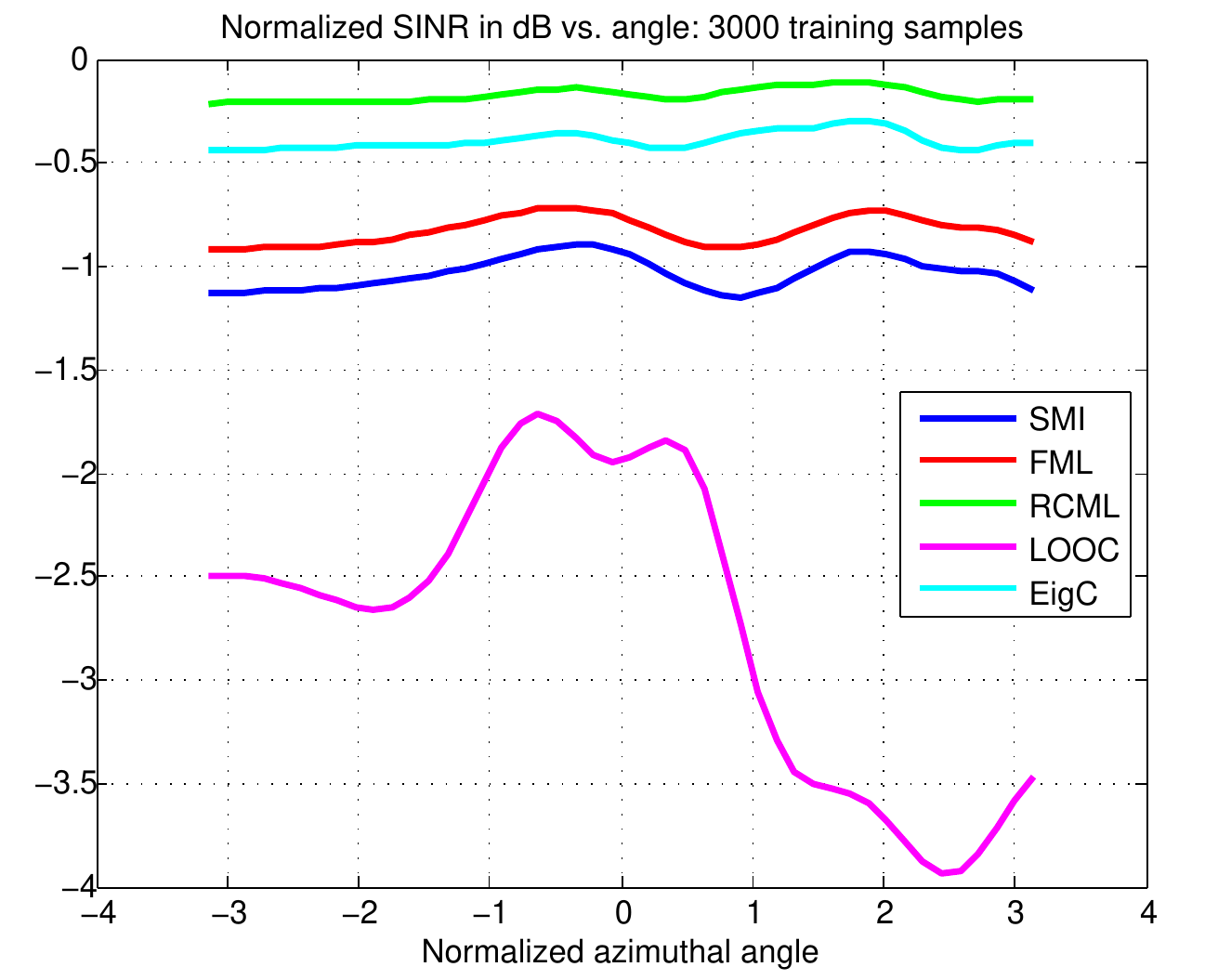}\label{Fig:SINR3000angle}}
\hfil
\subfigure[$K=3000$]{\includegraphics[width=2.5in]{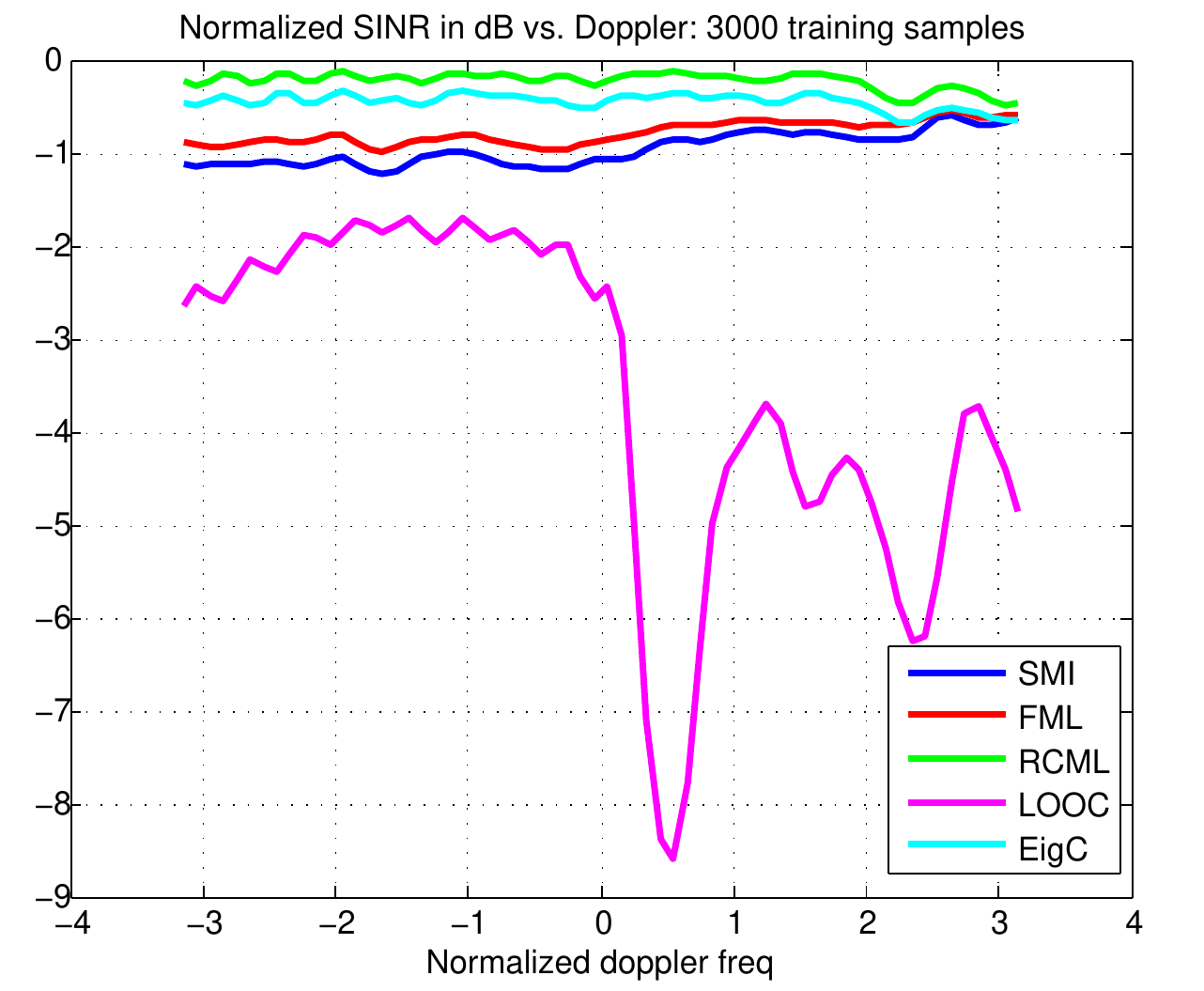}\label{Fig:SINR3000dop}}\\
\end{center}
\caption{Normalized SINR vs. normalized azimuthal angle and doppler frequency.}
\label{Fig:SINR_RCML2}
\end{figure}

\subsubsection{Performance vs. number of training samples}

While the results in Sections \ref{sec:SINR} do explore performance against training to some extent - here we present bar graphs to explore this issue with a finer granularity. To obtain a single scalar performance measure as a function of training, averaging was carried out over both the angle and Doppler variables.

Figure \ref{Fig:Trainingsamples} presents bar graphs that quantify the SINR (in dB) as a function of training samples $K$, where $K$ is varied from as low as $60$ to as high as $3000$.  Two trends are evident from Figure \ref{Fig:Trainingsamples}: 1.) as intuitively expected, the SINR values increases monotonically with an increase in the number of training samples for all methods (except for the sample covariance technique in the $K \leq N$ regime which is a well-known phenomena observed in past work as well \cite{Steiner00}) and 2.) the RCML estimator exhibits remarkably good performance in all training regimes.

\begin{figure}[!t]
\centering
\includegraphics[width=2.5in]{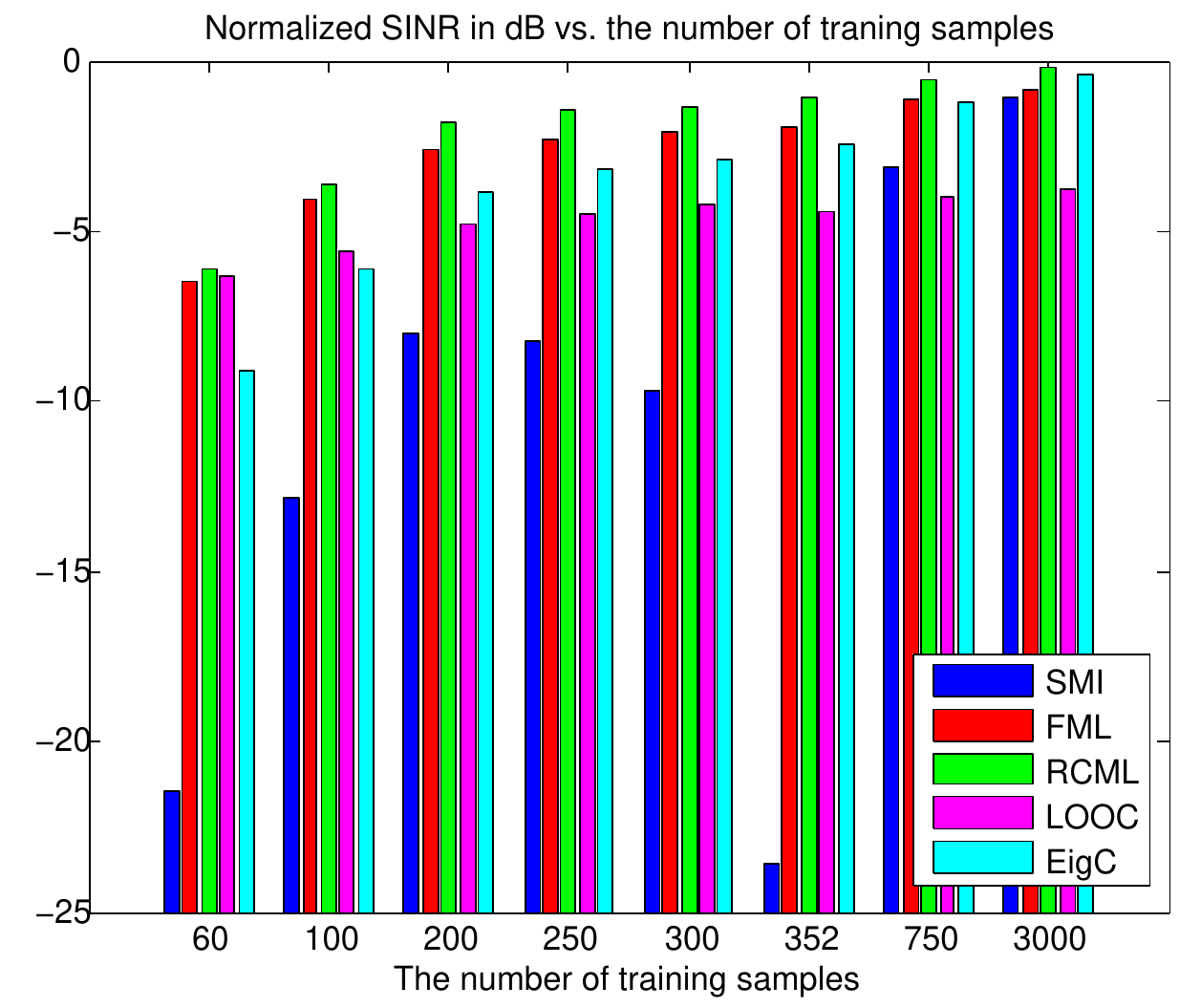}
\caption{Normalized SINR vs. the number of training samples.}
\label{Fig:Trainingsamples}
\end{figure}

\subsubsection{Rank Sensitivity}
\label{Sec:RankSensitivity}

The KASSPER data, the clutter rank conforms to Eq. (\ref{Eq:Brennanrule}) - the Brennan rule. For the parameters used in our experiments, this would lead to a predicted ideal rank of $r = J + P -1 = 42$. In a practical situation, departures from the ideal behavior are expected and hence we explore the performance our proposed RCML estimator even as \emph{incorrect} rank information is used.

The results in Figure \ref{Fig:RankSensitivity} demonstrate the robustness of RCML to perturbations in the clutter rank. Figure \ref{Fig:RankSensitivity} presents bar graphs that show averaged SINR results for $K = 352$ and $K = 750$ training samples. We determined numerically that the ``true'' rank of the clutter covariance for the range bin of choice was in fact 43 which is a mild departure from the 42 predicted by the Brennan rule. Comparisons are made between FML and RCML with the difference that seven variants of RCML are presented - with rank from $34$ to $45$. As Figure \ref{Fig:RankSensitivity} reveals, using the true rank of $43$ indeed yields the best covariance matrix estimator but the penalty of the small departure, i.e.\ using a rank of $40$ to $45$ which are close to the true rank $43$ leads to a very small performance loss. On the other hand, Figs.\ \ref{Fig:RankSensitivity} also shows variants of the RCML result with a somewhat bigger departure, i.e.\ a rank of $34$. In this case, the performance of RCML with rank $34$ is appreciably lower against using rank values around the true rank $43$. Remarkably, RCML with rank $34$ is still competitive with FML.  Overall Figure \ref{Fig:RankSensitivity} therefore provides two valuable insights: 1.) since rank information is predicted using the Brennan rule - small departures in practice are possible and our estimator exhibits desirable robustness against such small perturbations to rank, and 2.) the value of using the rank information is simultaneously revealed - because RCML with rank $34$ is competitive with FML which does not exploit the rank information.

\begin{figure}[!t]
\centering
\includegraphics[width=2.5in]{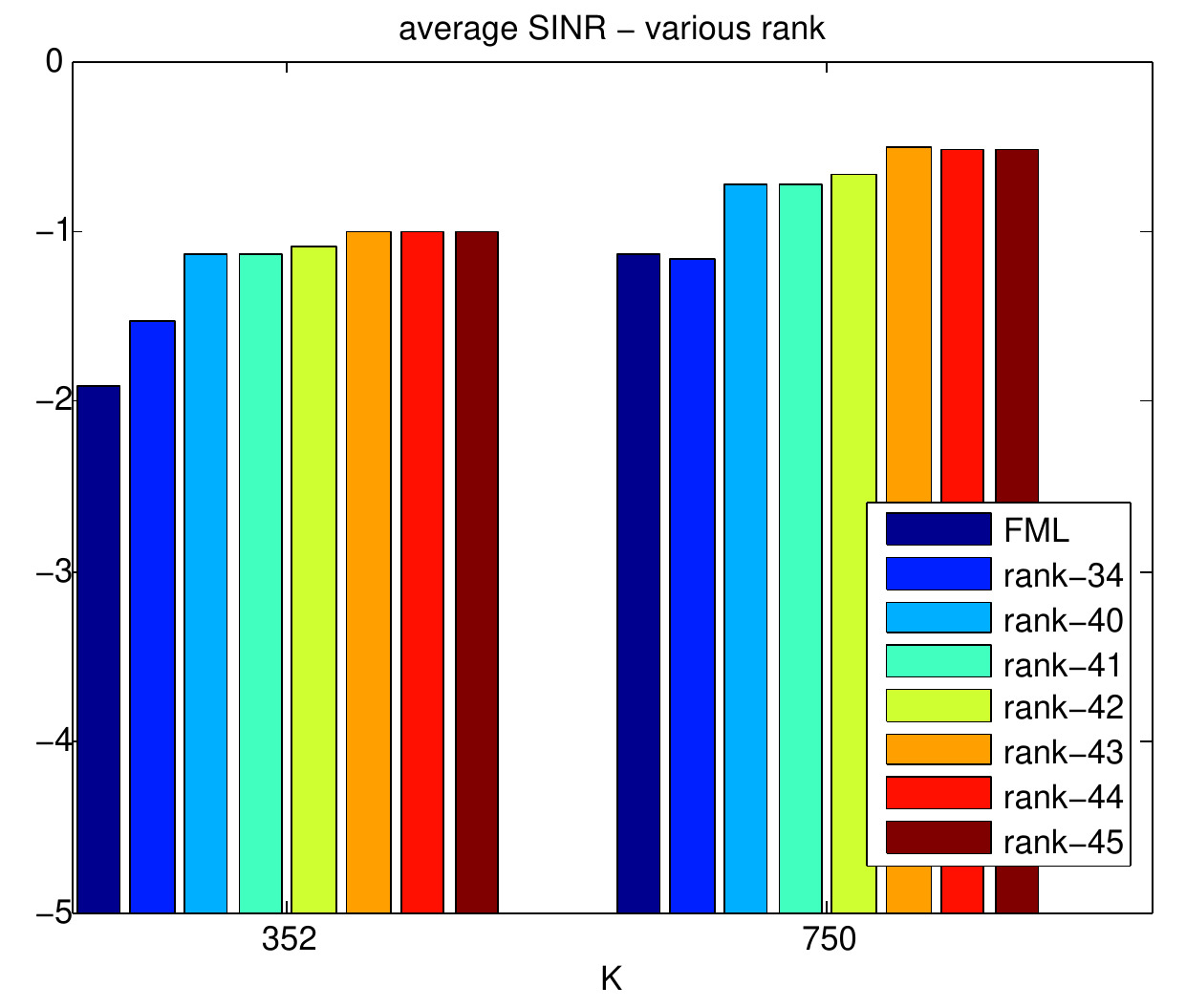}
\caption{Normalized SINR of rank constrained maximum likelihood (RCML) for various rank information}
\label{Fig:RankSensitivity}
\end{figure}

The experiments in Section \ref{sec:Eigenspectrum} to Section \ref{Sec:RankSensitivity}, however, assume that we have access to homogeneous training samples, which is often not available in practice. This section provides more realistic and challenging practical evaluation by means of two new flavors of experimental results: 1.) plots of probability of detection versus SNR for a variety of detection statistics, and 2.) normalized SINR performance in the presence of {\em heterogeneous} training samples which are corrupted by the target information. Here, we consider an experimental environment which reflects real-world scenarios by considering non-homogenous training. We perform two experimental investigations. First, we examine if incorporating rank-information really leads to better target detection. Second, robustness to target contamination in training samples is investigated - while outlier removal techniques have been proposed \cite{Rangaswamy05,Blunt04}, in practice target contamination of training data cannot be entirely ruled out. We evaluate and compare three different covariance estimation techniques, SMI, FML and our proposed RCML. We show that the RCML estimator can still outperform alternatives in that detection probability is second only to the theoretic upper bound when the true covariance is known, and the rank information is invaluable in yielding meaningful estimates even as almost all available training samples are corrupted.

\subsubsection{Probability of Detection Vs. SNR}

\begin{figure}[!t]
\begin{center}
\subfigure[AMF for $K=352$]{\includegraphics[width=2.4in]{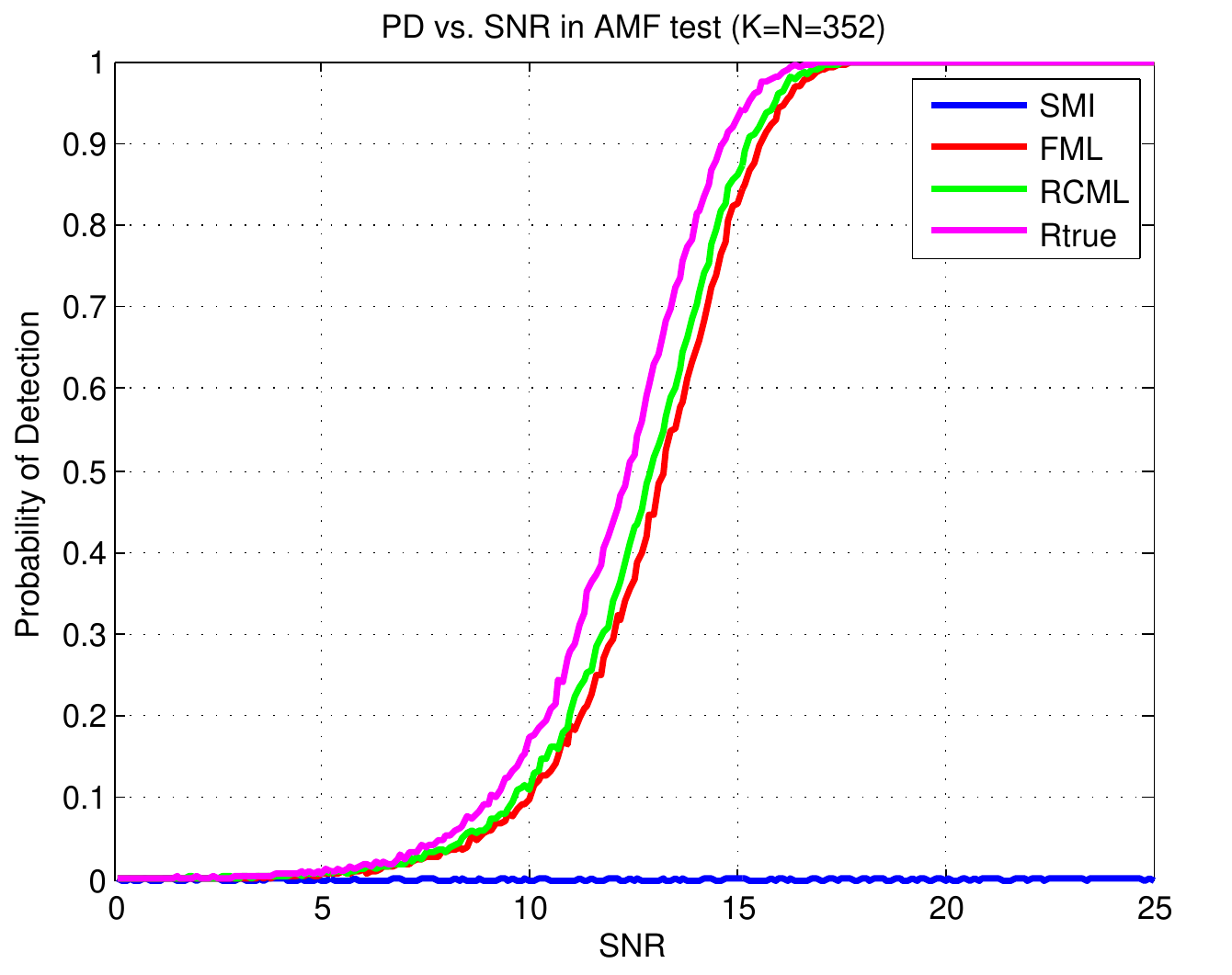}\label{Fig:AMF352}}
\hfil
\subfigure[AMF for $K=704$]{\includegraphics[width=2.4in]{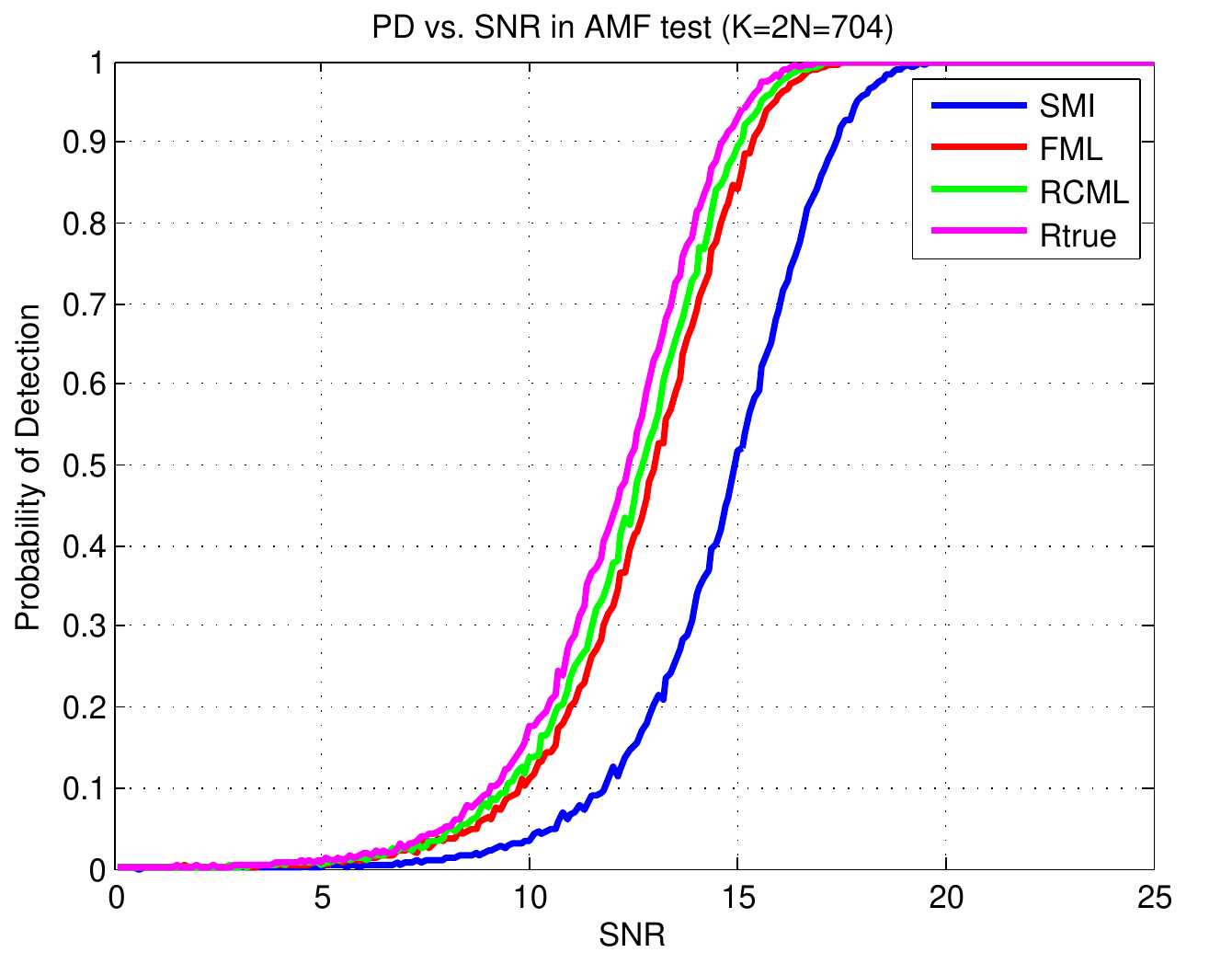}\label{Fig:AMF704}}\\
\subfigure[NMF for $K=352$]{\includegraphics[width=2.4in]{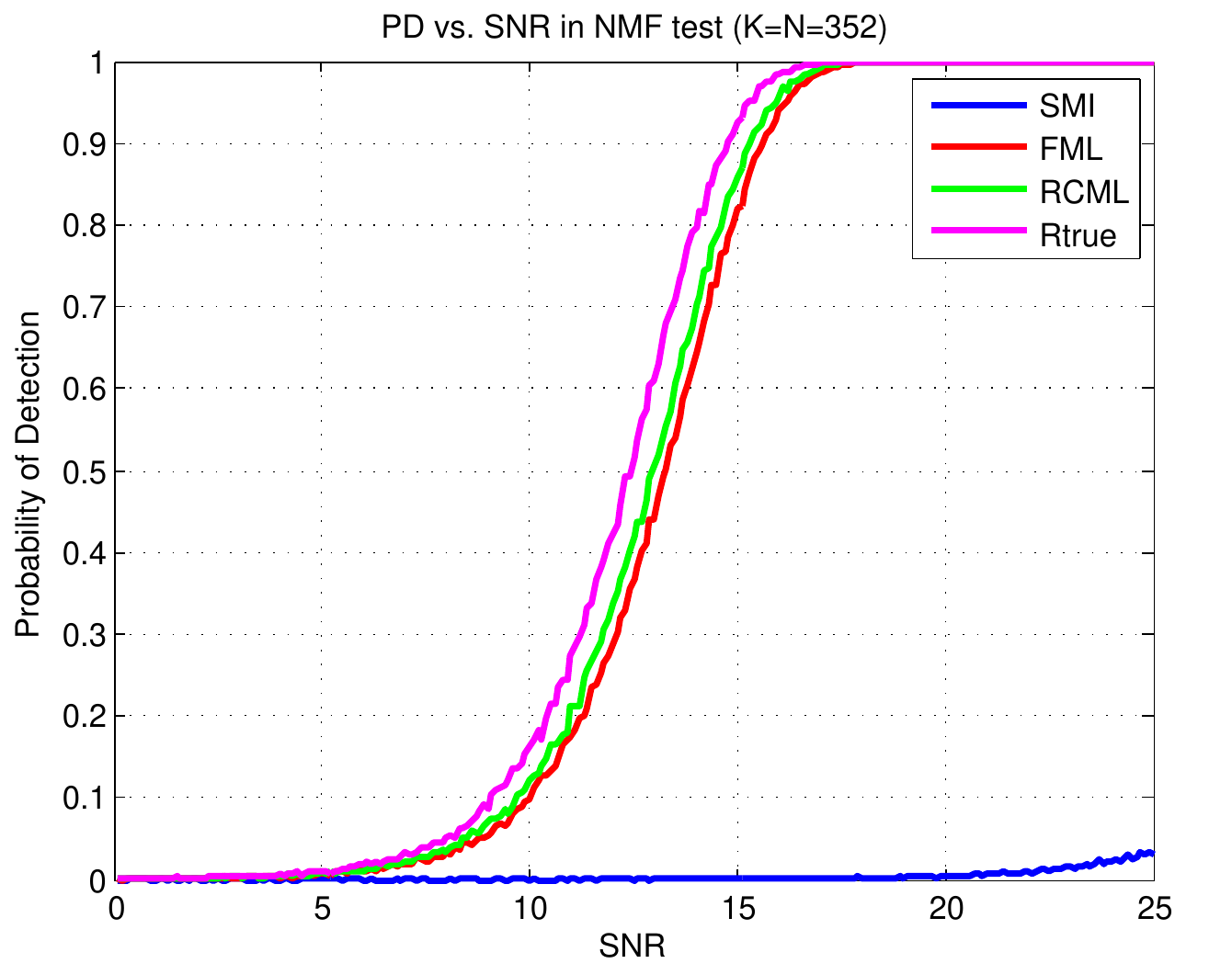}\label{Fig:NMF352}}
\hfil
\subfigure[NMF for $K=704$]{\includegraphics[width=2.4in]{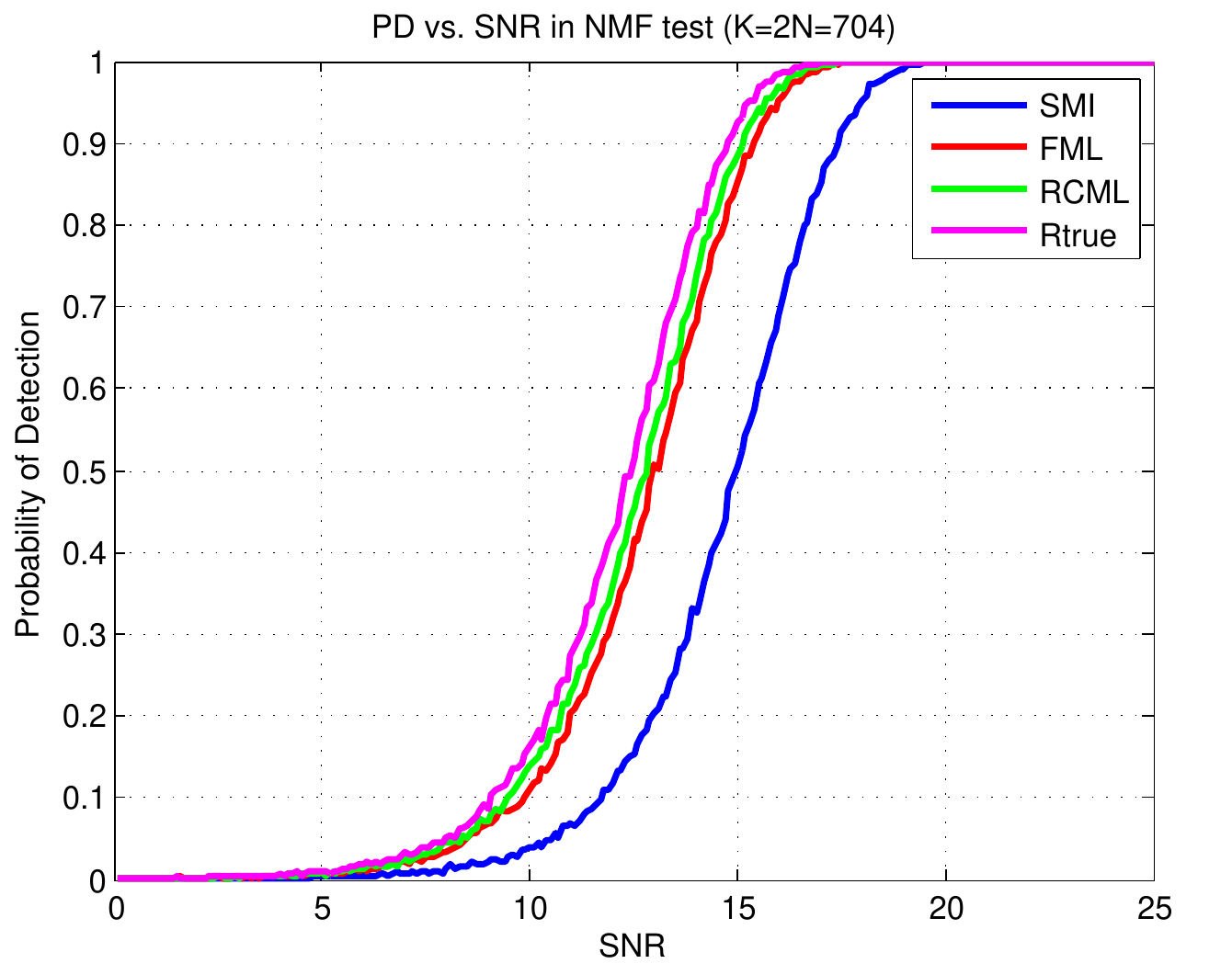}\label{Fig:NMF704}}\\
\subfigure[GLRT for $K=352$]{\includegraphics[width=2.4in]{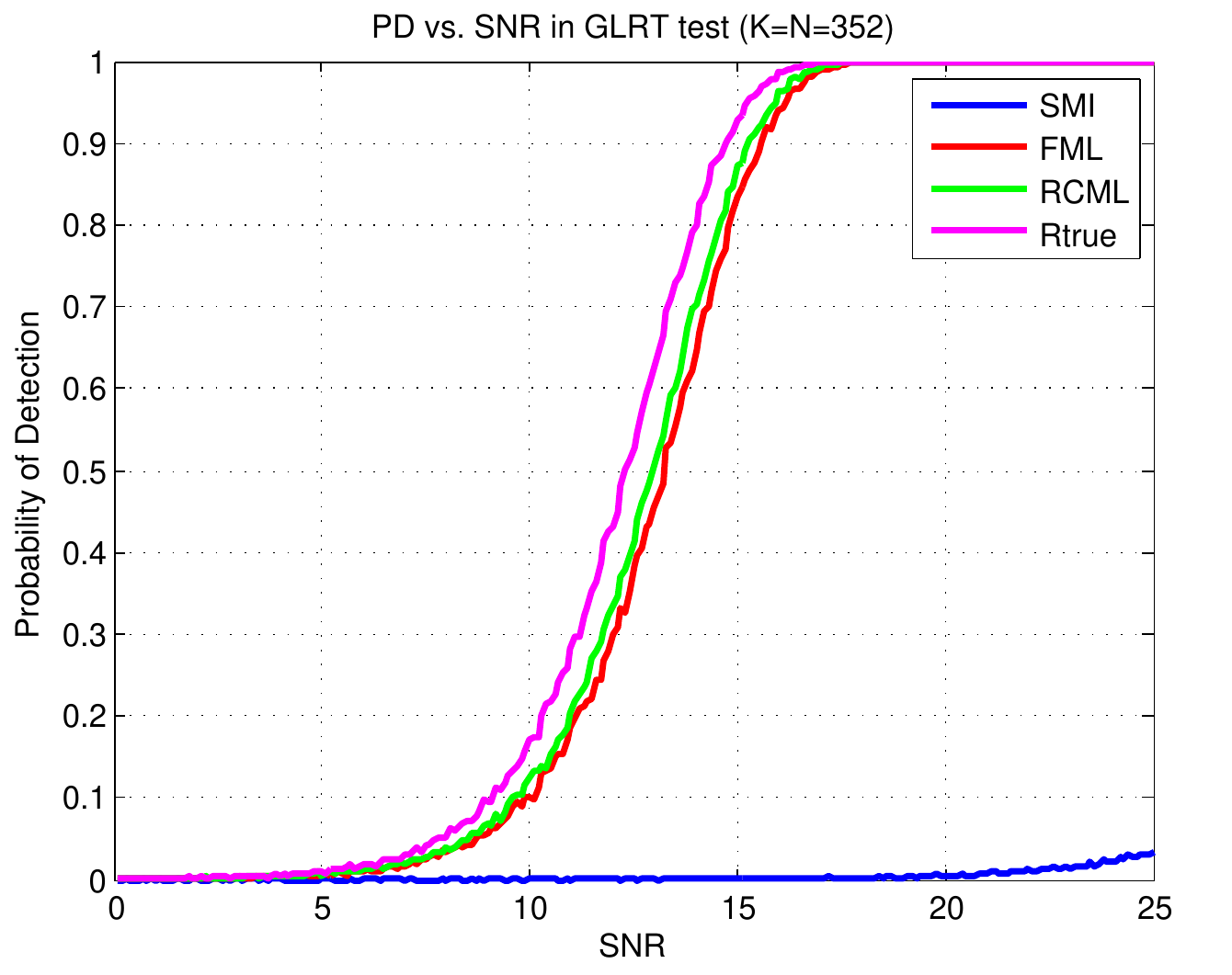}\label{Fig:GLRT352}}
\hfil
\subfigure[GLRT for $K=704$]{\includegraphics[width=2.4in]{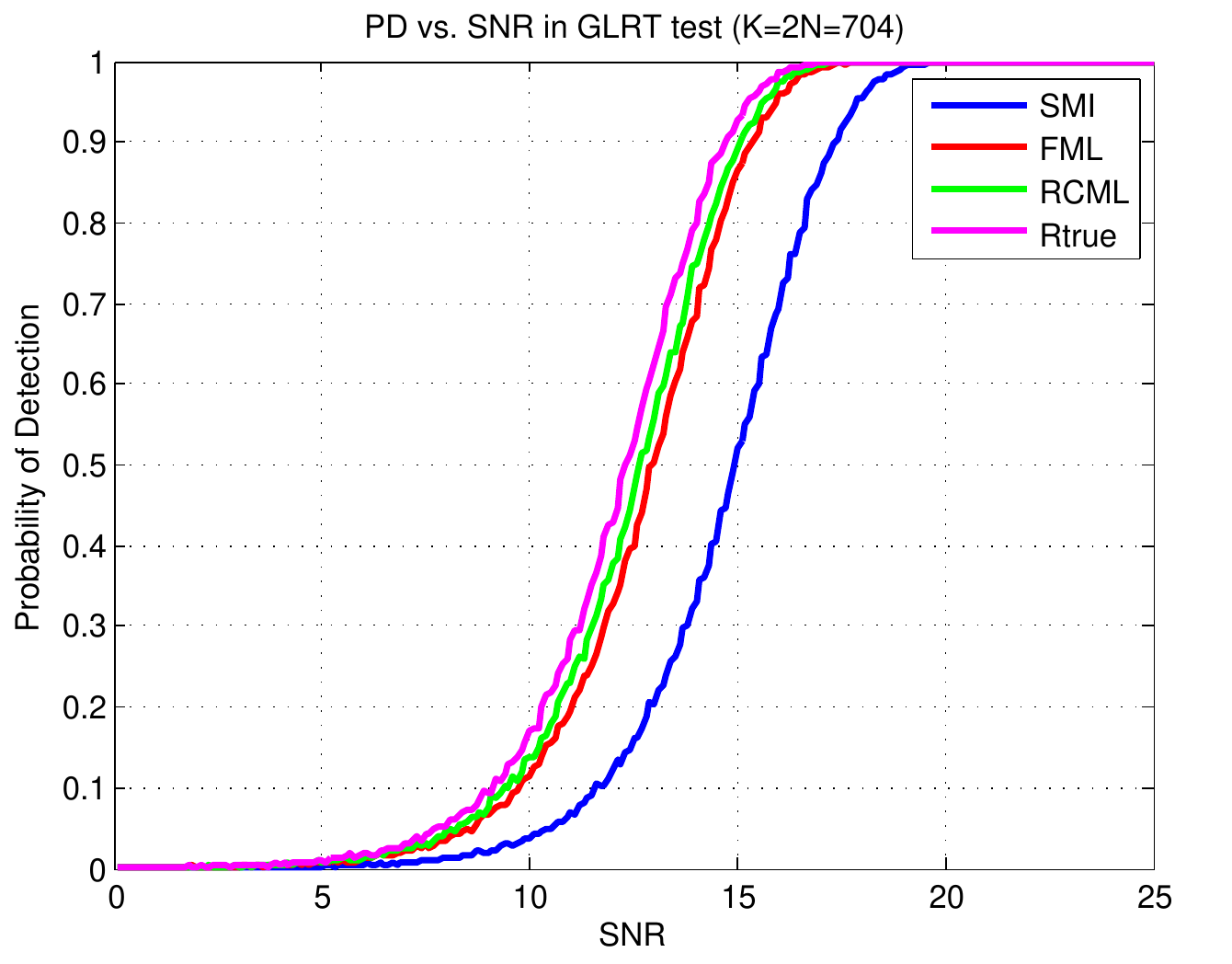}\label{Fig:GLRT704}}\\
\end{center}
\caption{Probability of detection vs. SNR.}
\label{Fig:PDvsSNR}
\end{figure}

We apply three test statistics, the normalized matched filter (NMF), the adaptive matched filter (AMF) \cite{Robey92}, and the generalized likelihood ratio test (GLRT) \cite{Kelly86}. The test statistics are given by
\bea
%\left\{ \begin{array}{cc}
\text{NMF:} & \dfrac{|\mb{s}^H \hat{\mb{R}}^{-1}\mb{e}|^2}{(\mb{s}^H \hat{\mb{R}}^{-1}\mb{s})(\mb{e}^H \hat{\mb{R}}^{-1}\mb{e})} \overset{H_1}{\underset{H_0}{\gtrless}} \lambda_{\text{NMF}}\label{Eq:NMF}\\
\text{AMF:} & \dfrac{|\mb{s}^H \hat{\mb{R}}^{-1}\mb{e}|^2}{\mb{s}^H \hat{\mb{R}}^{-1}\mb{s}} \overset{H_1}{\underset{H_0}{\gtrless}} \lambda_{\text{AMF}}\label{Eq:AMF}\\
\text{GLRT:} & \dfrac{|\mb{s}^H \hat{\mb{R}}^{-1}\mb{e}|^2}{\mb{s}^H \hat{\mb{R}}^{-1}\mb{s}\Big(1+\frac{1}{K}\mb{e}^H \hat{\mb{R}}^{-1}\mb{e}\Big)} \overset{H_1}{\underset{H_0}{\gtrless}} K \lambda_{\text{GLRT}}\label{Eq:GLRT} %\end{array} \right.
\eea
%\be
%\left\{ \begin{array}{cc}
%\text{NMF:} & \dfrac{|\mb{s}^H \hat{\mb{R}}^{-1}\mb{e}|^2}{(\mb{s}^H \hat{\mb{R}}^{-1}\mb{s})(\mb{e}^H \hat{\mb{R}}^{-1}\mb{e})} \overset{H_1}{\underset{H_0}{\gtrless}} \lambda_{\text{NMF}}\\
%\text{AMF:} & \dfrac{|\mb{s}^H \hat{\mb{R}}^{-1}\mb{e}|^2}{\mb{s}^H \hat{\mb{R}}^{-1}\mb{s}} \overset{H_1}{\underset{H_0}{\gtrless}} \lambda_{\text{AMF}}\\
%\text{GLRT:} & \dfrac{|\mb{s}^H \hat{\mb{R}}^{-1}\mb{e}|^2}{\mb{s}^H \hat{\mb{R}}^{-1}\mb{s}\Big(1+\frac{1}{K}\mb{e}^H \hat{\mb{R}}^{-1}\mb{e}\Big)} \overset{H_1}{\underset{H_0}{\gtrless}} K \lambda_{\text{GLRT}} \end{array} \right.
%\ee
where $\mb{s}$, $\hat{\mb{R}}$, $\mb{e}$, and $K$ are the steering vector, the estimated covariance matrix, the observation vector, and the number of training samples, respectively. The detection probability $P_{d}$ is defined as the probability that the value of test statistic is greater than a threshold conditioned on the hypothesis that the received data includes target information. Therefore, it depends on signal to noise ratio (SNR, by virtue of $\mb{s}$,) and the estimated covariance matrix.  Since $P_d$ does not typically admit a closed form, we first generate a number of samples from the L-band data set of KASSPER program to determine $\lambda$ corresponding to the fixed false alarm rate and then employ Monte Carlo simulations to evaluate $P_d$ corresponding to each estimator for each of the test statistics. We set a constant false alarm rate to $10^{-4}$.

Figure \ref{Fig:PDvsSNR} shows the detection probability $P_{d}$ plotted as a function of SNR  for different estimators and detection statistics. Figures \ref{Fig:AMF352} and \ref{Fig:AMF704} plot $P_{d}$ for AMF test, Figures \ref{Fig:NMF352} and \ref{Fig:NMF704} are corresponding to the NMF test, and Figures \ref{Fig:NMF352} and \ref{Fig:NMF704} plot results for the GLRT. We use $K=N=352$ and $K=2N=704$ training samples to estimate the covariance matrix for each of the test statistics. Figures \ref{Fig:AMF352}, \ref{Fig:NMF352}, and \ref{Fig:GLRT352} are for $K=352$ and Figures \ref{Fig:AMF704}, \ref{Fig:NMF704}, and \ref{Fig:GLRT704} are for $K=704$. It is well-known that $K=2N$ training samples are needed to keep the performance within 3dB. Indeed, we can see that the sample covariance matrix has about 3dB loss vs. the true covariance matrix in all of test statistics. The proposed RCML estimator is the closest to the $P_d$ achieved by using the true covariance matrix (upper bound) and FML follows RCML. As expected, each estimator shows higher detection probability when $K=2N$ vs. $K=N$, i.e. an increase in training. Note finally that the RCML estimator performs the best no matter which test statistic is applied and in every regime of training.

\begin{figure}
\begin{center}
\subfigure[$K=352$]{\includegraphics[width=2.4in]{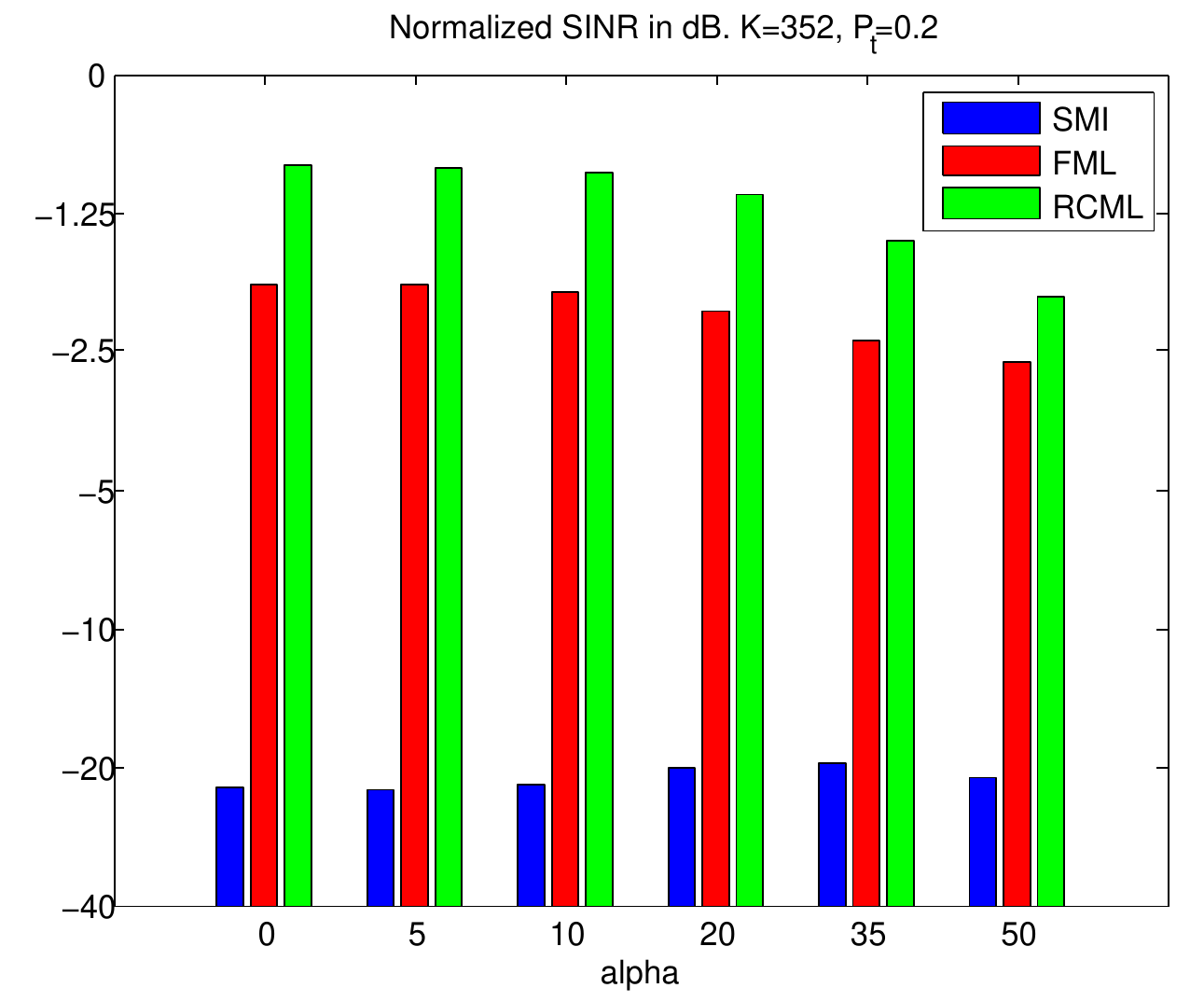}\label{Fig:SINR70of352}}
\hfil
\subfigure[$K=352$]{\includegraphics[width=2.4in]{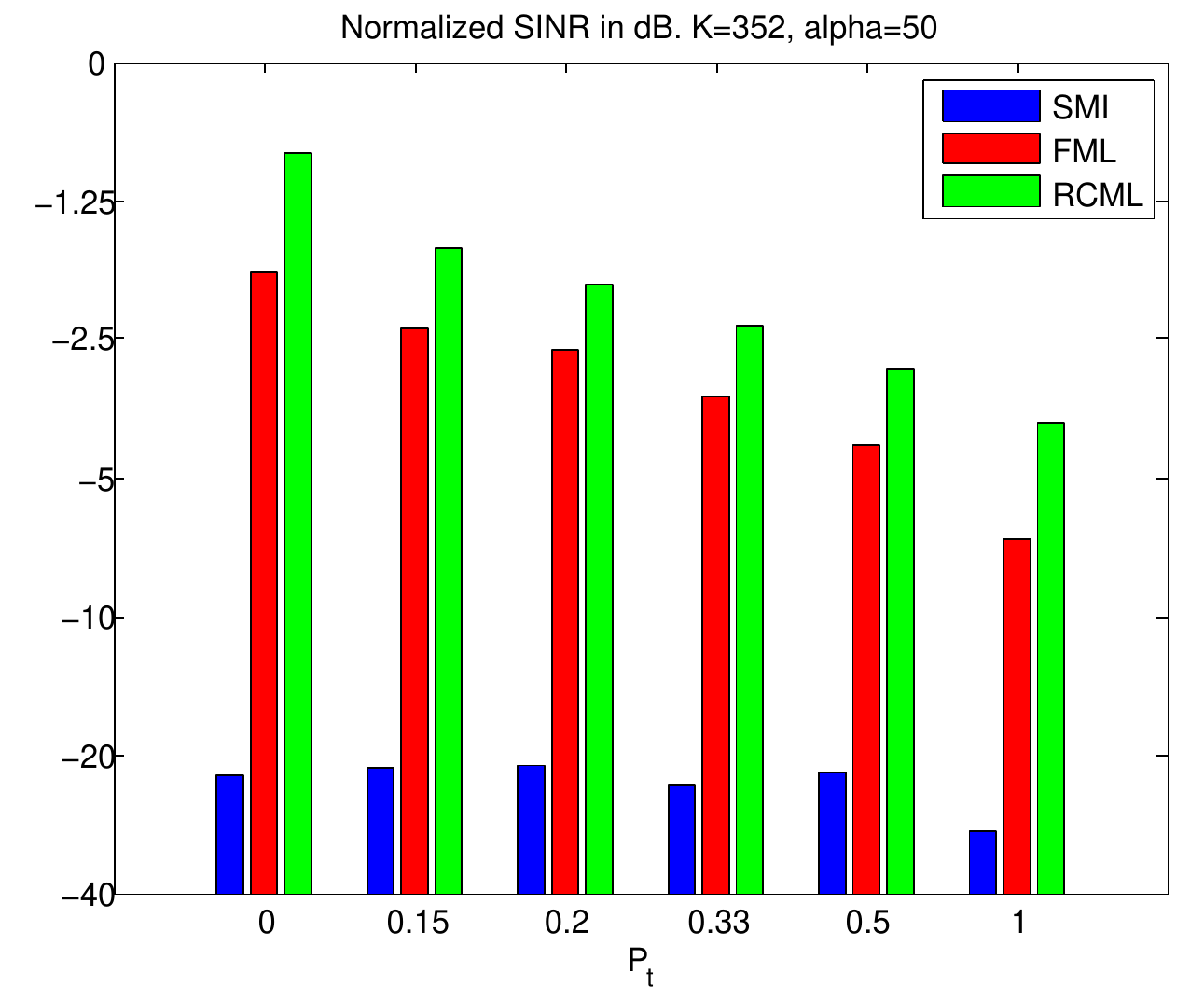}\label{Fig:SINR352w50}}\\
\subfigure[$K=704$]{\includegraphics[width=2.4in]{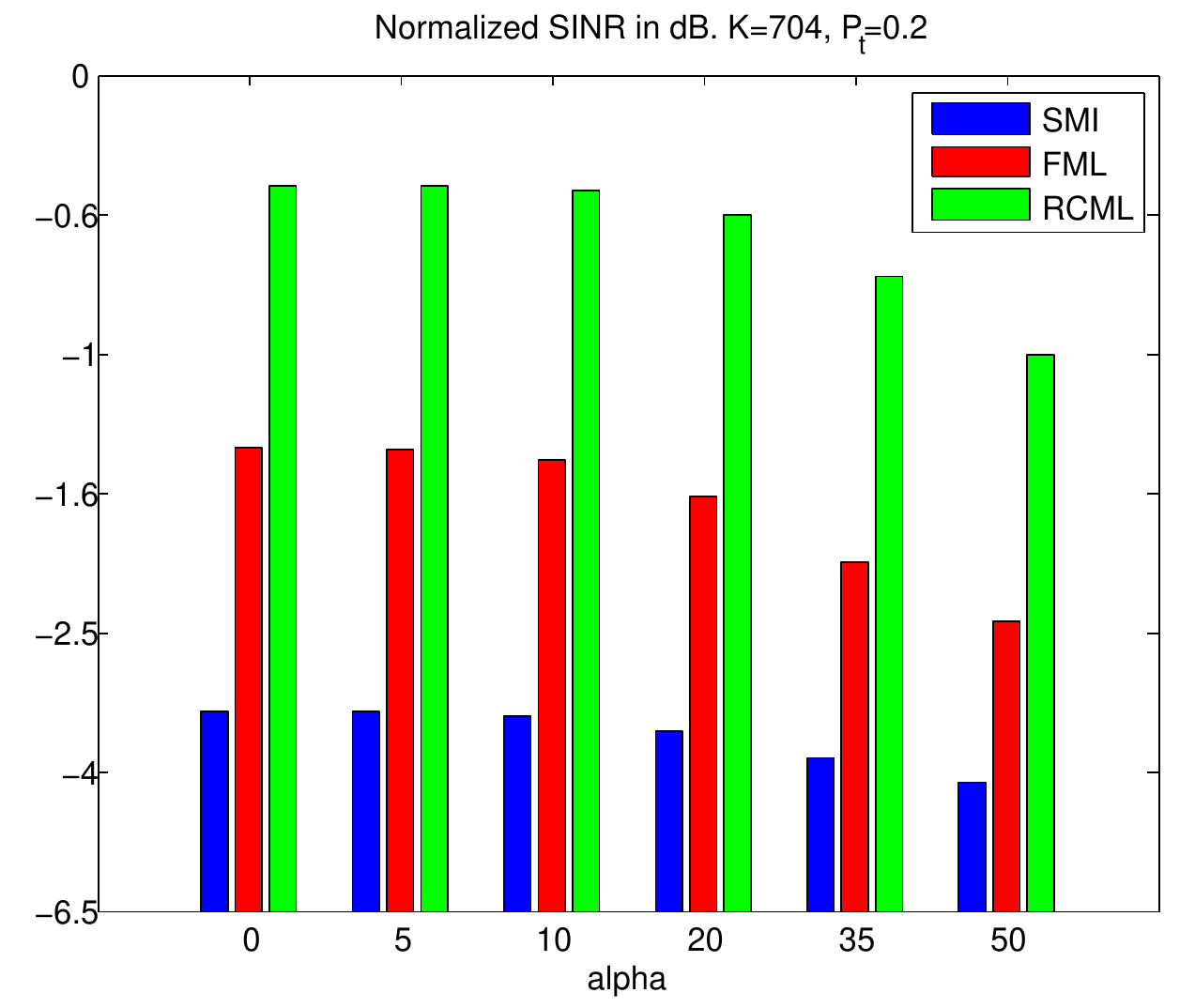}\label{Fig:SINR141of704}}
\hfil
\subfigure[$K=704$]{\includegraphics[width=2.4in]{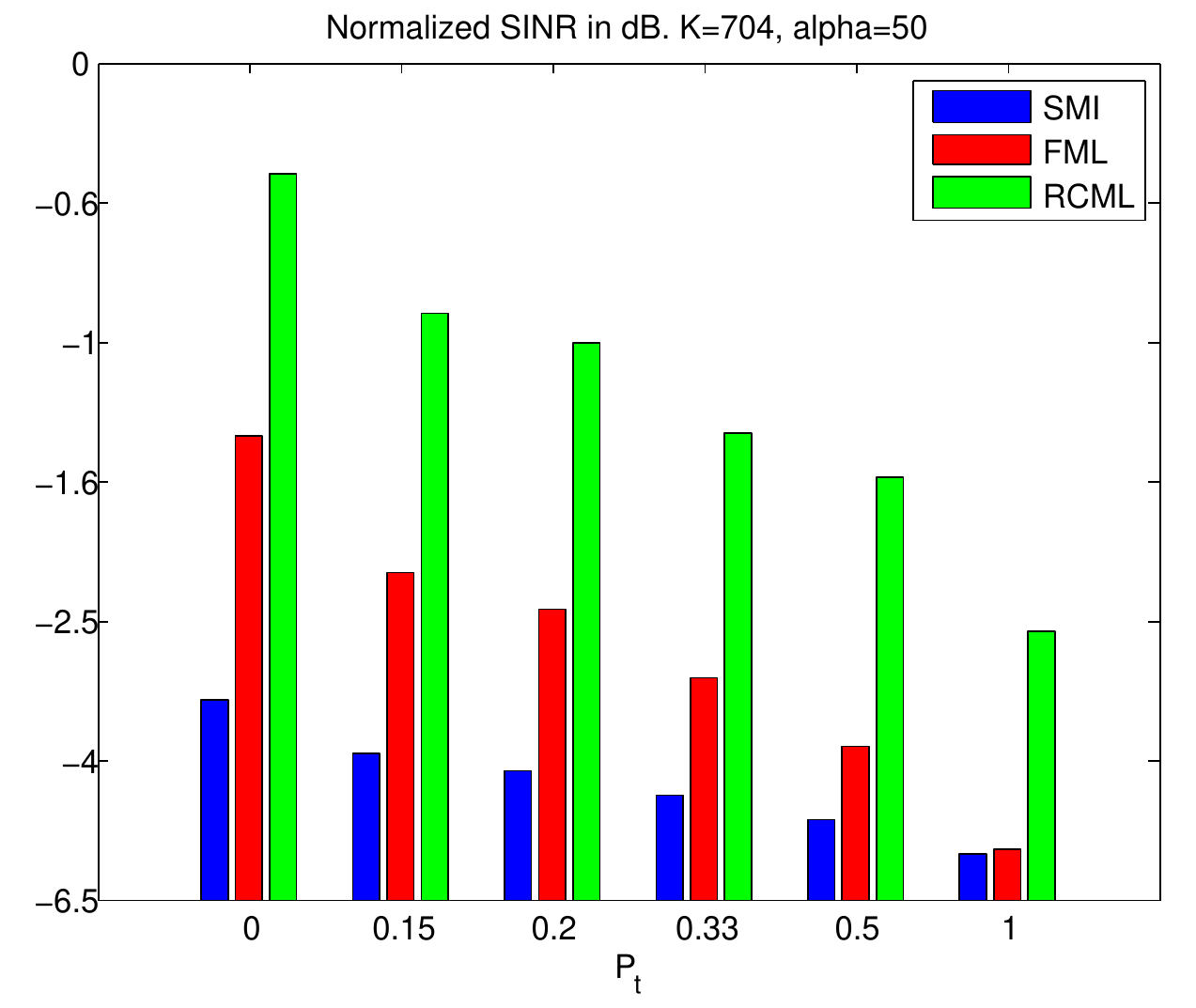}\label{Fig:SINR704w50}}\\
\end{center}
\caption{Normalized SINR vs. target intensity $\alpha$ and percentage corruption $P_t$.}
\label{Fig:SINRnonhomogeneous}
\end{figure}

\subsubsection{Complexity Comparison}
\begin{table}[!t]
\begin{center}
\caption{Running time (sec) for various estimators for KASSPER dataset. FML, RCML, and EigC involve eigenvalue decomposition which takes 0.0497 sec.}
\label{Tb:ComplexityRCML}
\begin{tabular}{|c|c|c|c|}
  \hline
  $K$ & 352 & 528 & 704\\
  \hline
  SMI & 0.0153 & 0.0241 & 0.0294\\
  \hline
  FML & 0.2877 & 0.3121 & 0.3292\\
  \hline
  RCML & 0.1054 & 0.1216 & 0.1311\\
  \hline
  EigC & 0.2853 & 0.3178 & 0.3319\\
  \hline
  LOOC & 12.0666 & 13.0476 & 14.5000\\
  \hline
\end{tabular}
\end{center}
\end{table}

We compare computational complexity of the compared methods. Table \ref{Tb:ComplexityRCML} shows running times in second for SMI, FML, RCML, eigencanceler, and LOOC. We take average values of results of 100 trials for each estimator and the experiments are performed on the desktop with Intel Core i7-2600 CPU 3.40 GHz and 8.00 GB RAM. As shown in the table, running times increase for all methods as the number of training samples increase. SMI is fastest is all training regimes as expected and RCML is the second best in terms of running time. FML and eigencanceler show running time close to each other. Please note that FML, RCML, and eigencanceler involve eigenvalue decomposition and running time of eigenvalue decomposition in this experiment is 0.0492 second. This result confirms that the RCML estimator not only outperforms the other estimators in the sense of the normalized SINR but also is even computationally cheaper than alternatives.

\subsubsection{Robustness to Nonhomogeneous Training Samples}
We investigate two different scenarios to evaluate robustness to nonhomogeneous training samples. First, we fix the ratio of the number of corrupted samples including target information to target-free samples. This ratio is given by
\be
P_t = \dfrac{\text{the number of corrupted samples by target information}}{K \text{(= the number of total training samples)}}
\ee
and the intensity of target signal by $\alpha$, that is, the received data $\mb{z}$ can be expressed by
\be
\mb{z} = \alpha\mb{s}(\theta_t,f_t) + \mb{d}
\ee
where $\mb{d} = \mb{c} + \mb{j} + \mb{n}$ represents the overall disturbance which is the sum of $\mb{c}$, clutter, $\mb{j}$, jammers, $\mb{n}$, the background white noise, and comes from a zero-mean complex circular Gaussian distribution. $\mb{s}$ is a known spatio-temporal steering vector \cite{Wicks06} which is drawn from a distribution independent of $\mb{d}$. In particular, we examine performance as the percentage of corrupted samples, i.e., $P_t$ is varied while keeping a fixed intensity of the target signal, $\alpha$. Our second investigation involves varying $\alpha$ for a fixed $P_t$.

We use two evaluation measures: the normalized SINR and a trace deviation measure, TRD($\hat{\mb{R}}$). Figure \ref{Fig:SINRnonhomogeneous} presents bar graphs that show averaged $SINR_{dB}$ results for $K=352$ and $K=704$ training samples. Because the steering vector is a function of both azimuthal angle and Doppler frequency, we evaluate the normalized SINR in both angle and Doppler domain and average over both domains to get the normalized SINR value represented by each bar. Figures \ref{Fig:SINR70of352} and \ref{Fig:SINR352w50} are corresponding to $K=N=352$ and Figures \ref{Fig:SINR141of704} and \ref{Fig:SINR704w50} plot results for $K=2N=704$. In particular, Figures \ref{Fig:SINR70of352} and \ref{Fig:SINR141of704} plot the variation of the normalized SINR for varying intensity of the steering vector $\alpha$, where $\alpha$ is varied from as low as $0$ to as high as $50$. We fixed $P_t = 0.2$ in these plots. Two trends are evident from Figures \ref{Fig:SINR70of352} and \ref{Fig:SINR141of704}: 1.) as intuitively expected, the SINR values decreases monotonically with an increase in $\alpha$ for all methods (except for the sample covariance technique in the $K=N$ regime) and 2.) the RCML estimator exhibits appreciably good performance in all training regimes. Figures \ref{Fig:SINR352w50} and \ref{Fig:SINR704w50} plot the SINR performance for varying $P_t$ where $\alpha$ remains a constant, $\alpha = 50$. The range of $P_t$ is from $0$ (no target corruption) to $1$ (all the samples are corrupted by target information). Similar trends are observed as well in Figures \ref{Fig:SINR352w50} and \ref{Fig:SINR704w50}. Again, the RCML estimator consistently outperforms the other methods. An interesting observation is that $SINR_{dB}$ drops more rapidly as a function of increasing $P_t$ vs. increasing $\alpha$, which reveals that $P_t$ is a more critical factor than $\alpha$ in influencing estimation with heterogeneous training.

 We define a trace deviation measure, TRD($\hat{\mb{R}}$) = $|tr\{\mb{R}\hat{\mb{R}}^{-1}\}/N - 1|$ that is an alternate way of evaluating the performance of covariance matrix estimators. Intuitively, we can see $tr\{\mb{R}\hat{\mb{R}}^{-1}\}/N = 1$ when $\hat{\mb{R}} = \mb{R}$. Therefore, we can say the goal of estimation is to keep TRD($\hat{\mb{R}}$) as small as possible, ideally close to $0$. Figure \ref{Fig:TraceNonhomogeneous} shows plots bar graphs in the same training regime as Figure \ref{Fig:SINRnonhomogeneous}. We plots values of TRD($\hat{\mb{R}}$) for varying $\alpha$ and $P_t$ and the number of training samples are $K=352$ and $K=704$.

We can observe trends similar to those in Figure \ref{Fig:SINRnonhomogeneous}. The TRD values monotonically increase as $\alpha$ and $P_t$ increase for all methods. The proposed RCML estimator consistently outperforms other techniques considered in all experiments. Additionally, the merits of RCML in robust estimation are brought out. The TRD values corresponding to both the sample covariance matrix and the FML estimator increase quite dramatically with an increase in $\alpha$ and, especially, $P_t$. However, in the case of the RCML estimator this increase is more gradual. The TRD values corresponding to RCML in Figure \ref{Fig:Trace704w50} are in fact still close to  $0$ even under severe target corruption, i.e.\ $P_t =1$.

\begin{figure}
\begin{center}
\subfigure[ $K=352$]{\includegraphics[width=2.4in]{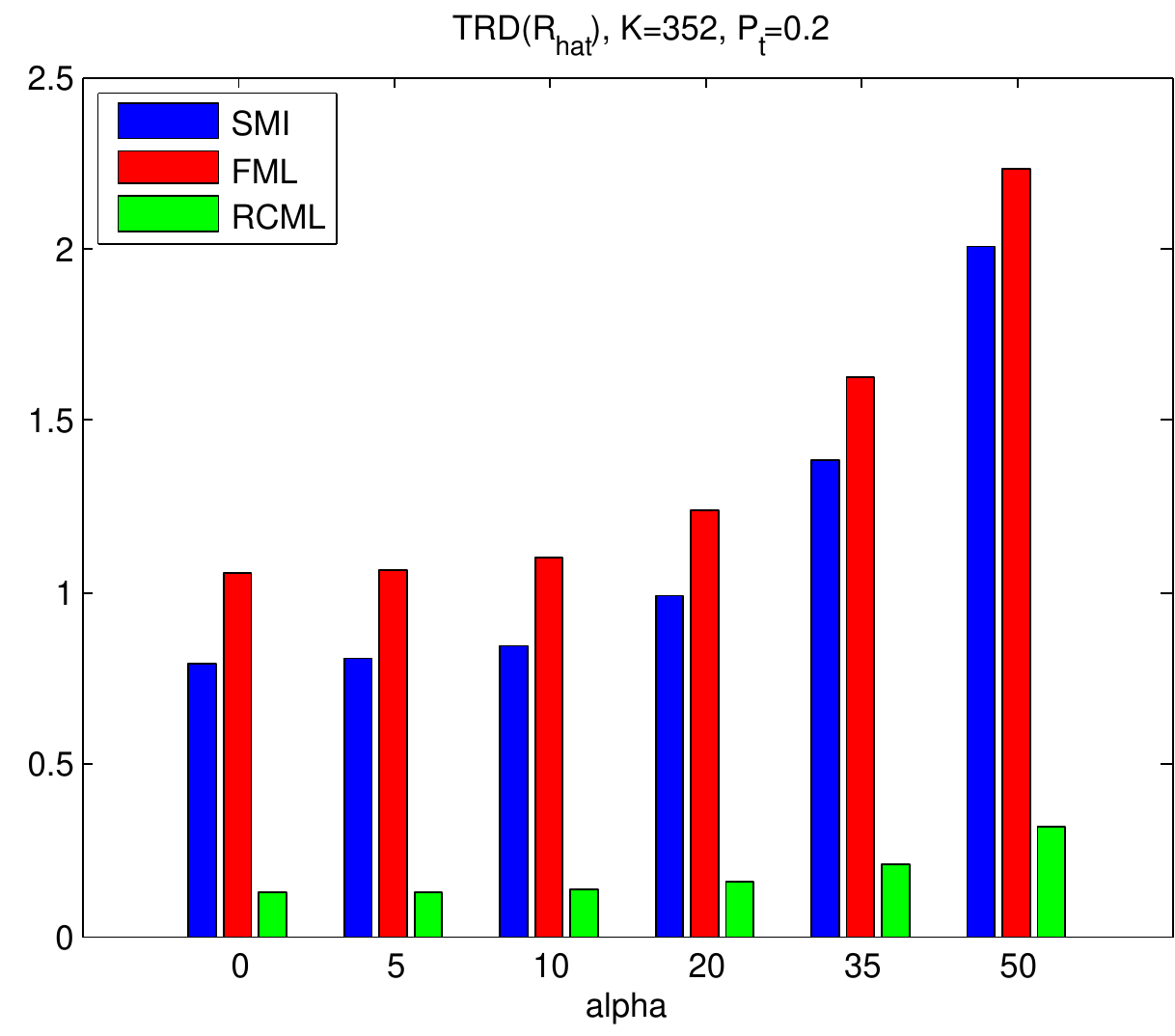}\label{Fig:Trace70of352}}
\hfil
\subfigure[ $K=352$]{\includegraphics[width=2.4in]{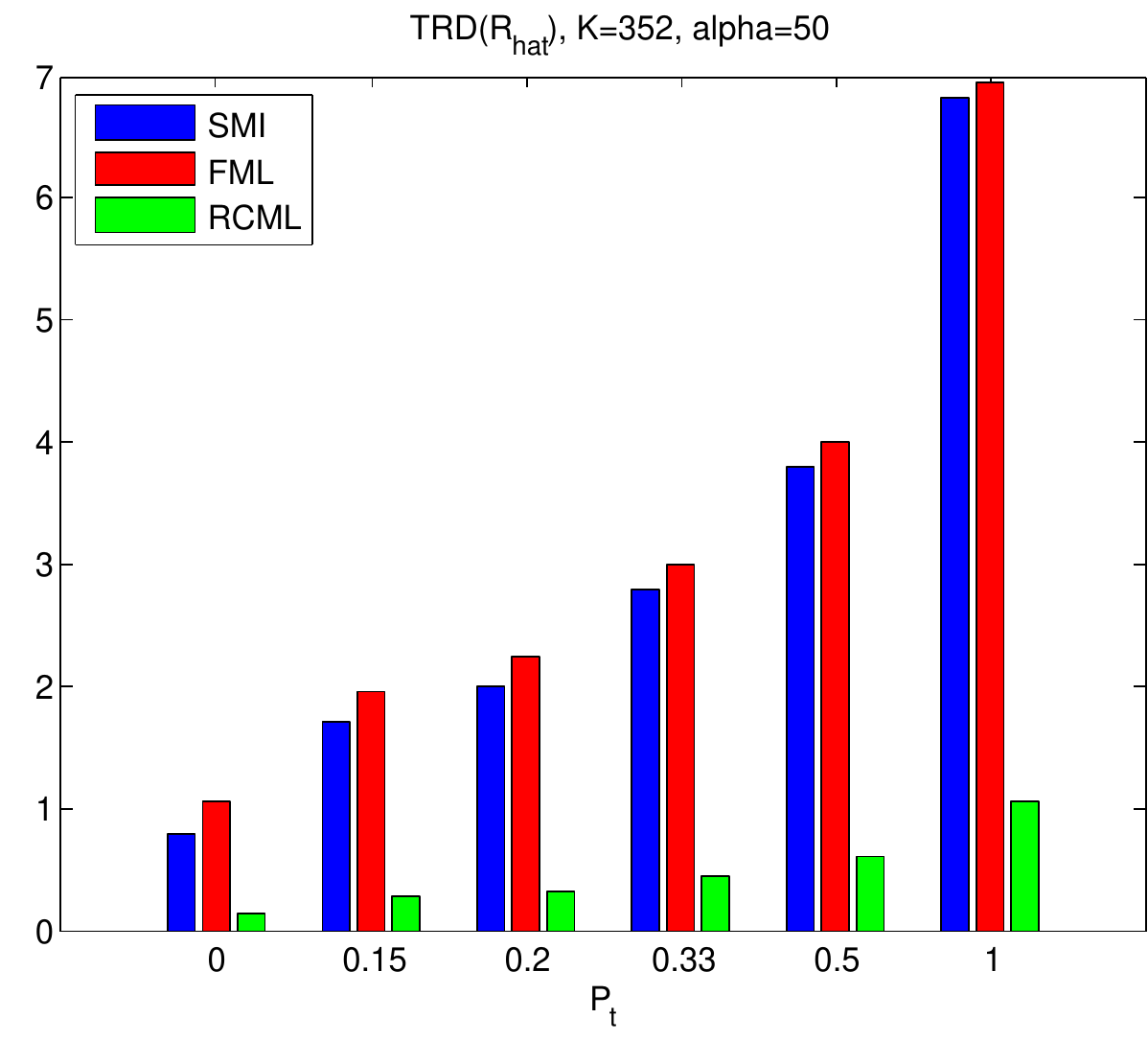}\label{Fig:Trace352w50}}\\
\subfigure[$K=704$]{\includegraphics[width=2.4in]{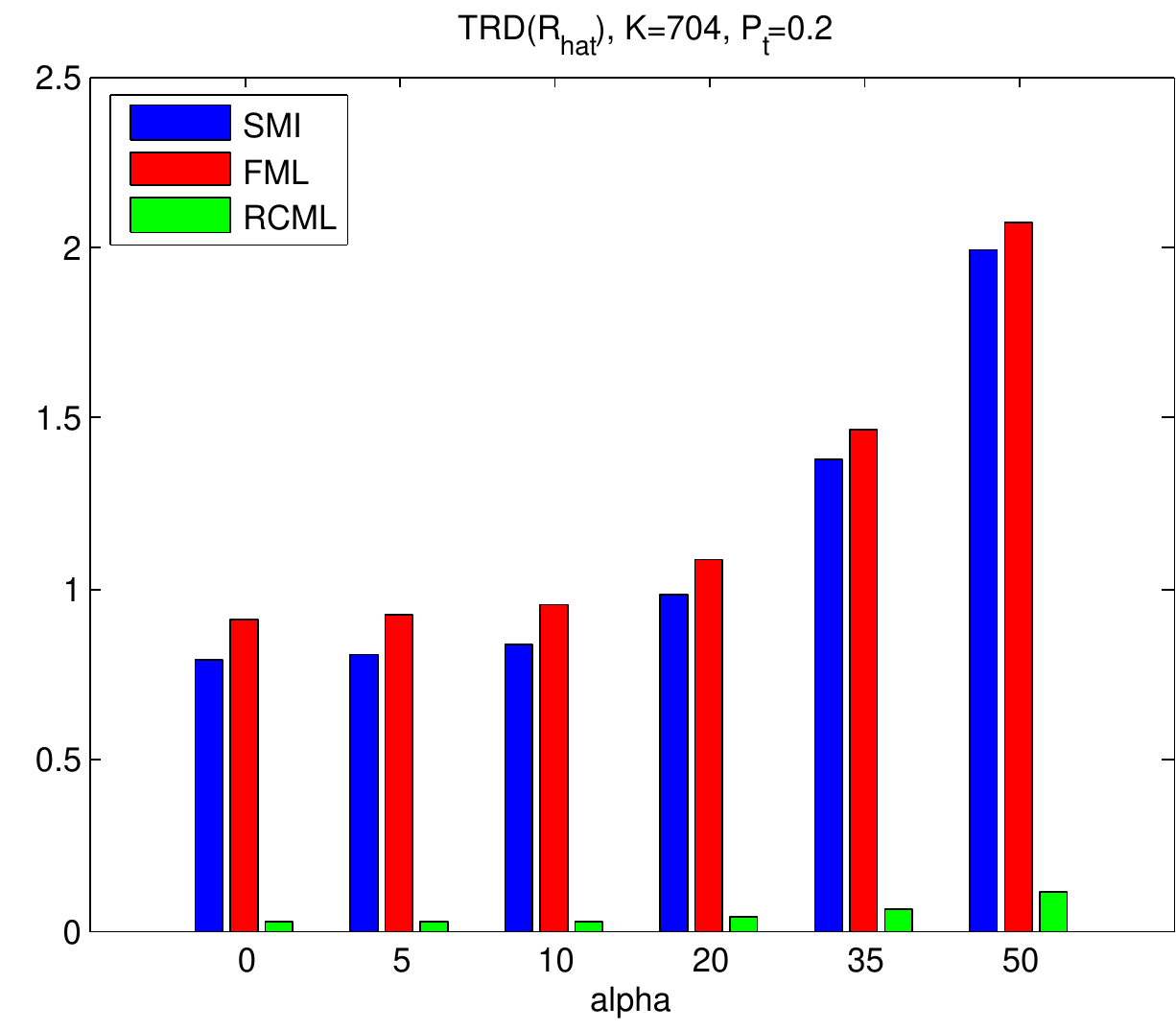}\label{Fig:Trace141of704}}
\hfil
\subfigure[$K=704$]{\includegraphics[width=2.4in]{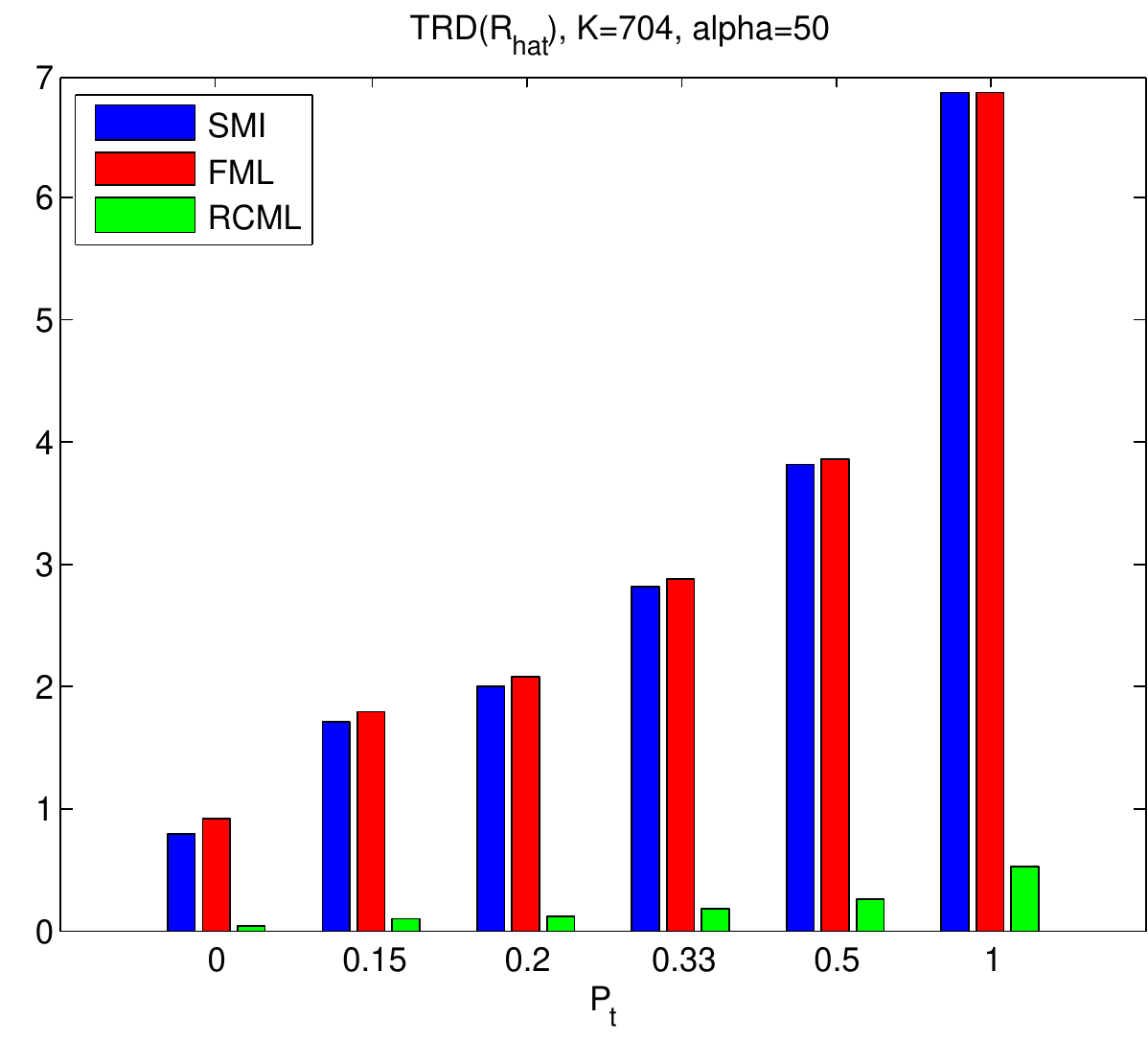}\label{Fig:Trace704w50}}\\
\end{center}
\caption{TRD($\hat{\mb{R}}$) vs. target intensity $\alpha$ and percentage corruption $P_t$.}
\label{Fig:TraceNonhomogeneous}
\end{figure}

\subsection{RCML vs. RCML$_{\text{LB}}$ and Wax and Kailath Estimator}
\label{Sec:Wax}
\begin{figure}
\begin{center}
\subfigure[$K=30$]{\includegraphics[width=2.5in]{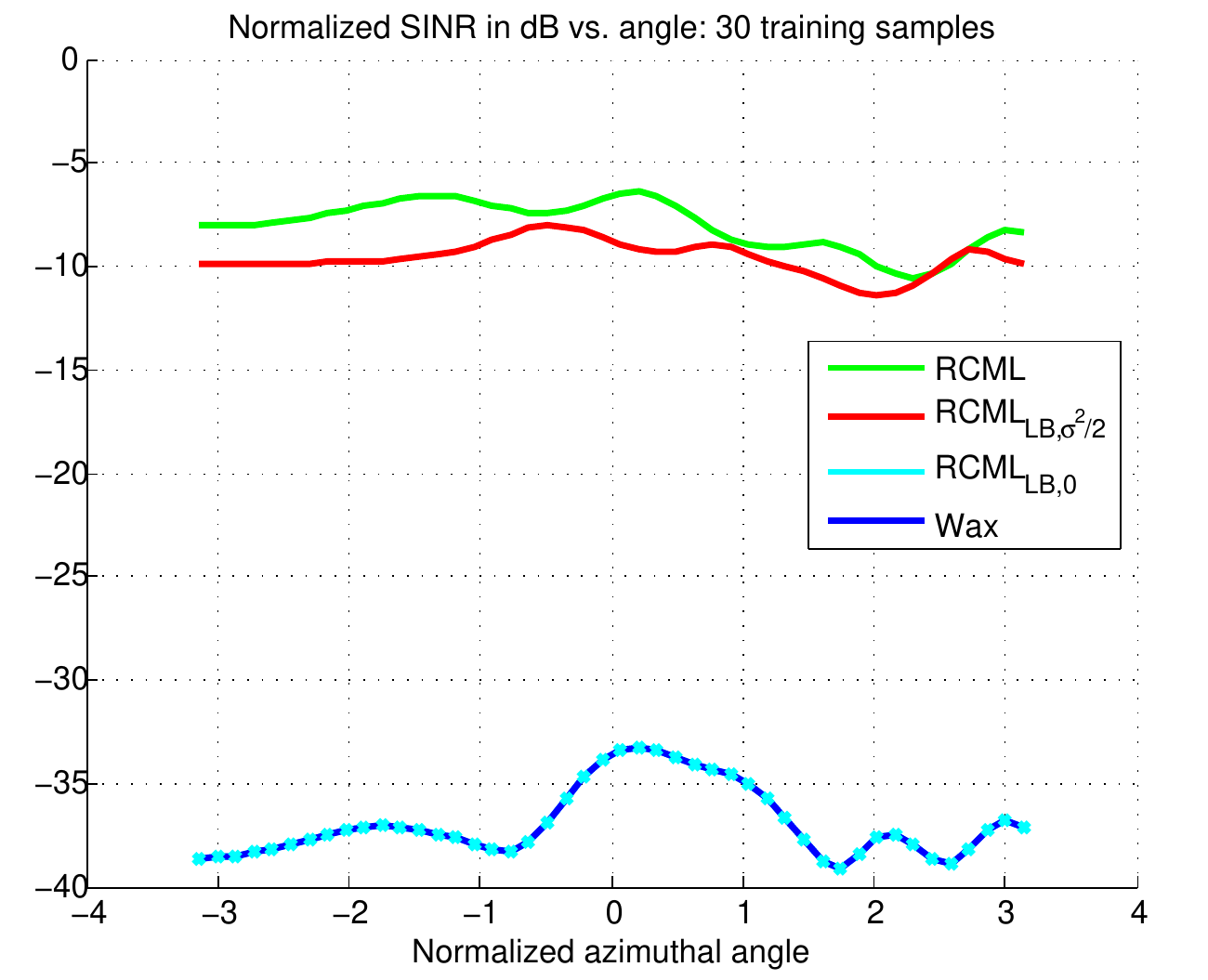}\label{Fig:SINR_Wax_30}}
\hfil
\subfigure[$K=300$]{\includegraphics[width=2.5in]{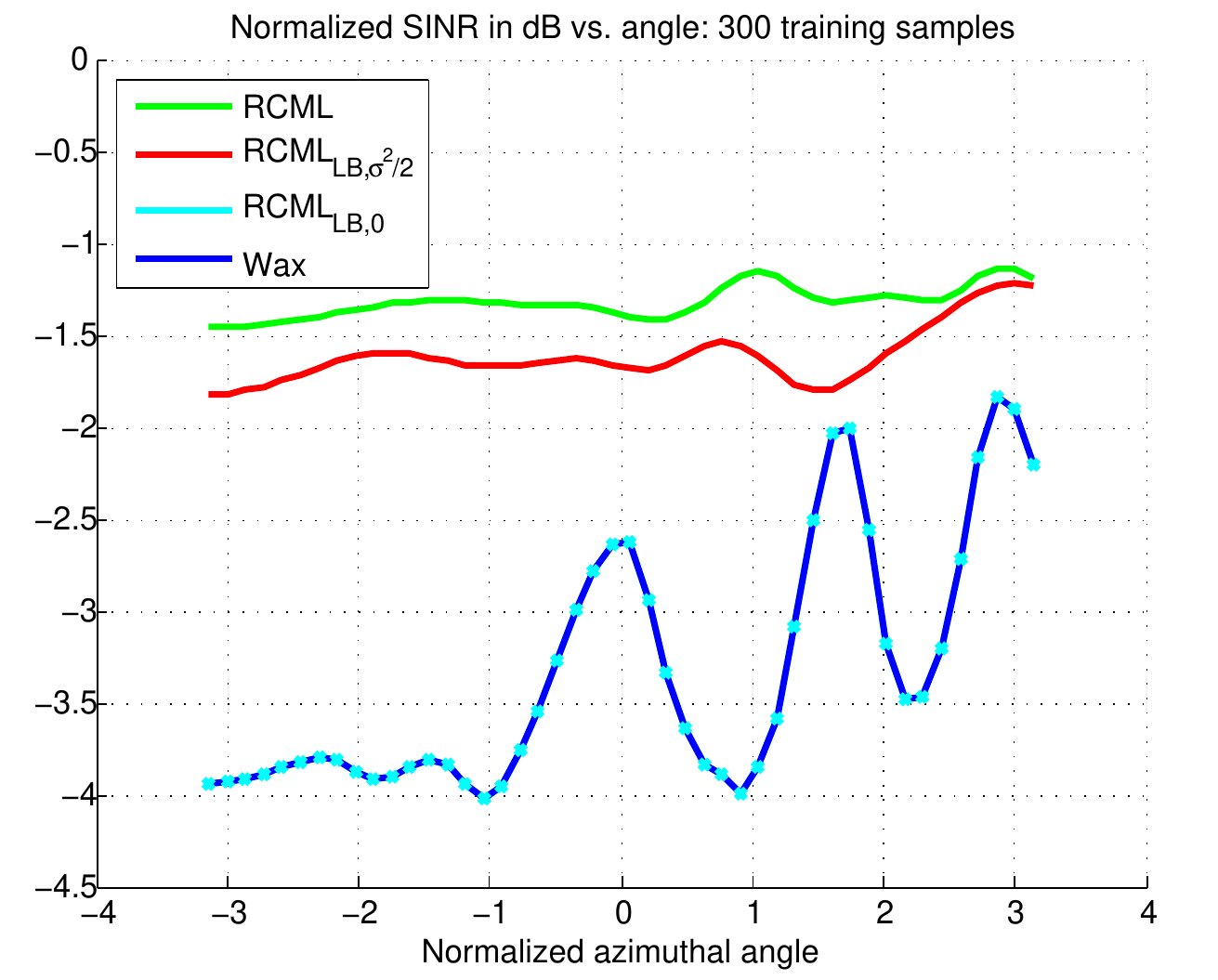}\label{Fig:SINR_Wax_300}}\\
\end{center}
\caption{Normalized SINR vs. normalized azimuthal angle.}
\label{Fig:SINR_Wax}
\end{figure}

Following Anderson's result \cite{Anderson63} in statistics, Wax and Kailath \cite{Wax85} reported an ML estimator of a structured covariance estimator also under the rank-constraint as follows:
\bea
\hat{\lambda}_i  & = & \bar{\lambda}_i\label{Eq:WaxLambda}\\
\hat\sigma^2 & = & \frac{1}{N-r}\sum_{i=r+1}^N \bar{\lambda}_i\label{Eq:WaxNoise}\\
\hat{\mb\Phi}_i & = & \mb V_i\label{Eq:WaxEigenvector}
\eea

where $ \bar{\lambda}_{1} \geq \bar{\lambda}_{2} \geq \cdots \geq \bar{\lambda}_{N}$ denote the eigenvalues of the sample covariance matrix. It is easy to see that the RCML$_{\text{LB}}$ estimator in Eq. (\ref{eqn:optimalLabmdainoptimalc}) is a generalization of this result by employing a lower bound on the noise floor.

To further emphasize the value of RCML in Eq. (\ref{eqn:RCML}) is the estimator of choice for practical radar STAP, we now perform an experimental comparison of RCML, RCML$_{\text{LB}}$ and Wax and Kailath \cite{Wax85} estimator in the challenging low training regime.
Figure \ref{Fig:SINR_Wax} shows the performance of RCML, $\text{RCML}_{\text{LB}}$, and Wax and Kailath estimators for $K = 30$ and $K=300$ training samples, respectively. Two versions of RCML$_{\text{LB}}$ are reported with the lower bound $\hat c$ set to $\sigma^2/2$ and $0$ and the estimators labeled as $\text{RCML}_{\text{LB},\sigma^2/2}$ and $\text{RCML}_{\text{LB},0}$ in the plots. Figures \ref{Fig:SINR_Wax_30} and \ref{Fig:SINR_Wax_300} clearly reveal that: 1.) RCML is clearly the best estimator, while the Wax and Kailath estimator is about $30$ dB worse for $K=30$ and about $3$ dB below for $K=300$ training samples. 2.) the knowledge of a lower bound helps RCML$_{\text{LB}}$ in that $\text{RCML}_{\text{LB},\sigma^2/2}$ is a better estimator than $\text{RCML}_{\text{LB},0}$, which in fact overlaps with the the Wax and Kailath estimator.

\section{Conclusion}
\label{conclusion}

We developed a new estimator of structured covariance matrices (identity plus a positive semi-definite component) which employs rank of the positive semi-definite matrix as an explicit constraint in ML estimation. In radar applications, the rank-deficient component corresponds to the clutter and its rank can be determined using the Brennan rule for airborne radar interacting with land clutter. We demonstrated that despite the presence of the challenging rank-constraint, the estimation problem can in fact be reduced to a convex optimization problem and admits a closed form solution. Experimentally, rank information plays a vital role and rigorous evaluation over the KASSPER data set establishes merits of the proposed estimator when evaluated via widely used figures of merits such as normalized SINR. Future work could consider the incorporation of more constraints on the clutter/disturbance matrix such as Toeplitz structure as well as the use of physically inspired probabilistic  priors in a Bayesian setting. 
\chapter{Computationally Efficient Toeplitz Approximation under a Rank Constraint}
\label{Ch:EASTR}

\section{Introduction}
\label{Sec:Introduction}
% no \PARstart
Radar systems using multiple antenna elements that coherently process multiple pulses offer significant benefits in many applications. The directivity and resolution limits of a single sensor can be overcome by using an adaptive array of spatially distributed sensors makes multiple temporal snapshots processing possible. Specifically, joint adaptive processing in the spatial and temporal domains \cite{Monzingo04,Guerci03,Klemm02} called space time adaptive processing (STAP) creates an ability to suppress interference signals while simultaneously preserving gain on the desired signal. For STAP to be successful though, interference statistics, in particular the covariance matrix of the disturbance or interference must be estimated from target free training data, and therefore training plays a pivotal role in adaptive radar systems.

To obtain accurate estimates of the disturbance covariance matrix, a large number of homogeneous (target free) disturbance training samples are required in the absence of any prior knowledge about the interference environment. A compelling challenge for radar STAP emerges since generous homogeneous training is often not available in practice \cite{Himed97}. This problem is exacerbated because the estimation process must be repeated for each range bin of interest. Much recent research in radar STAP has been proposed to overcome the lack of generous homogeneous training. One approach to this problem uses \emph{a priori} information about the radar environment and is widely referred to in the literature as knowledge-based processing \cite{Guerci06,Capraro06,Haykin06,Wicks06,Miranda06,DeMaio09,DeMaio10,Aubry13SP}. A subset of this technique deals with intelligent training selection for reducing both the number of required training samples and computational cost \cite{Wicks06,Wang91,Gini08}. Another approach to improve the target detection performance is data selection screening among the training data to excise potential outliers \cite{Chen99,Aubry13IET}.

Covariance matrix estimation techniques that enforce and exploit specific structure inherent to the disturbance phenomenon have merit in the regime of extremely limited training data. Examples of structure include persymmetry \cite{Nitzberg80}, eigenstructure \cite{Steiner00,Haimovich96}, circulant structure \cite{Conte98}, rank constraint \cite{Monga12,Kang13}, multichannel autoregressive models \cite{Roman00,Wang09}, physical constraints \cite{Kraay07} and so on. In particular, since the covariance matrix from a stationary stochastic signal is Hermitian and Toeplitz, estimating Toeplitz covariance benefits many applications such as array processing and time series analysis. Such a Hermitian Toeplitz matrix models the covariance of a random vector obtained by sampling a wide sense stationary noise field with a uniform linear array and uncorrelated narrow-band interferers \cite{Fuhrmann91}. The seminal work by Burg \emph{et al.} \cite{Burg82} proposed an {\em iterative} method for estimation of structured covariance matrices using the ML method in its full generality . Li \emph{et al.} developed the asymptotic maximum likelihood (AML) estimation for structured covariance matrices \cite{Li99} using the extended invariance principle (EXIP) \cite{Stoica89}. Approximation of arbitrary matrices by a (Hermitian) Toeplitz matrix using matrix decompositions and outer approximations has separately been pursued in applied mathematics \cite{AlHomidan02,Davis98,Shaw98,Suffridge93}. While the techniques in \cite{AlHomidan02,Davis98,Shaw98,Suffridge93} were not conceived for signal processing or radar STAP, they can potentially be used in conjunction with classical covariance estimation. Of particular interest is Al-Homidan's $l_1$ sequential quadratic programming (SQP) method to find the nearest symmetric positive semi-definite Toeplitz matrix to given a matrix \cite{AlHomidan02}.

\subsection{Motivation and Challenges}
Various estimation and approximation techniques of Toeplitz covariance matrices have been proposed \cite{Miller87,Little02,Forster89,Jansson00}. It is well known \cite{Fuhrmann91} though that there is no closed-form solution for the ML estimation of a Hermitian Toeplitz covariance matrix.
Many Toeplitz covariance estimation techniques need the assumption of large sample size (i.e.\ observed training) for computational tractability \cite{Li99},\cite{Jansson00}. In the regime of realistic training, methods rely on numerical optimization (often non-convex), are computationally involved and hence unsuitable for real-time/practical deployment.

Previous works, notably in statistics \cite{Tipping99,Stoica09,Anderson63,Wax85} (and references therein) have also shown that the rank of the structured interference can be exploited in a tractable manner. Rank is a powerful constraint in covariance estimation and can often be determined via underlying radar physics. Under nominal assumptions, the Brennan rule \cite{Ward94} may be used to determine the rank of the structured interference. Related work also addresses the problem of determining rank in non-ideal scenarios \cite{Goodman07}. Recently, Kang \emph{et al.} proposed the rank constrained ML (RCML) estimation of structured covariance matrices \cite{Kang14} which exploits the knowledge of the radar noise floor. Kang {\em et al.} \cite{Kang14} also report another estimator called RCML$_{\text{LB}}$ for the case when the noise floor is assumed unknown and only a lower bound (LB) is available. The RCML$_{\text{LB}}$ estimator generalizes the well-known result in statistics \cite{Anderson63}. In the radar context though, the noise variance is assumed known since it can be determined by placing the radar in receive only mode \cite{Skolnik08}. Notable contributions which deal with both the rank information and Toeplitz structure of the covariance matrix jointly include the iterated Toeplitz approximation method (ITAM) \cite{Wilkes88} proposed by Wilkes and Hayes and the iterative approach by Forster {\em et al.} \cite{Forster89}. Both approaches are based on a computationally expensive iterative procedure. The ITAM estimator in particular has been shown to be effective under very low training because of its ability to exploit structure but does not yield scalable performance improvements as realistic or generous training is made available.

\subsection{Our contributions}

It may be inferred that for adequate performance under limited training, computationally involved estimators such as ITAM \cite{Wilkes88} are needed but online covariance estimation is often needed in near real-time. While fast, closed form estimators such as AML \cite{Li99} can be used, they do not excel under low or realistic training. Our contribution aims to break this classical trade-off.  We develop a computationally efficient approximation of structured Toeplitz covariance under a rank (EASTR) constraint\footnote{Preliminary version of the work has appeared at IEEE Asilomar Conference on Signals, Systems, and Computers, November 2013 \cite{Kang13Asilomar} and Computational Advances in Multi-Sensor Adaptive Processing (CAMSAP), December 2013 \cite{Kang13CAMSAP}.}. Specifically, our key contributions are listed next.

\begin{itemize}

\item \textbf{Analytically tractable framework for exploiting both Toeplitz structure and the rank of the structured interference.} Our proposed estimator, i.e.\ EASTR, satisfies both Toeplitz structure property (at least approximately) and the rank information of the structured interference at the same time. Decades of research has shown that enforcing even each constraint individually can be quite onerous (e.g.\ rank is a non-convex constraint and no known closed form exists under the Toeplitz constraint for all training). The rank constrained ML estimation problem can be made convex as shown in \cite{Kang14} via a transformation of variables. However, this does not apply when the Toeplitz constraint is added. We propose to decouple the rank and Toeplitz constraints, which lends analytical tractability. Crucially, the EASTR solution does not need iterative steps like ITAM and as will be established in Section \ref{Sec:Experiment}, Furthermore, our results demonstrate that EASTR consistently outperforms ITAM.

\item \textbf{Computationally efficient and fast estimation and approximation.} Our proposed method, EASTR, essentially involves a cascade of two steps where a closed form solution is available in each step. First a closed form solution using maximum likelihood employing the rank constraint is obtained from the RCML \cite{Anderson63,Kang14} estimator. Next, we propose a new method to perturb the eigenvalues of the RCML estimator in a rank preserving manner so as to impose the Toeplitz structure. We formulate a new quadratic programming (QP) optimization problem that solves for the eigenvalues while incorporating Toeplitz constraints and demonstrate that this problem also admits a closed form solution.

\item \textbf{Experimental insights and improved performance in low training regimes.} The merits of EASTR are also verified experimentally over both simulated data and realistic data sets such as Knowledge Aided Sensor Signal Processing and Expert Reasoning (KASSPER). ITAM works well particularly in low training regimes but is numerically expensive. The asymptotic ML estimation gives us a fast closed form solution but shows good performances only in high training regimes. EASTR excels across all training regimes while still permitting closed form solutions attractive for practical deployment.
\end{itemize}

We consider two cases in this chapter: 1.) when the Toeplitz constraint is satisfied exactly, we obtain the exact Toeplitz estimate satisfying the rank constraint and Toeplitz property and 2.) when the Toeplitz constraint is not exactly satisfied, we make a slight modification to the Toeplitz constraint and derive a modified optimization problem to obtain approximately Toeplitz estimate. In practice, the available data dictates which of the two cases is invoked. Experimental investigation shows that EASTR can outperform alternatives in the sense of 1.) normalized SINR and 2.) the probability of detection.

The remainder of this chapter is organized as follows. Section \ref{Sec:EASTR} develops our proposed estimator, i.e. computationally efficient approximation of structured Toeplitz covariance under a rank constraint (EASTR). Experimental validation of the proposed method is provided in Section \ref{Sec:Experiment} wherein we report the performance of the proposed method and compare it against widely used existing radar STAP covariance estimators in terms of normalized SINR and the probability of detection. Validation is performed using a popular disturbance covariance simulation model as well as real-world data from the benchmark KASSPER dataset.

\section{Efficient Approximation of Structured Covariance}
\label{Sec:EASTR}

The maximum likelihood covariance estimate $\mb R$ is one which maximizes the likelihood function based on a zero-mean complex circular Gaussian distribution:
\be
\label{Eq:Likelihood}
f(\mb R ; \mb{Z}) = \frac{1}{\pi^{NK}}|\mb{R}|^{-K}\exp\big(-tr \{\mb{Z}^H\mb{R}^{-1}\mb{Z}\}\big)
\ee
under both Toeplitz and rank constraints. In (\ref{Eq:Likelihood}), $K$ is the number of training samples, $N$ is the dimension of observations, and $\mb Z$ is an $N \times K$ matrix whose each column is an i.i.d. observation vector. With some algebraic manipulations, the final optimization problem may be written as
\be
\label{Eq:InitialProblem}
\left\{ \begin{array}{cc}
\ds\min_{\mb{R}} & tr \{\mb{R}^{-1}\mb{S}\} + \log(|\mb{R}|)\\
s.t. & \mb{R} = \sigma^2 \mb{I} + \mb{R}_c\\
 & \rank(\mb{R}_c) = r\\
 & \mb{R}_c \in T \end{array} \right.
\ee
where $\mb S = \frac{1}{K}\mb Z \mb Z^H$ is the sample covariance matrix, $\mb R_c$ denotes the interference covariance matrix, $\mb I$ is an $N \times N$ identity matrix, and $\sigma^2$ is the radar noise floor which can be readily determined using standard techniques \cite{Skolnik08}, and lastly $T$ is the set of all $N \times N$ Hermitian positive semi-definite Toeplitz matrices,
\be
T = \{\mb{T}:\mb{T} \in \mathbb{C}^{N \times N}, \mb{T}^H = \mb{T}, \mb{T} \succeq \mb{0} \; \text{and} \; \mb{T} \in \mathcal{T}\}
\ee
where $\mathcal{T}$ is the set of all Toeplitz matrices. The optimization problem (\ref{Eq:InitialProblem}) is particularly hard to solve because 1.) the problem is not convex hence a global minimum would be difficult to find, 2.) from a numerical standpoint, solutions are known to be computationally burdensome under the Toeplitz constraint alone \cite{Burg82}, \cite{Wilkes88}. Adding the rank constraint only exacerbates the problem.

In view of the aforementioned challenges, we  focus on covariance matrix estimation that: 1.) is fast and based on analytical closed forms so as to facilitate practical deployment, and 2.) exploits previously known insights in radar STAP so that performance in the sense of high SINR and $P_{d}$ can be obtained across {\em all} training regimes.

Our proposed solution decouples the rank and Toeplitz constraints, and develops a cascade of two closed forms as the final estimator. The first closed form is obtained by employing the rank constrained ML (RCML) estimator of structured covariance \cite{Kang14,Anderson63}. The final RCML solution is given by
\be
\mb R^\star = \sigma^2 {\mb X^\star}^{-1} = \sigma^2 \mb \Phi {\mb \Lambda^\star}^{-1} \mb \Phi^H
\ee
where $\mb \Phi$ is the eigenvector matrix of the sample covariance matrix $\mb S$ and $\mb\Lambda^\star$ is a diagonal matrix with optimal diagonal entries $\lambda_i^\star$ which is given by
\be
\label{Eq:RCMLsolution}
\lambda_i^\star = \left\{ \begin{array}{cc}
\min (1,\dfrac{1}{d_i}) & \text{for} \; i=1,2,\ldots,r\\
1 & \text{for} \; i=r+1,r+2,\ldots,N \end{array} \right.
\ee
where $d_i$ is the $i$th eigenvalue of the sample covariance matrix normalized by $\sigma^2$, $\mb S' = \frac{1}{\sigma^2}\mb S$.

\subsection{Conditions for Eigenvalues of Toepiltz Covariance}
Our approach now involves enforcing the Toeplitz structure on top of the RCML estimator in \eqref{Eq:RCMLsolution}.
Let the eigenvector matrix of $\mb{S}$ be $\mb\Phi$ and the eigenvalues of $\mb{R}_c$ be $\lambda_1, \lambda_2, \ldots, \lambda_r, \ldots, \lambda_N$. Since we want to preserve the rank constraint of the structured interference $rank(\mb{R}_c) = r$, $\mb R_c$ should have only $r$ positive eigenvalues and the rest of them should be zero, that is
\be
\lambda_1 \geq \lambda_2 \geq \cdots \geq \lambda_r > \lambda_{r+1} = \lambda_{r+2} = \cdots = \lambda_N = 0
\ee
Therefore, $\mb{R}_c$ can be expressed as
\be
\mb{R}_c = \mb\Phi \mb\Lambda \mb\Phi^H
\ee
where
\be
\mb\Lambda = \left[
 \begin{array}{cc}
    \mb\Lambda_r & \mb 0_{r \times (N-r)}\\
    \mb 0_{(N-r) \times r}  & \mb 0_{N-r}\\
  \end{array}
\right]
\ee
and $\mb\Lambda_r$ is an $r \times r$ diagonal matrix with diagonal entries $\lambda_1, \ldots, \lambda_r$. Therefore, we know that $ij$th component of $\mb R_c$ is given by
\be
(\mb{R}_c)_{ij} = \ds\sum_{k=1}^r \lambda_k\phi_{ik}\phi_{jk}^*
\ee
where $\phi_{ij}$ is the $(i,j)$ element of $\mb\Phi$. Note that $\mb{R}_c$ is already Hermitian, that is, $(\mb R_c)_{ij} = (\mb R_c)_{ji}^*$. Now in order for $\mb{R}_c$ to be Toeplitz matrix, all entries on each diagonal in the lower triangular part in $\mb R_c$ must have same values, i.e., following equations must hold.
\be
\label{Eq:equations}
\left\{ \begin{array}{ccccccc}
(\mb{R}_c)_{11} & = & (\mb{R}_c)_{22} & = & \cdots & = & (\mb{R}_c)_{NN}\\
(\mb{R}_c)_{21} & = & (\mb{R}_c)_{32} & = & \cdots & = & (\mb{R}_c)_{N,N-1}\\
                &   &                 & \vdots &   &   & \\
                &   & (\mb{R}_c)_{N-1,1} & = & (\mb{R}_c)_{N2} & & \end{array} \right.
\ee
Let us examine the first condition in (\ref{Eq:equations}), $(\mb{R}_c)_{11} = (\mb{R}_c)_{22}$,
\be
\ds\sum_{k=1}^r \lambda_k\phi_{1k}\phi_{1k}^* = \ds\sum_{k=1}^r \lambda_k\phi_{2k}\phi_{2k}^*
\ee
It can be also expressed as
\be
\ds\sum_{k=1}^r \lambda_k (\phi_{1k}\phi_{1k}^* - \phi_{2k}\phi_{2k}^*) = 0
\ee
In vector form, the first equation is given by
\be
\label{Eq:VectorForm}
\left[ \begin{array}{ccc}
\phi_{11}\phi_{11}^*-\phi_{21}\phi_{21}^* & \cdots & \phi_{1r}\phi_{1r}^*-\phi_{2r}\phi_{2r}^* \end{array} \right]
\left[ \begin{array}{c}
\lambda_1\\
\vdots\\
\lambda_r \end{array} \right] = 0
\ee
Since the elements $\phi_{ij}$ of the eigenvalue matrix $\mb\Phi$ are  known ($\mb\Phi$ is the eigenvector matrix of the sample covariance matrix),  we now have the first constraint for Toeplitz covariance matrix as a linear combination of the eigenvalues.
Other equations in Eqs. (\ref{Eq:equations}) also can be expressed in a vector form as in (\ref{Eq:VectorForm}). Consequentially, we have a total of $N(N-1)/2$ equations and finally get the following equation which is the equality constraint of our optimization problem.
\be
\label{Eq:Constraint1}
\mb\Psi \bs\lambda = \mb{0}
\ee
where each row of $\mb\Psi \in \mathds{C}^{N(N-1)/2 \times r}$ denotes coefficients of $\lambda_i$ which come from each of equations in Eqs. (\ref{Eq:equations}) and $\bs\lambda = \left[ \begin{array}{cccc} \lambda_1 & \lambda_2 & \cdots & \lambda_r \end{array} \right]^T$.

Since $\mb\Psi$ Eq. (\ref{Eq:Constraint1}) is a tall matrix, (\ref{Eq:Constraint1}) in general is a overdetermined linear system, that is, we have more equations than unknowns. The solution set therefore depends on the rank of $\mb\Psi$. The first case is that we have an infinite set of solutions when the column rank of $\mb\Psi$ is less than $r$. On the other hand, when $\mb\Psi$ has a full column rank, we have the trivial solution, $\bs\lambda = \mb{0}$. That is, the covariance matrix can only be made approximately (and not exactly) Toeplitz in this case - a remedy for this case is discussed in Section \ref{Sec:ToeplitzApproximation}.

\subsection{Exact Toeplitz Solution}
\label{Sec:ExactToeplitz}
When the column rank of $\mb\Psi$ is less than $r$, Eq. (\ref{Eq:Constraint1}) has an infinite number of solutions. In this case, we can obtain the exact Toeplitz solution. First, let $\bs\lambda_\text{RCML}$ be the eigenvalues obtained from the RCML estimation, which is given by Eq. (\ref{Eq:RCMLsolution}). We already know the eigenvalues $\bs\lambda$ from the RCML estimate are the optimal ML estimate of the true structured covariance matrix under only the rank constraint. Therefore, we want the eigenvalues of the interference covariance matrix to satisfy Eq. (\ref{Eq:Constraint1}) and to be as close to the RCML solution as possible. Since Eq. (\ref{Eq:Constraint1}) has an infinite number of solutions, we can find the closest vector of the eigenvalues to $\bs\lambda_\text{RCML}$ by solving the following convex optimization problem.
\be
\label{Eq:OptimizationProblem1}
\begin{array}{cc}\ds\min_{\bs{\lambda}} & || \bs\lambda_{\text{RCML}} - \bs\lambda ||^2\\
\text{subject to :} & \mb\Psi \bs\lambda = \mb{0} \end{array}
\ee
The optimization problem (\ref{Eq:OptimizationProblem1}) is a well known quadratic programming (QP) optimization problem with an equality constraint and therefore the closed form solution is available using KKT condition \cite{Boyd04} and it is given by solving the following equation.
\be
\label{Eq:ClosedForm1}
\left[
  \begin{array}{cc}
    2\mb{I} & \mb\Psi^T \\
    \mb\Psi & \mb{0} \\
  \end{array}
\right] \left[ \begin{array}{c} \bs\lambda^\star \\ \bs\nu^\star
\end{array}\right] = \left[ \begin{array}{c} 2\bs\lambda_{\text{RCML}} \\ \mb{0}
\end{array}\right]
\ee
where $\bs\nu^\star$ is a vector of Lagrange multipliers.

However, the matrix on the left-hand side of Eq. (\ref{Eq:ClosedForm1}) is actually singular because $\mb\Psi$ has not full column rank. So we introduce a new matrix $\check{\mb\Psi}$ instead of $\mb\Psi$ to make the left matrix invertible when we solve it. That is,
\be
\label{Eq:ClosedForm2}
\left[
  \begin{array}{cc}
    2\mb{I} & \check{\mb\Psi}^T \\
    \check{\mb\Psi} & \mb{0} \\
  \end{array}
\right] \left[ \begin{array}{c} \bs\lambda^\star \\ \bs\nu^\star
\end{array}\right] = \left[ \begin{array}{c} 2\bs\lambda_{\text{RCML}} \\ \mb{0}
\end{array}\right]
\ee
where $\check{\mb\Psi}$ is a matrix consists of $\rank(\mb\Psi)$ linearly independent rows of $\mb\Psi$. Obviously, Eq. (\ref{Eq:ClosedForm1}) and Eq. (\ref{Eq:ClosedForm2}) have the same solution because linearly independent $\rank(\mb\Psi)$ rows of $\mb\Psi$ determine the set of solutions of the equation and removing redundant rows does not make any changes to the solution. It follows that the final closed form solution using blockwise inversion property is given by
\be
\bs\lambda^\star = \big(\mb I - \check{\mb\Psi}^T(\check{\mb\Psi}\check{\mb\Psi}^T)^{-1}\check{\mb\Psi}\big)\bs\lambda_{\text{RCML}}
\ee
and the final covariance matrix can be obtained by
\be
\mb R^\star = \sigma^2 \mb I + \mb\Phi diag(\bs \lambda^\star) \mb\Phi^H
\ee

\subsection{Toeplitz Approximation}
\label{Sec:ToeplitzApproximation}

In the case that $\mb\Psi$ has a full column rank, Eq. (\ref{Eq:Constraint1}) has the only one solution, $\bs\lambda = \mb 0$, which does not yield a meaningful covariance matrix. In this case, the optimization problem to enforce the Toeplitz structure must be modified. One possibility is to explicitly incorporate the eigenvector matrix into the optimization. This however, will lead to a computationally expensive problem because the optimization must constrain the eigenvector matrix to be unitary. Further, using an eigenvector matrix to agree with $\mb\Phi$, i.e.\ the one obtained from sample covariance has been known to be very successful in radar STAP \cite{Guerci03,Aubry12,Kang14}.

We therefore take the approach of building an {\em approximately} as opposed to exactly Toeplitz matrix.
This can be done by computing the closest rank deficient matrix $\tilde{\mb\Psi}$ to $\mb\Psi$. Consider the singular value decomposition of $\mb\Psi$,
\be
\mb\Psi = \mb U \mb\Sigma \mb V^H
\ee
The well-known theorem, Eckart-Young theorem \cite{Eckart36}, says that a matrix $\bar{\mb\Psi}$ with the column rank less than $r$ that minimizes $|| \mb\Psi - \tilde{\mb\Psi} ||_F$ is given by
\be
\tilde{\mb\Psi} = \mb U \tilde{\mb\Sigma} \mb {V}^H
\ee
where $\tilde{\mb\Sigma}$ is the diagonal matrix obtained from $\mb\Sigma$ by replacing the $r$-th diagonal element which is the smallest diagonal element by zero. By substituting $\mb\Psi$ with $\tilde{\mb\Psi}$ in Eq. (\ref{Eq:Constraint1}), we obtain the infinite number of solutions for $\bs\lambda$. Now, the optimization problem becomes
\be
\label{Eq:OptimizationProblem3}
\begin{array}{cc}\ds\min_{\bs{\lambda}} & || \bs\lambda_{\text{RCML}} - \bs\lambda ||^2\\
\text{subject to :} & \tilde{\mb\Psi} \bs\lambda = \mb{0} \end{array}
\ee
Finally, a Toeplitz matrix is obtained by solving the above optimization problem in the same way done in the case of exact Toeplitz solution, that is,
\be
\bs\lambda^\star = \big(\mb I - \breve{\mb\Psi}^T(\breve{\mb\Psi}\breve{\mb\Psi}^T)^{-1}\breve{\mb\Psi}\big)\bs\lambda_{\text{RCML}}
\ee
where $\breve{\mb\Psi}$ is a matrix consists of $r-1$ linearly independent rows of $\tilde{\mb\Psi}$.

\emph{Remark}: It should be noted that the actual rank of $\mb\Psi$ which is derived from $\mb\Phi$ depends on the training data. If the true covariance is indeed Toeplitz, we expect training samples to reflect that particularly in the regime of $ K >> N$ training samples (asymptotic regime), this is indeed what we observe in practice.

\section{Experimental Investigation}
\label{Sec:Experiment}

\subsection{Experimental Setup and Methods Compared}

In this section, we compare the performance of proposed estimator against state of the art Toeplitz STAP estimators. Two data sets are used: 1.) A radar covariance simulation model and 2.) the  well known KASSPER \cite{Bergin02} data set.

First, we model a radar system with an $N$-element uniform linear array. The overall disturbance is composed of jammer and white interference. Therefore, the external wideband noise environment via its input covariance matrix can been modeled by
\be
\mb{R}(n,m) = \sum_{i=1}^J \sigma_i^2 \sinc[0.5 \beta_i (n-m) \phi_i ] e^{j(n-m)\phi_i} + \sigma_a^2 \delta(n-m)
\ee
where $n,m \in \{1,\ldots,N\}$, $J$ is the number of jammers, $\sigma_i^2$ is the power associated with the $i$th jammer, $\phi_i$ is the jammer phase angle with respect to the antenna phase center, $\beta_i$ is the fractional bandwidth, $\sigma_a^2$ is the actual power level of the white disturbance term, and $\delta(n,m)$ has the value of 1 only when $n=m$ and 0 otherwise. This simulation model has in fact been widely and very successfully used in previous literature \cite{Steiner00,Aubry12,Pallotta12,DeMaio09} for performance analysis. It is easily seen that $\mb{R}$ is Hermitian and Toeplitz since $\mb{R}(n,m)$ depends on only $n-m$ and $\sinc$ function is an even function. In addition, $\mb{R}$ generally has a rank less than $N$. Therefore, this model can not only be used to simulate radar disturbance samples but also makes ground truth covariance available.

Data from the L-band data set of KASSPER program is the other data set used in our experiments. Note, the KASSPER dataset also makes the true ground truth covariance available and we picked range bins such that their covariance matrices were exactly or approximately Toeplitz. The L-band data set consists of a data cube of 1000 range bins corresponding to the returns from a single coherent processing interval from $11$ channels and $32$ pulses. Therefore, the dimension of observations (or the spatio temporal product) $N$ is $11 \times 32 = 352$. Other key parameters are detailed in Table \ref{Tb:parameters}.

We compare the following six different covariance estimation techniques: A host of competing techniques like FML, eigencanceller, and shrinkage estimators have been compared with the RCML method in Chapter \ref{Ch:RCML}. The results of Chapter \ref{Ch:RCML} demonstrate that RCML ouperforms these techniques under all conditions of training data support and hence they are not reproduced here.
\begin{itemize}
\item \textbf{Sample Covariance Matrix:} The sample covariance matrix is given by $\mb S = \dfrac{1}{K}\mb Z \mb Z^H$. It is well known that the sample covariance is the unconstrained maximum likelihood estimator under Gaussian disturbance statistics. We refer to the use of this technique as SMI.
\item \textbf{Iterated Toeplitz Approximation Method:} The iterated Toeplitz approximation method (ITAM) \cite{Wilkes88} alternatively estimates a rank deficient matrix using the eigenvalue decomposition and then makes the resulting matrix Toeplitz by substituting diagonal entries with the average value of themselves for each diagonal of the estimated matrix. After that, the same process is repeated until the estimated Toeplitz matrix has a desired rank. The estimated covariance satisfies both a desired rank and Toeplitz property and it is closer to the true covariance matrix in the sense of Frobenius norm than the sample covariance matrix.
\item \textbf{Asymptotic Maximum Likelihood:} The asymptotic maximum likelihood (AML) \cite{Li99} exploits Toeplitz property of the structured covariance matrix. The authors derived a closed-form formula for Toeplitz covariance matrix estimation and it facilitates computationally efficient implementation. However, they assumed a large number of training samples and their closed-form solution is asymptotically valid. That is, in the low/realistic training regime, estimation performance invariably suffers.

\item \textbf{Rank Constrained ML estimator:} The RCML estimator \cite{Kang14} proposed in Chapter \ref{Ch:RCML} exploits the clutter rank information of the structured covariance matrix but not Toeplitz property. It is also the first step of the closed form solution of our proposed method.

\item \textbf{Sequential Quadratic Programming:} Al-Homidan proposed a sequential quadratic programming (SQP) algorithm to find the nearest symmetric positive semi-definite Toeplitz matrix to given a matrix \cite{AlHomidan02}. There are many other Toeplitz approximation algorithms in applied mathematics \cite{Davis98,Shaw98,Suffridge93}. We choose the SQP algorithm largely because it guarantees a global minima in approximation error and the $l_1$ SQP method is considerably faster \cite{AlHomidan02} than alternatives. In practice, the estimator is developed by making a Toeplitz approximation to the RCML estimator. This makes the technique analogous to our proposal of decoupling the rank and Toeplitz constraints in Section \ref{Sec:EASTR}. However, using applied math approximations in a `black-box' manner has two major drawbacks: 1.) the approximation may not necessarily preserve rank and radar STAP specific structure (e.g. eigenvector matrix is perturbed as well), and 2.) the techniques are numerically involved particularly with an increase in data dimension.
%In fact, for this reason we will show SQP results for the simulation model ($N = 20$) but not for KASSPER data ($N= 352$) because the complexity is way too high in the latter case.

\item \textbf{EASTR:} The proposed Efficient Approximation of Structured covariance under joint Toeplitz and Rank (EASTR) constraints. It incorporates Toeplitz structure, the rank of the clutter component as well as the STAP structural constraint.
\end{itemize}

In the results to follow, the ITAM, RCML, SQP and EASTR exploit rank information. The clutter rank for the simulation model covariance is of course known and for the KASSPER data set was inferred via the Brennan rule.

\begin{figure}[!t]
\begin{center}
\subfigure[simulation model]{\includegraphics[scale=0.5]{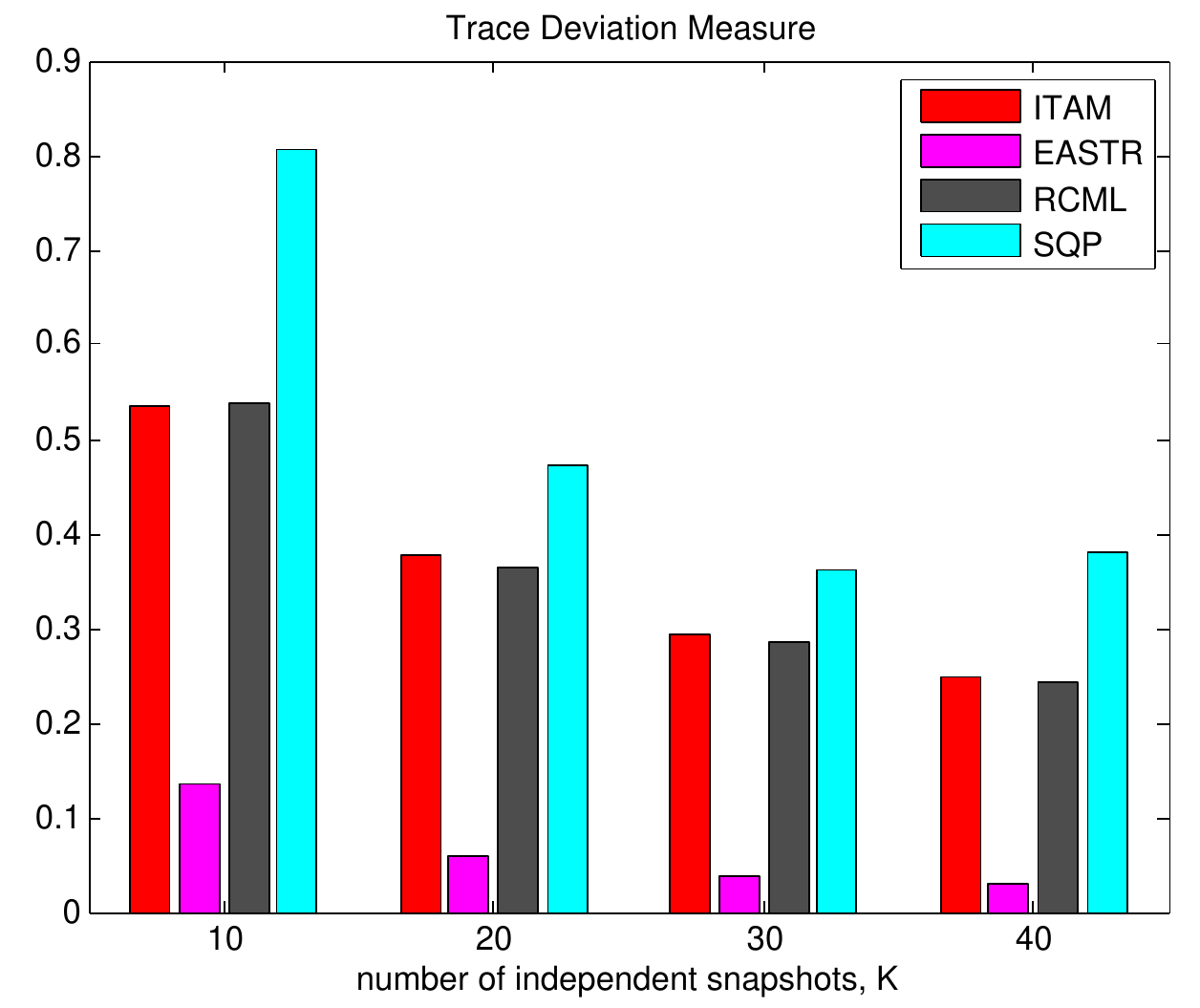}\label{Fig:TRD_model}}
\hfil
\subfigure[KASSPER data set]{\includegraphics[scale=0.5]{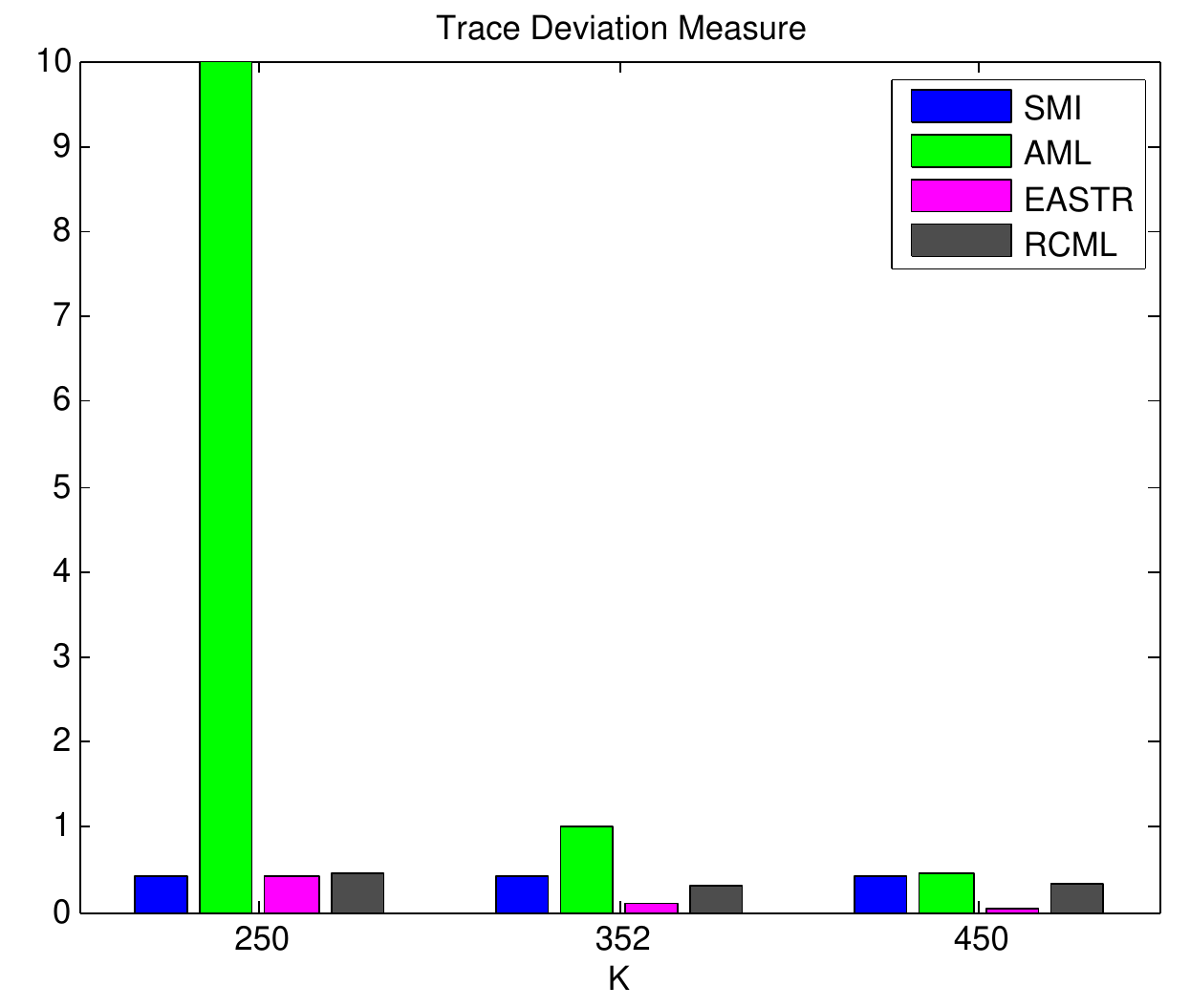}\label{Fig:TRD_KASSPER}}
\end{center}
\caption{Trace deviation measure vs. the number of training samples.}
\label{Fig:TRD}
\end{figure}

\subsection{Whiteness Test}

Before using popular radar STAP measures, we apply a `whiteness test'. The trace deviation measure \cite{Kang13} is one way of evaluating covariance matrix estimators since it captures the extent to which the estimated covariance matrix whitens the true covariance matrix. It is given by
\be
TRD(\hat{\mb R}) = | tr\{\mb R^{-1} \hat{\mb R}\}/N -1 |
\ee
Intuitively, we can see that its lower bound is zero when $\hat{\mb R} = \mb R$ and smaller value of TRD means better performance.

\begin{figure}[!t]
\centering
\includegraphics[scale=0.59]{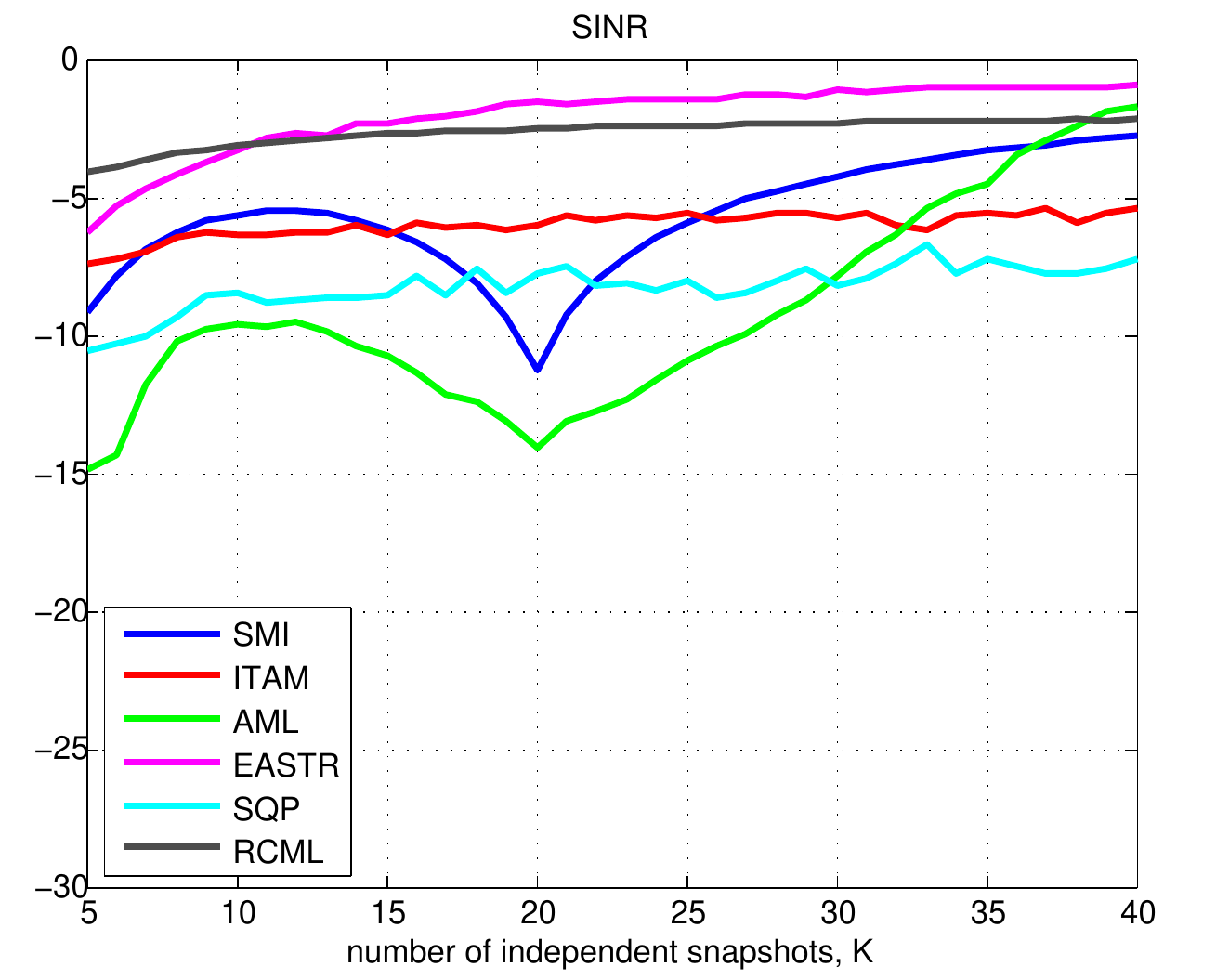}
\caption{Normalized SINR versus the number of training samples for the simulation model}
\label{Fig:SINR_model}
\end{figure}

Figure \ref{Fig:TRD} shows bargraphs of the performance of compared methods for simulation model and KASSPER data set respectively. Figure \ref{Fig:TRD_model} shows bargraphs of the performance in terms of TRD measure versus the number of training samples. Because the SMI and the AML show very high TRD values, we do not plot them in Figure \ref{Fig:TRD_model}.  Figure \ref{Fig:TRD_KASSPER} similarly shows the result of TRD measure across three training regimes for the KASSPER data set. The TRD measure results in Figures \ref{Fig:TRD_model} and \ref{Fig:TRD_KASSPER} reveal hence that EASTR is in fact ``structurally'' the closest to the true covariance matrix.

\begin{figure}[!t]
\begin{center}
\subfigure[$K(=250)<N$]{\includegraphics[scale=0.5]{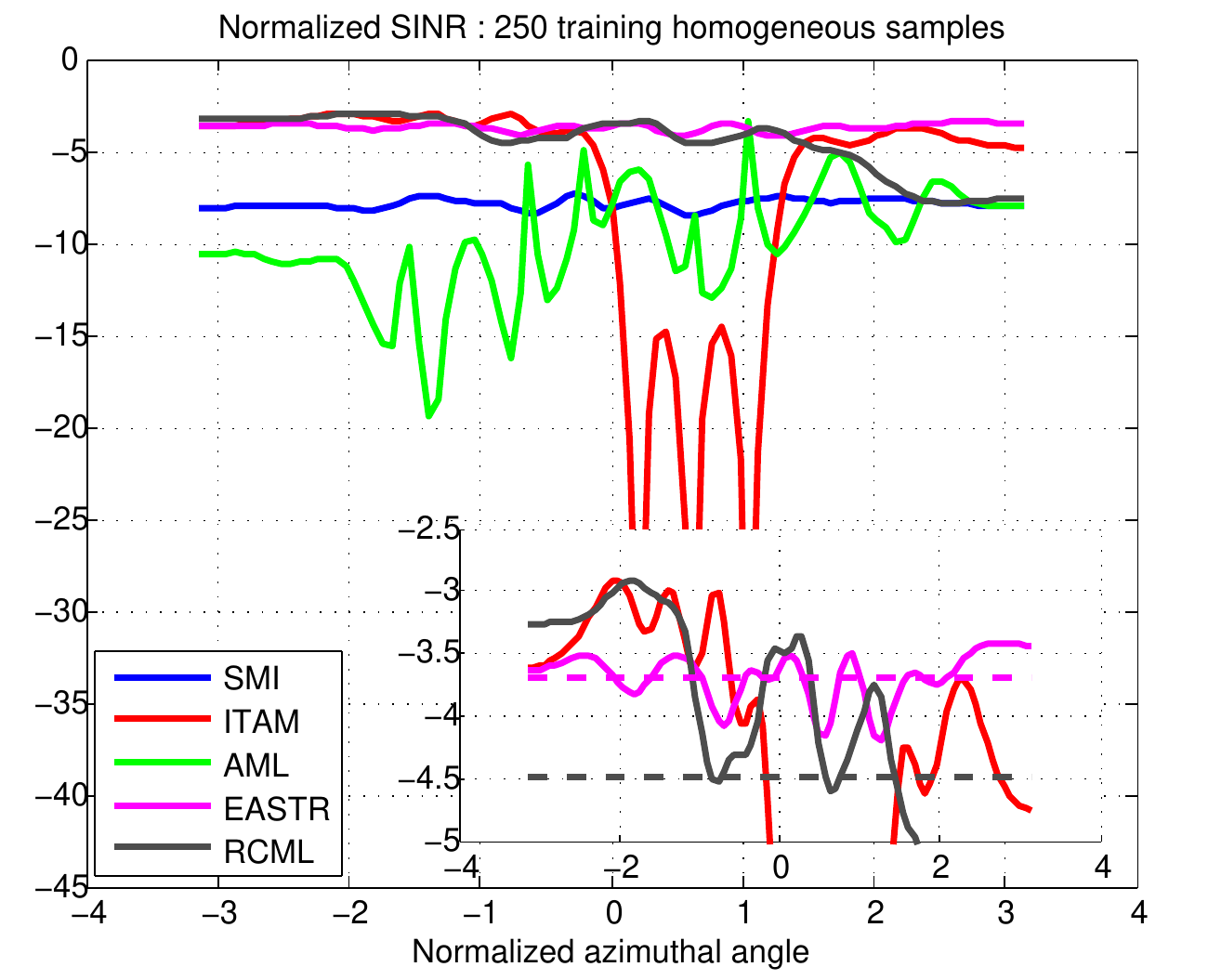}\label{Fig:SINR_KASSPER_angle_250}}
\hfil
\subfigure[$K(=250)<N$]{\includegraphics[scale=0.5]{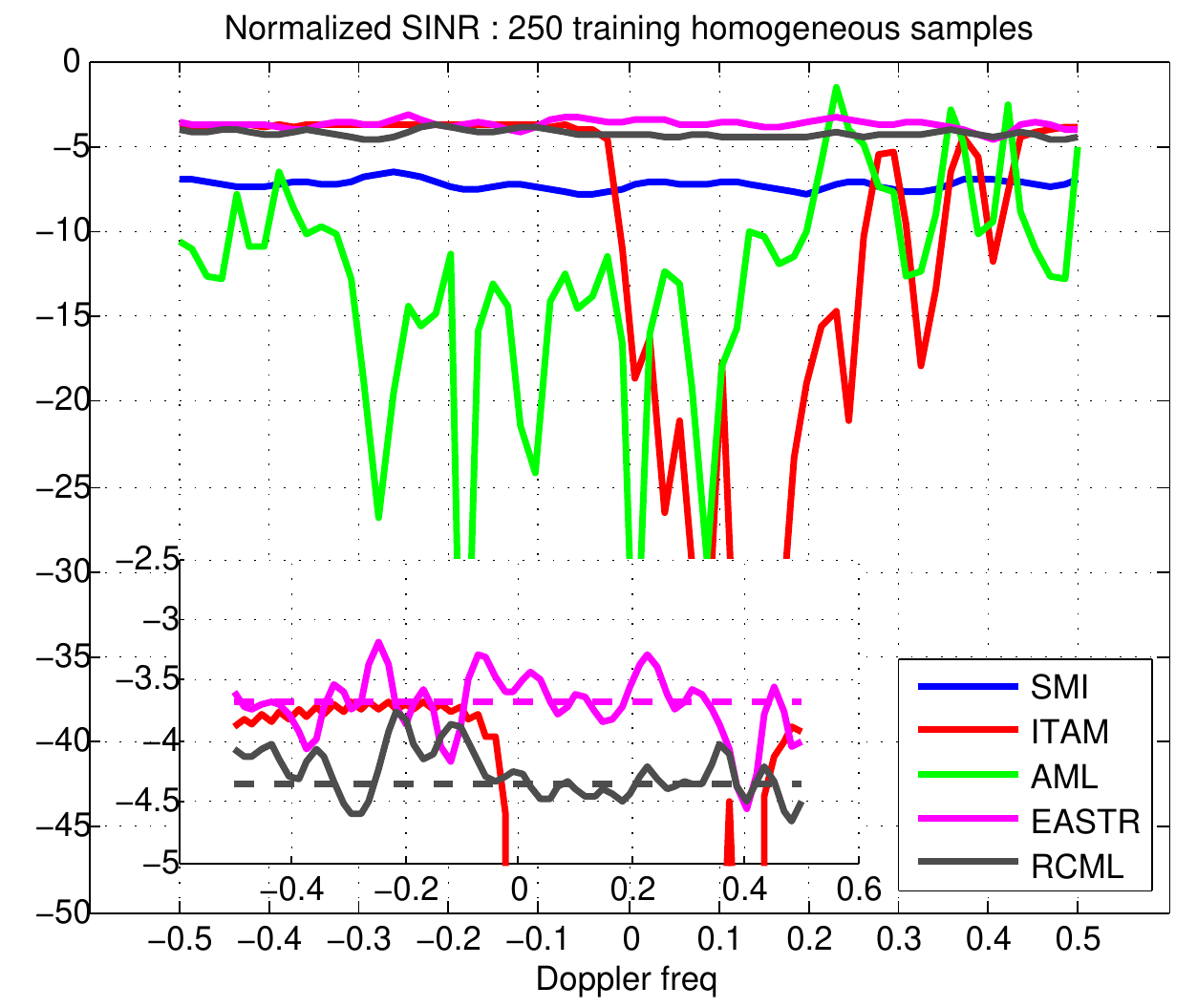}\label{Fig:SINR_KASSPER_dop_250}}\\
\subfigure[$K=N=352$]{\includegraphics[scale=0.5]{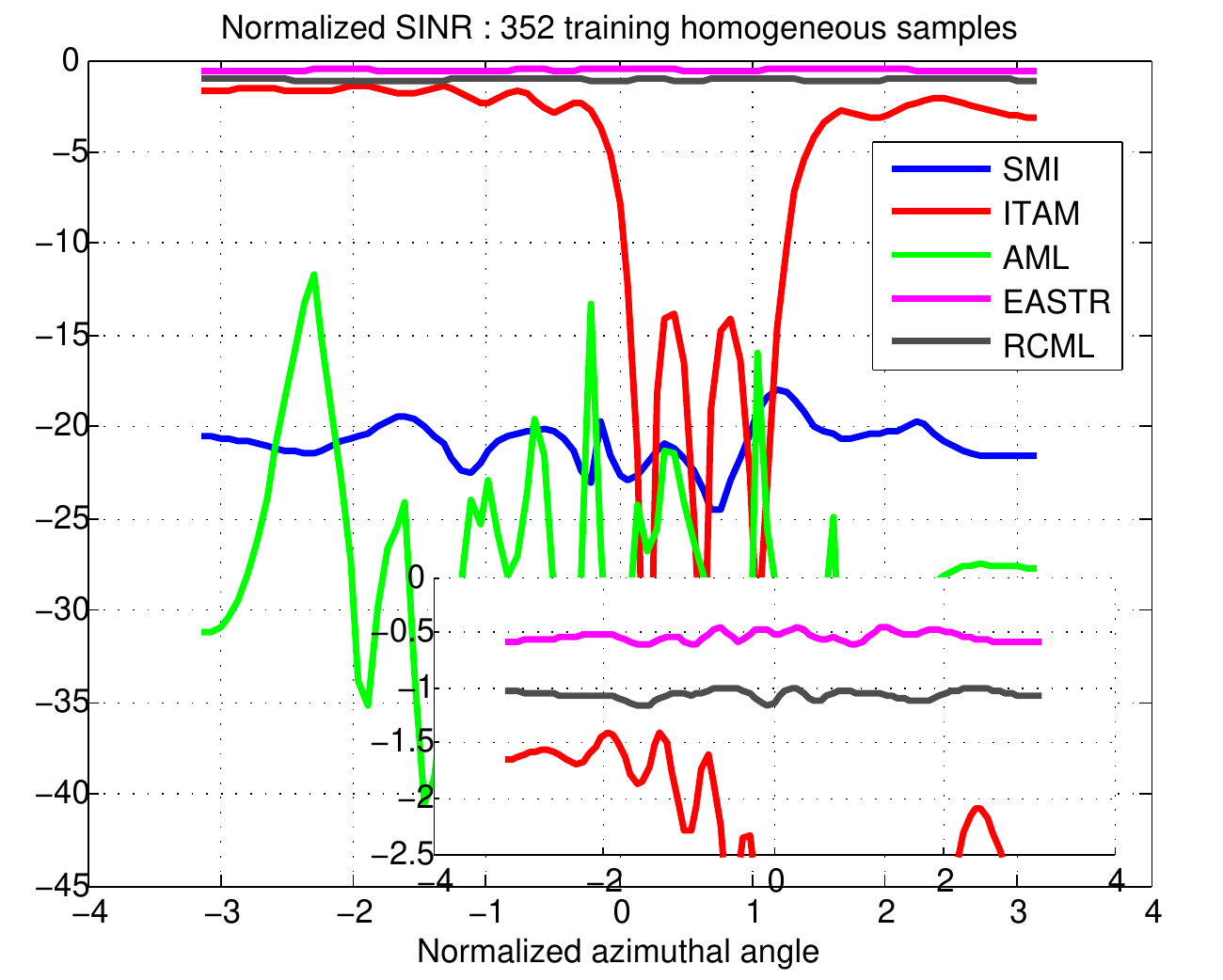}\label{Fig:SINR_KASSPER_angle_352}}
\hfil
\subfigure[$K=N=352$]{\includegraphics[scale=0.5]{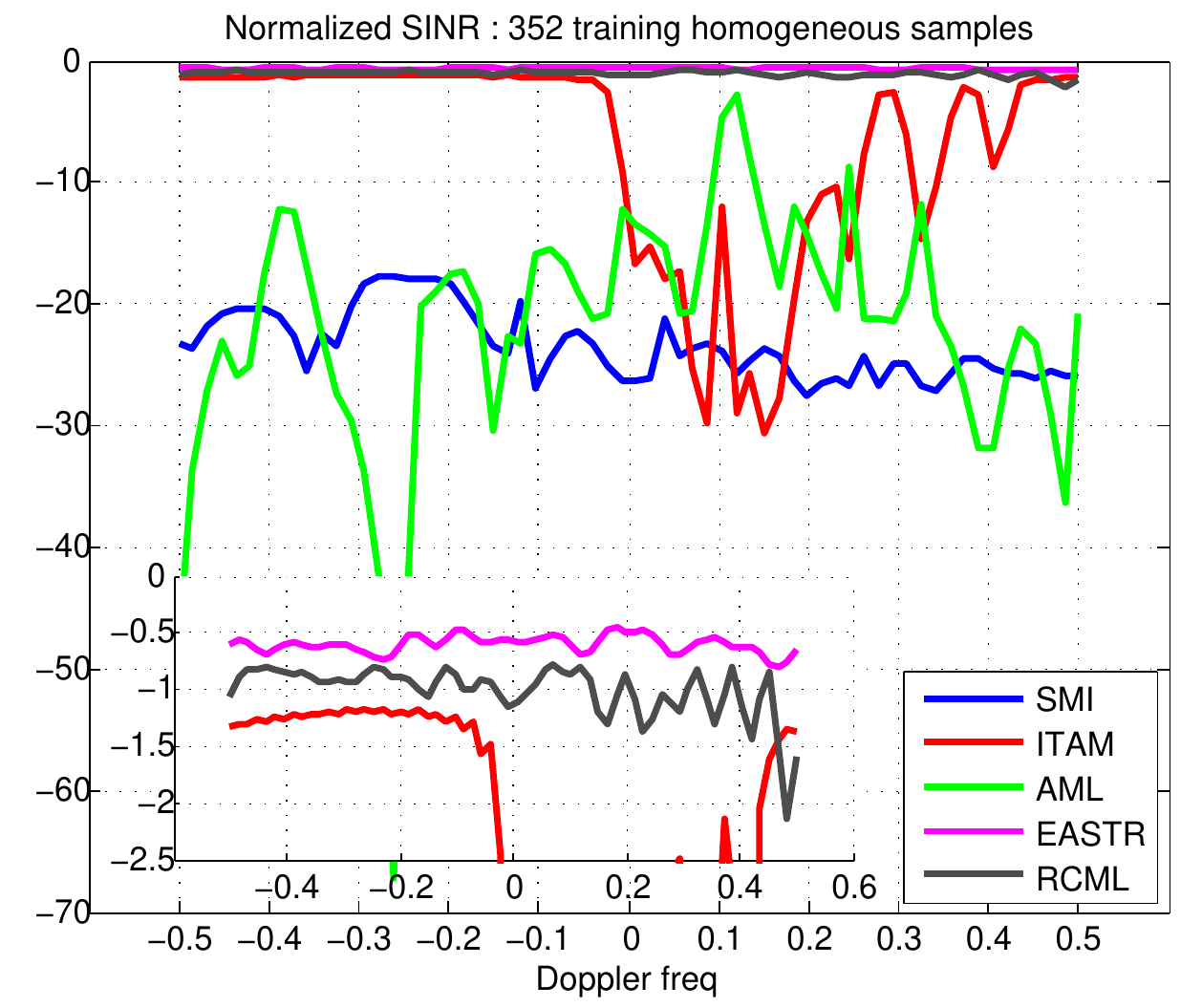}\label{Fig:SINR_KASSPER_dop_352}}\\
\subfigure[$K(=450)>N$]{\includegraphics[scale=0.5]{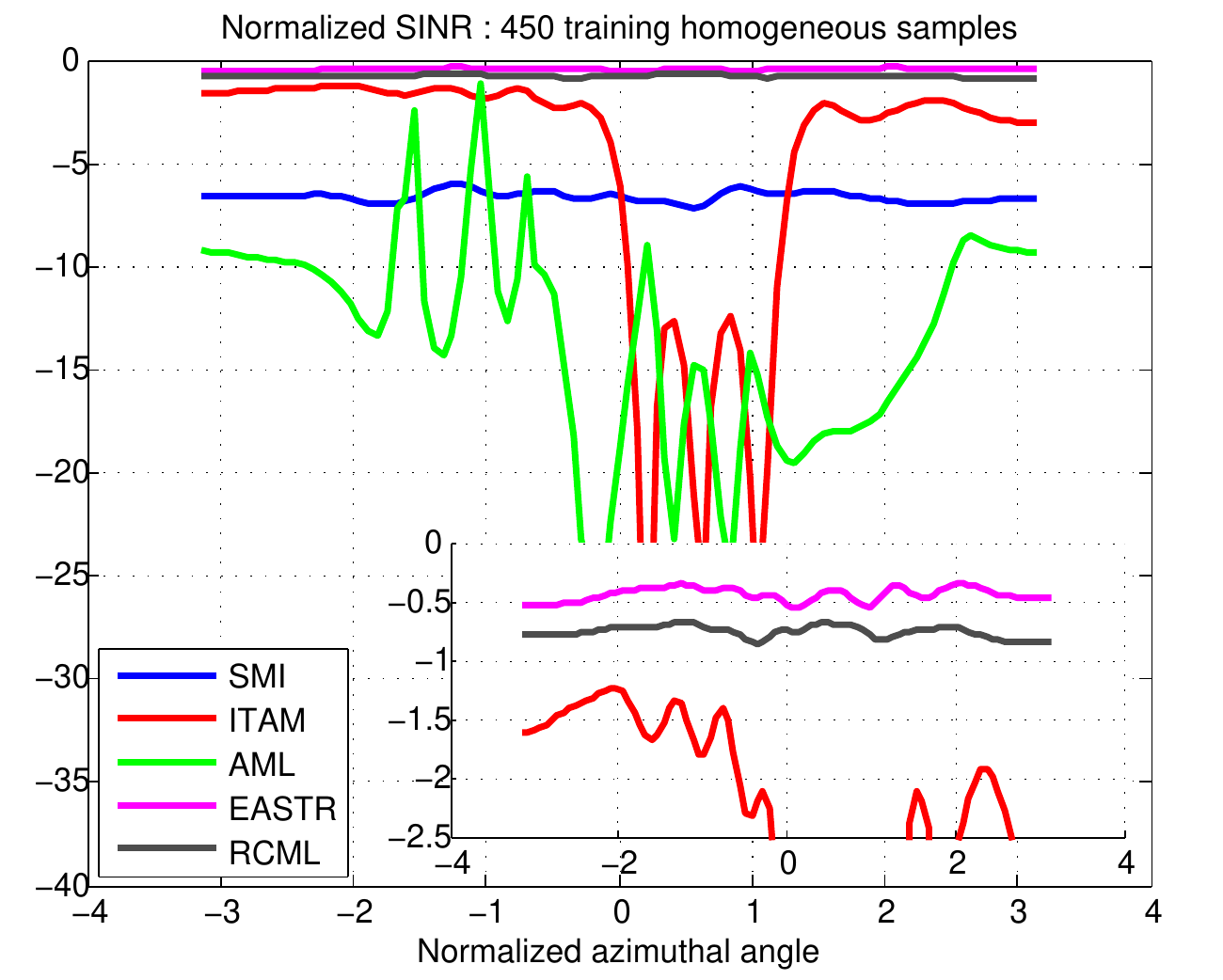}\label{Fig:SINR_KASSPER_angle_450}}
\hfil
\subfigure[$K(=450)>N$]{\includegraphics[scale=0.5]{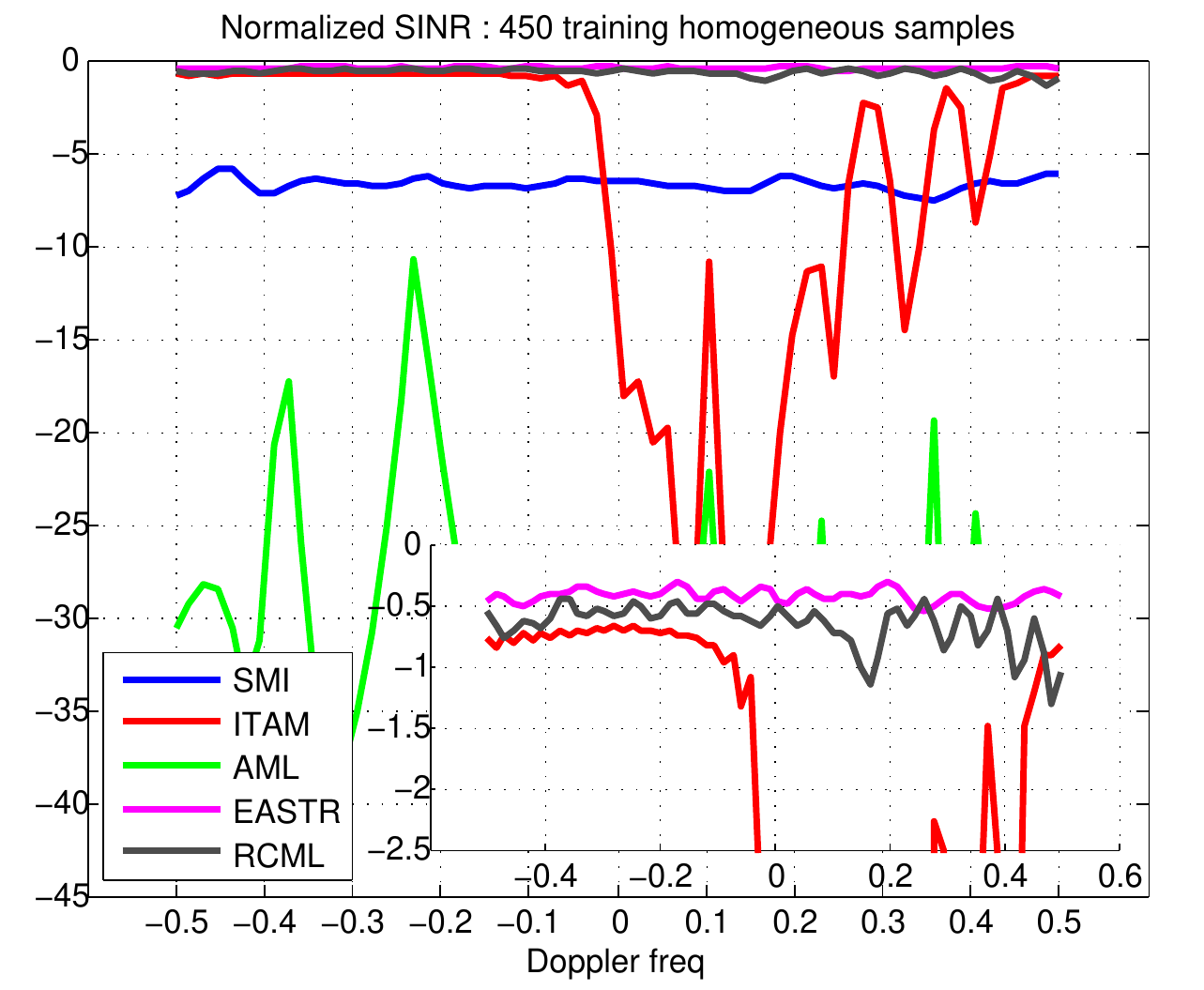}\label{Fig:SINR_KASSPER_dop_450}}\\
\end{center}
\caption{Normalized SINR versus azimuthal angle Doppler frequency for KASSPER data set.}
\label{Fig:SINR_KASSPER}
\end{figure}

\subsection{Normalized SINR}
\label{Sec:SINR}
The normalized SINR measure \cite{Monzingo04} is commonly used in the radar literature. We plot the normalized average SINR versus the number of training samples $K$ for the simulation model in Figure \ref{Fig:SINR_model}. In this case, we consider the presence of wideband jamming $J=3$. In particular, the fractional bandwidth $\beta_i = [0.2, 0, 0.3]$, the powers and phases of jammers are $10$ dB, $20$ dB, $30$ dB and $20$ deg, $40$ deg, and $60$ deg, respectively. $\sigma_a^2 =1$ and the rank of $r=7$ are used in the experiments. When $K<N$ the sample covariance is singular, therefore we used its pseudo-inverse instead of inverse itself. Note also that SMI and AML have a dip when $K=20$ due to numerical instabilities in the $K=N=20$ training regime. In contrast, ITAM, RCML, EASTR, and SQP guarantee nonsingularity in all training regimes. Interpreting the results in Figure \ref{Fig:SINR_model}, it is useful to start with AML which well performs when training is generous $K >> N$. However, because AML is asymptotically based - its performance is poor when $ K < N$ or $K \approx N$ due to numerical instabilities. Even, the SMI and SQP estimators outperform AML when training is low/realistic. ITAM is effective in very low training as expected because it exploits both rank and Toeplitz constraints (though in a largely heuristic way) - ITAM does not exhibit scalable improvements as training support is increased. However, EASTR performs the best overall, even better than RCML (which was recently demonstrated to be the most competitive radar STAP estimator \cite{Kang14}) by virtue of additionally capturing the Toeplitz structure on covariance.

Figure \ref{Fig:SINR_KASSPER} plots the normalized SINR results for KASSPER data set as a function of the azimuthal angle and Doppler frequency (averaged over all ranges of Doppler frequency and the azimuthal angle, respectively) for three different training regimes. Specifically, each row of Figure \ref{Fig:SINR_KASSPER} is corresponding to $K=250(<N)$, $K=N=352$, and $K > N = 450$ training samples, respectively. Zoomed in versions of ITAM, EASTR, and RCML are shown in each plot to make difference clearly seen. The sample covariance technique and the AML suffer tremendously when $K \leq N$. For low training, ITAM shows comparable performance to the EASTR and the RCML estimators in some ranges of the azimuthal angle but is worse in some other ranges. On an average (over azimuthal angle and Doppler frequency), EASTR is easily the best in Figure \ref{Fig:SINR_KASSPER}, even providing appreciably gains over the second best RCML estimator. Further, EASTR is stable and effective across all training regimes $K < N, K \approx N$ and $K > N$.

\subsection{Probability of Detection vs. SNR}

In order to compute probability of detection, $P_d$, we apply the normalized matched filter (NMF) \cite{Conte95} as the test statistic given by \eqref{Eq:NMF}. The detection probability $P_d$ is defined as the probability that the value of test statistic is grater than a threshold conditioned on the hypothesis that the received data includes target information. Therefore, it depends on signal to noise ratio (SNR, by virtue of $\mb s$,) and the estimated covariance matrix. Since $P_d$ does not typically admit a closed form, we first generate a number of samples from the true covariance to determine $\lambda$ corresponding to the fixed false alarm rate and then employ Monte Carlo simulations to evaluate $P_d$ corresponding to each estimator. We set a constant false alarm rate to $10^{-4}$.

Figure \ref{Fig:PD_model} shows the detection probability $P_d$ for simulation model plotted as a function of SNR for different estimators. We use $K=N=20$ and $K=2N=40$ training samples to estimate the covariance matrix in Figure \ref{Fig:PD_N_model} and Figure \ref{Fig:PD_2N_model}, respectively. It is well-known that $K=2N$ training samples are needed to keep the performance within 3 dB. Indeed, we see that the sample covariance matrix has about 3 dB loss vs. the true covariance matrix in Figure \ref{Fig:PD_2N_model}. The proposed EASTR is the closest to the $P_d$ achieved by using the true covariance matrix (upper bound) for both cases. In Figure \ref{Fig:PD_N_model}, we do not plot for ITAM and SQP because they do not guarantee positive semi-definiteness of final estimate in the case of $K=N=20$, so we cannot calculate the detection probabilities for them.

Figure \ref{Fig:PD_KASSPER} also shows the probability of detection versus SNR plots. We use the same training regimes as used in Section \ref{Sec:SINR}. Figure \ref{Fig:PD_KASSPER_352} and Figure \ref{Fig:PD_KASSPER_704} plot results for $K=352$ and $K=2N=704$, respectively. We can see similar trends in Figure \ref{Fig:PD_KASSPER} hence for KASSPER data to the ones for the simulation model in Figure \ref{Fig:PD_model}. EASTR exhibits the best performance in both plots.

\begin{figure}[!t]
\begin{center}
\subfigure[$K=N=20$]{\includegraphics[scale=0.6]{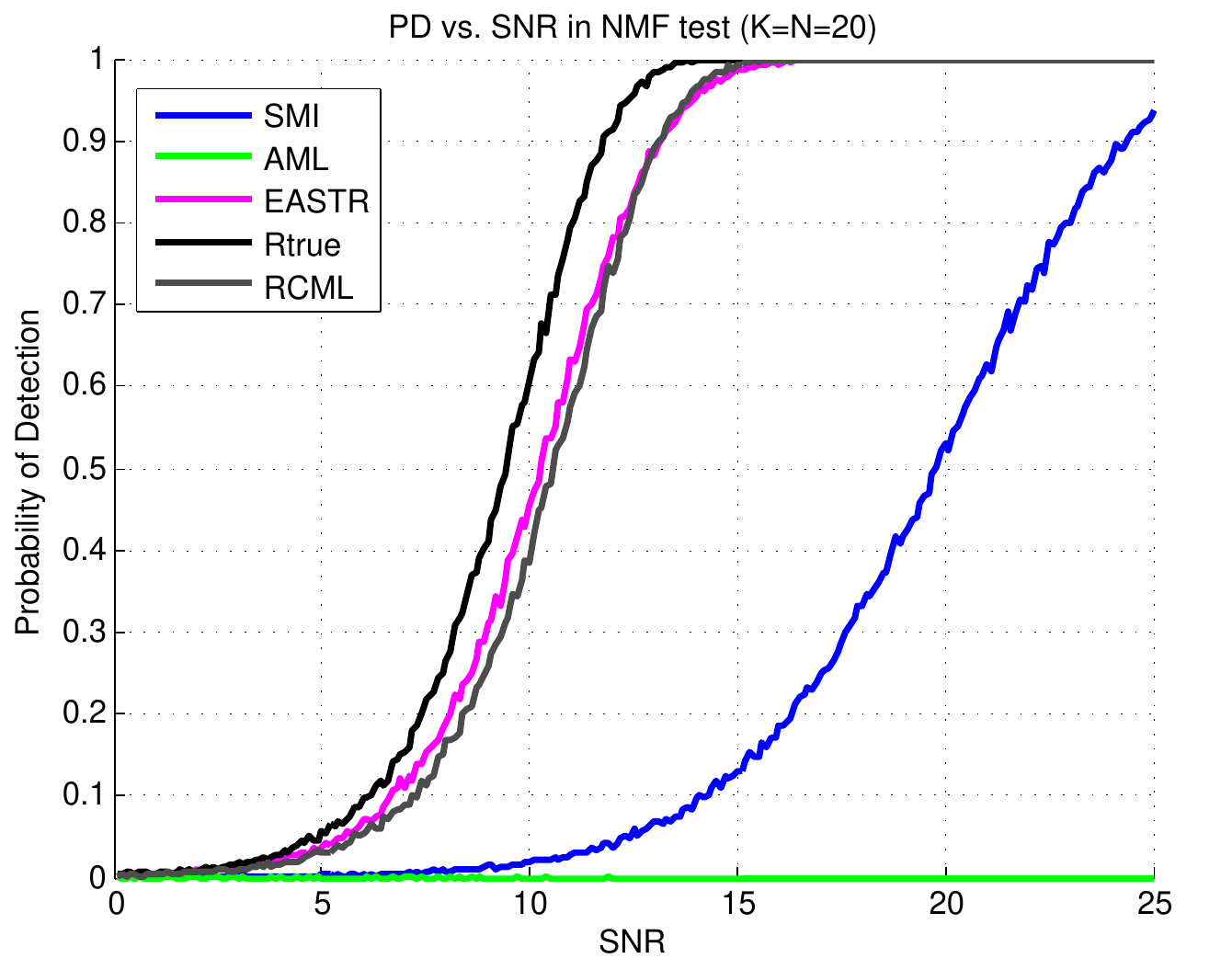}\label{Fig:PD_N_model}}
\hfil
\subfigure[$K=2N=40$]{\includegraphics[scale=0.6]{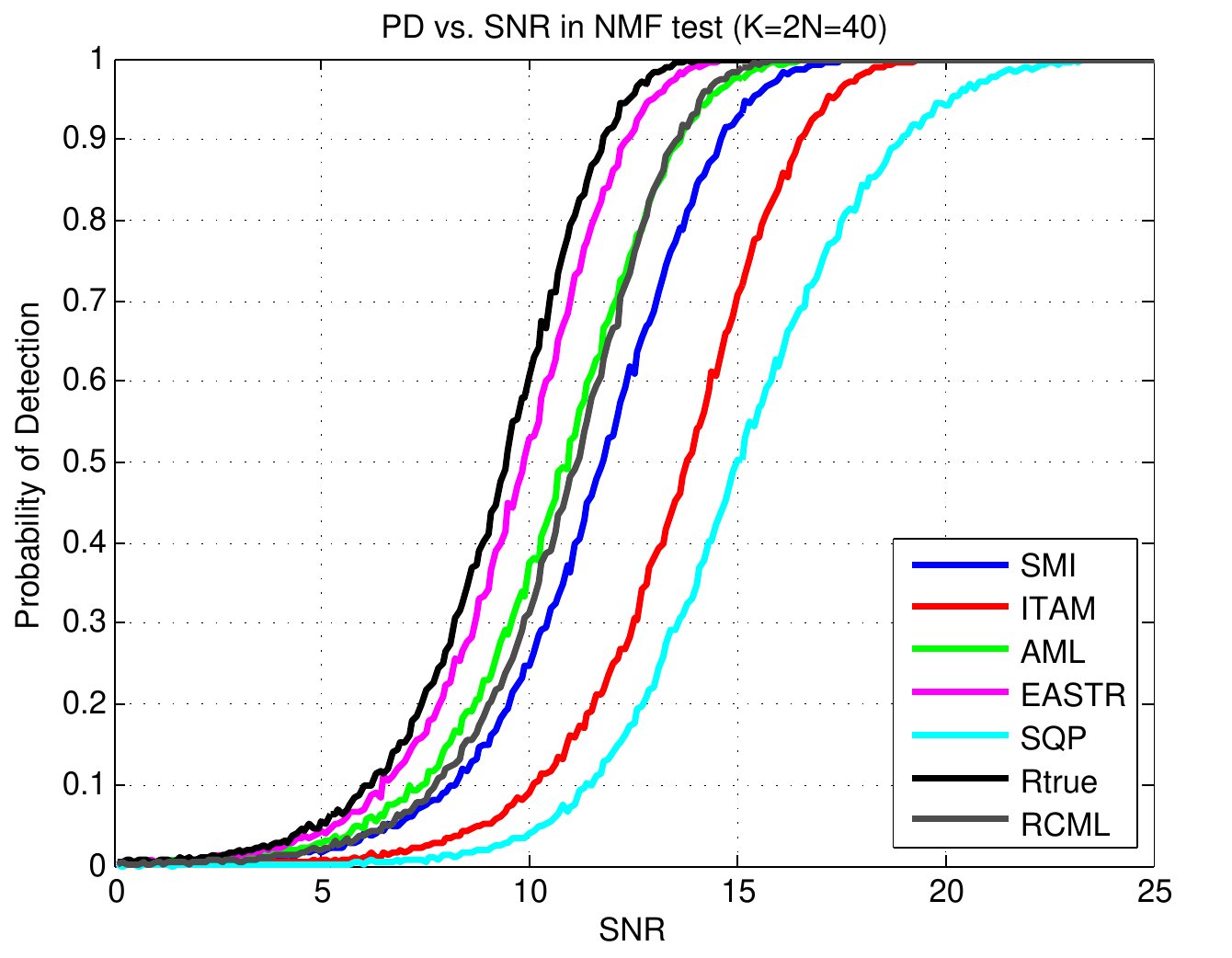}\label{Fig:PD_2N_model}}
\end{center}
\caption{Probability of detection vs. SNR for simulation model via normalized matched filter (NMF) test.}
\label{Fig:PD_model}
\end{figure}

%\begin{figure}
%\centering
%\includegraphics[scale=0.6]{figs/PD_N_All_crop.pdf}
%\caption{Probability of detection vs. SNR vis normalized matched filter (NMF) test. $K=2N=40$ is used.}
%\label{Fig:PD_N_model}
%\end{figure}

%\begin{figure}
%\centering
%\includegraphics[scale=0.6]{figs/PD_2N_All_crop.pdf}
%\caption{Probability of detection vs. SNR vis normalized matched filter (NMF) test. $K=2N=40$ is used.}
%\label{Fig:PD_2N_model}
%\end{figure}

\begin{figure}[!t]
\begin{center}
\subfigure[$K=N=352$]{\includegraphics[scale=0.6]{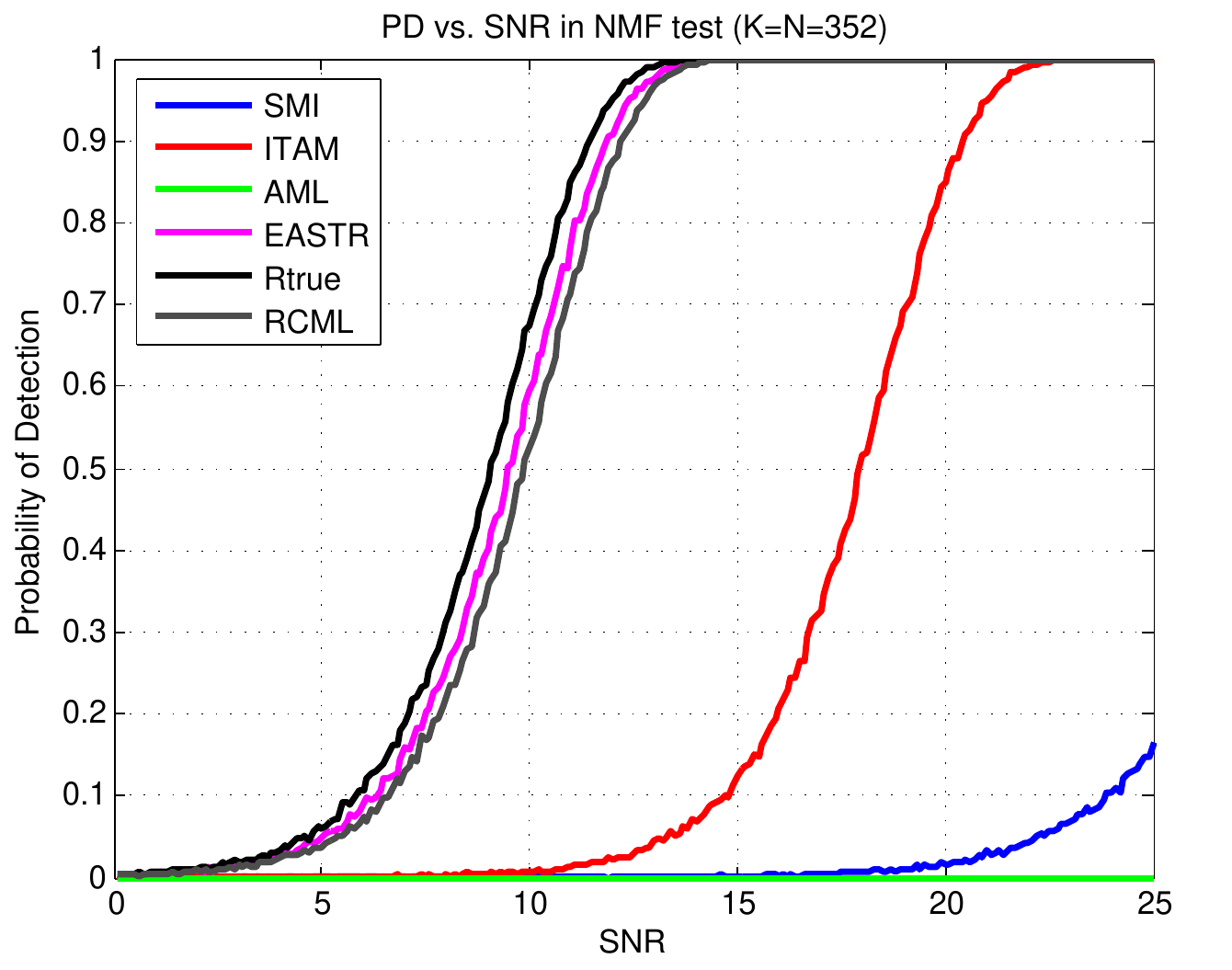}\label{Fig:PD_KASSPER_352}}
\hfil
\subfigure[$K=2N=704$]{\includegraphics[scale=0.6]{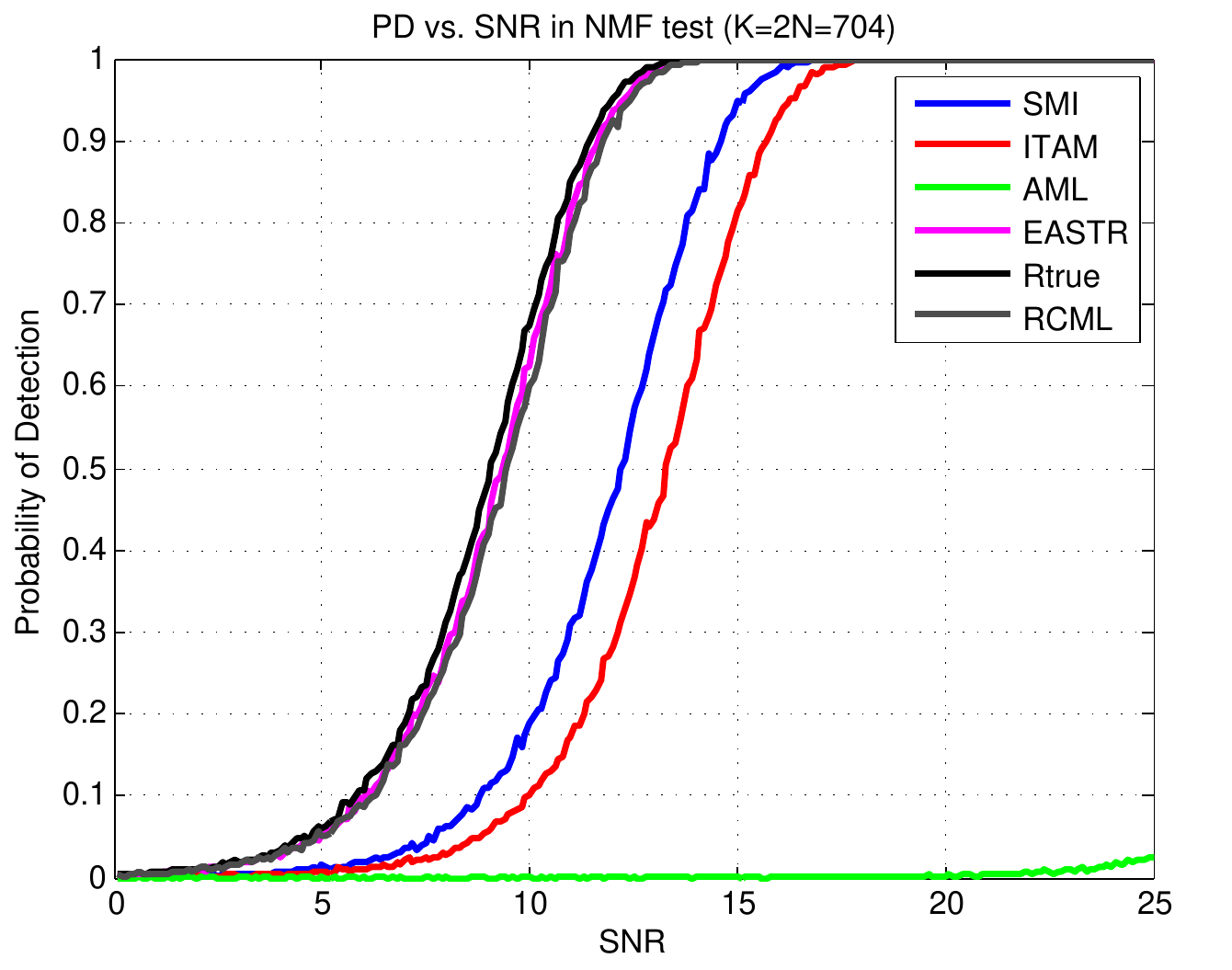}\label{Fig:PD_KASSPER_704}}
\end{center}
\caption{Probability of detection vs. SNR for KASSPER data set via normalized matched filter (NMF) test.}
\label{Fig:PD_KASSPER}
\end{figure}

%\begin{figure}
%\centering
%\includegraphics[scale=0.6]{figs/PD_352_KASSPER_crop.pdf}
%\caption{Probability of detection vs. SNR vis normalized matched filter (NMF) test. $K=2N=40$ is used.}
%\label{Fig:PD_KASSPER_352}
%\end{figure}

%\begin{figure}
%\centering
%\includegraphics[scale=0.6]{figs/PD_450_KASSPER_crop.pdf}
%\caption{Probability of detection vs. SNR vis normalized matched filter (NMF) test. $K=2N=40$ is used.}
%\label{Fig:PD_KASSPER_450}
%\end{figure}

\subsection{Complexity Comparison}
\begin{table}[!t]
\begin{center}
\caption{Running time (sec) for various estimators for simulation model}
\label{Tb:ComplexityEASTR}
\subfigure[$N=20$. Eigenvalue decomposition: $3.4546 \times 10^{-4}$]{ \begin{tabular}{|c|c|c|c|}
  \hline
  $K$ & 20 & 30 & 40\\
  \hline
  SMI & 7.8362 $\times$ $10^{-5}$ & 1.0345 $\times$ $10^{-4}$ & 1.3710 $\times$ $10^{-4}$\\
  \hline
  ITAM & 0.0939 & 0.1024 & 0.1157\\
  \hline
  AML & 0.0615 & 0.0585 & 0.0514\\
  \hline
  SQP & 0.2265 & 0.2128 & 0.2006\\
  \hline
  EASTR & 0.0212 & 0.0210 & 0.0234\\
  \hline
\end{tabular}\label{Tb:ComplexityEASTR20}}\\

\subfigure[$N=40$. Eigenvalue decomposition: 0.0017]{ \begin{tabular}{|c|c|c|c|}
  \hline
  $K$ & 40 & 60 & 80\\
  \hline
  SMI & 1.4791 $\times$ $10^{-4}$ & 5.5810 $\times$ $10^{-4}$ & 3.7152 $\times$ $10^{-4}$\\
  \hline
  ITAM & 0.3915 & 0.4541 & 0.4749\\
  \hline
  AML & 0.6593 & 0.7182 & 0.7012\\
  \hline
  SQP & 2.7752 & 2.8421 & 2.8463\\
  \hline
  EASTR & 0.1115 & 0.1167 & 0.1119\\
  \hline
\end{tabular}\label{Tb:ComplexityEASTR40}}\\

\subfigure[$N=80$. Eigenvalue decompostion: 0.0071]{ \begin{tabular}{|c|c|c|c|}
  \hline
 $K$ & 80 & 120 & 160\\
  \hline
  SMI & 4.7273 $\times$ $10^{-4}$ & 5.1667 $\times$ $10^{-4}$ & 6.7018 $\times$ $10^{-4}$\\
  \hline
  ITAM & 2.6020 & 2.8289 & 2.5226\\
  \hline
  AML & 15.8855 & 16.2998 & 15.9827\\
  \hline
  SQP & 110.3448 & 109.8137 & 114.0377\\
  \hline
  EASTR & 0.4939 & 0.5162 & 0.5447\\
  \hline
\end{tabular}\label{Tb:ComplexityEASTR80}}\\

\end{center}
\end{table}

We compare computational complexity of the compared methods for the simulation model. We take average values of results of 100 trials for each estimator and the experiments are performed on the desktop with Intel Core i7-2600 CPU 3.40 GHz and 8.00 GB RAM. Table \ref{Tb:ComplexityEASTR} shows running times in second for SMI, ITAM, AML, SQP, and EASTR. Tables \ref{Tb:ComplexityEASTR20}, \ref{Tb:ComplexityEASTR40}, and \ref{Tb:ComplexityEASTR80} show the results for various data dimensions $N= 20, 40, 80$, respectively. ITAM and EASTR involve eigenvalue decomposition and running times of eigenvalue decomposition for these experiments are shown in caption lines. EASTR is the cheapest method except the sample covariance matrix even though it shows the best performance in the sense of the normalized SINR and the probability of detection. Though a closed form solution is available for AML, the running time of AML increases exponentially as the data dimension increases since AML involves calculations of matrices of which size is $N^2$. This results shows that EASTR is the best estimator for the performance as well as for computational complexity.

\section{Conclusion}
\label{Sec:Conclusion}

Our work focuses on jointly exploiting a Toeplitz structure as well as a rank constraint on the clutter covariance for radar STAP. The problem is inherently hard because it is well known that there is no closed form solution for ML estimation under Toeplitz constraint for all training regimes. While past work has provided iterative often expensive solutions, we develop a new estimator that is based on a cascade of two closed forms. The first closed form is the recently proposed RCML estimator. Our core contribution, the second step of Toeplitz approximation performs constrained optimization of eigenvalues to either exactly or approximately satisfy the Toeplitz constraint without compromising the rank. Crucially, this optimization also has a closed form making the overall estimator very friendly from a computational standpoint. Via performance analysis evaluating probability of detection, normalized SINR, and trace deviation measure, our estimator is shown to outperform traditional efforts in Toeplitz and low rank covariance estimation including those based on expensive numerical solutions. Recently, the optimality of the fast maximum likelihood \cite{Steiner00} covariance estimator has been proven with respect to cost functions involving the Frobenius or the spectral norm \cite{Aubry13}. EASTR can also be investigated for similar notions of optimality. In addition, more analysis of our estimator such as asymptotic convergence can be performed. Finally, practical evaluation may be performed on other radar data sets involving departures from idealized scenarios. 
\chapter{Robust Covariance Estimation under Imperfect Constraints using Expected Likelihood Approach}
\label{Ch:EL}

\section{Introduction}
\label{Sec:Introduction}
Radar systems using multiple antenna elements and processing multiple pulses are widely used in modern radar signal processing since it helps overcome the directivity and resolution limits of a single sensor. Joint adaptive processing in the spatial and temporal domains for the radar systems, called space time adaptive processing (STAP) \cite{Guerci03,Klemm02,Monzingo04}, enables to suppress interfering signals as well as to preserve gain on the desired signal. Interference statistics, in particular the covariance matrix of the disturbance, which must be estimated from secondary training samples in practice plays a critical role on success of STAP. To obtain a reliable estimate of the disturbance covariance matrix, a large number of homogeneous training samples are necessary. This gives rise to a compelling challenge for radar STAP because such generous homogeneous (target free) training is generally not available in practice \cite{Himed97}.

Much recent research for radar STAP has been developed to overcome this practical limitation of generous homogeneous training. Specifically, the knowledge-based processing which uses \emph{a priori} information about the interference environment is widely referred in the literature \cite{Guerci06,Wicks06} and has merit in the regime of limited training data. These techniques include intelligent training selection \cite{Guerci06} and the spatio-temporal degrees of freedom reduction \cite{Wicks06,Wang91,Gini08}. In addition, covariance matrix estimation techniques the enforce and exploit a particular structure have been pursed as one approach of these techniques. Examples of structure include persymmetry \cite{Nitzberg80}, Toeplitz structure \cite{Li99,Fuhrmann91,Abramovich98}, circulant structure \cite{Conte98}, and eigenstructure \cite{Steiner00,Kang14,Aubry12}. In particular, the fast maximum likelihood (FML) method \cite{Steiner00} which enforces a special eigenstructure that the disturbance covariance matrix represents a scaled identity matrix plus a rank deficient and positive semidefinite clutter component also falls in this category and is shown to be the most competitive technique experimentally.

Recently, the works by Kang \emph{et al.} \cite{Kang14} and Aubry \emph{et al.} \cite{Aubry12} have also improved upon the FML by exploiting practical constraints inspired by physical radar environment, specifically the eigenstructure of the disturbance covariance matrix for radar STAP. They employed a rank of the clutter subspace and a condition number of the interference covariance matrix respectively as a constraint as well as the structural constraint used in the FML into the optimization problem. For both methods, though the initial optimization problems are non-convex, the estimation problems are reduced to a convex optimization problems and admit closed-form solutions. Their methods have also been shown to enable higher normalized SINR over the state-of-the art alternatives for the simulation model and the knowledge-aided sensor signal processing and expert reasoning (KASSPER) data set.

In \cite{Kang14}, the authors assume the rank of the clutter is given by Brennan rule \cite{Ward94} under ideal conditions of no coupling. However, in practice (under non-ideal conditions) the clutter rank departs from the Brennan rule prediction due to antenna errors and internal clutter motion. In this case, the rank is not known precisely and needs to be determined before using with the RCML estimator. Determination of the number of signals in a measurement record is a classical eigenvalue problem, which has received considerable attention in the past 60 years. It is important to note that the problem does not have a simple and unique solution. Consequently, a number of techniques have been developed to address this problem \cite{Akaike74,Rissanen78,Wax85,Yin87,Tufts94}. In addition, the noise level and the condition number should be estimated as well if they are unknown or non precisely known in practice.

Expected likelihood (EL) approach \cite{Abramovich07} has been proposed to determine a regularization parameter based on the statistical invariance property of the likelihood ratio (LR) values. More specifically, the probability distribution function (pdf) of LR values for the true covariance matrix depends on only the number of training samples ($K$) and the dimension of the true covariance matrix ($N$), not the true covariance itself under a Gaussian assumption on the observations. This statistical independence of LR values on the true covariance itself enables pre-calculation of LR values even though the true covariance is unknown. Finally, the regularization parameters are selected so that the LR value of the estimate agrees as closely as possible with the {\em median} LR value determined via its pre-characterized pdf.

\textbf{Contributions:} In view of the aforementioned observations, we develop covariance estimation methods which automatically and adaptively determines the values of practical constraints via an expected likelihood approach for practical radar STAP. Our main contributions are:
\begin{itemize}
    \item \textbf{A method of choice of constraints using the EL approach:} We propose a method of a choice of practical constraints employed in the optimization problems for covariance estimation in radar STAP using the expected likelihood approach. The proposed method guides the selection of the constraints via the expected likelihood criteria in the case that the knowledge of the constraints is imperfectly known in practice. We consider three different cases of the constraints in this chapter: 1) only the clutter rank constraint, 2) both the clutter rank and the noise power constraints, and 3) the condition number constraint.
    \item \textbf{Analytical results with formal proofs for three different cases of imperfect constraints:} For each case mentioned above, we develop significant analytical results. We first formally prove that the rank selection problem based on the expected likelihood approach has a unique solution. This guarantees there is only one rank which is the best (global optimal) rank in the sense of the EL approach. Second, we derive a closed form solution of the optimal noise power in the sense of the EL approach for a given rank. This means we do not need iterative or numerical method to find the optimal noise power and enables fast implementation. Finally, we also prove there exists the unique condition number for the condition number selection method via the EL approach.
    \item \textbf{Experimental Results through simulated model and the KASSPER data set:} Experimental investigation on a simulation model and on the KASSPER data set shows that the proposed methods for three different cases outperform alternatives such as the FML, leading rank selection methods in radar literature and statistics, and the ML estimation of the condition number constraint in the sense of normalized SINR.
\end{itemize}

The rest of the chapter is organized as follows. We provide the proposed methods of the constraint selection problems via the EL approach in Section \ref{Sec:Proposed}. Experimental validation of our method is provided in Section \ref{Sec:Experiments} wherein we report the performance of the proposed method and compare it against existing methods in terms of normalized SINR on both the simulation model and the KASSPER data set.

\section{Constraints selection method via Expected Likelihood Approach}
\label{Sec:Proposed}

\subsection{Imperfect rank constraint}
\label{Sec:Rankonly}

\begin{figure}
\centering
\includegraphics[scale=0.43]{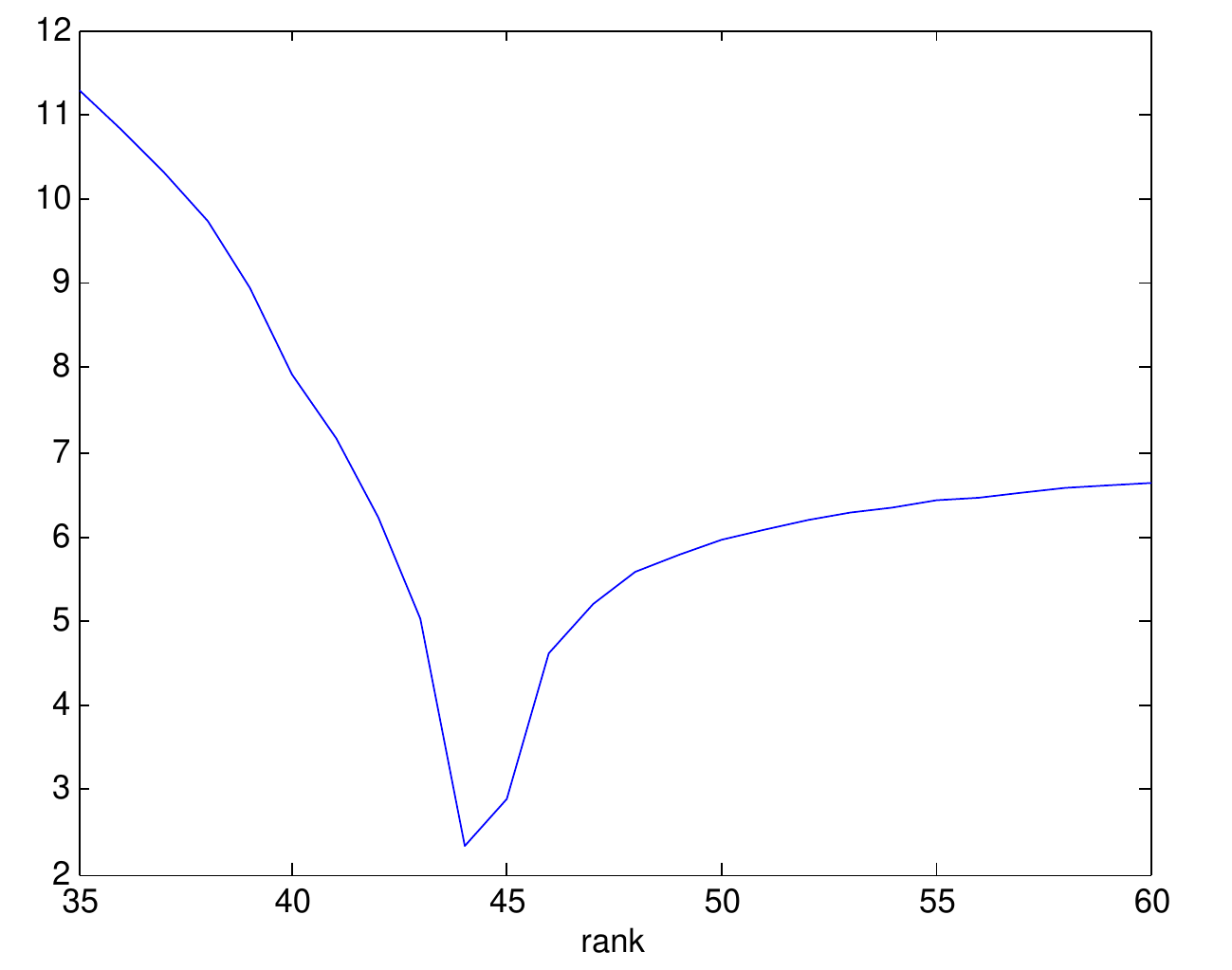}
\caption{$\bigg( \log\Big( \lr\big(\mb R_\text{RCML}(r), \mb Z\big)/\lr_0\Big)\bigg)^2$ versus $r$ for KASSPER dataset ($K=2N=704$)}
\label{Fig:LRdifference}
\end{figure}

In Chapter \ref{Ch:RCML}, we discuss that the RCML estimator is not only powerful in practice but also computationally cheap and the EL approach is shown to be useful to select parameters so that the estimate is consistent with the true covariance matrix in the sense of the LR value in Section \ref{Sec:EL_Background}. From Eq. \eqref{Eq:RCMLsolution}, we see the RCML solution is a function of the rank $r$ and $d_i$'s which are given in the problem. We propose to use the EL approach to refine and find the {\em optimal} rank when the rank determined by underlying physics is not necessarily accurate.

Now we set up the optimization criterion to find the rank via the EL approach. Since the rank is an integer, there may not exist the rank which exactly satisfies Eq. \eqref{Eq:OptimalBeta}. Therefore, we instead find a rank which makes its corresponding LR value closer to $\lr_0$ than any other ranks. That is,
\be
\hat{\mb R}_{\text{RCML}_\text{EL}} = \sigma^2 \mb{V} {\mb\Lambda^\star}^{-1}(\hat r) \mb{V}^H
\ee
where
\be
\label{Eq:OptimalRank}
\hat r \equiv \arg\min_{r \in \mathds{Z}} \Big| \lr\big(\mb R_\text{RCML}(r), \mb Z\big)  - \lr_0 \Big|^2
\ee
and $\lr\big(\mb R_\text{RCML}(r), \mb Z\big)$ is given by Eq. \eqref{Eq:LR_RCML1}.

Now we investigate the optimization problem \eqref{Eq:OptimalRank} for the rank selection. Since the eigenvectors of $\mb R_\text{RCML}$ are identical to those of the sample covariance matrix $\mb S$ as shown in Eq. \eqref{Eq:RCMLsolution2}, the LR value of $\mb R_\text{RCML}$ in Eq. \eqref{Eq:OptimalRank} can be reduced to the function of the eigenvalues of $\mb R_\text{RCML}$ and $\mb S$. Let the eigenvalues of $\mb R_\text{RCML}$ and $\mb S$ be $\lambda_i$ and $d_i$ (arranged in descending order). Then the LR value of $\mb R_\text{RCML}$ can be simplified to a function of ratio of $d_i$ to $\lambda_i$, $\dfrac{d_i}{\lambda_i}$. That is,
\bea
\lr\big(\mb R_\text{RCML}(r), \mb Z\big) & = & \dfrac{|\hat{\mb R}_{\text{RCML}}^{-1}(r) \mb S| \exp N}{\exp \Big( \tr\big[\hat{\mb R}_{\text{RCML}}^{-1}( r) \mb S\big] \Big)}\label{Eq:LR_RCML1}\\
& = & \frac{\ds\prod_{i=1}^N \dfrac{d_i}{\lambda_i} \cdot \exp N}{\exp\Big[\ds\sum_{i=1}^N \dfrac{d_i}{\lambda_i}\Big]}\label{Eq:SimplifiedLR}
\eea
\begin{lem}
\label{Lemma1}
The LR value of the RCML estimator, $\lr\big(\mb R_\text{RCML}(r), \mb Z\big)$, is a monotonically increasing function with respect to the rank $r$ and there is only one unique $\hat r$ in the optimization problem \eqref{Eq:OptimalRank}.
\end{lem}
\begin{proof}
First, let $r$ be the largest $i$ such that $d_{i+1} \geq \sigma^2$. Then, from the closed form solution of the RCML estimator, the eigenvalues of the RCML estimator with rank $i$ and $i+1$ for given $i  < r$ will be
\begin{itemize}
  \item $\hat{\mb R}_\text{RCML}(i)$ : $d_1, \; \; d_2,  \; \; \ldots,  \; \; d_i,  \; \; \sigma^2,  \; \; \ldots,  \; \; \sigma^2$
  \item $\hat{\mb R}_\text{RCML}(i+1)$ : $d_1, \; \; d_2,  \; \; \ldots,  \; \; d_i,  \; \; d_{i+1}, \; \; \sigma^2,  \; \; \ldots,  \; \; \sigma^2$
\end{itemize}
since $d_{i+1} \geq \sigma^2$.
Then $\dfrac{d_i}{\lambda_i}$ should be
\begin{itemize}
  \item $\hat{\mb R}_\text{RCML}(i)$ : $1, \; \; 1,  \; \; \ldots,  \; \; 1_i,  \; \; \dfrac{d_{i+1}}{\sigma^2},  \; \; \ldots,  \; \; \dfrac{d_N}{\sigma^2}$
  \item $\hat{\mb R}_\text{RCML}(i+1)$ : $1, \; \; 1,  \; \; \ldots,  \; \; 1_i,  \; \; 1_{i+1}, \; \; \dfrac{d_{i+2}}{\sigma^2},  \; \; \ldots,  \; \; \dfrac{d_N}{\sigma^2}$
\end{itemize}
From Eq. \eqref{Eq:SimplifiedLR}, the LR values of the RCML estimators with the ranks $i$ and $i+1$ are
\be
\label{Eq:LR_r}
\lr(i) = \frac{\dfrac{\exp N}{\sigma^{2(N-i)}}\ds\prod_{k=i+1}^N d_k }{\exp (i + \dfrac{1}{\sigma^2}\ds\sum_{k=i+1}^N d_k)}
\ee
\be
\label{Eq:LR_r+1}
\lr(i+1) = \frac{\dfrac{\exp N}{\sigma^{2(N-i-1)}}\ds\prod_{k=i+2}^N d_k }{\exp (i + 1 +\dfrac{1}{\sigma^2}\ds\sum_{k=i+2}^N d_k)}
\ee

From Eq. \eqref{Eq:LR_r} and Eq. \eqref{Eq:LR_r+1}, we obtain
\bea
\lr(i+1) & = & \frac{\dfrac{\exp N}{\sigma^{2(N-i-1)}}\ds\prod_{k=i+2}^N d_k }{\exp (i + 1 +\dfrac{1}{\sigma^2}\ds\sum_{k=i+2}^N d_k)}\\
& = & \frac{\dfrac{\exp N}{\sigma^{2(N-i)}}\ds\prod_{k=i+1}^N d_k \cdot \frac{\sigma^2}{d_{i+1}}}{\exp (i +\dfrac{1}{\sigma^2}\ds\sum_{k=i+1}^N d_k) \exp(1-\dfrac{d_{i+1}}{\sigma^2})}\\
& = & \lr(i) \cdot \frac{\sigma^2}{d_{i+1}} \cdot \exp(\frac{d_{i+1}}{\sigma^2}-1)\label{Eq:LRrelation}
\eea
Eq. \eqref{Eq:LRrelation} tells us $\lr (i+1)$ can be calculated by multiplying $\lr (i)$ by the coefficient $\dfrac{\sigma^2}{d_{i+1}} \cdot \exp(\dfrac{d_{i+1}}{\sigma^2}-1)$. Fig. \ref{Fig:increase} shows that
\be
\dfrac{\sigma^2}{d_{i+1}} \cdot \exp(\dfrac{d_{i+1}}{\sigma^2}-1) \geq 1
\ee
for all values of $\dfrac{\sigma^2}{d_{i+1}}$. Therefore, it is obvious that
\be
\lr(i+1) \geq \lr(i),
\ee
which means the LR value monotonically increases with respect to $i$.

Now, let's consider the other case, $i \geq r$. In this case, since $d_{i+1} < \sigma^2$, it is easily shown that
\be
\mb R_\text{RCML}(i) = \mb R_\text{RCML}(i+1)
\ee
Therefore,
\be
\lr(i+1) = \lr(i)
\ee

This proves that $\lr(i)$ monotonically increases for all $1 \leq i \leq N$.

\end{proof}

\begin{figure}
\centering
\includegraphics[scale=0.5]{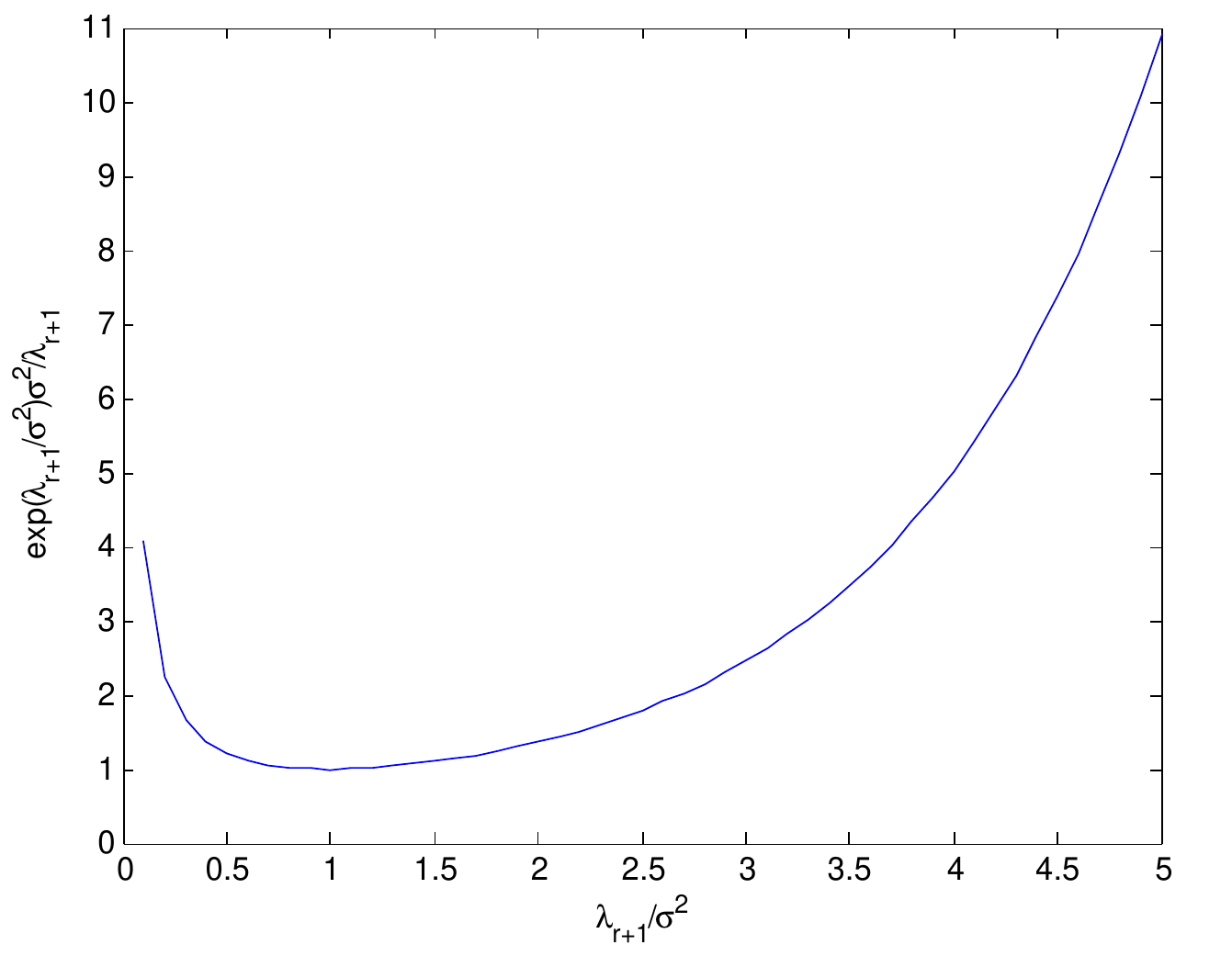}
\caption{The value of the coefficient $\dfrac{\sigma^2}{\lambda_{r+1}} \cdot \exp(\dfrac{\lambda_{r+1}}{\sigma^2}-1)$}
\label{Fig:increase}
\end{figure}

Lemma \ref{Lemma1} gives us a significant analytical result that is the EL approach leads to a unique value of the rank, i.e., when searching over the various values of the rank it is impossible to come up with multiple choices. That also means that it is guaranteed that we can always find the global optimum of $r$ not local optima (minima) for the optimization problem \eqref{Eq:OptimalRank} regardless of an initial value of $r$. We plot the values of $\bigg( \log\Big( \lr\big(\mb R_\text{RCML}(r), \mb Z\big)/\lr_0\Big)\bigg)^2$ versus the rank $r$ for one realization for the KASSPER dataset ($K=2N=704$) in Fig. \ref{Fig:LRdifference}. Since the LR values are too small in this case, we use a log scale and the ratio between two instead of the distance to see the variation clearly. Note that monotonic increase of the value of $\lr\big(\mb R_\text{RCML}(r), \mb Z\big)$ w.r.t $r$ guarantees a unique optimal rank even if the optimization function as defined in (\ref{Eq:OptimalRank}) is not necessarily convex in $r$.

The algorithm to find the optimal rank is simple and not computationally expensive due to the analytical results above. For a given initial rank such as Brennan rule for the KASSPER data set and the number of jammers for a simulation model, we first determine a direction of searching and then find the optimal rank. The procedure of finding the optimal rank is shown in Algorithm 1 in detail.

\begin{algorithm}[t]
\caption{The proposed algorithm to select the rank via EL }
\label{Alg:Rank}
    \begin{algorithmic}[1]
        \STATE Initialize the rank $r$ by physical environment such as Brennan rule.
        \STATE Evaluate $\lr(r-1)$, $\lr(r)$, $\lr(r+1))$, the LR values of RCML estimators for the ranks $r-1$, $r$, $r+1$, respectively.
            \begin{itemize}
                \item if $|\lr(r+1) - \lr_0| < |\lr(r) - \lr_0|$\\
                $\rightarrow$ increase $r$ by 1 until $|\lr(r) - \lr_0|$ is minimized to find $\hat r$.
                \item elseif $|\lr(r-1) - \lr_0| < |\lr(r) - \lr_0|$\\
                $\rightarrow$ decrease $r$ by 1 until $|\lr(r) - \lr_0|$ is minimized to find $\hat r$.
                \item else $\hat r = r$, the initial rank.
            \end{itemize}
    \end{algorithmic}
\end{algorithm}

\subsection{Imperfect rank and noise power constraints}
\label{Sec:Both}

In this section, we investigate the second case that both the rank $r$ and the noise power $\sigma^2$ are not perfectly known. We propose the estimation of both the rank and the noise level based on the EL approach. The estimator with both the rank and the noise power obtained by the EL approach is given by
\be
\hat{\mb R}_{\text{RCML}_\text{EL}} = \hat \sigma^2 \mb{V} {\mb\Lambda^\star}^{-1}(\hat r) \mb{V}^H
\ee
where
\be
\label{Eq:OptimalRankSigma}
(\hat r, \hat \sigma^2) \equiv \arg\min_{r \in \mathds{Z}, \sigma^2 > 0} \Big| \lr\big(\mb R_\text{RCML}(r,\sigma^2), \mb Z\big)  - \lr_0 \Big|^2
\ee

In section \ref{Sec:Rankonly}, we have shown that the optimal rank via the EL approach is uniquely obtained for a fixed $\sigma^2$. Now we analyze the LR values of the RCML estimator for various $\sigma^2$ and a fixed rank.

\begin{lem}
\label{Lemma2}
For a fixed rank, the LR value of the RCML estimator, which is a function of $\sigma^2$, has a maximum value at $\sigma^2 = \sigma_{\text{ML}}^2$. It monotonically increases for $\sigma^2 < \sigma_{\text{ML}}^2$ and monotonically decreases for $\sigma^2 > \sigma_{\text{ML}}^2$.
\end{lem}
\begin{proof}
In this section, I investigate the LR values for varying noise level $\sigma^2$ and a given rank $r$. From Eq. \eqref{Eq:LR_r} we obtain the LR value when the rank is $r$,
\be
\label{Eq:LR_sigma}
\lr(\sigma^2) = \frac{\dfrac{\exp N}{\sigma^{2(N-r)}}\ds\prod_{k=r+1}^N d_k }{\exp (r + \dfrac{1}{\sigma^2}\ds\sum_{k=r+1}^N d_k)}
\ee
For simplicity, let $\sigma^2 = t$ then Eq. \eqref{Eq:LR_sigma} can be simplified as
\be
\lr(t) = \dfrac{e^{N-r}\ds\prod_{k=r+1}^N d_k }{t^{N-r} e^{\dfrac{\sum_{k=r+1}^N d_k}{t}}}
\ee
Now let $\ds\sum_{k=r+1}^N d_k = d_s$ and $\ds\prod_{k=r+1}^N d_k = d_p$, then
\bea
\lr(t) & = & \dfrac{e^{N-r}d_p }{t^{N-r} e^{\frac{d_s}{t}}}\\
& = & d_p e^{N-r} t^{r-N} e^{-\frac{d_s}{t}}\label{Eq:LR_t}
\eea
To analyze increasing or decreasing property Eq. \eqref{Eq:LR_t}, I calculate its first derivative. Since $d_p e^{N-r}$ is a positive constant, it does not affect increasing or decreasing of the function. Therefore,
\bea
\lefteqn{(t^{r-N} e^{-\frac{d_s}{t}})^\prime}\nonumber\\
& = & (r-N)t^{r-N-1}e^{-d_s/t} + t^{r-N} e^{-d_s/t} \dfrac{d_s}{t^2}\\
& = & (r-N)t^{r-N-1}e^{-d_s/t} + t^{r-N-2} e^{-d_s/t} d_s\\
& = & t^{r-N-2}\big((r-N)t + d_s\big) e^{-d_s/t}
\eea
Since $t^{r-N-2}$ and $e^{-d_s/t}$ are always positive, the first derivative $(t^{r-N} e^{-\frac{d_s}{t}})^\prime = 0$ if and only if
\be
t = \dfrac{d_s}{N-r} = \dfrac{\sum_{k=r+1}^N d_k}{N-r}
\ee
and it is positive when $t<\dfrac{\sum_{k=r+1}^N d_k}{N-r}$ and negative otherwise. This means that $\lr(\sigma^2)$ increases for $\sigma^2<\dfrac{\sum_{k=r+1}^N d_k}{N-r}$ and decreases for $\sigma^2>\dfrac{\sum_{k=r+1}^N d_k}{N-r}$. The LR value is maximized when $\sigma^2 = \dfrac{\sum_{k=r+1}^N d_k}{N-r}$. Note that $\dfrac{\sum_{k=r+1}^N d_k}{N-r}$ is the average value of $N-r$ smallest eigenvalues of the sample covariance matrix and in fact a maximum likelihood solution of $\sigma^2$ as shown in the RCML estimator \cite{Kang14}.
\end{proof}

Fig. \ref{Fig:LR_sigma} shows an example of the LR values as a function of the noise level $\sigma^2$. As shown in Lemma \ref{Lemma2}, we see that the LR value is maximized for the ML solution of $\sigma^2$ and monotonically increases and decreases for each direction. It is obvious that we have three cases of the solution of the optimal noise power from Lemma \ref{Lemma2}: 1) no solution, 2) only one solution, and 3) two optimal solution. Now we discuss how to obtain the optimal noise power for a fixed rank.
\begin{lem}
\label{Lemma3}
The noise power obtained by the expected likelihood approach, $\hat{\sigma}_{\text{EL}}^2$, is given by
\be
\hat{\sigma}_{\text{EL}}^2 = \exp\Bigg(W_k\bigg(\frac{b}{a} e^{-\frac{c}{a}}\bigg) + \frac{c}{a}\Bigg)
\ee
where $W_k (z)$ is the $k$-th branch of Lambert $W$ function and
\be
\left\{\begin{array}{l}
a = r - N\\
b = \sum_{k=r+1}^N d_k\\
c = \log \lr_0 - \log\Big (\prod_{k=r+1}^N d_k\Big) + a
\end{array} \right.
\ee
\end{lem}
\begin{proof}
For a given rank $r$, the optimal solution of the noise power via the EL approach, $\hat t (=\hat{\sigma}_{\text{EL}}^2)$, is the solution of $\lr(t) = \lr_0$. From Eq. \eqref{Eq:LR_t}, that is, $\hat t$ is the solution of the equation given by
\be
d_p e^{N-r} t^{r-N} e^{-\frac{d_s}{t}} = \lr_0
\ee
Taking $\log$ on both side leads
\be
\label{Eq:OptimalT}
\log d_p + N-r + (r-N)\log t -\frac{d_s}{t} = \log \lr_0
\ee
For simplification, we take substitutions of variables,
\be
\left\{ \begin{array}{l}
a = r - N\\
b = \sum_{k=r+1}^N d_k\\
c = \log \lr_0 - \log\Big (\prod_{k=r+1}^N d_k\Big) + a
\end{array} \right.
\ee
Then, Eq. \eqref{Eq:OptimalT} is simplified to an equation of $t$,
\be
a \log t - \frac{b}{t} = c
\ee
Again, let $u = \log t$. Then, since $t = e^u$, we obtain
\bea
au - be^{-u} = c\\
e^{-u} = \frac{a}{b} u - \frac{c}{b}
\eea
Now let $s=u-\frac{c}{a}$. Then, the equation is
\bea
e^{-s-\frac{c}{a}} = \frac{a}{b} s\\
s e^s = \frac{b}{a} e^{-\frac{c}{a}}\label{Eq:S}
\eea
The solution of Eq. \eqref{Eq:S} is known to be obtained using Lambert $W$ function \cite{Corless96}. That is,
\be
s = W\bigg(\frac{b}{a} e^{-\frac{c}{a}}\bigg)
\ee
where $W(\cdot)$ is a Lambert $W$ function which is defined to be the function satisfying
\be
W(z) e^{W(z)} = z
\ee
Finally, we obtain
\be
u = W\bigg(\frac{b}{a} e^{-\frac{c}{a}}\bigg) + \frac{c}{a}
\ee
and
\be
\hat{\sigma}_{\text{EL}}^2 = \hat t = \exp\Bigg(W\bigg(\frac{b}{a} e^{-\frac{c}{a}}\bigg) + \frac{c}{a}\Bigg)
\ee
\end{proof}
Lemma \ref{Lemma3} shows that there is a closed form solution of the optimal noise power for a fixed rank. Therefore we do not need any iterative and numerical algorithms to obtain both the optimal rank and noise power.

Now we propose the method to alternately find the optimal solution of both the rank and the noise power. For a fixed $\sigma^2$, we can obtain the optimal rank via Algorithm 1. For a fixed rank, we should consider three cases described above. The first case implies that the LR value corresponding $\sigma_{\text{ML}}^2$ is less than $\lr_0$ and therefore, we increase the rank until the solution of $\sigma^2$ exists. In the second case, we can easily determine $\hat \sigma^2 = \sigma_{\text{ML}}^2$. For the third case that there are two solutions of $\sigma^2$, we have to choose one among two EL solutions and the ML solution. We experimentally observe that the threshold in the test statistics such as the normalized matched filter is typically smaller for the better estimator in the sense of the normalized SINR and the probability of detection from our experiments. Therefore, we choose one of $\sigma_{\text{ML}}^2$, $\sigma_{\text{EL1}}^2$, $\sigma_{\text{EL2}}^2$, which generates the smallest value of the test statistics. The detail procedure of the algorithm is described in Algorithm 2.

\begin{figure}
\centering
\includegraphics[scale=0.5]{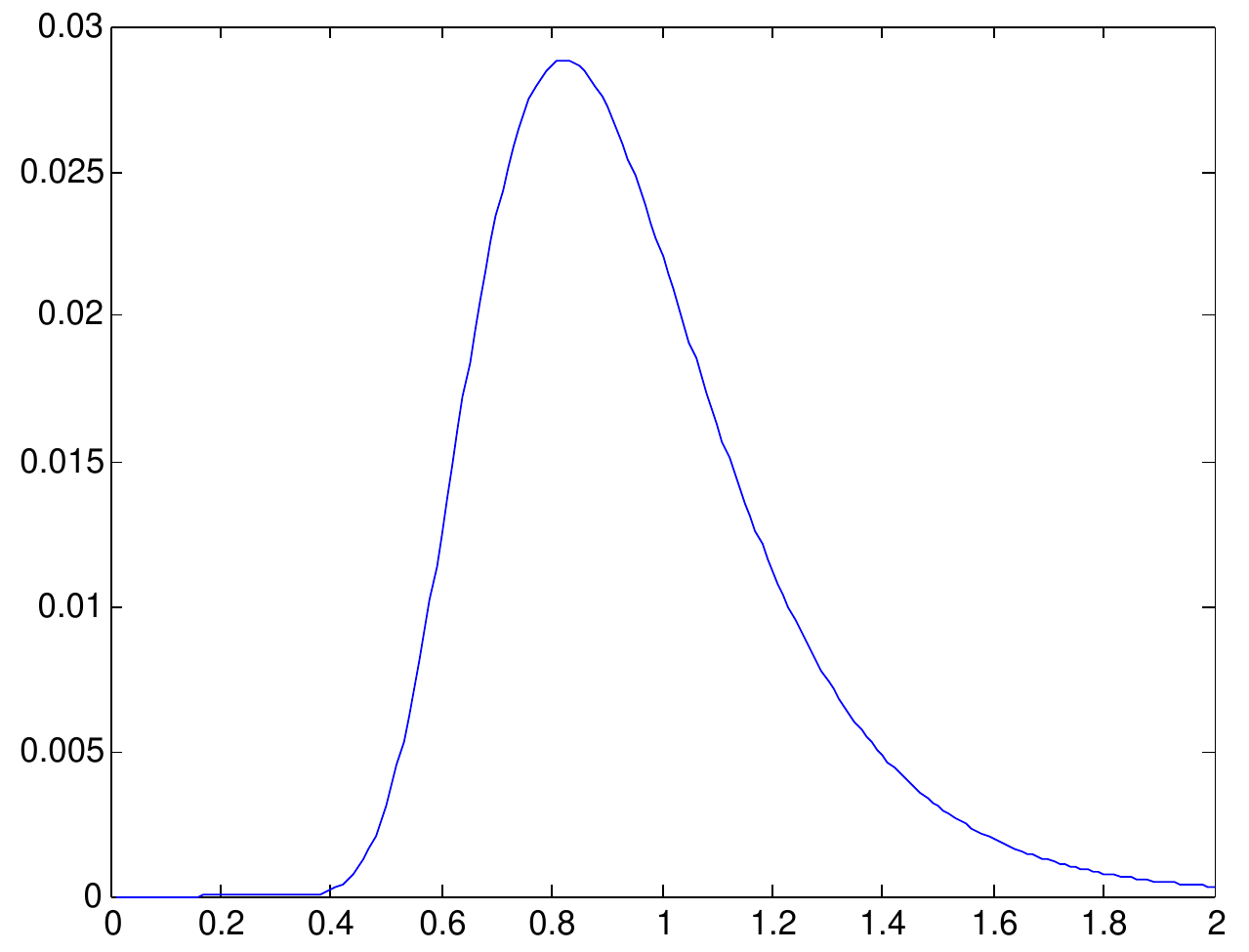}
\caption{The LR value versus $\sigma^2$ for the simulation model, $N=20$, $K=40$, $r=5$}
\label{Fig:LR_sigma}
\end{figure}

\begin{algorithm}[t]
\caption{The proposed algorithm to select the rank and the noise level via EL}
\label{Alg:RankNoise}
    \begin{algorithmic}[1]
        \STATE Initialize the rank $r$ by physical environment such as Brennan rule or the number of jammers.
        \STATE If there is no solution of $\sigma^2$ for given $r$, increase $r$ until the solution of $\sigma^2$ exists.
        \STATE Obtain $\sigma_{\text{ML}}^2 =\frac{1}{N-r} \sum_{i=r+1}^N d_i$.
        \STATE For given $\sigma_{\text{ML}}^2$, find a new $r$ using Algorithm 1.
        \STATE Repeat Step 3 and Step 4 until the rank $r$ converges.
        \STATE After $r$ is determined, choose $\hat \sigma^2$ among $\sigma_{\text{ML}}^2$, $\sigma_{\text{EL1}}^2$, $\sigma_{\text{EL2}}^2$.
    \end{algorithmic}
\end{algorithm}

\subsection{Imperfect condition number constraint}
\label{Sec:ConditionNumber}

Now we discuss the proposed method to determine the condition number constraint through the EL approach in this section. As shown in Eq. \eqref{Eq:CN1} through Eq. \eqref{Eq:Giu2}, the condition number constrained ML estimator is a function of $u$ which is a function of the condition number $K_{\max}$. Therefore, the final estimate is also a function of $K_{\max}$. Similar to what we have done in previous sections, we find an optimal condition number so that the LR value of the estimated covariance matrix should be same as a statistical median value of the LR value of the true covariance matrix, that is
\be
\hat{\mb R}_{\text{CNCML}_\text{EL}} = \hat \sigma^2 \mb{V} {\mb\Lambda^\star}^{-1}(\hat K_{\max}) \mb{V}^H
\ee
where
\be
\label{Eq:OptimalCN}
\hat K_{\max} \equiv \arg\min_{\Kmax \geq 1} \Big| \lr\big(\mb R_\text{CNCML}(\Kmax), \mb Z\big)  - \lr_0 \Big|^2
\ee

Before we discuss the algorithm to find the optimal condition number, we analyze the closed form solution for the condition number constrained ML estimation which is proposed in \cite{Aubry12}. We derive a more explicit closed form solution.

\begin{lem}
\label{Lemma4}
The more simplified closed form solution of the condition number constrained ML estimator is given by
   \begin{enumerate}
    \item $d_1\leq \sigma^2$,
    \be
    \hat{\mb R}_{\text{CN}} = \sigma^2 \mb I
    \ee

    \item $\sigma^2 \leq d_1 \leq \sigma^2 K_{\max}$,
    \be
    \hat{\mb R}_{\text{CN}} = \hat{\mb R}_{\text{FML}}
    \ee

    \item $d_1 > \sigma^2 K_{\max}$ and $K_{\max} \geq \frac{\sum_{i=1}^c d_i}{c - \sum_{\bar N + 1}^N (d_i -1)}$,
    \be
    \hat{\mb R}_{\text{CN}} = \mb\Phi \diag(\bs\lambda^*) \mb\Phi^H
    \ee
    where
    \be
    \bs\lambda^\star = \big[ \sigma^2 K_{\max}, \ldots, \sigma^2 K_{\max}, d_{c+1}, \ldots, d_{\bar N}, \sigma^2,\ldots,\sigma^2  \big],
    \ee
     $c$ and $\bar N$ are the vector of the eigenvalues of the estimate, the largest indices so that $d_c > \sigma^2 K_{\max}$, and $d_{\bar N} \geq \sigma^2$

    \item $d_1 > \sigma^2 K_{\max}$ and $K_{\max} < \frac{\sum_{i=1}^c d_i}{c - \sum_{\bar N + 1}^N (d_i -1)}$,
    \be
    \bs\lambda^\star = \big[ \frac{\sigma^2}{u}, \ldots, \frac{\sigma^2}{u}, d_{p+1}, \ldots, d_q,\frac{\sigma^2}{uK_{\max}},\ldots,\frac{\sigma^2}{uK_{\max}}  \big]
    \ee
    \end{enumerate}
And the condition numbers of the estimates are $1$, $\frac{d_1}{\sigma^2}$, $\Kmax$, and $\Kmax$, respectively.
\end{lem}
\begin{proof}
We consider 5 cases provided in \cite{Aubry12}.
\begin{enumerate}
    \item $d_1 \leq \sigma^2 \leq \sigma^2 \Kmax$\\
    Since $u^\star = \frac{1}{\Kmax}$,
    \bea
    \lambda_i^\star & = & \min (\min (\Kmax  u^\star,1),\max(u^\star,\frac{1}{\bar d_i}))\\
    & = & \min (\min (1,1),\max(\frac{1}{\Kmax},\frac{1}{\bar d_i}))\\
    & = & \min (1,\frac{1}{\bar d_i}) = 1
    \eea
    Therefore,
    \be
    \hat{\mb R}_{\text{CN}} = \sigma^2 \mb I
    \ee
    and the condition number is $1$.
    \item $\sigma^2 < d_1 \leq \Kmax$\\
    Since $u^\star = \frac{1}{\bar d_1}$,
    \bea
    \lambda_i^\star & = & \min (\min (\Kmax  u^\star,1),\max(u^\star,\frac{1}{\bar d_i}))\\
    & = & \min (\min (\frac{\Kmax}{\bar d_1},1),\max(\frac{1}{\bar d_1},\frac{1}{\bar d_i}))\\
    & = & \min (1,\frac{1}{\bar d_i})\\
    & = & \left\{ \begin{array}{cc} \frac{1}{\bar d_i} & \bar d_i \geq 1\\ 1 & \bar d_i < 1 \end{array} \right.
    \eea
    Therefore,
    \be
    \hat{\mb R}_{\text{CN}} = \hat{\mb R}_{\text{FML}}
    \ee
    and the condition number is $\frac{d_1}{\sigma^2}$.

    \item $d_1 > \sigma^2 \Kmax$ and $u^\star = \frac{1}{\bar d_1}$\\
    Since $u^\star$ is the optimal solution of the optimization problem \eqref{Eq:OptimizationU}, $\frac{dG(u)}{du}|_{u=\frac{1}{\bar d_1}}$ must be zero if $u^\star = \frac{1}{\bar d_1}$. From, Eq. \eqref{Eq:Giu1} and Eq. \eqref{Eq:Giu2}, the first derivative of $G_i(u)$ is given by
    \be
    \label{Eq:Gpiu1}
    G_i^\prime(u) = \left\{ \begin{array}{ll}
    -\frac{1}{u} + \Kmax \bar d_i & \text{if} \quad 0 < u \leq \frac{1}{\Kmax}\\
    0 & \text{if} \quad \frac{1}{\Kmax} \leq u \leq 1 \end{array} \right.
    \ee
    for $\bar d_i \leq 1$, and
    \be
    \label{Eq:Gpiu2}
    G_i^\prime(u) = \left\{ \begin{array}{ll}
    -\frac{1}{u} + \Kmax \bar d_i & \text{if} \quad 0 < u \leq \frac{1}{\Kmax \bar d_i}\\
    0 & \text{if} \quad \frac{1}{\Kmax \bar d_i} < u \leq \frac{1}{\bar d_i}\\
    -\frac{1}{u} + \bar d_i & \text{if} \quad \frac{1}{\bar d_i} \leq u \leq 1 \end{array} \right.
    \ee
    for $\bar d_i > 1$. Therefore,
    \be
   \frac{dG(u)}{du}|_{u=\frac{1}{\bar d_1}} = \sum_{i=\bar N +1}^N (\Kmax \bar d_i - \bar d_1) + \sum_{i=p}^{\bar N} (\Kmax \bar d_i - \bar d_1)
    \ee
    where $p$ is the greatest index such that $\frac{1}{\bar d_1} < \frac{1}{\Kmax \bar d_p}$. For $i=\bar N, \ldots, N$, since $\bar d_i \leq 1$,
    \be
    \Kmax \bar d_i - \bar d_1 < \Kmax - \bar d_1 < 0
    \ee
    and for $i=p, \ldots, \bar N -1$, since $\bar d_1 > \Kmax \bar d_i$, $\Kmax \bar d_i - \bar d_1 < 0$. Therefore, in this case, it is obvious that
    \be
    \frac{dG(u)}{du}|_{u=\frac{1}{\bar d_1}} < 0
    \ee
    which implies $u = \frac{1}{\bar d_1}$ can not be the optimal solution of \eqref{Eq:OptimizationU}.

    \item $d_1 > \sigma^2 \Kmax$ and $u^\star = \frac{1}{\Kmax}$\\
    Aubry \emph{et al.} \cite{Aubry12} showed that $u^\star = \frac{1}{\Kmax}$ if $\frac{dG(u)}{du}|_{u=\frac{1}{\Kmax}} \leq 0$. From Eq. \eqref{Eq:Gpiu1} and Eq. \eqref{Eq:Gpiu2},
    \be
    \frac{dG(u)}{du}|_{u=\frac{1}{\Kmax}} = \sum_{i=\bar N +1}^N \Kmax (\bar d_i - 1) + \sum_{i=1}^p (\bar d_i - \Kmax)
    \ee
    where $p$ is the greatest index such that $\bar d_p > \Kmax$. Therefore,
   \bea
    & \frac{dG(u)}{du}|_{u=\frac{1}{\Kmax}} \leq 0\\
    \Leftrightarrow & \ds\sum_{i=\bar N +1}^N \Kmax (\bar d_i - 1) + \ds\sum_{i=1}^p (\bar d_i - \Kmax) \leq 0\\
    \Leftrightarrow & \Kmax(\sum_{i=\bar N +1}^N (\bar d_i - 1) - p) + \sum_{i=1}^p \bar d_i \leq 0\\
    \Leftrightarrow & \Kmax(\sum_{i=\bar N +1}^N (\bar d_i - 1) - p) \leq -\sum_{i=1}^p \bar d_i\\
    \Leftrightarrow & \Kmax \geq \frac{\sum_{i=1}^p \bar d_i}{p-\sum_{i=\bar N +1}^N (\bar d_i - 1)}
    \eea
    In this case,
    \bea
    \lambda_i^\star & = & \min (\min (\Kmax  u^\star,1),\max(u^\star,\frac{1}{\bar d_i}))\\
    & = & \min (\min (1,1),\max(\frac{1}{\Kmax},\frac{1}{\bar d_i}))\\
    & = & \min (1,\max(\frac{1}{\Kmax},\frac{1}{\bar d_i}))\\
     & = & \left\{ \begin{array}{cc} \min(1,\frac{1}{\Kmax}) & \bar d_i \geq \Kmax\\ \min(1,\frac{1}{\bar d_i}) & \bar d_i < \Kmax \end{array} \right.\\
     & = & \left\{ \begin{array}{cc} \frac{1}{\Kmax} & \bar d_i \geq \Kmax\\ \frac{1}{\bar d_i} & \bar 1 \leq \bar d_i < \Kmax\\ 1 & \bar d_i < 1 \end{array} \right.
    \eea
    Finally we obtain
    \be
    \bs\lambda^\star = \big[ \sigma^2 K_{\max}, \ldots, \sigma^2 K_{\max}, d_{p+1}, \ldots, d_{\bar N}, \sigma^2,\ldots,\sigma^2  \big],
    \ee
     where $p$ and $\bar N$ are the largest indices so that $d_p > \sigma^2 K_{\max}$ and $d_{\bar N} \geq \sigma^2$, respectively.

     \item $d_1 > \sigma^2 K_{\max}$ and $ \Kmax < \frac{\sum_{i=1}^p \bar d_i}{p-\sum_{i=\bar N +1}^N (\bar d_i - 1)}$\\
     In this case, since $\frac{1}{\bar d_1} < u^\star < \frac{1}{\Kmax}$,
      \bea
    \lambda_i^\star & = & \min (\min (\Kmax  u^\star,1),\max(u^\star,\frac{1}{\bar d_i}))\\
    & = & \min (\Kmax  u^\star,\max(u^\star,\frac{1}{\bar d_i}))\\
     & = & \left\{ \begin{array}{cc} \min(\Kmax  u^\star,u^\star) & \bar d_i \geq \frac{1}{u^\star}\\ \min(\Kmax  u^\star,\frac{1}{\bar d_i}) & \bar d_i < \frac{1}{u^\star} \end{array} \right.\\
     & = & \left\{ \begin{array}{cc} u^\star & \bar d_i \geq \frac{1}{u^\star}\\ \frac{1}{\bar d_i} & \frac{1}{\Kmax u^\star} \leq \bar d_i \leq \frac{1}{u^\star}\\ \Kmax u^\star & \bar d_i < \frac{1}{\Kmax u^\star} \end{array} \right.
    \eea
    Therefore, we obtain
     \be
    \bs\lambda^\star = \big[ \frac{\sigma^2}{u^\star}, \ldots, \frac{\sigma^2}{u^\star}, d_{p+1}, \ldots, d_q,\frac{\sigma^2}{u^\star K_{\max}},\ldots,\frac{\sigma^2}{u^\star K_{\max}}  \big]
    \ee
    where $p$ and $q$ are the largest indices so that $d_p > \frac{\sigma^2}{u}$ and $d_q > \frac{\sigma^2}{u \Kmax}$, respectively.
\end{enumerate}

\end{proof}

From Lemma \ref{Lemma4}, for the first two cases that is $d_1 \leq \sigma^2 K_{\max}$, the estimator is either a scaled identity matrix or the FML. Therefore, there is no need to find an optimal condition number in these cases since the estimator is not a function of the condition number.

Now we investigate uniqueness of the optimal condition number as we have done in the case of only rank constraint for the last two cases where the optimal eigenvalues are functions of the condition number.

\begin{lem}
\label{Lemma5}
The LR value of the condition number ML estimator is a monotonically increasing function with respect to the condition number $K_{\max}$ and there is only one unique $K_{{\max}_{\text{EL}}}$.
\end{lem}
\begin{proof}
\begin{enumerate}
   \item $d_1\leq \sigma^2$
    \be
    \hat{\mb R}_{\text{CN}} = \sigma^2 \mb I
    \ee
    In this case, $\hat{\mb R}_{\text{CN}}$ does not change, so $\lr(\Kmax)$ is a constant.

    \item $\sigma^2 \leq d_1 \leq \sigma^2 K_{\max}$
    \be
    \hat{\mb R}_{\text{CN}} = \hat{\mb R}_{\text{FML}}
    \ee
    In this case, $\hat{\mb R}_{\text{CN}}$ does not change, so $\lr(\Kmax)$ is a constant.

    \item $d_1 > \sigma^2 K_{\max}$ and $K_{\max} \geq \frac{\sum_{i=1}^p d_i}{c - \sum_{\bar N + 1}^N (d_i -1)}$
    \be
    \hat{\mb R}_{\text{CN}} = \mb\Phi \diag(\bs\lambda^*) \mb\Phi^H
    \ee
    where
    \be
    \bs\lambda^\star = \big[ \sigma^2 K_{\max}, \ldots, \sigma^2 K_{\max}, d_{p+1}, \ldots, d_{\bar N}, \sigma^2,\ldots,\sigma^2  \big],
    \ee
     $p$ and $\bar N$ are the largest indices so that $d_p > \sigma^2 K_{\max}$ and $d_{\bar N} \geq \sigma^2$, respectively.
     \beaa
     \lefteqn{\lr(\Kmax)}\nonumber\\ & = & \frac{\prod_{i=1}^N \frac{d_i}{\lambda_i} e^N }{ \exp (\sum_{i=1}^N \frac{d_i}{\lambda_i}) }\\
     & = & \frac{\ds\prod_{i=1}^p \frac{d_i}{\sigma^2 \Kmax} \cdot \ds\prod_{i=p+1}^{\bar N} 1\cdot \ds\prod_{i=\bar N + 1}^N \frac{d_i}{\sigma^2} \cdot e^N}{ \exp(\ds\sum_{i=1}^p \frac{d_i}{\sigma^2 \Kmax} + \ds\sum_{i=p+1}^{\bar N} 1 + \ds\sum_{i=\bar N + 1}^N \frac{d_i}{\sigma^2} )}\\
     & = & \frac{\prod_{i=1}^p \frac{d_i}{\sigma^2 \Kmax} \cdot \prod_{i=\bar N + 1}^N \frac{d_i}{\sigma^2} \cdot e^N}{ \exp(\ds\sum_{i=1}^p \frac{d_i}{\sigma^2 \Kmax}) \cdot e^{\bar N - p}  \cdot \exp( \ds\sum_{i=\bar N + 1}^N \frac{d_i}{\sigma^2} )}
     \eeaa
%     \bea
%     \lefteqn{\lr(\Kmax)}\nonumber\\ & = & \frac{\prod_{i=1}^N \frac{d_i}{\lambda_i} e^N }{ \exp (\sum_{i=1}^N \frac{d_i}{\lambda_i}) }\\
%     & = & \frac{\ds\prod_{i=1}^p \frac{d_i}{\sigma^2 \Kmax} \cdot \ds\prod_{i=p+1}^{\bar N} 1\cdot \ds\prod_{i=\bar N + 1}^N \frac{d_i}{\sigma^2} \cdot e^N}{ \exp(\ds\sum_{i=1}^p \frac{d_i}{\sigma^2 \Kmax} + \ds\sum_{i=p+1}^{\bar N} 1 + \ds\sum_{i=\bar N + 1}^N \frac{d_i}{\sigma^2} )}\\
%     & = & \frac{\prod_{i=1}^p \frac{d_i}{\sigma^2 \Kmax} \cdot \prod_{i=\bar N + 1}^N \frac{d_i}{\sigma^2} \cdot e^N}{ \exp(\ds\sum_{i=1}^p \frac{d_i}{\sigma^2 \Kmax}) \cdot e^{\bar N - p}  \cdot \exp( \ds\sum_{i=\bar N + 1}^N \frac{d_i}{\sigma^2} )}
%     \eea

     \begin{enumerate}
        \item within the range where $p$ remains same
       \beaa
        \lefteqn{\lr(\Kmax)}\nonumber\\  & = & \frac{\prod_{i=1}^p \frac{d_i}{\sigma^2 \Kmax} \cdot \prod_{i=\bar N + 1}^N \frac{d_i}{\sigma^2} \cdot e^N}{ \exp(\ds\sum_{i=1}^p \frac{d_i}{\sigma^2 \Kmax}) \cdot e^{\bar N - p}  \cdot \exp( \ds\sum_{i=\bar N + 1}^N \frac{d_i}{\sigma^2} )}\\
         & = & c_1 \frac{\prod_{i=1}^p \frac{d_i}{\sigma^2 \Kmax}}{\exp(\sum_{i=1}^p \frac{d_i}{\sigma^2 \Kmax})}\\
        & = & c_1 \frac{\frac{1}{(\sigma^2 \Kmax)^p} \prod_{i=1}^p d_i}{\exp( \frac{1}{\sigma^2 \Kmax} \sum_{i=1}^p d_i)}\\
        & = & c_1 \frac{\frac{1}{(\sigma^2 \Kmax)^p} \prod_{i=1}^p d_i}{(\exp( \sum_{i=1}^p d_i))^\frac{1}{\sigma^2 \Kmax}}\\
        & = & c_2 \frac{ (\frac{1}{\Kmax})^p}{c_3^\frac{1}{\Kmax}}\\
        & = & c_2 \frac{1}{(\Kmax)^p \cdot c_3^\frac{1}{\Kmax}}\label{Eq:LRKmax}
        \eeaa
  %      \bea
%        \lefteqn{\lr(\Kmax)}\nonumber\\  & = & \frac{\prod_{i=1}^p \frac{d_i}{\sigma^2 \Kmax} \cdot \prod_{i=\bar N + 1}^N \frac{d_i}{\sigma^2} \cdot e^N}{ \exp(\ds\sum_{i=1}^p \frac{d_i}{\sigma^2 \Kmax}) \cdot e^{\bar N - p}  \cdot \exp( \ds\sum_{i=\bar N + 1}^N \frac{d_i}{\sigma^2} )}\\
%        & = & c_1 \frac{\prod_{i=1}^p \frac{d_i}{\sigma^2 \Kmax}}{\exp(\sum_{i=1}^p \frac{d_i}{\sigma^2 \Kmax})}\\
%        & = & c_1 \frac{\frac{1}{(\sigma^2 \Kmax)^p} \prod_{i=1}^p d_i}{\exp( \frac{1}{\sigma^2 \Kmax} \sum_{i=1}^p d_i)}\\
%        & = & c_1 \frac{\frac{1}{(\sigma^2 \Kmax)^p} \prod_{i=1}^p d_i}{(\exp( \sum_{i=1}^p d_i))^\frac{1}{\sigma^2 \Kmax}}\\
%        & = & c_2 \frac{ (\frac{1}{\Kmax})^p}{c_3^\frac{1}{\Kmax}}\\
%        & = & c_2 \frac{1}{(\Kmax)^p \cdot c_3^\frac{1}{\Kmax}}\label{Eq:LRKmax}
%       \eea
        where $c_1 = \frac{\prod_{i=\bar N + 1}^N \frac{d_i}{\sigma^2} \cdot e^N}{\exp (\bar N - p)  \cdot \exp( \sum_{i=\bar N + 1}^N \frac{d_i}{\sigma^2} )}$, $c_2 = c_1 \frac{ \prod_{i=1}^p d_i}{\sigma^{2p}}$, and $c_3 = \exp( \frac{1}{\sigma^2}\sum_{i=1}^p d_i)$.\\
        Now let's evaluate the first derivative of the denominator of Eq. \eqref{Eq:LRKmax}.
        \beaa
        \lefteqn{((\Kmax)^p \cdot c_3^\frac{1}{\Kmax})^\prime}\nonumber\\ & = & p (\Kmax)^{p-1} c_3^\frac{1}{\Kmax} + (\Kmax)^p \frac{c_3^\frac{1}{\Kmax} \log c_3}{-(\Kmax)^2}\\
        & = & p (\Kmax)^{p-1} c_3^\frac{1}{\Kmax} - (\Kmax)^{p-2} c_3^\frac{1}{\Kmax} \log c_3\\
        & = & (\Kmax)^{p-2} c_3^\frac{1}{\Kmax} (p\Kmax- \log c_3)\\
        & = & (\Kmax)^{p-2} c_3^\frac{1}{\Kmax} (p\Kmax- \frac{1}{\sigma^2}\sum_{i=1}^p d_i)
        \eeaa
        Since $d_1 > d_2 > \cdots > d_p > \sigma^2 \Kmax$,
        \be
        p\Kmax- \frac{1}{\sigma^2}\sum_{i=1}^p d_i < 0
        \ee
        This implies the denominator of Eq. \eqref{Eq:LRKmax} is a decreasing function, and therefore, $LR(\Kmax)$ is a increasing function with respect to $\Kmax$.

        \item $p \rightarrow p+1$ as $\Kmax$ decreases\\
        The $\lr(\Kmax)$ is a continuous function since $\lambda_{p+1} = d_{p+1}$ at the moment that $\sigma^2 \Kmax = d_{p+1}$ and there is no discontinuity of $\lambda_i$. Therefore, $\lr(\Kmax)$ is an increasing function in this case.
    \end{enumerate}

    \item $d_1 > \sigma^2 K_{\max}$ and $K_{\max} < \frac{\sum_{i=1}^c d_i}{c - \sum_{\bar N + 1}^N (d_i -1)}$\\
    \be
    \hat{\mb R}_{\text{CN}} = \mb\Phi \diag(\bs\lambda^*) \mb\Phi^H
    \ee
    where
    \be
    \bs\lambda^\star = \big[ \frac{\sigma^2}{u}, \ldots, \frac{\sigma^2}{u}, d_{p+1}, \ldots, d_q,\frac{\sigma^2}{u K_{\max}},\ldots,\frac{\sigma^2}{uK_{\max}}  \big]
    \ee
    $p$, $q$, and $\bar N$ are the vector of the eigenvalues of the estimate, the largest indices so that $d_p > \frac{\sigma^2}{u}$, $d_q > \frac{\sigma^2}{u \Kmax}$, and $d_{\bar N} \geq \sigma^2$, respectively.

    Before we prove the increasing property of $\lr(\Kmax)$, we show $u$ decreases as $\Kmax$ increases. $u$ is the optimal solution of the optimization problem. In this case, $u^\star$, the optimal solution of the optimization problem \eqref{Eq:OptimizationU} is obtained by making the first derivative of the cost function 0. Let $u_1$ and $u_2$ be the optimal solutions for $\Kmax_1$ and $\Kmax_2$, respectively. Then, $\sum_{i=1}^N G_i^\prime(u_1) = 0$ for $\Kmax_1$. Since $\frac{1}{d_i} \leq u_1 \leq \frac{1}{\Kmax_1}$ in this case, for $\Kmax_2 < \Kmax_1$, the value of $G_i^\prime (u_1)$ decreases for $d_i \leq 1$. $G_i^\prime(u)$ also decreases for $d_i >1$ and $u \leq \frac{1}{\Kmax d_i}$ and remain same for $d_i > 1$ and $\frac{1}{\Kmax d_i} < u$. Therefore, $\sum_{i=1}^N G_i^\prime(u_1) < 0$ for $\Kmax_2$. Finally, since $\sum_{i=1}^N G_i^\prime(u_2)$ must be zero for $\Kmax_2$, it is obvious that $u_1 < u_2$. This shows that $u$ decreases as $\Kmax$ increases.

    Now we show the increasing property of $\lr(\Kmax)$.
    \begin{enumerate}
        \item within the range where $p$ and $q$ remain same\\
        In this case, We show $\lr(u)$ is a decreasing function of $u$ and an increasing function of $\Kmax$ for each of $u$ and $\Kmax$.
        \begin{enumerate}
        \item Proof of $\lr(u)$ is a decreasing function.
        \beaa
        \lefteqn{\lr(u)}\nonumber\\ & = & \frac{\prod_{i=1}^p \frac{u d_i}{\sigma^2} \cdot \prod_{i=q+1}^{\bar N} \frac{\Kmax u d_i}{\sigma^2} \cdot e^N}{ \exp(\sum_{i=1}^p \frac{u d_i}{\sigma^2} + \sum_{i=p+1}^{q} 1}\nonumber\\
         && \: \frac{}{+ \sum_{i=q + 1}^N \frac{\Kmax u d_i}{\sigma^2} )}\\
        & = & \frac{u^p \prod_{i=1}^p \frac{d_i}{\sigma^2} \cdot u^{N-q} \prod_{i=q+1}^{\bar N} \frac{\Kmax d_i}{\sigma^2} \cdot e^N}{ \exp( u (\sum_{i=1}^p \frac{d_i}{\sigma^2} + \sum_{i=q+1}^N \frac{\Kmax d_i}{\sigma^2})}\nonumber\\
        && \: \frac{}{ + q - p )}\\
        & = & \frac{c_1 u^{N-q+p}}{ \exp( c_2 u + c_3 )}\\
        & = & c_4 \frac{u^{N-q+p}}{c_5^u}\label{Eq:LR_u}
        \eeaa
        where $c_1 = \prod_{i=1}^p \frac{d_i}{\sigma^2} \cdot \prod_{i=q+1}^{\bar N} \frac{\Kmax d_i}{\sigma^2}\cdot e^N$, $c_2 = \sum_{i=1}^p \frac{d_i}{\sigma^2} + \sum_{i=q+1}^N \frac{\Kmax d_i}{\sigma^2}$, $c_3 = q-p$, $c_4 = \frac{c_1}{e^{c_3}}$, and $c_5 = e^{c_2}$.
        The first derivative of Eq. \eqref{Eq:LR_u} is obtained by
        \beaa
        \lefteqn{\lr^\prime(u)}\nonumber\\ & = & (N-q+p) u^{N-q+p-1}c_5^{-u}\nonumber\\
        &&  - \:  u^{N-q+p} \log c_5 \cdot c_5^{-u}\\
        & = & u^{N-q+p-1}c_5^{-u}(N-q+p  - u\log c_5)\\
        & = & u^{N-q+p-1}c_5^{-u}(N-q+p  - c_2 u)\\
        & = & u^{N-q+p-1}c_5^{-u}(N-q+p\nonumber\\
        && - \:  u(\sum_{i=1}^p \frac{d_i}{\sigma^2} + \sum_{i=q+1}^N \frac{\Kmax d_i}{\sigma^2}))
        \eeaa
        Since $\frac{\sigma^2}{u} \leq d_p$,
        \beaa
        \lefteqn{N-q+p  - u(\sum_{i=1}^p \frac{d_i}{\sigma^2} + \sum_{i=q+1}^N \frac{\Kmax d_i}{\sigma^2})}\nonumber\\ & \leq & N-q+p - u(\frac{p}{u} + \frac{N-q}{u}\cdot \Kmax)\\
        & = & N - q - \Kmax(N-q)
        \eeaa
        Since $\Kmax > 1$, $\lr^\prime(u)<0$ which implies $\lr(u)$ is a decreasing function with respect to $u$.

        \item Proof of $\lr(\Kmax)$ is an increasing function.
        \beaa
        \lefteqn{\lr(\Kmax)}\nonumber\\ & = & \frac{\prod_{i=1}^p \frac{u d_i}{\sigma^2} \cdot \prod_{i=q+1}^{\bar N} \frac{\Kmax u d_i}{\sigma^2} \cdot e^N}{ \exp(\sum_{i=1}^p \frac{u d_i}{\sigma^2} + \sum_{i=p+1}^{q} 1}\nonumber\\
         && \frac{}{+\sum_{i=q + 1}^N \frac{\Kmax u d_i}{\sigma^2} )}\\
        & = & \frac{c_1 \Kmax^{N-q}}{ \exp( c_2 \Kmax + c_3)}\\
        & = & c_4 \frac{\Kmax^{N-q}}{c_5^{\Kmax}}\label{Eq:LR_Kmax}
        \eeaa
        where $c_1 = \prod_{i=1}^p \frac{u d_i}{\sigma^2} \cdot \prod_{i=q+1}^{\bar N} \frac{u d_i}{\sigma^2} \cdot e^N$, $c_2 = \sum_{i=q+1}^N \frac{ud_i}{\sigma^2}$, $c_3 = \sum_{i=1}^p \frac{u d_i}{\sigma^2} + q - p$, $c_4 = \frac{c_1}{e^{c_3}}$, and $c_5 = e^{c_2}$. The first derivative is
      \bea
        \lefteqn{\lr^\prime(\Kmax)}\nonumber\\ & = & (N-q) \Kmax^{N-q-1}c_5^{-\Kmax} \nonumber\\
        && -\: \Kmax^{N-q} \log c_5 \cdot c_5^{-\Kmax}\\
        & = & \Kmax^{N-q-1}\nonumber\\
        && \times \: c_5^{-\Kmax}(N-q  - \Kmax\log c_5)\\
        & = & \Kmax^{N-q+p-1}\nonumber\\
        && \times \: c_5^{-u}(N-q  - c_2 \Kmax)\\
        & = & \Kmax^{N-q+p-1}\nonumber\\
        && \times \: c_5^{-u}(N-q  - \Kmax \sum_{i=q+1}^N \frac{ud_i}{\sigma^2})
        \eea
        Since $\frac{\sigma^2}{u\Kmax} \leq d_{q+1}$,
        \bea
        \lefteqn{N-q  -\Kmax \sum_{i=q+1}^N \frac{ud_i}{\sigma^2}}\nonumber\\ & \geq & N-q - \Kmax (\frac{N-q}{\Kmax}) = 0
        \eea
        Therefore, $\lr^\prime(\Kmax) \geq 0$ and $\lr(\Kmax)$ is an increasing function with respect to $\Kmax$.
        \end{enumerate}
        These two proofs show that $\lr(u,\Kmax)$ is an increasing function with respect to $\Kmax$.

   \item $p$ and $q$ changes as $\Kmax$ decreases\\
        The $\lr(u, \Kmax)$ is a continuous function, and therefore, $\lr(u,\Kmax)$ is an increasing function in this case.
    \end{enumerate}

\end{enumerate}
\end{proof}

Lemma \ref{Lemma5} formally proves that the there exist only one optimal condition number and therefore we can find the optimal condition number numerically. The algorithm of finding the optimal condition number is shown in Algorithm \ref{Alg:ConditionNumber}. We first set the initial condition number as the ML condition number obtained by \cite{Aubry12}. Then we increase or decrease the condition number to the direction where the LR value decreases. Reducing the stepsize as the direction is reversed, we find the optimal condition number as precisely as we want.

\begin{algorithm}[t]
\caption{The proposed algorithm to select condition number via EL }
\label{Alg:ConditionNumber}
    \begin{algorithmic}[1]
        \STATE Obtain the ML solution of the condition number $K_{\max_{\text{ML}}}$ by the method in \cite{Aubry12} and set the initial value of $K_{\max} = K_{\max_{\text{ML}}}$
        \STATE Set the initial step, $\Delta = K_{\max}/100$
        \STATE Evaluate $\lr(K_{\max}-\Delta)$, $\lr(K_{\max})$, $\lr(K_{\max}+\Delta)$
            \begin{itemize}
                \item if $|\lr(K_{\max_{\text{ML}}}+\Delta) - \lr_0| < |\lr(K_{\max_{\text{ML}}}) - \lr_0|$\\
                $\rightarrow$ increase $K_{\max}$ by $\Delta$ until it does not hold.\\
                $\rightarrow$ then $\Delta = -\Delta/10$
                \item elseif $|\lr(K_{\max_{\text{ML}}}+\Delta) - \lr_0| > |\lr(K_{\max_{\text{ML}}}) - \lr_0|$\\
                $\rightarrow$ decrease $K_{\max}$ by $\Delta$ until it does not hold.\\
                $\rightarrow$ then $\Delta = -\Delta/10$
            \end{itemize}
         \STATE Repeat Step 3 until $\Delta < 0.0001$.
    \end{algorithmic}
\end{algorithm}

\section{Experimental Validation}
\label{Sec:Experiments}

\subsection{Experimental setup}

In this section, we compare the proposed methods with alternative covariance estimation algorithms and parameter estimation algorithms. Two data sets are used in the experiments: 1) a radar covariance simulation model and 2) the KASSPER dataset \cite{Bergin02}.

First, we consider a radar system with an $N$-element uniform linear array for the simulation model. The overall covariance which is composed of jammer and additive white noise can be modeled by
\be
\label{Eq:SimulationModel}
\mb{R}(n,m) = \sum_{i=1}^J \sigma_i^2 \sinc[0.5 \beta_i (n-m) \phi_i ] e^{j(n-m)\phi_i} + \sigma_a^2 \delta(n,m)
\ee
where $n,m \in \{1,\ldots,N\}$, $J$ is the number of jammers, $\sigma_i^2$ is the power associated with the $i$th jammer, $\phi_i$ is the jammer phase angle with respect to the antenna phase center, $\beta_i$ is the fractional bandwidth, $\sigma_a^2$ is the actual power level of the white disturbance term, and $\delta(n,m)$ has the value of 1 only when $n=m$ and 0 otherwise. This simulation model has been widely and very successfully used in previous literature \cite{Steiner00,Aubry12,Pallotta12,DeMaio09} for performance analysis.

Data from the L-band data set of KASSPER program is the other data set used in our experiments. Note that the KASSPER data set exhibits two desirable characteristics: 1) the low-rank structure of clutter and 2) the true covariance matrices for each range bin have been made available. These two characteristics facilitate comparisons via powerful figures of merit. The L-band data set consists of a data cube of 1000 range bins corresponding to the returns from a single coherent processing interval from $11$ channels and $32$ pulses. Therefore, the dimension of observations (or the spatio-temporal product) $N$ is $11 \times 32 = 352$. Other key parameters are detailed in Table \ref{Tb:parameters}.

We measure the normalized signal to interference and noise ratio (SINR). The normalized SINR measure is commonly used in the radar literature and given by
\be
\eta = \dfrac{|\mb{s}^H\hat{\mb{R}}^{-1}\mb{s}|^2}{|\mb{s}^H\hat{\mb{R}}^{-1}\mb{R}\hat{\mb{R}}^{-1}\mb{s}||\mb{s}^H\mb{R}^{-1}\mb{s}|}
\ee
where $\mb s$ is the spatio-temporal steering vector, $\hat{\mb R}$ is the data-dependent estimate of $\mb R$, and $\mb R$ is the true covariance matrix. It is easily seen that $0<\eta<1$ and $\eta=1$ if and only if $\hat{\mb R} = \mb R$. The SINR is plotted in decibels in all our experiments, that is, $\text{SINR} \text{(dB)} = 10 \log_{10} \eta$. Therefore, $\text{SINR} \text{(dB)} \leq 0$. For the KASSPER data set, since the steering vector is a function of both azimuthal angle and Doppler frequency, we obtain plots as a function of one variable (azimuthal angle or Doppler) by marginalizing over the other variable. We evaluate and compare different covariance estimation techniques and parameter selection methods in the following three subsections:
\begin{itemize}
    \item \textbf{Sample Covariance Matrix:} The sample covariance matrix is given by $\mb S = \frac{1}{K} \mb Z \mb Z^H$. It is well known that $\mb S$ is the unconstrained ML estimator under Gaussian disturbance statistics. We refer to this as SMI.
    \item \textbf{Fast Maximum Likelihood:} The fast maximum likelihood (FML) \cite{Steiner00} uses the structural constraint of the covariance matrix. The FML method just involves the eigenvalue decomposition of the sample covariance and perturbing eigenvalues to conform to the structure. The FML also can be considered as the RCML estimator with the rank which is the greatest index $i$ satisfying $\lambda_i > \sigma^2$ where $\lambda_i$'s are the eigenvalues of the sample covariance in descending order. Therefore, a rank can be considered as an output of the FML. The FML's success in radar STAP is widely known \cite{Rangaswamy04Sep}.
    \item \textbf{Rank Constrained ML Estimators:} The RCML estimator with the rank or the rank and the noise level obtained by the proposed methods using the expected likelihood approach. The rank is obtained by the EL approach in the case of the imperfect rank constraint and both of the rank and the noise level are obtained by the EL approach in the case of imperfect rank and noise power constraints. We refer to these as \RCMLEL.
    \item \textbf{Chen \emph{et al.} Rank Selection Method:} Chen \emph{et al.} \cite{Chen01} proposed a statistical procedure for detecting the multiplicity of the smallest eigenvalue of the structured covariance matrix using statistical selection theory. The rank can be estimated from their methods using pre-calculated parameters. We refer to this method as \RCMLC.
    \item \textbf{AIC:} Akaike \cite{Akaike74} proposed the information theoretic criteria for model selection. The AIC selects the model that best fits the data for given a set of observations and a family of models, that is, a parameterized family of probability densities. Wax and Kailath \cite{Wax85} proposed the method to determine the number of signals from the observed data based on the AIC. We compare Wax and Kailath's method.
    \item \textbf{Condition number constrained ML estimators:} The maximum likelihood estimation method of the covariance matrix with a condition number \cite{Aubry12} proposed by Aubry \emph{et al.} is considered for evaluating the performance with three different condition numbers. 1) \CNCML: the condition number obtained by the proposed method in \cite{Aubry12}, 2) \CNCEL: the condition number obtained by the expected likelihood approach, and 3) \CNCtrue: the true condition number.
\end{itemize}

\subsection{Imperfect rank constraint}
\label{Sec:ResultRank}

\begin{figure}
\centering
\includegraphics[scale=0.5]{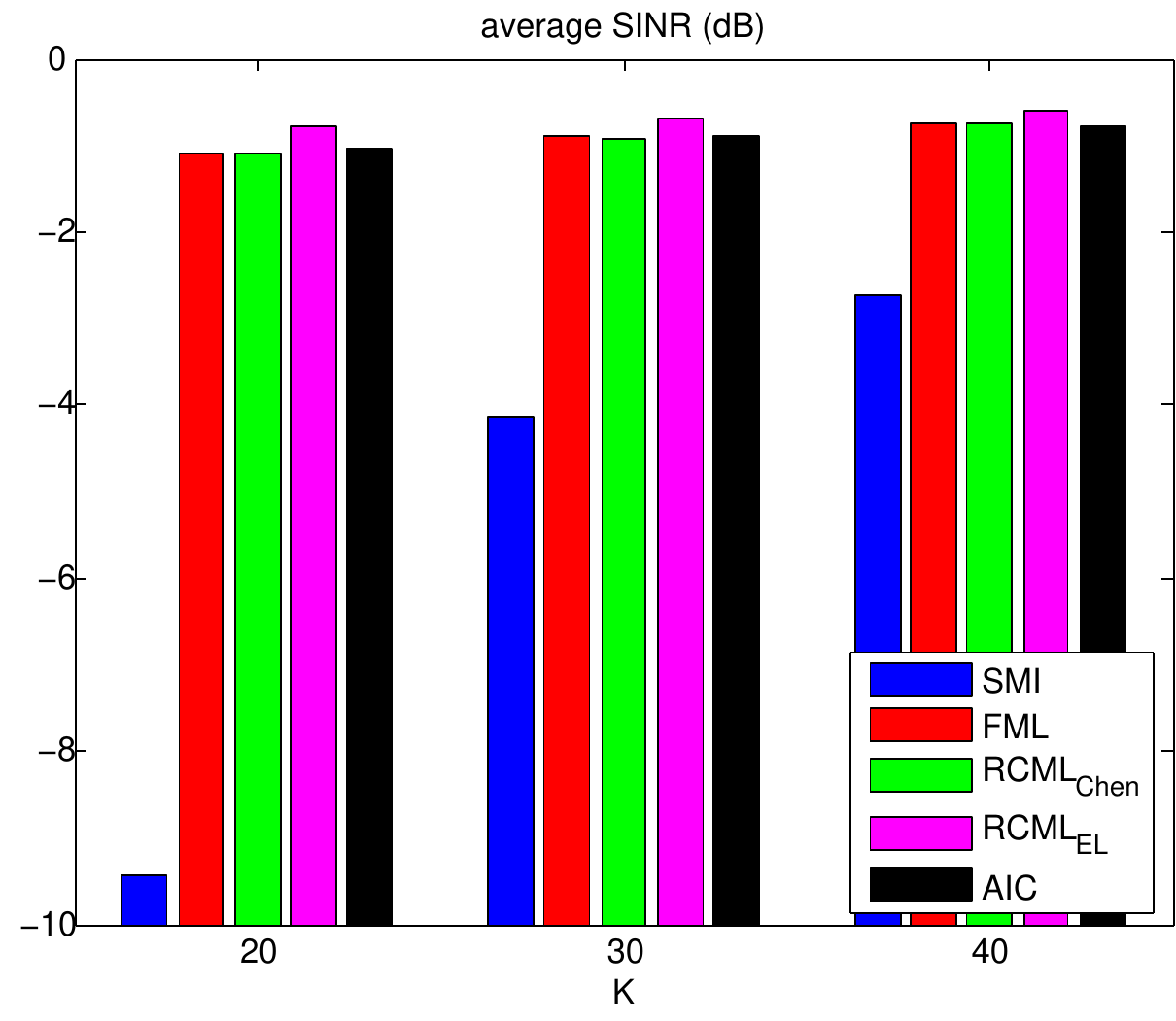}\label{Fig:SINR_Simulation}
\caption{Normalized SINR in dB versus number of training samples $K$ $(N=20)$ for the simulation model.}
\label{Fig:Simulation_rank}
\end{figure}

\begin{figure}[!t]
\begin{center}
\subfigure[$K=N=352$]{\includegraphics[scale=0.5]{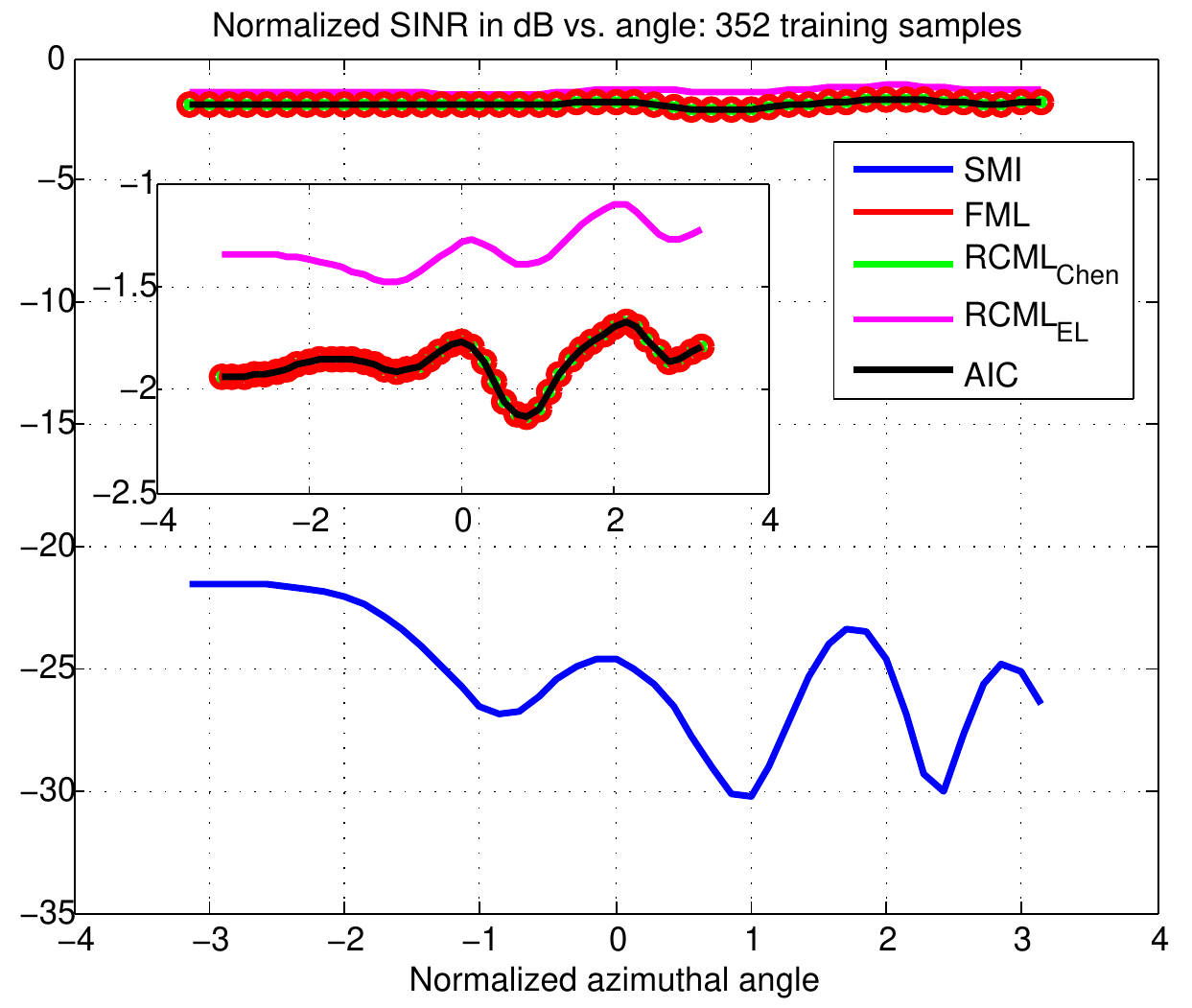}\label{Fig:KASSPER_rank_angle_352}}
\hfil
\subfigure[$K=N=352$]{\includegraphics[scale=0.5]{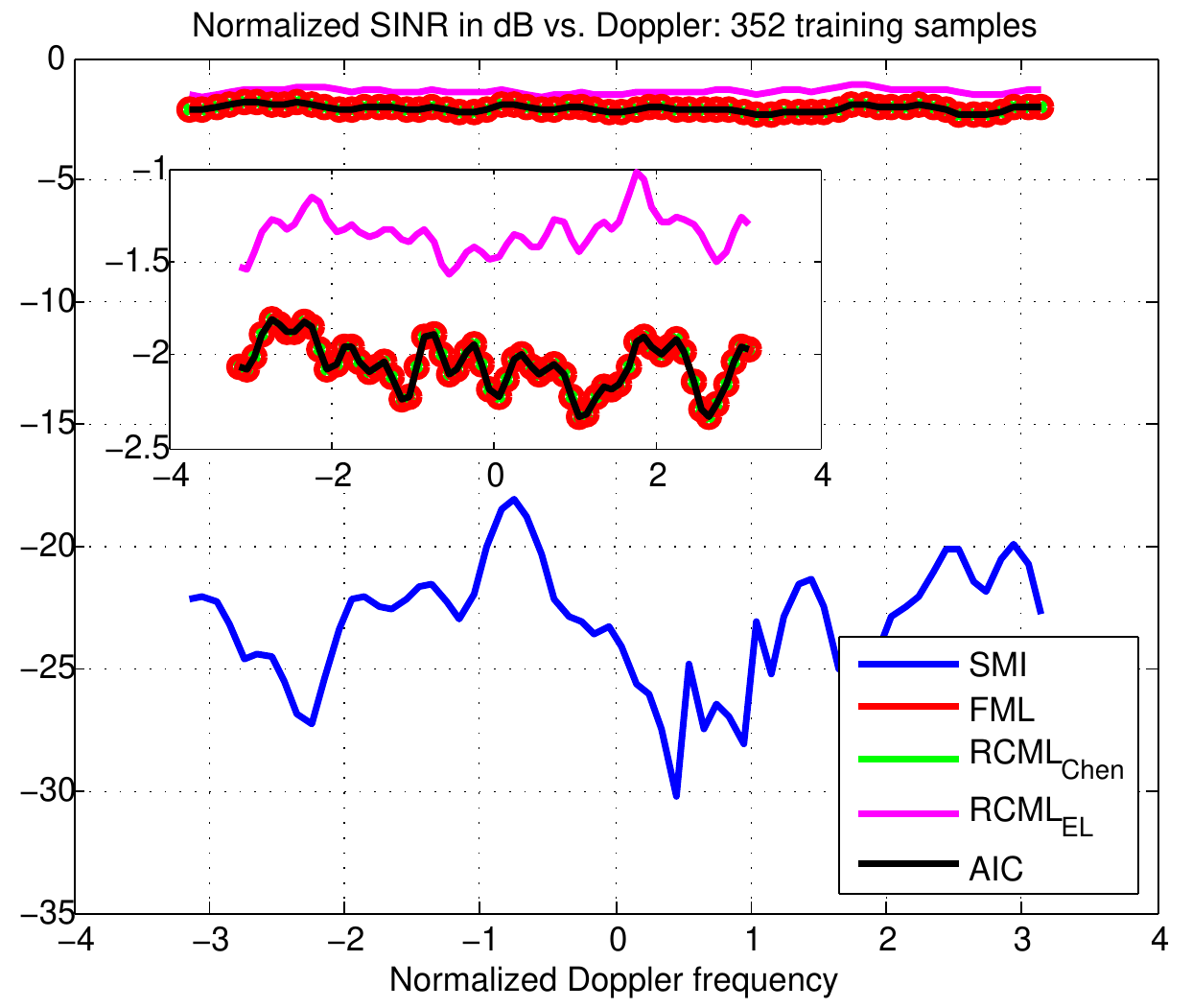}\label{Fig:KASSPER_rank_dop_352}}\\
\subfigure[$K=1.5N=528$]{\includegraphics[scale=0.5]{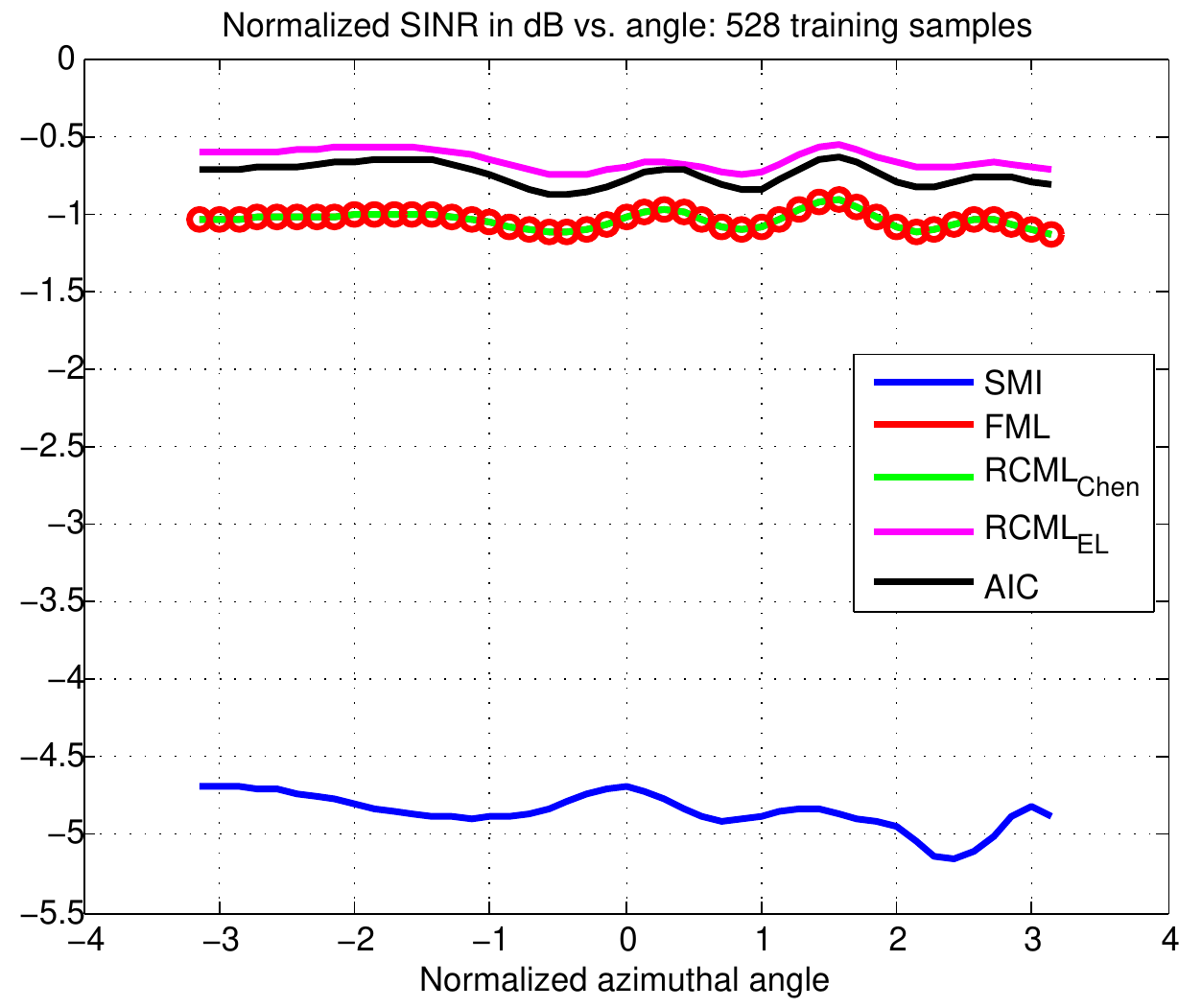}\label{Fig:KASSPER_rank_angle_528}}
\hfil
\subfigure[$K=1.5N=528$]{\includegraphics[scale=0.5]{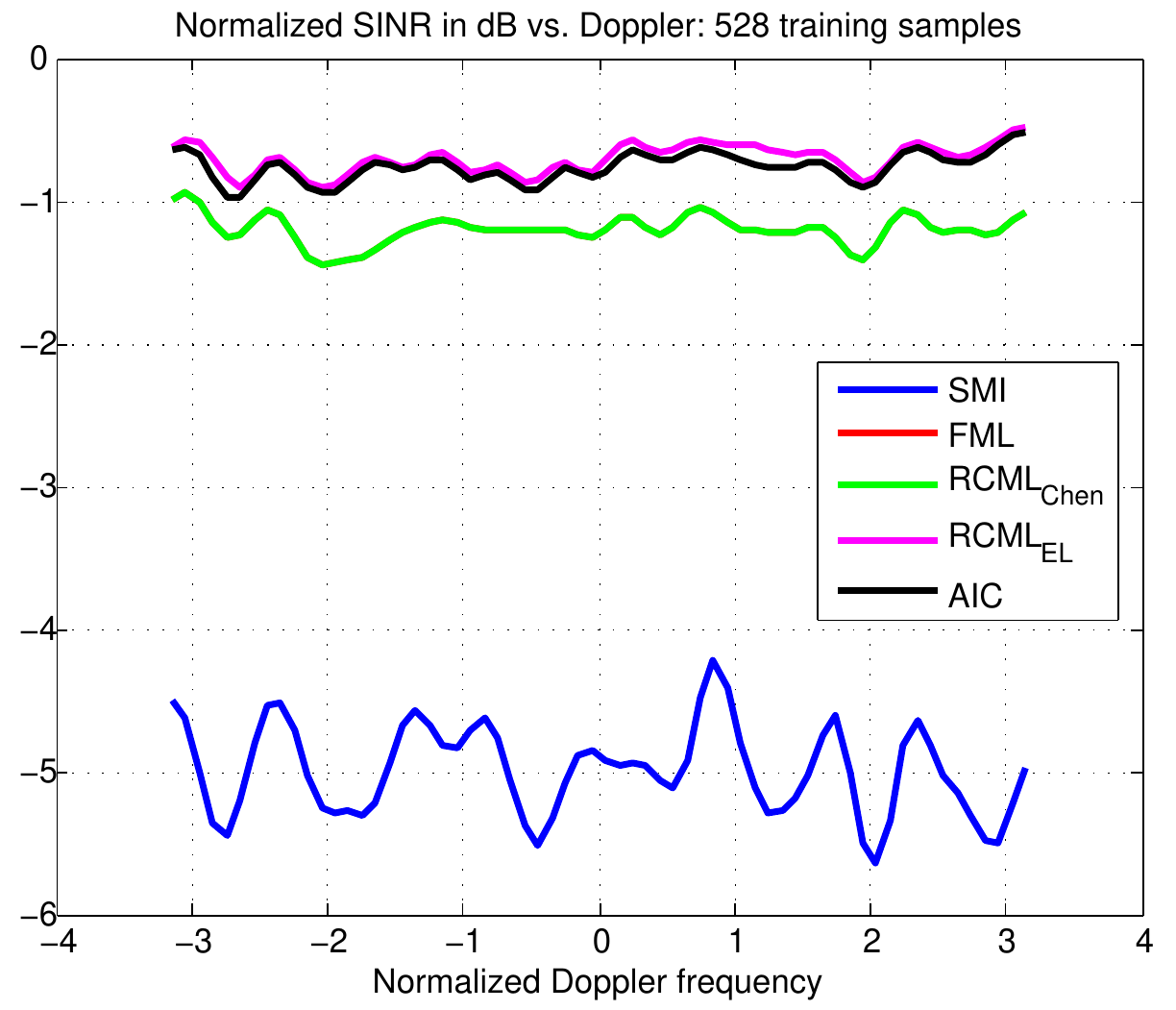}\label{Fig:KASSPER_rank_dop_528}}\\
\subfigure[$K=2N=704$]{\includegraphics[scale=0.5]{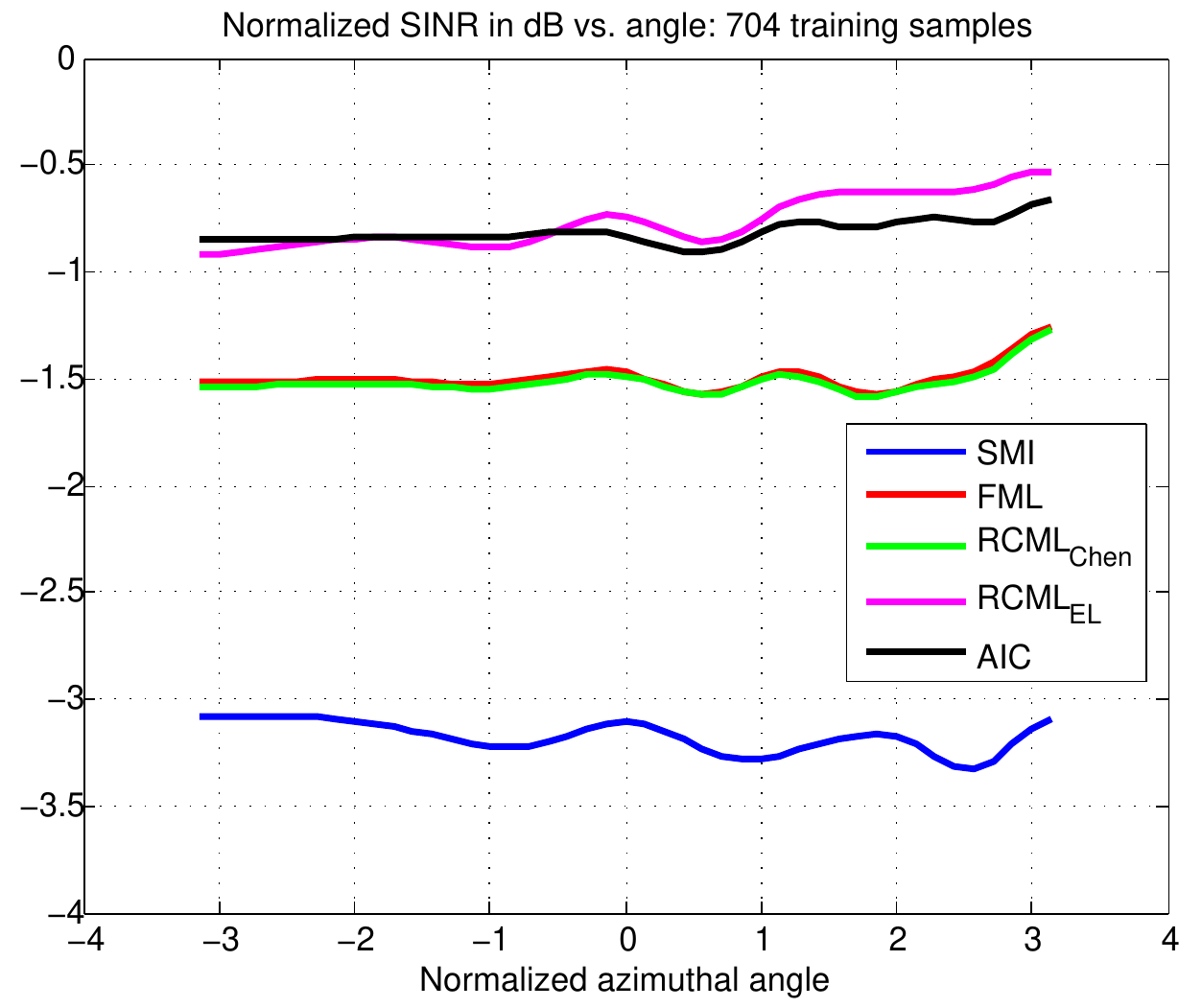}\label{Fig:KASSPER_rank_angle_704}}
\hfil
\subfigure[$K=2N=704$]{\includegraphics[scale=0.5]{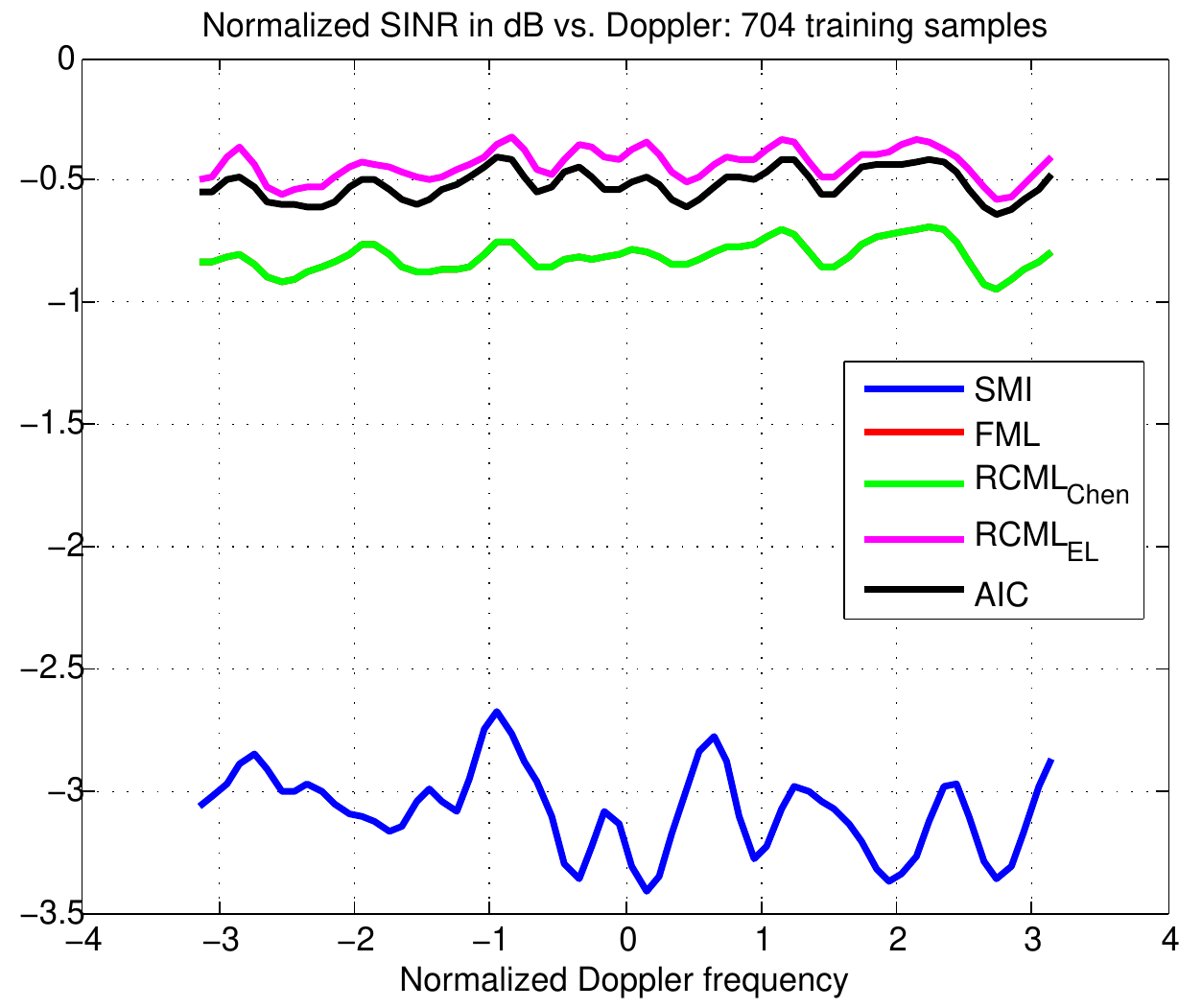}\label{Fig:KASSPER_rank_dop_704}}
\end{center}
\caption{Normalized SINR versus azimuthal angle and Doppler frequency for the KASSPER data set.}
\label{Fig:KASSPER_rank}
\end{figure}

\begin{figure}[!t]
\begin{center}
\subfigure[$K=N=352$]{\includegraphics[scale=0.5]{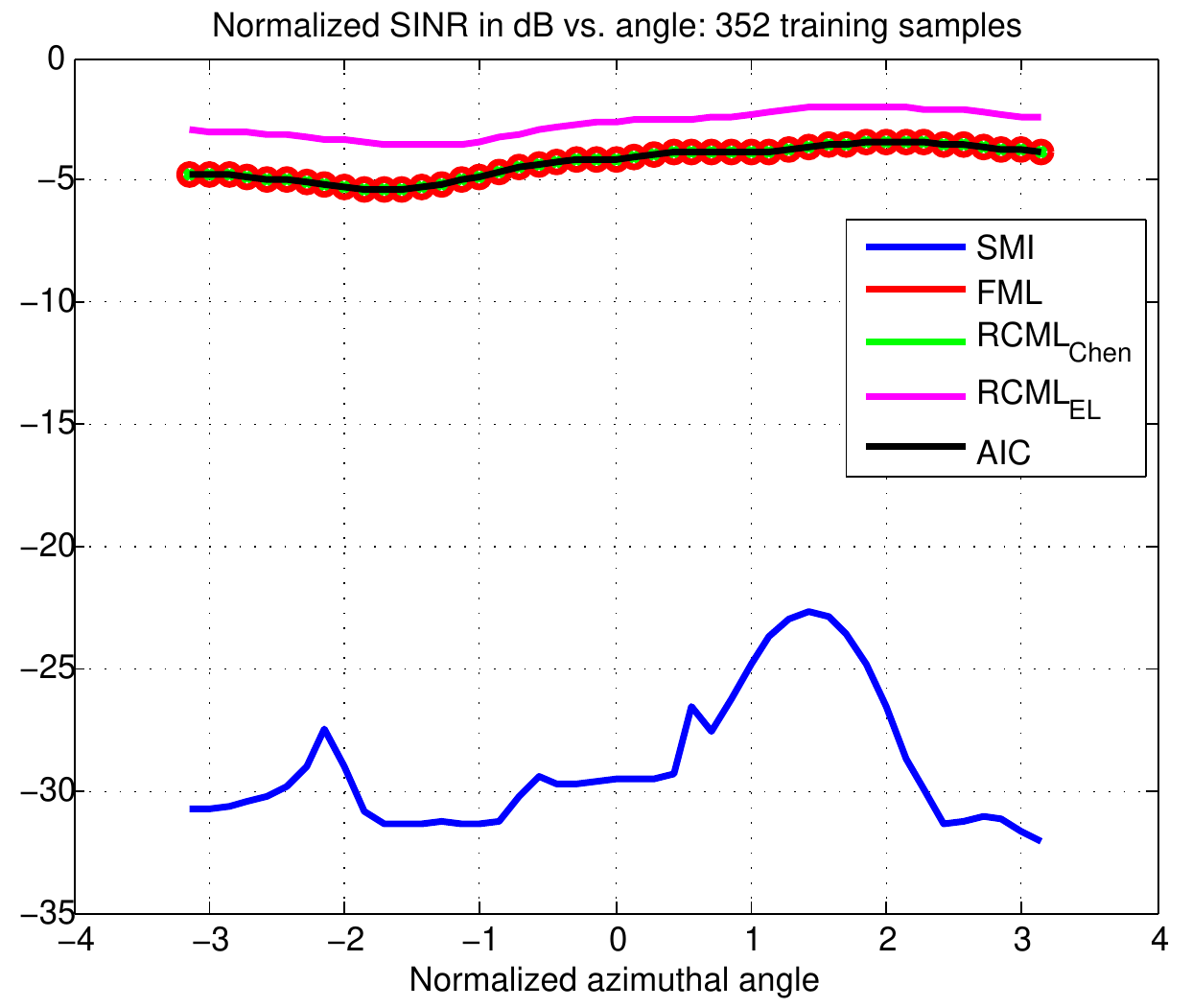}\label{Fig:KASSPER_angle_nonhomo_352}}
\hfil
\subfigure[$K=N=352$]{\includegraphics[scale=0.5]{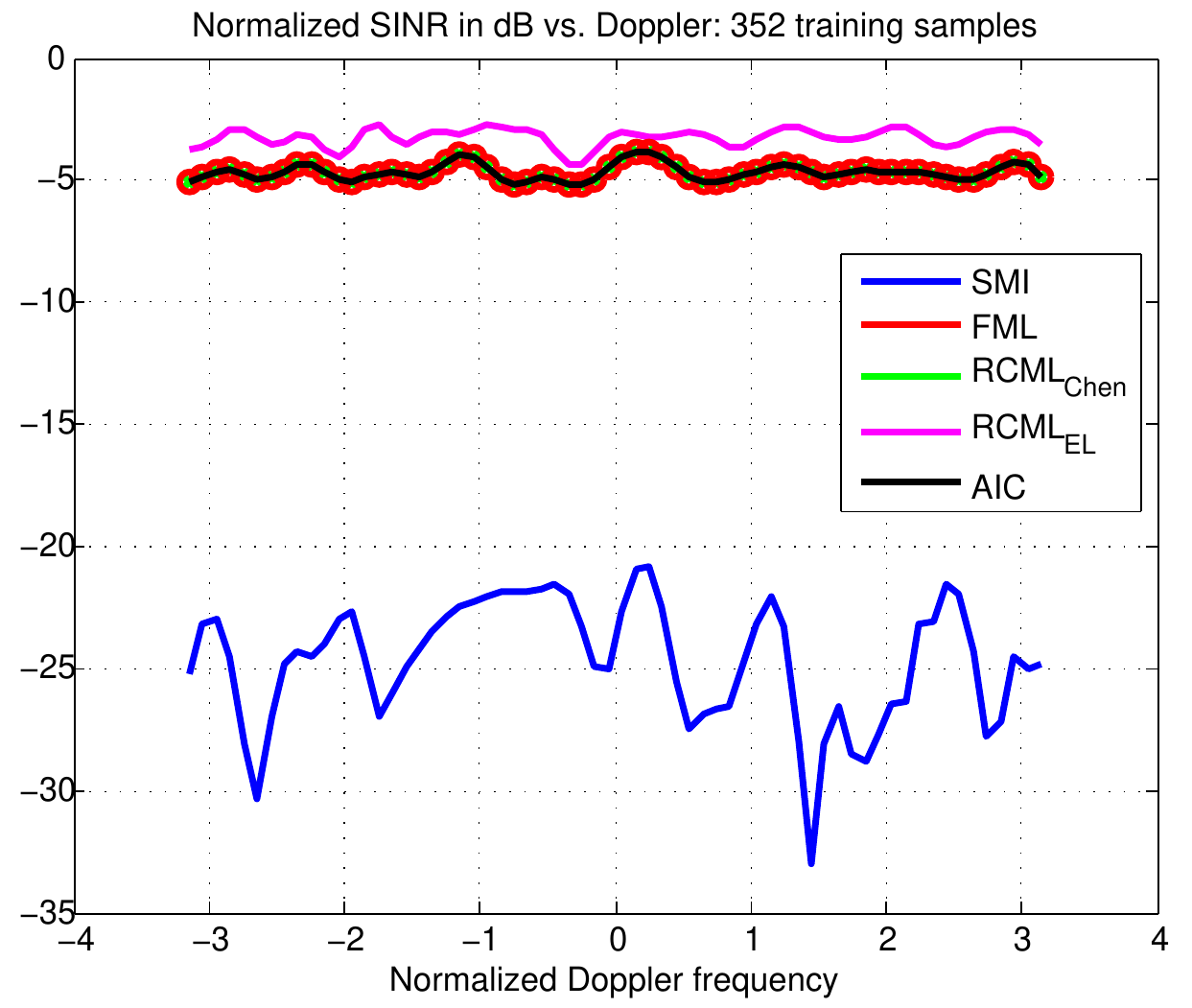}\label{Fig:KASSPER_dop_nonhomo_352}}\\
\subfigure[$K=1.5N=528$]{\includegraphics[scale=0.5]{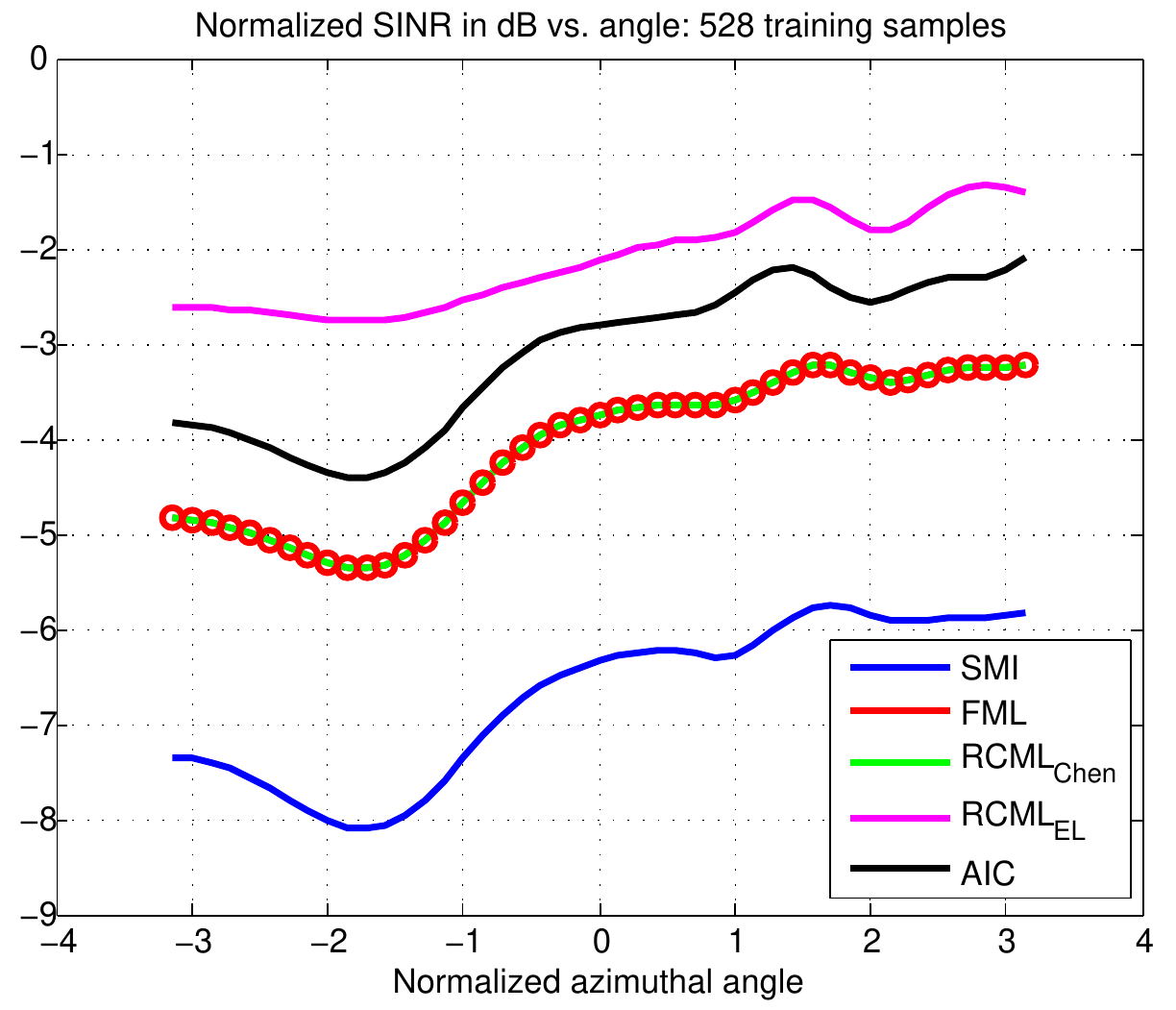}\label{Fig:KASSPER_angle_nonhomo_528}}
\hfil
\subfigure[$K=1.5N=528$]{\includegraphics[scale=0.5]{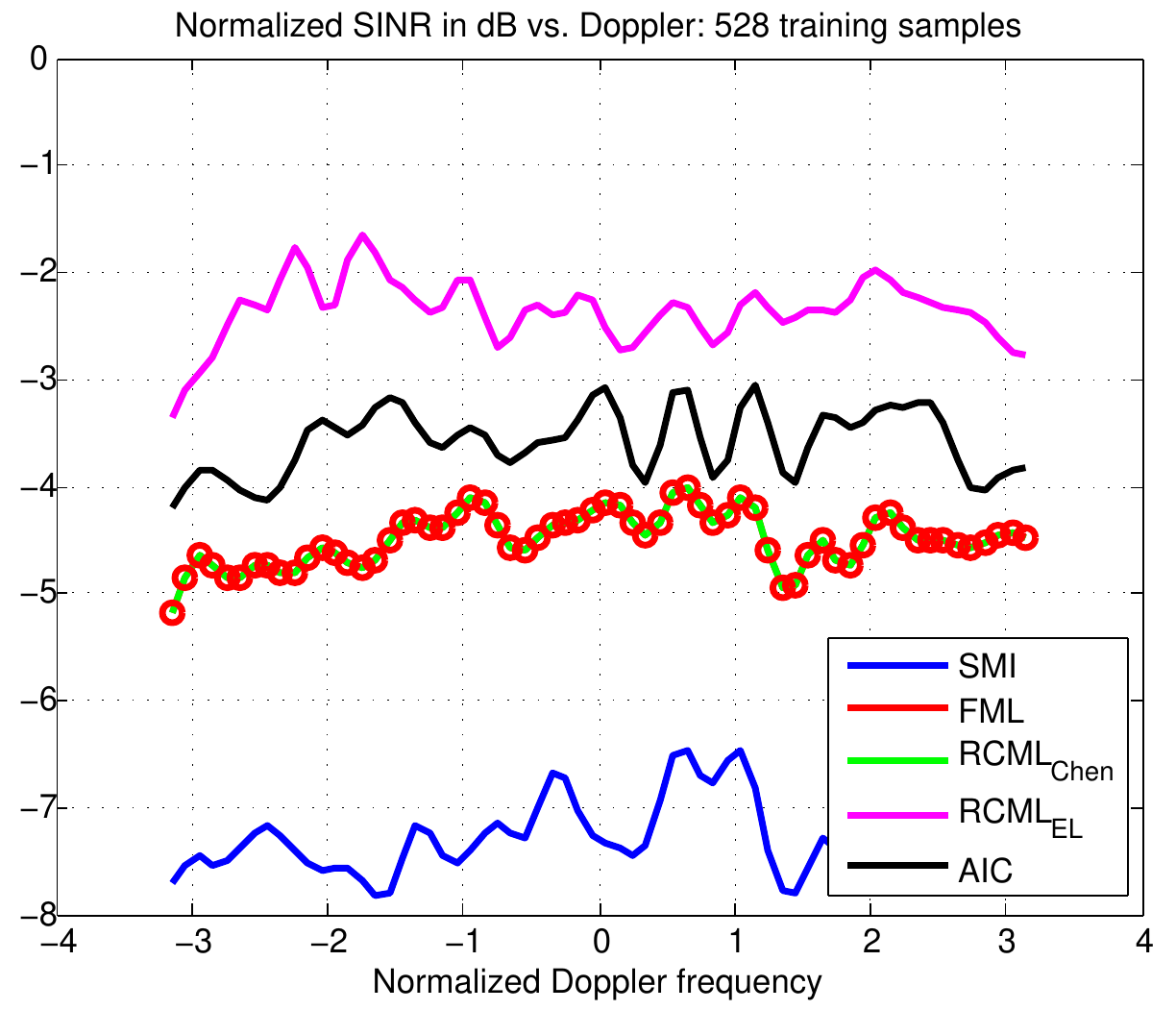}\label{Fig:KASSPER_dop_nonhomo_528}}\\
\subfigure[$K=2N=704$]{\includegraphics[scale=0.5]{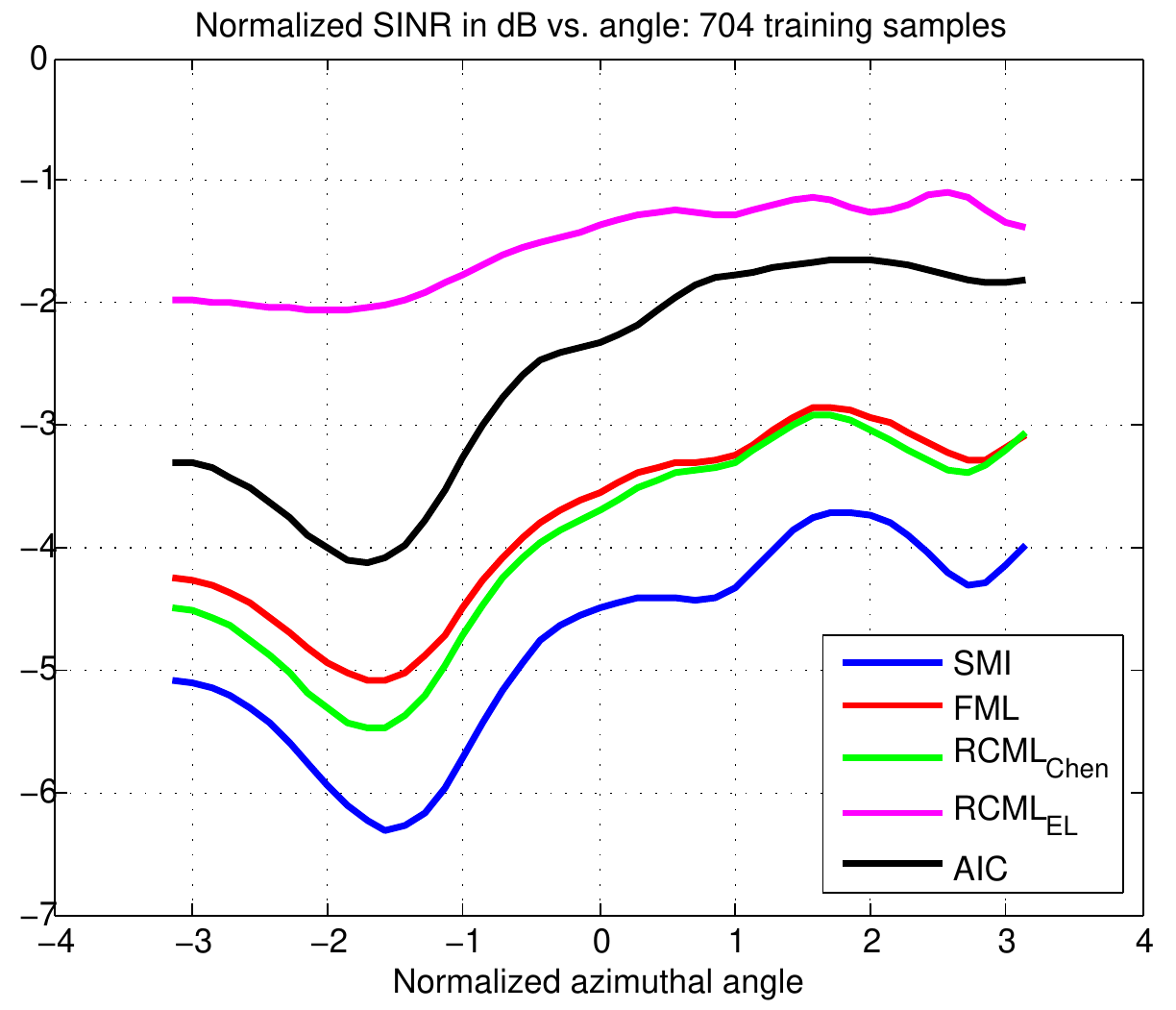}\label{Fig:KASSPER_angle_nonhomo_704}}
\hfil
\subfigure[$K=2N=704$]{\includegraphics[scale=0.5]{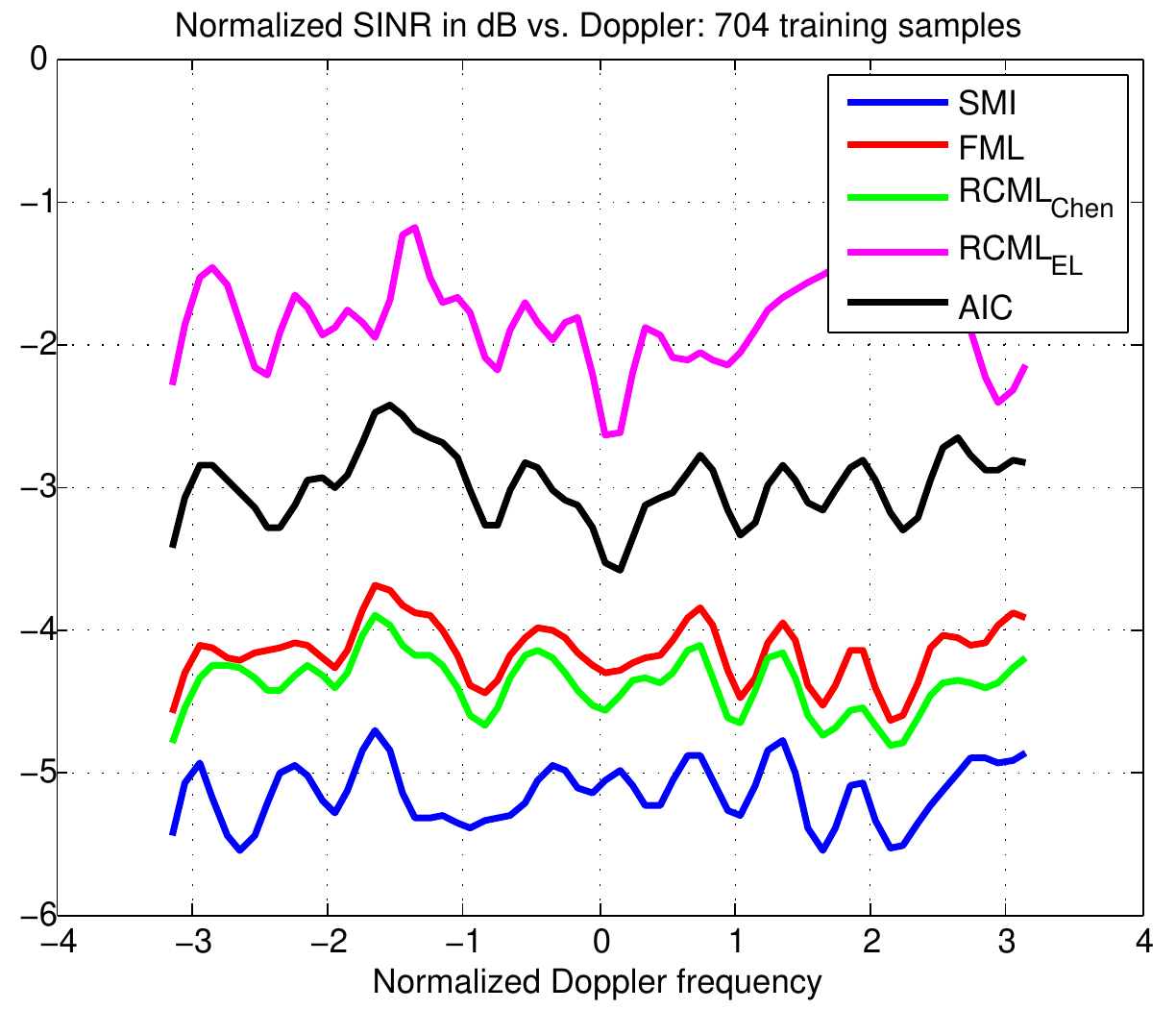}\label{Fig:KASSPER_dop_nonhomo_704}}
\end{center}
\caption{Normalized SINR versus azimuthal angle and Doppler frequency for the KASSPER data set.}
\label{Fig:KASSPER_rank_nonhomogeneous}
\end{figure}

First, we compare the rank estimation method proposed in Section \ref{Sec:Rankonly} with alternative algorithms including SMI, FML, AIC, and Chen's algorithm. We plot the normalized SINR (in dB) versus the number of training samples, 20, 30, and 40 in Fig. \ref{Fig:Simulation_rank} for the simulation model. The SINR values are obtained by averaging SINR values from 500 Monte Carlo trials. It is shown that the SINR values increases monotonically as $K$ increases. The \RCMLEL exhibits the best performance in all training regimes. Particularly, the difference between \RCMLEL and other methods increases when training samples are limited.

Figure \ref{Fig:KASSPER_rank} shows the normalized SINR values for various number of training samples for the KASSPER data set. We plot the averaged SINR values in decibel over either azimuth angle or Doppler frequency domain. The left and right column show the results for angle and Doppler, respectively. Similarly to the results for the simulation model, \RCMLEL outperforms all the other compared methods in all training regime. This implies that the rank obtained by the EL approach is more accurate and closer to the rank predicted by Brennan rule ($M + P - 1 = 42$) than any other methods.

\textbf{Realistic case of contaminated observations:} In practice, homogeneous training samples are hard to obtain and the received signals are often corrupted by target information. Therefore, it is meaningful to compare the performance for nonhomogeneous observation to investigate which algorithm indeed works well and is robust in practice. In this case, the received signal is given by
\be
\mb z = \alpha \mb s + \mb d
\ee
where $\mb s$ and $\mb d$ are the deterministic steering vector and stochastic disturbance vector, respectively. Figure \ref{Fig:KASSPER_rank_nonhomogeneous} shows the normalized SINR values when a half of the training samples contain $\mb s$ with $\alpha = 50$. The gap between \RCMLEL and the others is bigger than that in Figure \ref{Fig:KASSPER_rank}. In particular, AIC shows compatible performance in homogeneous cases though, Figure \ref{Fig:KASSPER_angle_nonhomo_704} and Figure \ref{Fig:KASSPER_dop_nonhomo_704} show that the difference between AIC and \RCMLEL is larger for non-homogeneous case. Therefore the results show remarkably that \RCMLEL still excels under target contamination or heterogeneous training where other techniques face severe degradation in performance.

\subsection{Imperfect rank and noise power constraints}
\label{Sec:ResultBoth}

\begin{figure}
\centering
\includegraphics[scale=0.5]{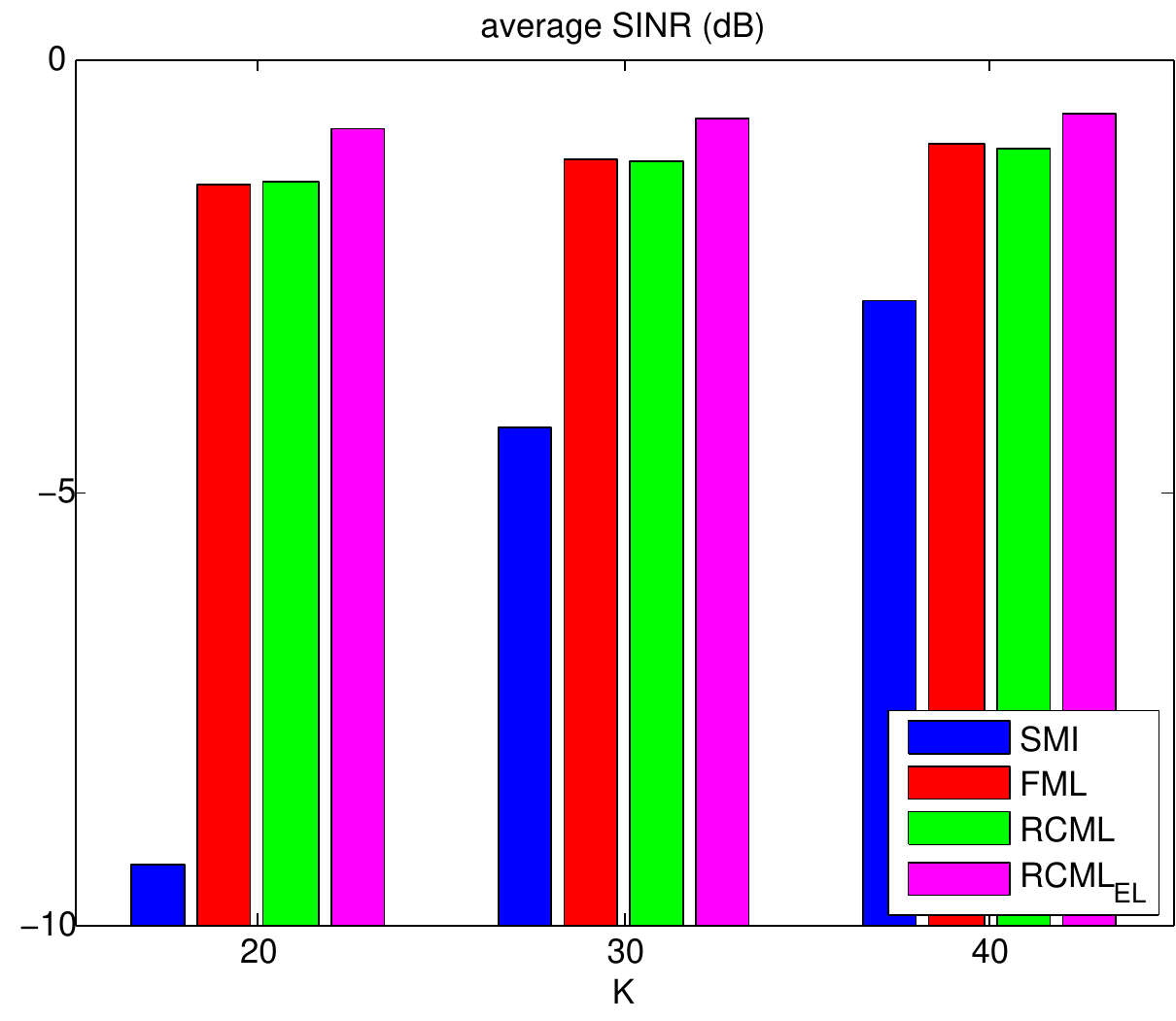}\label{Fig:SINR_Simulation_both}
\caption{Normalized SINR in dB versus number of training samples $K$ $(N=20)$ for the simulation model.}
\label{Fig:Simulation_both}
\end{figure}

\begin{figure}[!t]
\begin{center}
\subfigure[$K=N=352$]{\includegraphics[scale=0.5]{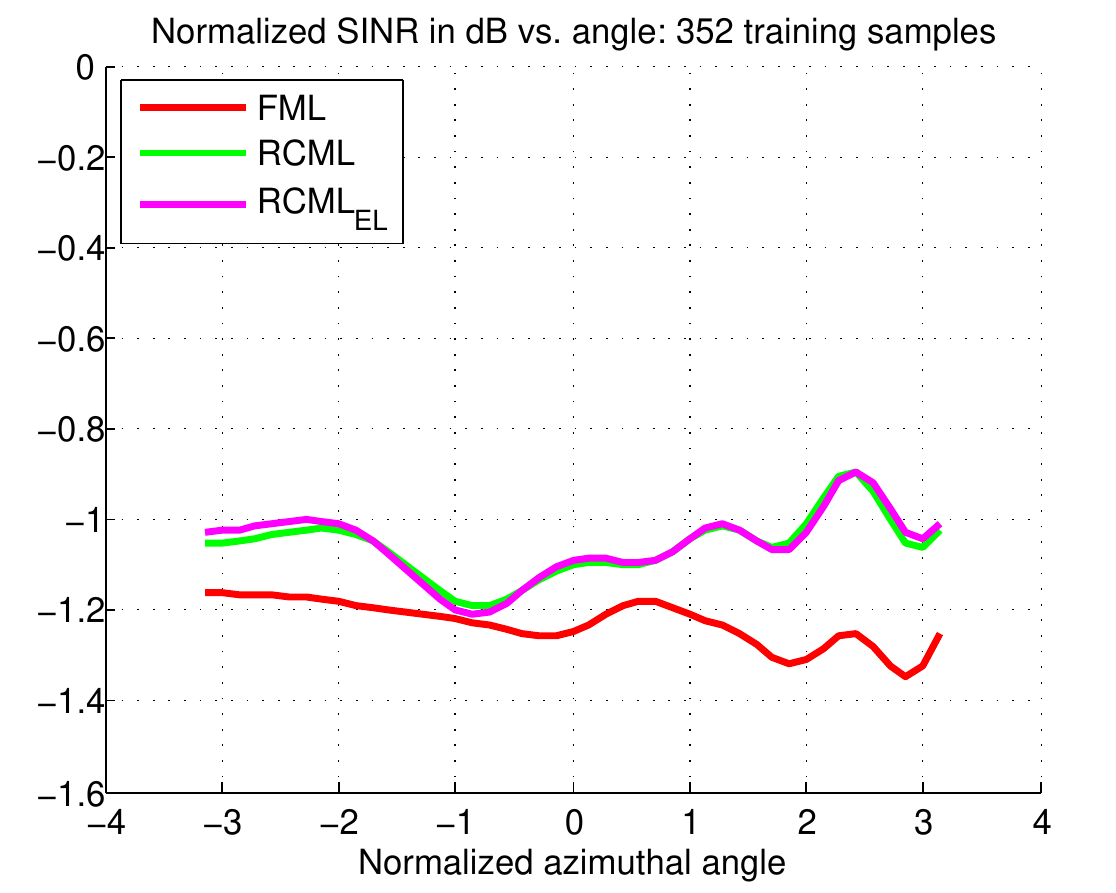}\label{Fig:SINR_angle_352}}
\hfil
\subfigure[$K=N=352$]{\includegraphics[scale=0.5]{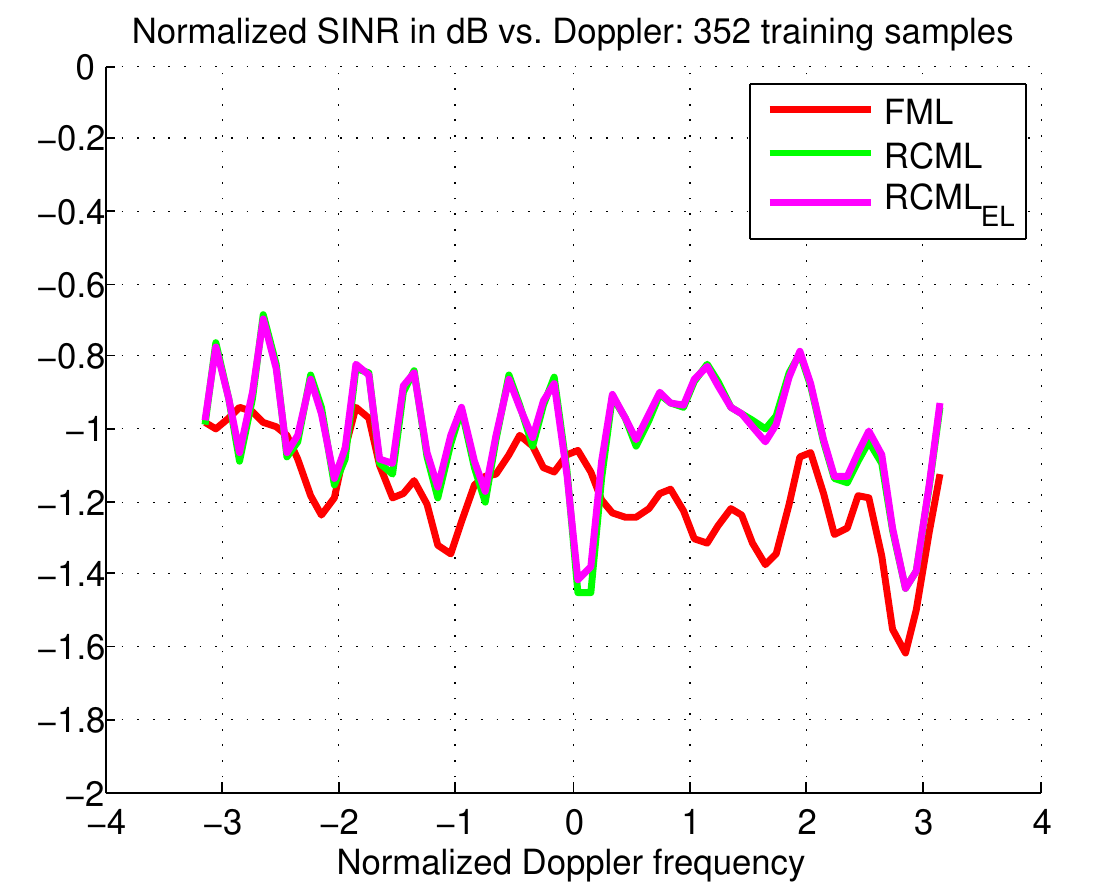}\label{Fig:PD_Simulation_20}}\\
\subfigure[$K=1.5N=528$]{\includegraphics[scale=0.5]{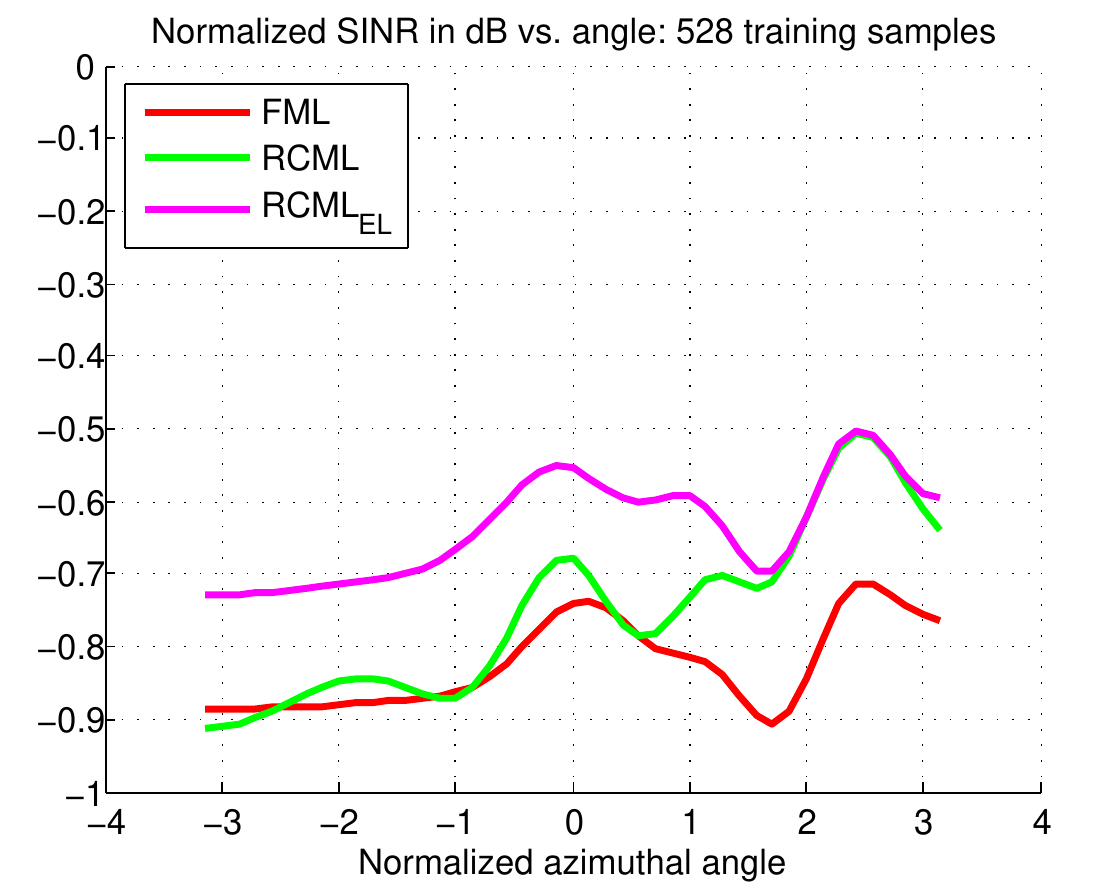}\label{Fig:SINR_angle_352}}
\hfil
\subfigure[$K=1.5N=528$]{\includegraphics[scale=0.5]{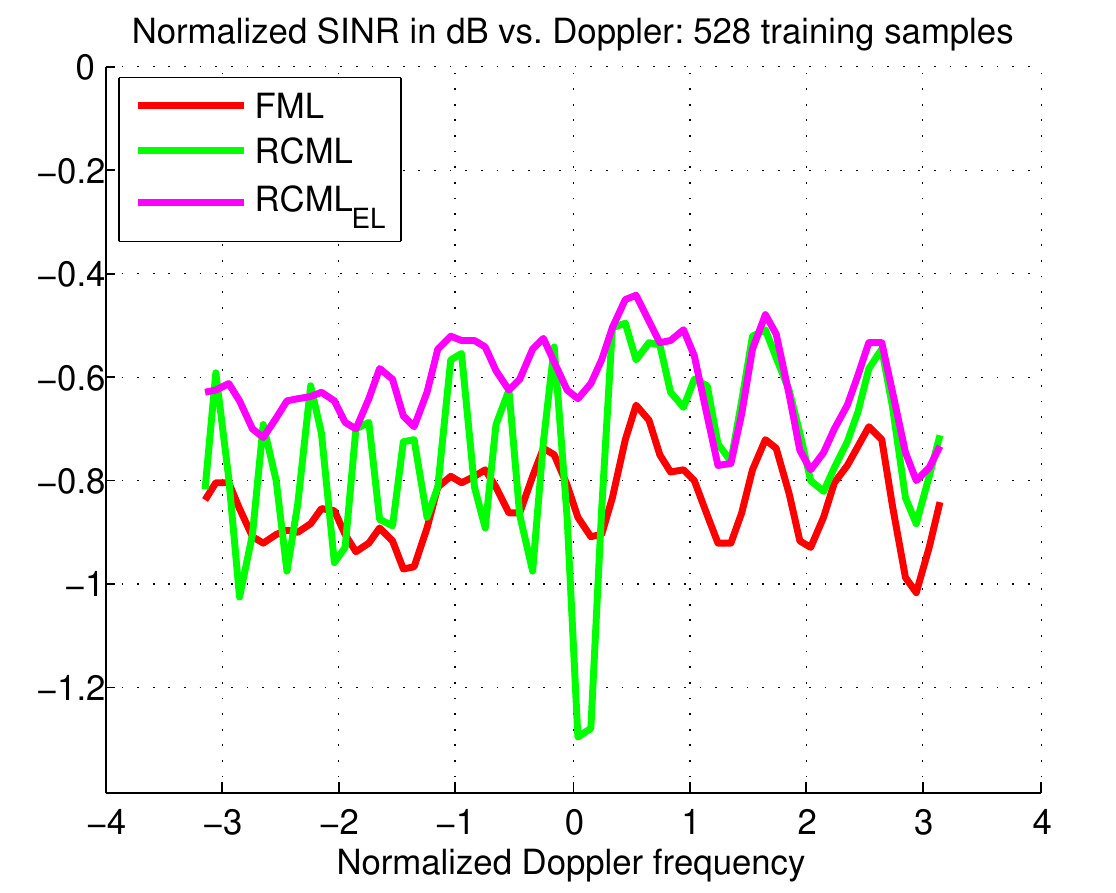}\label{Fig:PD_Simulation_20}}\\
\subfigure[$K=2N=704$]{\includegraphics[scale=0.5]{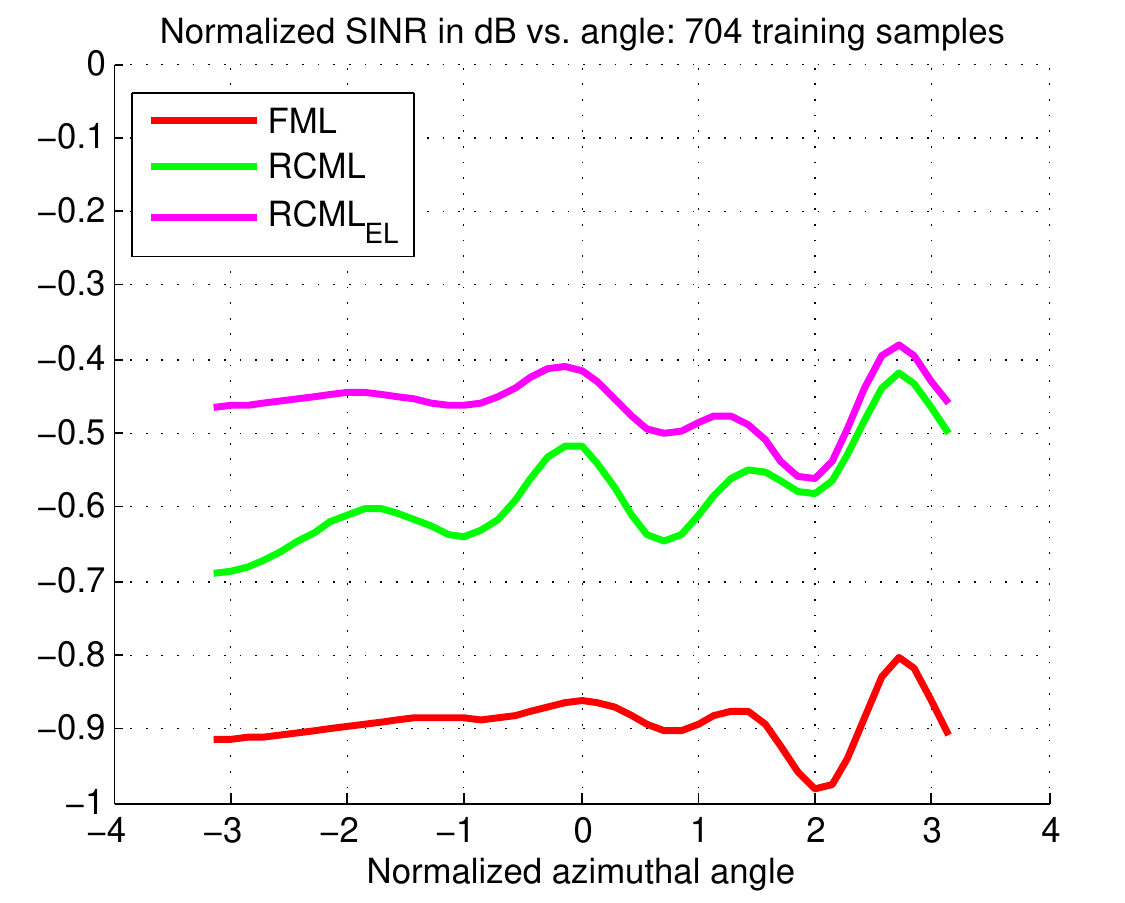}\label{Fig:SINR_angle_352}}
\hfil
\subfigure[$K=2N=704$]{\includegraphics[scale=0.5]{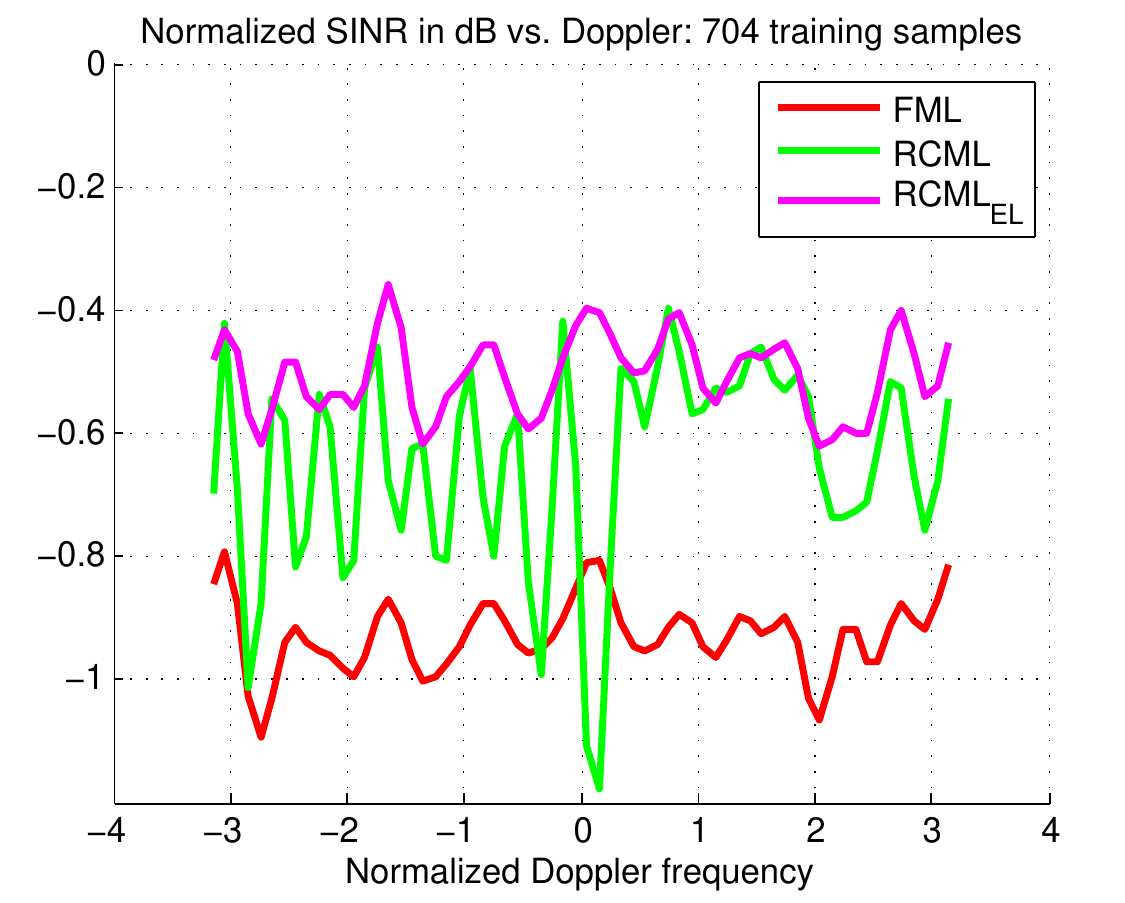}\label{Fig:PD_Simulation_20}}
\end{center}
\caption{Normalized SINR versus azimuthal angle and Doppler frequency for the KASSPER data set.}
\label{Fig:KASSPER_both}
\end{figure}

In this section, we show experimental results for estimation of both a rank and a noise power via the expected likelihood approach, which is proposed in Section \ref{Sec:Both}. In this case, we assume that both the rank and the noise power are unknown and to be estimated for both the simulation model and the KASSPER data set. Since the previous works such as AIC and Chen's algorithm are for only estimating the rank and can not be extended to estimate both the rank and the noise power, we compare the proposed EL method with the sample covariance, FML, and the RCML estimator with a prior knowledge of the rank. For the RCML estimator, we employ the number of jammers $(r=5)$ and the Brennan rule $(r=42)$ as the clutter rank for the simulation model and the KASSPER data set, respectively. In addition, since the FML method requires a prior knowledge of the noise power, we calculate and use the maximum likelihood estimate of the noise power for a rank given by a prior knowledge for the FML.

Figure \ref{Fig:Simulation_both} shows the performance of various estimators in the sense of the normalized SINR values for the simulation model. Similarly to the case of only rank estimation, the proposed method show the best performance in all training regimes.

Figure \ref{Fig:KASSPER_both} shows the performance of the methods in the normalized SINR for the KASSPER data set. The proposed method is comparable with or slightly better than the RCML estimator using the rank by Brennan rule. This means that the proposed method estimates both the rank and the noise power adaptively from training samples whereas the rank by Brennan rule is fixed regardless of the training samples.

\subsection{Imperfect condition number constraint}
\label{Sec:ResultConditionNumber}

\begin{table}[!t]
\begin{center}
\subfigure[]{             \begin{tabular}{|c|c|c|c|c|c|c|}
               \hline
               $\sigma^2$ & K & SMI & FML & CNCML & $\text{CNCML}_\text{EL}$ & $\text{CNCML}_\text{true}$ \\
               \hline
               & 20 & -9.3785 & -0.5195 & -0.5212 & \textbf{-0.4822} & -0.5200 \\
               -5 &30 & -4.2579 & \textbf{-0.4242} & -0.4257 & -0.4256 & -0.4250 \\
               & 40 & -2.7424 & \textbf{-0.3460} & -0.3476 & -0.3476 & -0.3468 \\
               \hline
                \hline
               & 20 & -9.3196 & -0.5511 & -0.5521 & \textbf{-0.5141} & -0.5515 \\
               0 & 30 & -4.2276 & \textbf{-0.4202} & -0.4221 & -0.4220 & -0.4210 \\
               & 40 & -2.7649 & \textbf{-0.3513} & -0.3530 & -0.3528 & -0.3521 \\
               \hline
              \hline
               & 20 & -9.0922 & -0.5269 & -0.5279 & \textbf{-0.4875} & -0.5272 \\
               5 & 30 & -4.2172 & \textbf{-0.4348} & -0.4364 & -0.4362 & -0.4357 \\
               & 40 & -2.7300 & \textbf{-0.3484} & -0.3503 & -0.3505 & -0.3493 \\
               \hline
               \hline
               & 20 & -9.3511 & -0.5355 & -0.5305 & \textbf{-0.4998} & -0.5360 \\
               10 & 30 & -4.1955 & \textbf{-0.4164} & -0.4180 & -0.4175 & -0.4175 \\
               & 40 & -2.7491 & \textbf{-0.3501} & -0.3515 & -0.3518 & -0.3509 \\
               \hline
             \end{tabular}\label{Tb:CN_simulation_a}}\\

\subfigure[]{              \begin{tabular}{|c|c|c|c|c|c|c|}
               \hline
               $\sigma^2$ & K & SMI & FML & CNCML & $\text{CNCML}_\text{EL}$ & $\text{CNCML}_\text{true}$ \\
               \hline
               & 20 & -9.3069 & -1.7371 & \textbf{-1.7322} & -1.7358 & -1.7350 \\
               -5 & 30 & -4.1795 & -1.2399 & -1.2388 & \textbf{-1.2347} & -1.2397 \\
               & 40 & -2.7535 & -0.9496 & -0.9492 & \textbf{-0.9456} & -0.9493 \\
               \hline
                \hline
               & 20 & -9.1354 & -1.6944 & \textbf{-1.6928} & -1.7027 & -1.6940 \\
               0 & 30 & -4.2345 & -1.2986 & -1.2987 & \textbf{-1.2955} & -1.2990 \\
               & 40 & -2.7545 & -1.0041 & -1.0043 & \textbf{-1.0023} & -1.0046 \\
               \hline
              \hline
               & 20 & -9.2524 & -1.3976 & -1.4016 & \textbf{-1.3244} & -1.4000 \\
               5 & 30 & -4.2309 & -1.0737 & -1.0784 & \textbf{-1.0666} & -1.0762 \\
               & 40 & -2.7523 & -0.8848 & -0.8876 & \textbf{-0.8818} & -0.8866 \\
               \hline
               \hline
               & 20 & -9.3660 & -1.2567 & -1.2569 & \textbf{-1.2115} & -1.2570 \\
               10 & 30 & -4.3013 & -0.9526 & -0.9545 & \textbf{-0.9450} & -0.9537 \\
               & 40 & -2.7350 & -0.7171 & -0.7197 & \textbf{-0.7139} & -0.7186 \\
               \hline
             \end{tabular}  \label{Tb:CN_simulation_b} }\\
             \subfigure[]{
\begin{tabular}{|c|c|c|c|c|c|c|}
               \hline
               $\sigma^2$ & K & SMI & FML & CNCML & $\text{CNCML}_\text{EL}$ & $\text{CNCML}_\text{true}$ \\
               \hline
               & 20 & -9.3702 & -0.5340 & -0.5349 & \textbf{-0.4925} & -0.5340 \\
               -5 & 30 & -4.2791 & \textbf{-0.4302} & -0.4316 & -0.4315 & -0.4308 \\
               & 40 & -2.7856 & \textbf{-0.3493} & -0.3510 & -0.3509 & -0.3501 \\
               \hline
                \hline
               & 20 & -9.2898 & -0.5485 & -0.5501 & \textbf{-0.5104} & -0.5491 \\
               0 & 30 & -4.2648 & \textbf{-0.4202} & -0.4219 & -0.4220 & -0.4209 \\
               & 40 & -2.7274 & \textbf{-0.3604} & -0.3621 & -0.3621 & -0.3611 \\
               \hline
              \hline
               & 20 & -9.0582 & -0.5318 & -0.5328 & \textbf{-0.4899} & -0.5322 \\
               5 & 30 & -4.1548 & \textbf{-0.4142} & -0.4155 & -0.4152 & -0.4149 \\
               & 40 & -2.7655 & \textbf{-0.3515} & -0.3531 & -0.3533 & -0.3521 \\
               \hline
               \hline
               & 20 & -9.3632 & -0.5352 & -0.5363 & \textbf{-0.4974} & -0.5360 \\
               10 & 30 & -4.2728 & \textbf{-0.4328} & -0.4348 & -0.4349 & -0.4337 \\
               & 40 & -2.7577 & \textbf{-0.3538} & -0.3554 & -0.3547 & -0.3547 \\
               \hline
             \end{tabular}\label{Tb:CN_simulation_c}}\\
        \end{center}
        \caption{Normalized SINR for various values of parameters for the simulation model.}
        \label{Tb:CN_simulation1}
\end{table}

\begin{table}[!t]
\begin{center}
\subfigure[]{\begin{tabular}{|c|c|c|c|c|c|c|}
               \hline
               $\sigma^2$ & K & SMI & FML & CNCML & $\text{CNCML}_\text{EL}$ & $\text{CNCML}_\text{true}$ \\
               \hline
               & 20 & -9.0316 & -1.7161 & \textbf{-1.7131} & -1.7634 & -1.7150 \\
               -5 & 30 & -4.1465 & -1.1704 & -1.1691 & \textbf{-1.1659} & -1.1693 \\
               & 40 & -2.7727 & -0.9390 & -0.9384 & \textbf{-0.9351} & -0.9387 \\
               \hline
                \hline
               & 20 & -9.2091 & -1.6706 & -1.6701 & \textbf{-1.6674} & -1.6706 \\
               0 & 30 & -4.2004 & -1.2681 & -1.2682 & \textbf{-1.2633} & -1.2682 \\
               & 40 & -2.7423 & -1.0102 & -1.0117 & \textbf{-1.1009} & -1.0106 \\
               \hline
              \hline
               & 20 & -9.3538 & -1.3980 & -1.4028 & \textbf{-1.3216} & -1.4004 \\
               5 & 30 & -4.2203 & -1.0869 & -1.0910 & \textbf{-1.0785} & -1.0889 \\
               & 40 & -2.7079 & -0.8694 & -0.8721 & \textbf{-0.8666} & -0.8713 \\
               \hline
               \hline
               & 20 & -9.221 & -1.2446 & -1.2455 & \textbf{-1.1982} & -1.2452 \\
               10 & 30 & -4.2116 & -0.9428 & -0.9460 & \textbf{-0.9382} & -0.9444 \\
               & 40 & -2.7563 & -0.7235 & -0.7264 & \textbf{-0.7226} & -0.7253 \\
               \hline
             \end{tabular}\label{Tb:CN_simulation_d}}\\
             \subfigure[]{            \begin{tabular}{|c|c|c|c|c|c|c|}
               \hline
               $\sigma^2$ & K & SMI & FML & CNCML & $\text{CNCML}_\text{EL}$ & $\text{CNCML}_\text{true}$ \\
               \hline
               & 20 & -9.2679 & -1.1593 & -1.1616 & \textbf{-1.1150} & -1.1610 \\
               -5 & 30 & -4.2234 & -0.9262 & -0.9286 & \textbf{-0.9242} & -0.9278 \\
               & 40 & -2.8271 & \textbf{-0.7705} & -0.7729 & -0.7712 & -0.7723 \\
               \hline
                \hline
               & 20 & -9.2934 & -0.9052 & -0.9051 & \textbf{-0.8422} & -0.9059 \\
               0 & 30 & -4.1617 & -0.6909 & -0.6920 & \textbf{-0.6862} & -0.6924 \\
               & 40 & -2.7387 & -0.5711 & -0.5724 & \textbf{-0.5676} & -0.5725 \\
               \hline
              \hline
               & 20 & -9.4154 & -0.8398 & -0.8334 & \textbf{-0.7909} & -0.8399 \\
               5 & 30 & -4.2284 & -0.6273 & -0.6231 & \textbf{-0.6070} & -0.6278 \\
               & 40 & -2.7208 & -0.5034 & -0.5022 & \textbf{-0.4945} & -0.5046 \\
               \hline
               \hline
               & 20 & -9.1447 & -0.7388 & -0.7225 & \textbf{-0.6815} & -0.7392 \\
               10 & 30 & -4.2046 & -0.5931 & -0.5803 & \textbf{-0.5535} & -0.5931 \\
               & 40 & -2.7241 & -0.4821 & -0.4738 & \textbf{-0.4576} & -0.4827 \\
               \hline
             \end{tabular}\label{Tb:CN_simulation_e}}\\
             \renewcommand{\arraystretch}{1.3}
             \subfigure[]{\begin{tabular}{|c|c|c|c|c|}
               \hline
               & $J$ & $\sigma_J^2$ & $\phi$ & $B_f$\\
               \hline
               (a) & 1 & 30 & 20$^{\circ}$ & 0\\
                \hline
               (b) & 1 & 30 & 20$^{\circ}$ & 0.3\\
              \hline
               (c) & 3 & [30 30 30] & [20$^{\circ}$ 40$^{\circ}$ 60$^{\circ}$] & [0 0 0]\\
               \hline
               (d) & 3 & [30 30 30] & [20$^{\circ}$ 40$^{\circ}$ 60$^{\circ}$] & [0.3 0.3 0.3]\\
               \hline
               (e) & 3 & [10 20 30] & [20$^{\circ}$ 40$^{\circ}$ 60$^{\circ}$] & [0.2 0 0.3]\\
               \hline
             \end{tabular}\label{Tb:CN_parameter}}
        \end{center}
        \caption{Normalized SINR for various values of parameters for the simulation model.}
        \label{Tb:CN_simulation2}
\end{table}

\begin{figure}[!t]
\begin{center}
\subfigure[$K=N=352$]{\includegraphics[scale=0.5]{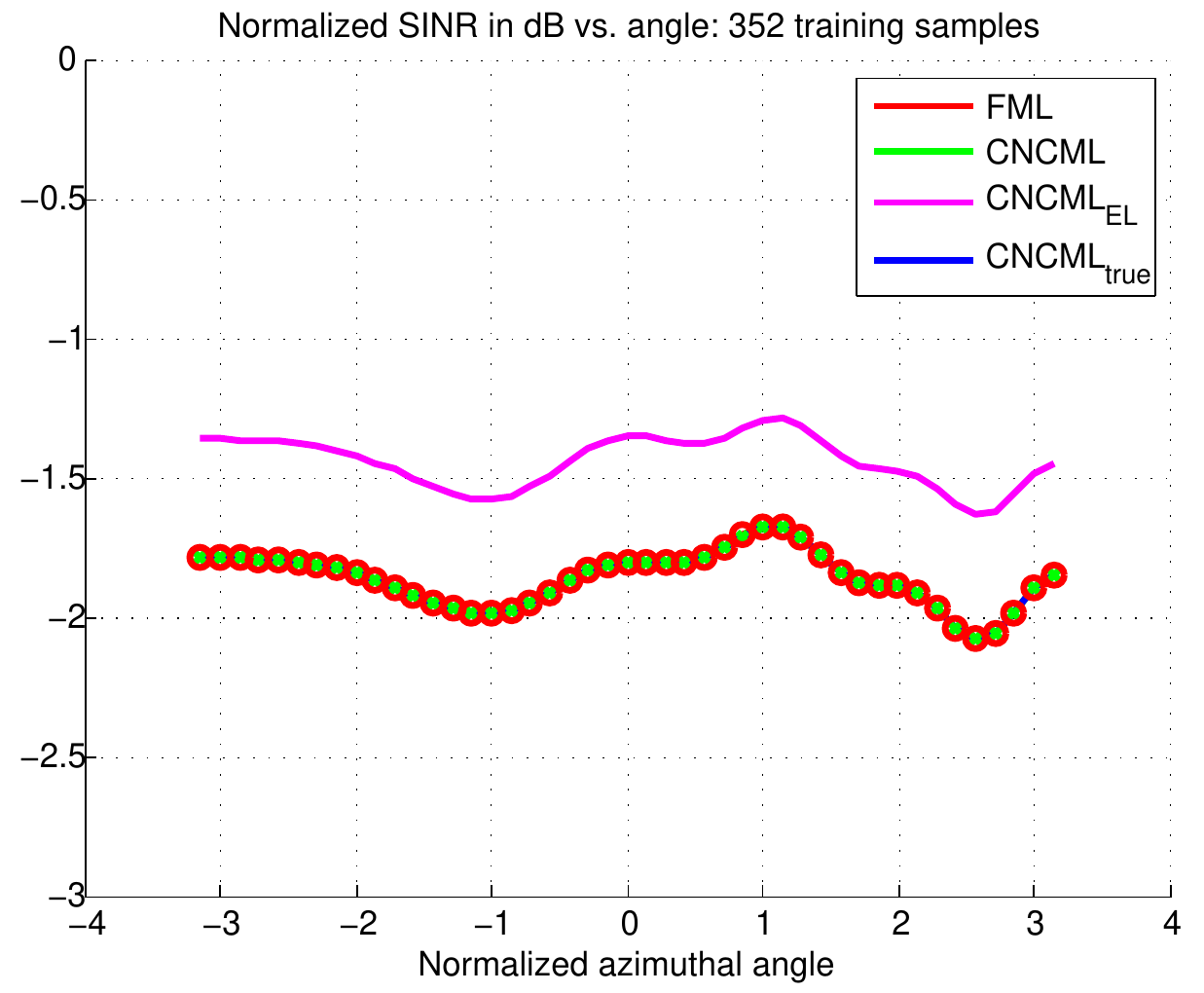}\label{Fig:SINR_angle_352}}
\hfil
\subfigure[$K=N=352$]{\includegraphics[scale=0.5]{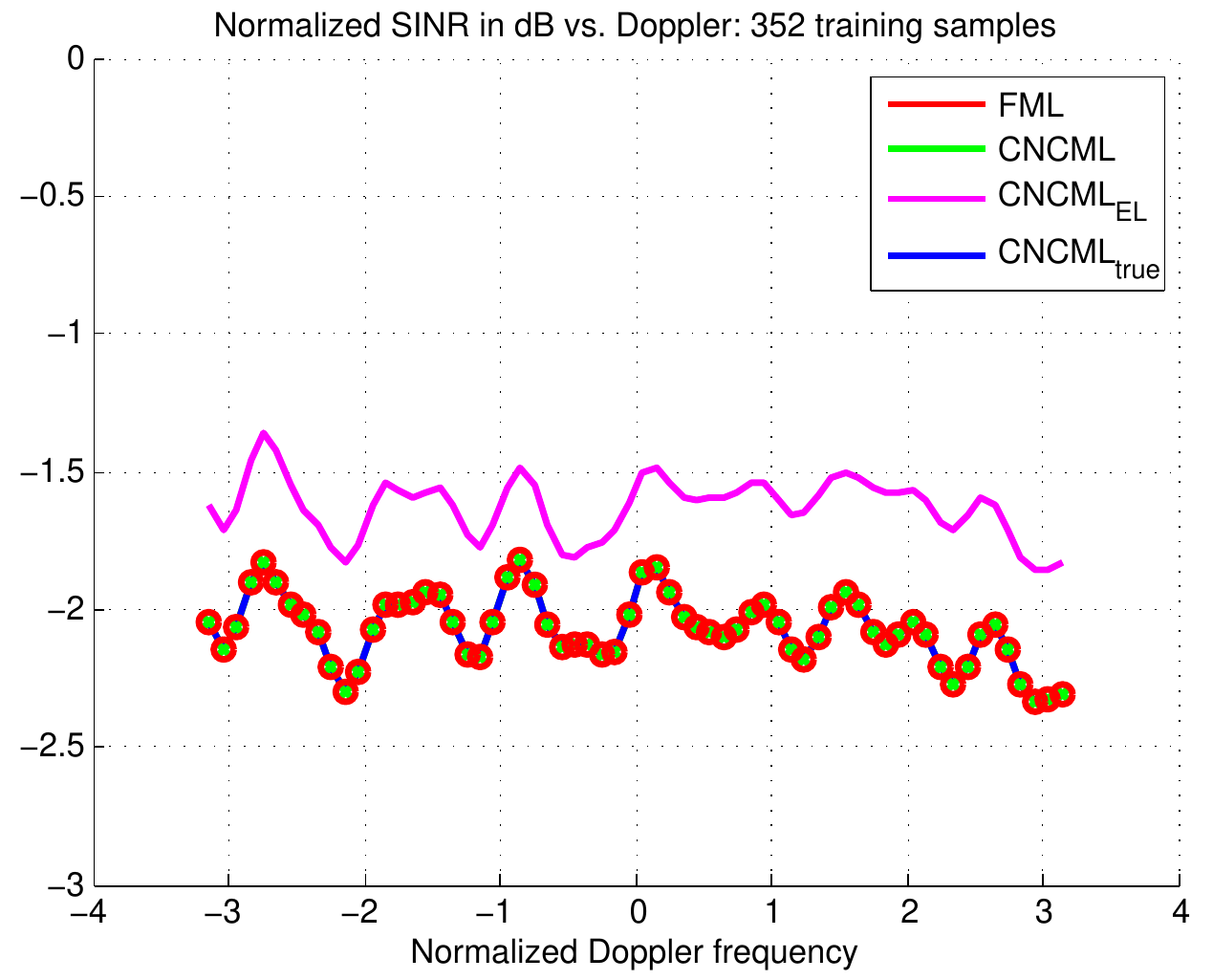}\label{Fig:PD_Simulation_20}}\\
\subfigure[$K=1.5N=528$]{\includegraphics[scale=0.5]{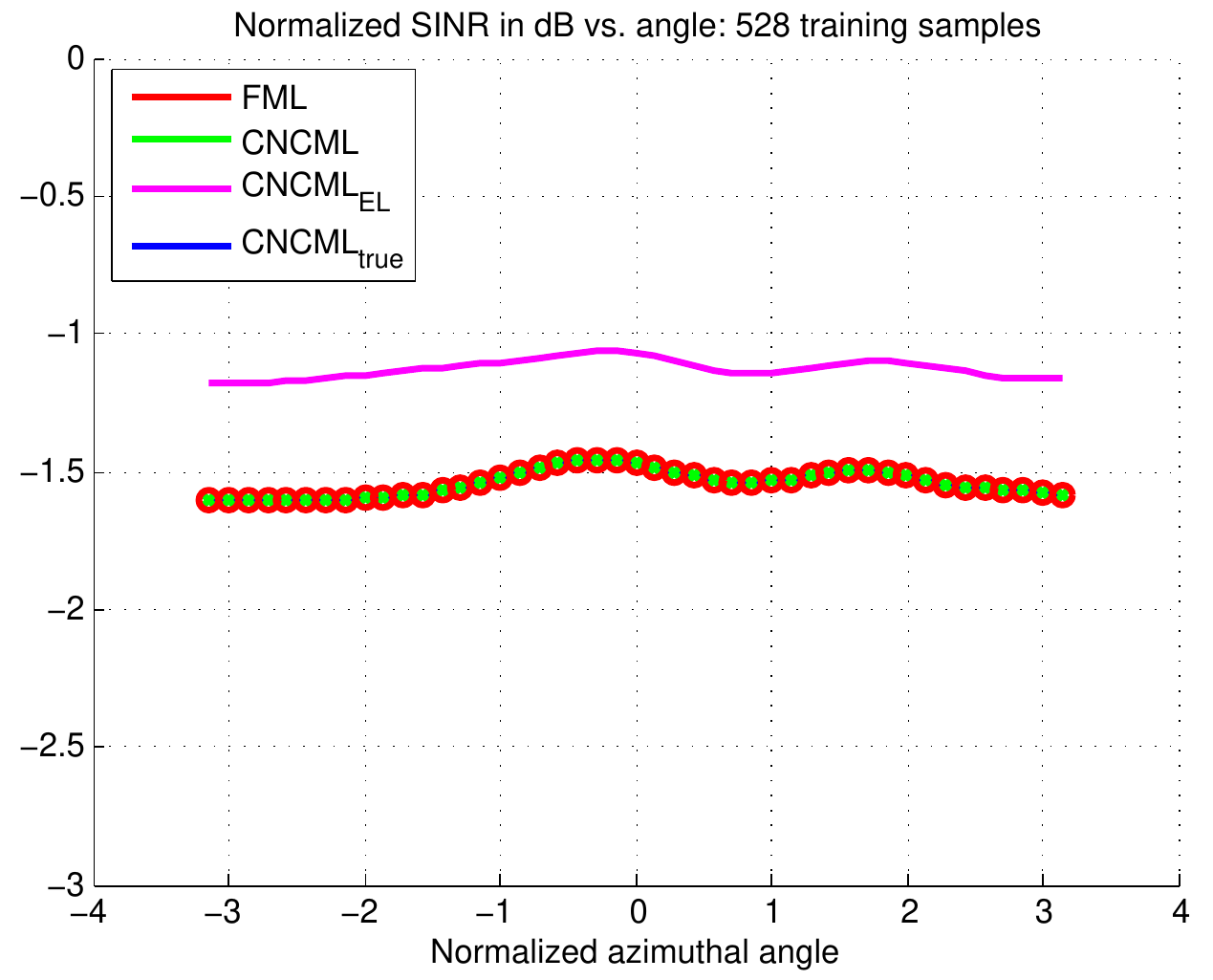}\label{Fig:SINR_angle_352}}
\hfil
\subfigure[$K=1.5N=528$]{\includegraphics[scale=0.5]{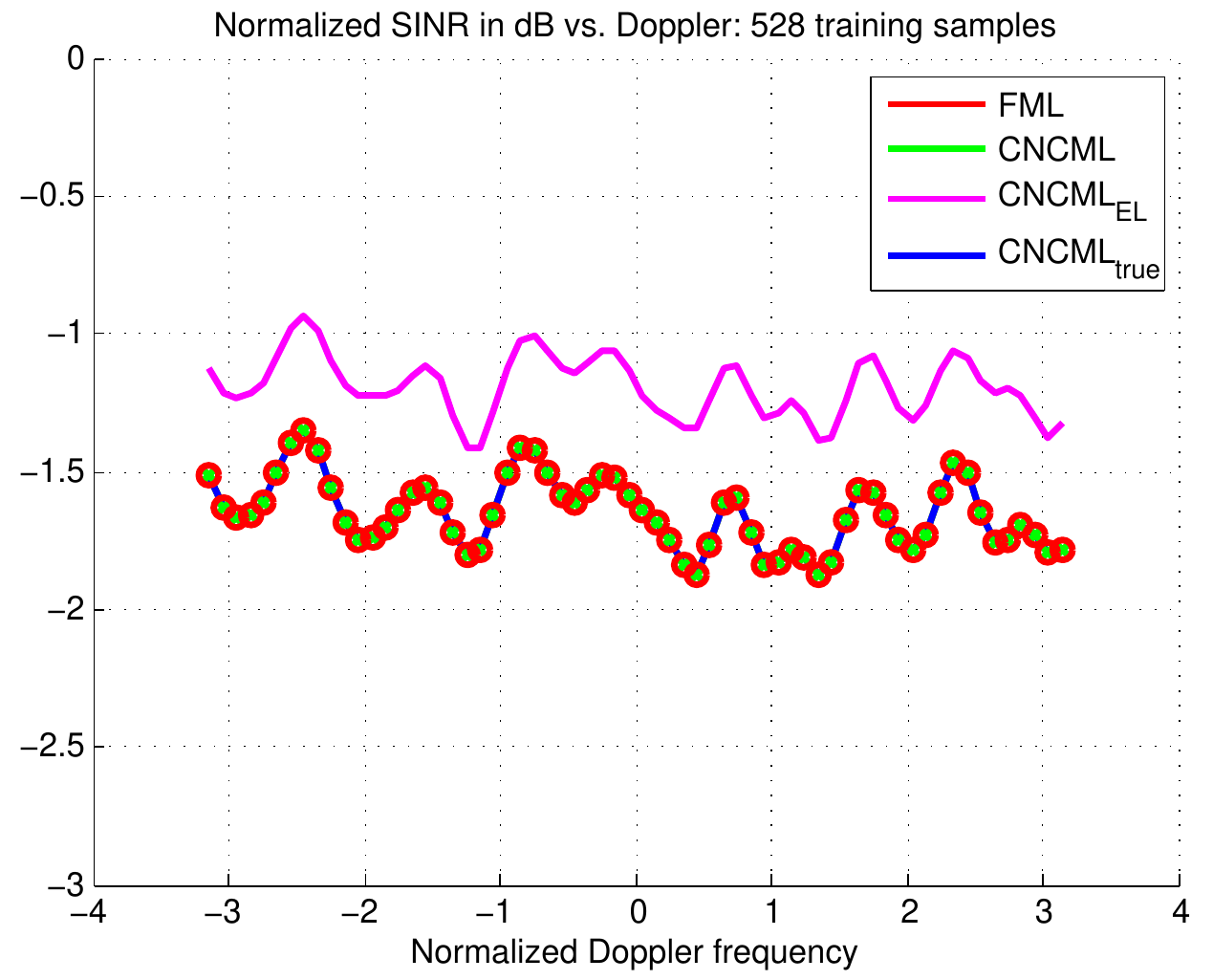}\label{Fig:PD_Simulation_20}}\\
\subfigure[$K=2N=704$]{\includegraphics[scale=0.5]{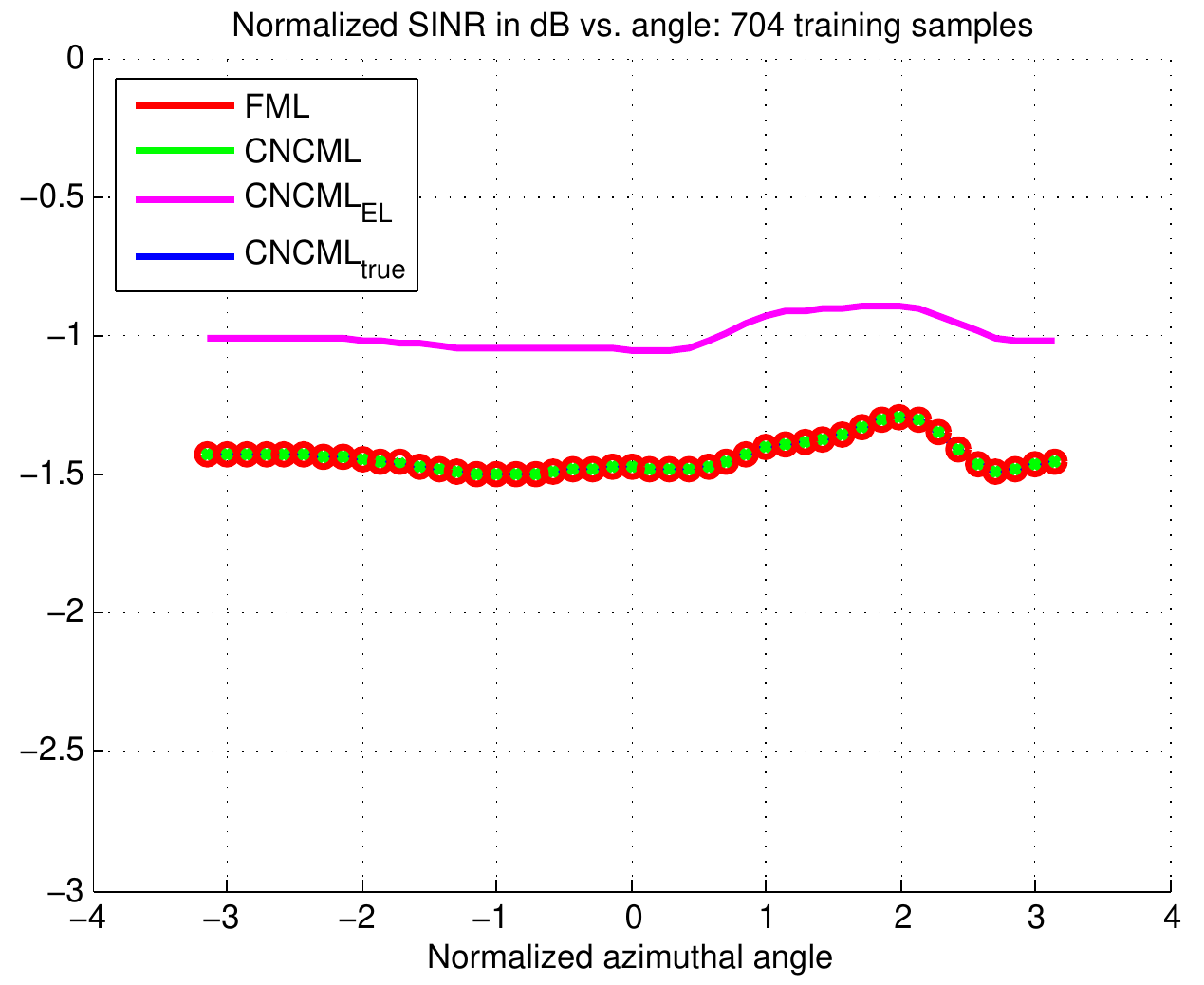}\label{Fig:SINR_angle_352}}
\hfil
\subfigure[$K=2N=704$]{\includegraphics[scale=0.5]{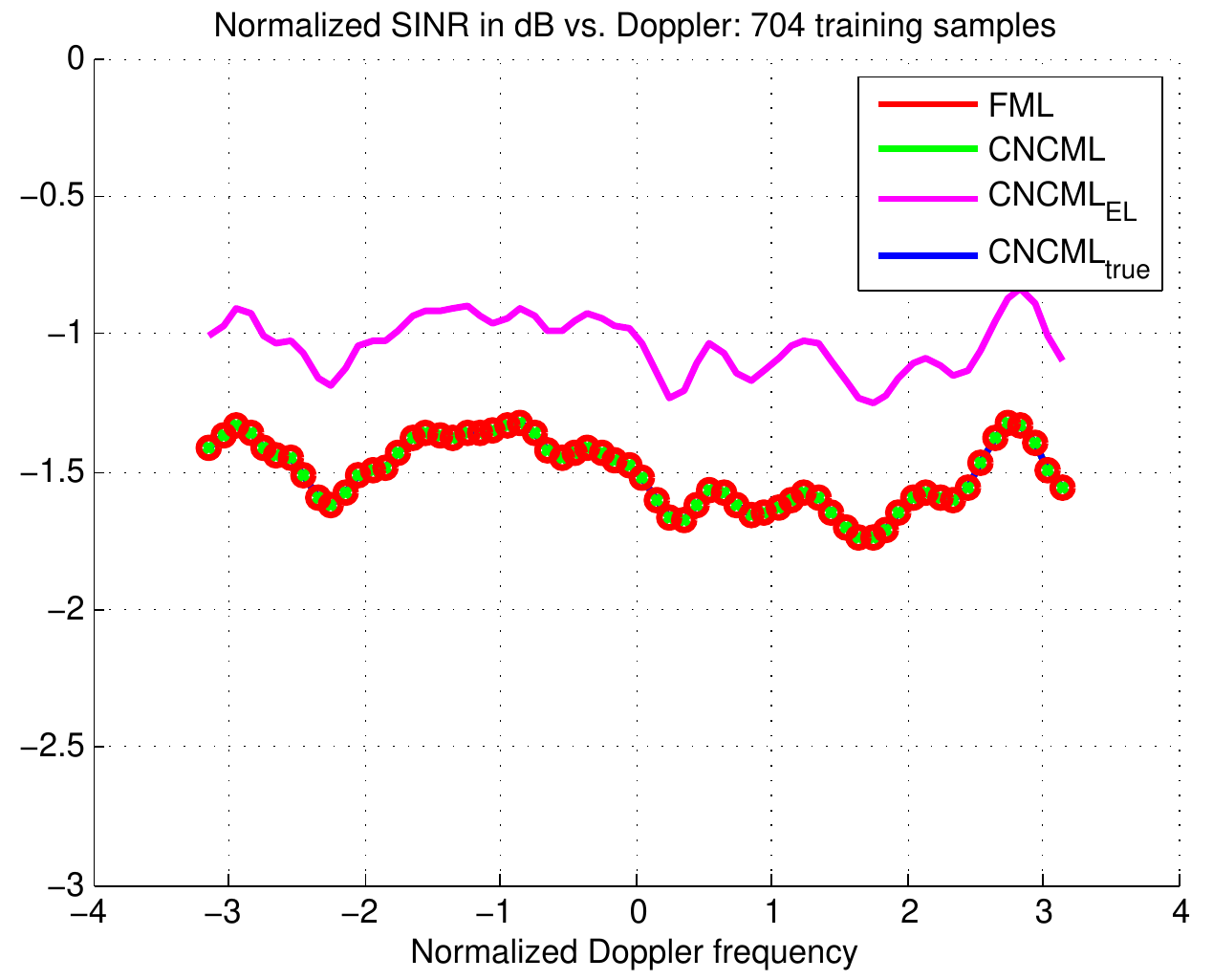}\label{Fig:PD_Simulation_20}}
\end{center}
\caption{Normalized SINR versus azimuthal angle and Doppler frequency for the KASSPER data set. (a) and (b) for $K=N=352$, (c) and (d) for $K=1.5N=528$, and (e) and (f) for $K=2N=704$}
\label{Fig:KASSPER_CN}
\end{figure}

Now we show experimental results for the condition number estimation method proposed in Section \ref{Sec:ConditionNumber}. We compare the proposed method, denoted by $\text{CNCML}_{\text{EL}}$, with four different covariance estimation methods, the sample covariance matrix (SMI), FML, \CNCML, and \CNCtrue.

Tables \ref{Tb:CN_simulation1} and \ref{Tb:CN_simulation2} show the normalized SINR values for the simulation model. We analyze five different scenarios with different parameters of the simulated covariance model given by Eq. \eqref{Eq:SimulationModel}. We use the same parameters as those used in \cite{Aubry12} to evaluate the performances and they are shown in Table 5.2(c).

For the narrowband scenarios ($B_f = 0$) in Table 5.1(a) and Table 5.1(c), $\text{CNCML}_{\text{EL}}$ outperforms the alternatives for the limited training regime and FML is the best in other training regimes. Note that the gap between \CNCEL and FML (at most 0.002) is much smaller than that of the limited training regime (at least 0.3).

On the other hand, for the wideband scenarios in Tables 5.1(b), 5.2(a), and 5.2(b), $\text{CNCML}_{\text{EL}}$ shows the best performance in most cases including \CNCtrue using the true condition number.

The experimental results for the KASSPER data set are shown in Figure \ref{Fig:KASSPER_CN}. We do not plot the sample covariance matrix to clarify the difference among the estimators. In every case, FML, $\text{CNCML}$, and \CNCtrue are very close to one another and $\text{CNCML}_{\text{EL}}$ is the best estimator. Note that $\text{CNCML}_{\text{EL}}$ is based on the same algorithm as \CNCML and \CNCtrue and differs from them in the point that \CNCEL uses a different condition number which is estimated by the expected likelihood approach. Again, this shows that the expected likelihood criterion is really useful and powerful to estimate parameters which is imperfectly known and leads to an adaptive and robust covariance estimator.

\section{Conclusion}
\label{Sec:Conclusion}

We propose robust covariance estimation algorithms which automatically determine the optimal values of practical constraints via the expected likelihood criterion for radar STAP in this chapter. Three different cases of practical constraints which is exploited in recent works including the rank constrained ML estimation and the condition number constrained ML estimation are investigated. Significant analytic results and proofs are derived for each case. Uniqueness of the optimal values of the rank constraint and the condition number constraint is formally proved and a closed form solution of the noise level is derived. Experimental results show that the estimators with the constraints obtained by the expected likelihood outperform previous works which uses constraints obtained by other parameter estimation methods including the maximum likelihood solution of the constraints.

\chapter{Conclusion}
\label{Ch:Conclusion}

\section{Summary of Contributions}
\label{Sec:Contribution}

This dissertation developed robust covariance estimation methods for radar STAP. Here, I briefly summarize the contributions of this dissertation.

\textbf{Covariance matrix estimation under practical constraints:} In Chapter \ref{Ch:RCML} and \ref{Ch:EASTR}, we focused on maximum likelihood estimation methods of structured covariance estimation exploiting practical constraints for radar STAP. It is shown that introducing practical constraints into covariance estimation problems leads to powerful estimators in the radar literatures. In particular, the rank of the clutter subspace is really powerful constraint because the rank can be either calculated by the well-known Brennan rule under ideal condition or estimated by robust rank estimation methods. We develop the RCML estimator which exploits the structural constraint and the rank of clutter subspace. Though the rank is a challenging constraint, we reduce the estimation problem to a convex optimization problem and derive a closed form solution. We also derive an estimator for when the noise level is unknown but only a lower bound is available.

Chapter \ref{Ch:EASTR} focuses on exploiting a Toeplitz structure as well as the rank constraint on the structured interference for radar STAP. The problem is inherently hard because there is no closed form solution for ML estimation of Toeplitz covariance matrices for all training regimes. We develop a new estimator that is based on a cascade of two closed forms. The first closed form is the RCML estimator and the second step of the Toeplitz approximation performs constrained optimization of eigenvalues to either exactly or approximately satisfy the Toeplitz constraint with preserving the rank.

\textbf{Robust covariance estimation under imperfect constraints:} In Chapter \ref{Ch:EL}, we develop robust covariance estimation algorithms under imperfect constraints. The proposed algorithms automatically determine the optimal values of practical constraints which are imperfectly or approximately known via the expected likelihood approach. We consider three different cases of practical constraints, 1) the rank constraint, 2) the rank constraint and the noise level, and 3) the condition number constraint, which are exploited in recent works including the RCML estimator and the condition number ML estimator. We develop significant analytical results and derive formal proof for each case of the constraints.

\section{Future Research}
\label{Sec:Future}

\subsection{Newer Constraint: Sparsity or Block Sparsity}
The principal merit of recent radar STAP estimator including the methods proposed in this dissertation is that they employ constrained optimization theory to yield tractable solutions. We can take this theme to fruition by investigating newer increasingly relevant constraints/scenarios that represent covariance estimation in modern radar STAP. The recently proposed \cite{Friedman08} graphical lasso technique is an $L_1$ regularized maximum likelihood estimate, such that
\be
\mb R = \arg\max_{\mb R \in \psi} \Big\{\log\big(p(\mb Z | \mb R)\big) - \rho  \| \mb R^{-1} \|_1 \Big\}
\ee
Here, $\psi$ denotes the set of covariance matrices of interest, i.e., for our problem the set of positive definite matrices which can be expressed as a sum of an identity and a rank-deficient positive semi-definite (the clutter) matrix. Note the cost function consists of a standard log likelihood term plus a regularization term. The regularization term $\rho  \| \mb R^{-1} \|$ encourages the inverse of the covariance matrix to be sparse. That is, if the $ij$th component of $\mb R^{-1}$ is zero, then clutter variables $i$ and $j$ are conditionally independent given the other variables. The sparsity constraint may often capture real-world clutter scenarios. In fact, physical effects such as locally pronounced clutter correlation over the spatio-temporal range may in fact give rise to ``block sparsity."

\subsection{Non-Homogeneity Detection}
The presence of outliers in the training data for radar STAP causes target cancelation resulting in degraded output signal to interference ratio and perforce degraded detection performance. A common signal processing method in this context is to excise outliers from the training data and use the resulting outlier free training data for covariance matrix estimation. Several algorithms for outlier removal have been proposed recently \cite{Blunt04,Michels04,Rangaswamy04May,Rangaswamy04Sep,Rangaswamy05} in a variety of dense target environments. We take the perspective of observing interference or clutter specific training data as normal and outlier contamination as an anomaly. In machine learning, anomaly detection is an extensively studied area \cite{Chandola09}. Anomaly detection refers to the problem of finding patterns in data that do not conform to expected behavior. By treating outlier (target) contamination as an anomaly we can benefit from machine learning techniques devoted to identifying which data samples are anomalous and subsequently use the remaining normal data for estimation purposes. We can particularly focus on methods that can identify anomalies even as a small number of observations is presented because training samples are not abundant in practice. We can investigate and apply two categories of anomaly detection techniques for the radar non-homogeneity detection (NHD) problem, 1.) supervised learning based anomaly detection \cite{Chawla04,Phua04,Ratsch02} and 2.) unsupervised learning or clustering based anomaly detection \cite{Tan05,Levy08}.

\subsection{Expansion of the EL Approach}
As shown in this dissertation, the expected likelihood approach is powerful to estimate regularization parameters employed as constraints in covariance estimation problems. Unsurprisingly, other constraints which have been commonly exploited in previous covariance estimation methods in radar literature can be also selected or estimated using the expected likelihood approach. For example, the condition number constraint and the noise level can be jointly estimated and the intensity parameter in the shrinkage estimator also can be estimated by the expected likelihood approach. Furthermore, we investigate experimental results for only $K \geq N$, i.e., the number of training samples are greater than or equal to the data dimension since the complex Wishart distribution is defined only for $K \geq N$. Analyzing theoretical properties of the likelihood ratio values for true covariance and focusing on how to apply the expected likelihood approach to the case of $K < N$ can be another future work. Finally, using the expected likelihood approach not just for tuning parameters but for direct estimation of disturbance covariance can be investigated in the future. For example, the EL approach can be applicable to estimate a whole set of eigenvalues and eigenvectors of the covariance matrix.

} % End of the \allowdisplaybreak command %
%%%%%%%%%%%%%%%%%%%%%%%%%%%%%%%%%%%%%%%%%%%

%%%%%%%%%%%%%%%%
% BIBLIOGRAPHY %
%%%%%%%%%%%%%%%%
% You can use BibTeX or other bibliography facility for your
% bibliography. LaTeX's standard stuff is shown below. If you
% bibtex, then this section should look something like:
 % \pagestyle{empty}
  \begin{singlespace}

   \bibliographystyle{IEEEtran}%,FG-bibstyle
   \addcontentsline{toc}{chapter}{Bibliography}
   \bibliographystyle{C:/Users/BXK265/BOXSYN\string~2/BOXSYN\string~1/LaTex/Bibliography/IEEEbib}
   \bibliography{C:/Users/BXK265/BOXSYN\string~2/BOXSYN\string~1/LaTex/Bibliography/IEEEabrv,C:/Users/BXK265/BOXSYN\string~2/BOXSYN\string~1/LaTex/Bibliography/Bosung}
   %\bibliography{IEEEabrv,Bosung}
   \end{singlespace}

\backmatter

% Vita
% Place your vita below.

Mr. Kang received the B. S. and M. S. degrees from Yonsei University, Seoul, Korea, in 2005 and 2007, respectively. He worked at LG Electronics as a research engineer, Seoul, Korea, from 2007 to 2011. He developed image and video signal processing algorithms for mobile cameras and monitors in LG Electronics. He was an intern in Air Force Research Lab, WPAFB, OH in 2013. His research interest include statistical signal processing, detection and estimation, and convex optimization. Mr. Kang received the Best Student Paper Award at IEEE International Radar Conference in 2014
\vspace{1cm}

\textbf{Journal Publications}:

\begin{enumerate}
  \item Bosung Kang, Vishal Monga, and Muralidhar Rangaswamy, ``Rank Constrained Maximum Likelihood Estimation of Structured Covariance Matrices," IEEE Transactions on Aerospace and Electronic Systems, vol. 50, no. 1,  pp.501-515, Jan. 2014.
  \item Bosung Kang, Vishal Monga, and Muralidhar Rangaswamy, ``Computationally Efficient Toeplitz Approximation of Structured Covariance Under a Rank Constraint," IEEE Transactions on Aerospace and Electronic Systems, vol. 51, no. 1, pp. 775-785, Jan. 2015.
  \item Bosung Kang, Vishal Monga, Muralidhar Rangaswamy, and Yuri Abramovich, ``Robust Covariance Estimation under Imperfect Constraints using Expected Likelihood Approach," in preparation.
\end{enumerate}

\end{document}